\documentclass[12pt]{article}

\usepackage[english]{babel}
\usepackage[utf8]{inputenc}
\usepackage{amsmath}
\usepackage{graphicx}
\usepackage[colorinlistoftodos]{todonotes}
\usepackage{amssymb}
\usepackage{mathtools}
\usepackage{amsthm}
\usepackage{natbib}
\usepackage{tikz}
\usepackage{bbm}
\usepackage{adjustbox}
\usepackage{etoolbox}
\usepackage{enumerate}

\usepackage{pdflscape}
\usepackage{adjustbox}
\usepackage{scrextend}
\usepackage{titling}
\usepackage{setspace}
\usepackage{subcaption}
\usepackage{caption}
\usepackage{float}
\usepackage{bm}
\usepackage{mathrsfs}
\usepackage[margin=1.25in]{geometry}
\usepackage[normalem]{ulem}
\usepackage{xr-hyper}
\usepackage{hyperref}
\usepackage{mathtools}
\usepackage{arydshln}
\usepackage{xcolor}

\DeclarePairedDelimiter\floor{\lfloor}{\rfloor}
\hypersetup{pdfborder = 0 0 0,%
    pdftitle = Bootstrapping Risk Measures,%
    pdfauthor = Beutner Heinemann Smeekes,%
}
\makeatletter

\newcommand*{\addFileDependency}[1]{ 
	\typeout{(#1)}
	\@addtofilelist{#1}
	\IfFileExists{#1}{}{\typeout{No file #1.}}
}
\makeatother

\newcommand*{\myexternaldocument}[1]{%
	\externaldocument{#1}%
	\addFileDependency{#1.tex}%
	\addFileDependency{#1.aux}%
}
\myexternaldocument{Supplement-r4}

\usepackage{scalerel,stackengine}
\stackMath
\newcommand\reallywidehat[1]{%
\savestack{\tmpbox}{\stretchto{%
  \scaleto{%
    \scalerel*[\widthof{\ensuremath{#1}}]{\kern-.6pt\bigwedge\kern-.6pt}%
    {\rule[-\textheight/2]{1ex}{\textheight}}
  }{\textheight}%
}{0.5ex}}%
\stackon[1pt]{#1}{\tmpbox}%
}
\numberwithin{equation}{section}

\newcommand\blfootnote[1]{%
  \begingroup
  \renewcommand\thefootnote{}\footnote{#1}%
  \addtocounter{footnote}{-1}%
  \endgroup
}

\newcommand{\EE}{\mathbb{E}}
\newcommand{\PP}{\mathbb{P}}
\newcommand{\Cov}{\mathbb{C}\mbox{ov}}
\newcommand{\Var}{\mathbb{V}\mbox{ar}}

\newcommand{\Z}{\mathbb{Z}}
\newcommand{\R}{\mathbb{R}}
\newcommand{\N}{\mathbb{N}}

\newcommand{\plim}{\text{p}\!\!\lim}

\theoremstyle{definition}
\newtheorem{assumption}{Assumption}
\newtheorem{assumption*}{Assumption}

\newtheorem{algorithm}{Algorithm}
\newtheorem{example}{Example}

\theoremstyle{plain}
\newtheorem{theorem}{Theorem}
\newtheorem{lemma}{Lemma}
\newtheorem{corollary}{Corollary}
\newtheorem{proposition}{Proposition}

\theoremstyle{remark}
\newtheorem{remark}{Remark}

\title{Journal of Econometrics}

\author{Alexander Heinemann, Eric Beutner and Stephan Smeekes}

\date{\today}

\begin{document}

\begin{center}

\scshape
\LARGE{\textbf{A Residual Bootstrap for Conditional Value-at-Risk}}
\normalfont
\end{center}
\begin{center}
\par\vspace{0.5cm}
\text{Eric Beutner$^*$}
\hspace{1.5cm}
\text{Alexander Heinemann$^{**}$}
\hspace{1.4cm}
\text{Stephan Smeekes$^{***}$}
\par\vspace{2cm} 
 
\par
 
\par
\text{\today}
\vspace{1cm}
\end{center}

\begin{abstract}
A fixed-design residual bootstrap method is proposed for the two-step estimator of \cite{francq2015risk} associated with the conditional Value-at-Risk. The bootstrap's consistency is proven for a general class of volatility models and intervals are constructed for the conditional Value-at-Risk. A simulation study reveals that the equal-tailed percentile bootstrap interval tends to fall short of its nominal value. In contrast, the reversed-tails bootstrap interval yields accurate coverage. We also compare the theoretically analyzed fixed-design bootstrap with the recursive-design bootstrap. It turns out that the fixed-design bootstrap performs equally well in terms of average coverage, yet leads on average to shorter intervals in smaller samples. An empirical application illustrates the interval estimation.\\ \\

\end{abstract}

\blfootnote{\hspace{-0.7cm}
$^{\textcolor{white}{**}*}$Department of Econometrics and Data Science, Vrije Universiteit Amsterdam, De Boelelaan 1105 
1081 HV Amsterdam,  Netherlands. E-mail address: \href{mailto:e.a.beutner@vu.nl}{e.a.beutner@vu.nl}\\
$^{\textcolor{white}{*}**}$a.s.r. Archimedeslaan 10, 3584 BA Utrecht. E-mail address: \href{mailto:alexander.heinemann@asr.nl}{alexander.heinemann@asr.nl}\\
$^{***}$Department of Quantitative Economics, Maastricht University, Tongersestraat 53, 6211 LM Maastricht, Netherlands. E-mail address: \href{mailto:s.smeekes@maastrichtuniversity.nl}{s.smeekes@maastrichtuniversity.nl}  (corresponding author)
}

\newpage

\doublespacing

\section{Introduction}
\label{sec:4.1}
Risk management has tremendously developed in past decades becoming an increasing practice. With minimum capital requirements being enforced by current legislation (Basel III and Solvency II), financial institutions and insurance companies monitor risk by using conventional measures such as Value-at-Risk (VaR). Typically, the volatility dynamics are specified by a (semi-)parametric model leading to conditional risk measure versions. For GARCH-type models the conditional VaR reduces to the conditional volatility scaled by a quantile of the innovations' distribution. The latter is conventionally treated as additional parameter and forms together with the others  the \textit{risk parameter} \citep{francq2015risk}. The true parameters are generally unknown and need to be estimated to obtain an estimate for the conditional VaR. 
Clearly, this VaR evaluation is subject to estimation risk that needs to be quantified for appropriate risk management.

Whereas an estimator based on a single step is available after re-parameterization \citep{francq2015risk}, a widely used approach is the following two-step estimation procedure. First, the parameters of the stochastic volatility model are estimated. Arguably the most popular estimation method in a GARCH-type setting is the Gaussian quasi-maximum-likelihood (QML) method. Based on the model's residuals the quantile is estimated by its empirical counterpart in a second step. For realistic sample sizes (e.g.\ $500$ or $1{,}000$ daily observations) the estimators are subject to considerable estimation risk. In particular, the estimation uncertainty associated with the quantile estimator is substantial for extreme quantiles (e.g.\  $\leq 5\%$). 

To quantify the uncertainty around the point estimators, one traditionally relies on asymptotic theory while replacing the unknown quantities in the limiting distribution by consistent estimates. An alternative approach -- frequently employed in practice -- is based on a bootstrap approximation. Regarding the estimators of the GARCH parameters, various bootstrap methods have been studied to approximate the estimators' finite sample distribution including the subsample bootstrap \citep{hall2003inference}, the block bootstrap \citep{corradi2008bootstrap}, the wild bootstrap \citep{shimizu2009bootstrapping} and the residual bootstrap. The residual bootstrap method is particularly popular and can be further divided into recursive \citep{pascual2006bootstrap,hidalgo2007goodness,jeong2017residual} and fixed \citep{shimizu2009bootstrapping,cavaliere2018fixed} design. Whereas in the former the bootstrap observations are generated recursively using the estimated volatility dynamics, the latter design keeps the dynamics of the bootstrap samples fixed  at the value of the original series. Further recent applications of fixed or recursive bootstrap designs or variants thereof to conditional volatility models can be found in \cite{HETLAND2021}, \cite{FRANCQ202247} and \cite{cavaliere2020bootstrap}.

The estimation of the quantile and the conditional VaR have received only selected attention in the bootstrap literature and proposed bootstrap methods have been, to the best of our knowledge, exclusively investigated by means of simulation. \cite{christoffersen2005estimation} examine various quantile estimators and construct intervals for the conditional VaR using a recursive-design residual bootstrap method. In addition, \cite{hartz2006accurate} presume the innovation distribution to be standard normal such that the quantile parameter is known; they propose a resampling method based on a residual bootstrap and a bias-correction step to account for deviations from the normality assumption. In contrast, \cite{spierdijk2016confidence} develops an $m$-out-of-$n$ without-replacement bootstrap to construct confidence intervals for ARMA-GARCH VaR. 

This paper proposes a fixed-design residual bootstrap method to mimic the finite sample distribution of the two-step estimator and provides an algorithm for the construction of bootstrap intervals for the conditional VaR. The proposed bootstrap method is proven to be consistent for a general class of volatility models. In particular, our framework does not only encompass GARCH but also several GARCH extensions such as the threshold GARCH (T-GARCH) of \cite{zakoian1994threshold} and the GJR-GARCH named after Glosten, Jagannathan and Runkle (\citeyear{glosten1993relation}). The bootstrap consistency is established under a set of mild assumptions, which relaxes moment conditions on the innovations imposed in the GARCH bootstrap literature. To the best of our knowledge this paper is the first to theoretically validate the residual bootstrap for the quantile and the conditional VaR.

The remainder of the paper is organized as follows. Section \ref{sec:4.2} specifies the model and the conditional VaR is derived. The two-step estimation procedure is described in Section \ref{sec:4.3} and asymptotic theory is provided under mild assumptions. In Section \ref{sec:4.4}, a fixed-design residual bootstrap method is proposed and proven to be consistent. Further, bootstrap intervals are constructed for the conditional VaR and extensions to the bootstrap methods presented here are discussed. A simulation study is conducted in Section \ref{sec:4.5} and an empirical application illustrates the interval estimation based on the fixed-design residual bootstrap. Section \ref{sec:4.6} concludes. Appendix \ref{app:1} contains proof of the main results, whereas Supplementary Appendix \ref{app:4.A} contains auxiliary results and their proofs. Finally, Supplementary Appendix \ref{app:4.B} is devoted to the related recursive-design residual bootstrap and Supplementary Appendix \ref{sec add simulation} contains additional simulation results.

\section{Model}
\label{sec:4.2}
We consider a conditional volatility model of the form
\begin{align}
\label{eq:4.2.1}
\epsilon_t = \sigma_t\eta_t
\end{align}
with $t\in \Z$, where $\{\epsilon_t\}$ denotes the sequence of log-returns, $\{\sigma_t\}$ is a volatility process and $\{\eta_t\}$ is a sequence of independent and identically distributed (iid) variables satisfying $\EE\big[\eta_t^2\big]=1$. The  volatility is presumed to be a measurable function of past observations
\begin{align}
\label{eq:4.2.2}
\sigma_{t}=\sigma_{t}(\theta_0)=\sigma(\epsilon_{t-1},\epsilon_{t-2},\dots;\theta_0)
\end{align}
with $\sigma:\R^\infty\times \Theta\to(0,\infty)$ and $\theta_0$ denotes the true parameter vector belonging to the parameter space $\Theta \subset \R^r$, $r \in \N$. Subsequently, we consider two examples for the functional form of \eqref{eq:4.2.2}: the well-known GARCH model (\citeauthor{engle1982autoregressive}, \citeyear{engle1982autoregressive}; \citeauthor{bollerslev1986generalized}, \citeyear{bollerslev1986generalized}) and the T-GARCH model of \cite{zakoian1994threshold}. Whereas the first is frequently applied in practice, the second is motivated by our empirical application (see Section \ref{sec:4.5.2}).
\begin{example}
\label{ex:4.1}
Suppose $\{\epsilon_t\}$ follows a GARCH$(1,1)$ process given by \eqref{eq:4.2.1} and $\sigma_{t}^2 = \omega_0 + \alpha_0 \epsilon_{t-1}^2+ \beta_0 \sigma_{t-1}^2$, where $\theta_0 = (\omega_0,\alpha_0,\beta_0)'\in (0,\infty)\times[0,\infty)\times[0,1)$. The recursive structure implies $\sigma_{t}=\sigma(\epsilon_{t-1},\epsilon_{t-2},\dots;\theta_0) = \sqrt{\sum_{k=1}^\infty \beta_0^{k-1}  \big(\omega_0 + \alpha_0 \epsilon_{t-k}^2\big)}$.
\end{example}
\begin{example}
\label{ex:4.2}
Suppose $\{\epsilon_t\}$ follows a T-GARCH$(1,1)$ process given by \eqref{eq:4.2.1} and $\sigma_{t} =\: \omega_0 + \alpha_0^+ \epsilon_{t-1}^+ +\alpha_0^- \epsilon_{t-1}^- + \beta_0 \sigma_{t-1}$ with parameters $\theta_0 =(\omega_0,\alpha_0^+,\alpha_0^- ,\beta_0)'\in (0,\infty)\times[0,\infty)\times[0,\infty)\times[0,1)$ and $ \epsilon_{t}^+=\max\{\epsilon_{t},0\}$ and $ \epsilon_{t}^-=\max\{-\epsilon_{t},0\}$. The model's recursive structure yields $\sigma_{t}=\sigma(\epsilon_{t-1},\epsilon_{t-2},\dots;\theta_0) =\! \sum_{k=1}^\infty \beta_0^{k-1}  \big(\omega_0 + \alpha_0^+ \epsilon_{t-k}^+ +\alpha_0^-\epsilon_{t-k}^-\big)$.
\end{example}
Throughout the paper, for any cumulative distribution function (cdf), say $G$, we define the generalized inverse by $G^{-1}(u)=\inf\big\{\tau \in \R :  G(\tau)\geq u\big\}$ and write $G(\cdot-)$ to denote its left limit. Generally, for an arbitrary real-valued random variable $X$ 
(e.g.~stock return) with cdf $F_X$, the VaR at level $\alpha \in (0,1)$, is given by $VaR_\alpha(X)=-F_X^{-1}(\alpha)$.

Let $\mathcal{F}_n$ denote the $\sigma$-algebra generated by $\{\epsilon_t, t\leq n\}$. It follows that the conditional VaR of $\epsilon_{n+1}$ given $\mathcal{F}_n$ at level $\alpha \in (0,1)$ is $VaR_{\alpha}(\epsilon_{n+1}|\mathcal{F}_n) = \sigma (\epsilon_n,\epsilon_{n-1},\dots;\theta_0) VaR_\alpha(\eta_{n+1})$.
For given $\alpha$, the quantile of $\eta_{n+1}$ is constant and can be treated as a parameter. Thus, denoting the cdf of the $\eta_t$'s by $F$ and setting $\xi_\alpha = F^{-1}(\alpha)$, the conditional VaR of $\epsilon_{n+1}$ given $\mathcal{F}_n$ at level $\alpha$ reduces to 
\begin{align}
\label{eq:4.2.4}
VaR_{\alpha}(\epsilon_{n+1}|\mathcal{F}_n) = -\xi_\alpha\: \sigma_{n+1}(\theta_0).
\end{align}
Typically, $\alpha$ is fixed at a sufficiently small level such that $\xi_\alpha<0$. Except for special cases (e.g. normality of $\eta_t$), $\xi_\alpha$ is unknown and needs to estimated just like $\theta_0$.

\section{Estimation}
\label{sec:4.3}

We estimate the parameters $\theta_0$ and $\xi_\alpha$ following the two-step procedure of \citeauthor{francq2015risk} (\citeyear{francq2015risk}, Section 4.2).
In the first step, we estimate the conditional volatility parameter $\theta_0$ by Gaussian QML. This approach is motivated as follows: if the innovations $\{\eta_t\}$ were Gaussian, the variables $\eta_t(\theta) = \epsilon_t/\sigma_t(\theta)$ would be iid $N(0,1)$ whenever $\theta=\theta_0$, where $\sigma_{t}(\theta) = \sigma(\epsilon_{t-1},\dots,\epsilon_{1}, \epsilon_{0},\epsilon_{-1},\dots;\theta)$.
The 'Q' in QML stands for 'quasi' and refers to the fact that $F$ does not need to be the standard normal distribution function. Obviously, given a sample $\epsilon_1, \dots, \epsilon_n$, we generally cannot determine $\sigma_{t}(\theta)$ completely. Replacing the unknown presample observations by arbitrary values, say $\tilde{\epsilon}_t$, $t\leq 0$, we obtain $\tilde{\sigma}_{t}(\theta) = \sigma(\epsilon_{t-1},\dots,\epsilon_{1}, \tilde{\epsilon}_{0},\tilde{\epsilon}_{-1},\dots;\theta)$, which serves as an approximation for $\sigma_{t}(\theta)$. The QML estimator of $\theta_0$ is defined by 
\begin{align}
\label{eq:4.3.3}
\hat{\theta}_n=\arg\max_{\theta \in \Theta} \frac{1}{n}\sum_{t=1}^n \tilde{\ell}_t(\theta) \qquad \text{with} \qquad \tilde{\ell}_t(\theta)=-\frac{1}{2}\bigg(\frac{\epsilon_t}{\tilde{\sigma}_t(\theta)}\bigg)^2-\log \tilde{\sigma}_t(\theta).
\end{align}
In the second step, we estimate $\xi_\alpha$ on the basis of the first-step residuals, i.e.~$\hat{\eta}_t=\epsilon_t/\tilde{\sigma}_t(\hat{\theta}_n)$. The empirical $\alpha$-quantile of $\hat{\eta}_1,\dots,\hat{\eta}_n$ is given by
\begin{align}
\label{eq:4.3.4}
\hat{\xi}_{n,\alpha}= \arg \min_{z \in \R} \frac{1}{n} \sum_{t=1}^n \rho_\alpha(\hat{\eta}_t-z),
\end{align}
where $\rho_\alpha(u)=u(\alpha-\mathbbm{1}_{\{u< 0\}})$ is the usual asymmetric absolute loss function (cf.\ \citeauthor{koenker2006quantile}, \citeyear{koenker2006quantile}). Equivalently, we can write $\hat{\xi}_{n,\alpha}=\hat{\mathbbm{F}}_n^{-1}(\alpha)$ with $\hat{\mathbbm{F}}_n(x)=\frac{1}{n}\sum_{t=1}^n \mathbbm{1}_{\{\hat{\eta}_t\leq x\}}$ being the empirical distribution function (edf) of the residuals. 

Having obtained estimators for $\theta_0$ and $\xi_\alpha$, we turn to the estimation of the conditional VaR of the one-period ahead observation at level $\alpha$. Whereas the notation $VaR_{\alpha}(\epsilon_{n+1}|\mathcal{F}_n)$ stresses the object's conditional nature, we henceforth proceed with the abbreviation $VaR_{n,\alpha}$ for notational convenience. Employing \eqref{eq:4.3.3} -- \eqref{eq:4.3.4} we can estimate $VaR_{n,\alpha}$ by
\begin{align}
\label{eq:4.3.5}
\reallywidehat{VaR}_{n,\alpha}=-\hat{\xi}_{n,\alpha}\: \tilde{\sigma}_{n+1}\big(\hat{\theta}_n\big).
\end{align}
Clearly,  the estimator's large sample properties cannot be studied using traditional tools such as consistency since \eqref{eq:4.3.5} does not permit a limit.

For the subsequent asymptotic analysis, we introduce the following assumptions.
\begin{assumption}{(Compactness)}
\label{as:4.1}
$\Theta$ is a compact subset of $\R^r$.
\end{assumption}

\begin{assumption}{(Stationarity \& Ergodicity)}
\label{as:4.2}
$\{\epsilon_t\}$ is a strictly stationary and ergodic solution of \eqref{eq:4.2.1} with \eqref{eq:4.2.2}.
\end{assumption}

\begin{assumption}{(Volatility process)}
\label{as:4.3}
The function $\sigma:\R^\infty\times \Theta\to(0,\infty)$ is known and for
any real sequence $\{x_i\}$, the function $\theta\to\sigma(x_1,x_2,\dots;\theta)$ is continuous.  Almost surely, $\sigma_t(\theta)>\underline{\omega}$ for any $\theta \in \Theta$ and some $\underline{\omega}>0$ and $\EE[\sigma_t^s(\theta_0)]<\infty$ for some $s>0$. Moreover, for any $\theta \in \Theta$, we assume $\sigma_t(\theta_0)/\sigma_t(\theta)=1$ almost surely (a.s.) if and only if $\theta=\theta_0$.
\end{assumption}

\begin{assumption}{(Initial conditions)}
\label{as:4.4}
There exists a constant $\rho \in (0,1)$ and a random variable $C_1$ measurable with respect to $\mathcal{F}_0$ and $\EE[|C_1|^s]<\infty$ for some $s>0$ such that
\begin{enumerate}[(i)]
\item \label{as:4.4.1} $\sup_{\theta \in \Theta}|\sigma_t(\theta)-\tilde{\sigma}_t(\theta)|\leq C_1 \rho^t$;

\item \label{as:4.4.2} $\theta\to \sigma(x_1, x_2, \dots;\theta)$ has continuous second-order derivatives satisfying
\begin{align*}
\sup_{\theta \in \Theta}\bigg|\bigg|\frac{\partial \sigma_t(\theta)}{\partial \theta}-\frac{\partial \tilde{\sigma}_t(\theta)}{\partial \theta}\bigg|\bigg|\leq  C_1 \rho^t, \qquad \quad \sup_{\theta \in \Theta}\bigg|\bigg|\frac{\partial^2 \sigma_t(\theta)}{\partial \theta \partial \theta'}-\frac{\partial^2 \tilde{\sigma}_t(\theta)}{\partial \theta\partial \theta'}\bigg|\bigg|\leq C_1 \rho^t,
\end{align*}
where $||\cdot||$ denotes the Euclidean norm.
\end{enumerate}
\end{assumption}

\begin{assumption}{(Innovation process)} 
\label{as:4.5}
The innovations $\{\eta_t\}$ satisfy

\begin{enumerate}[(i)]
\item  \label{as:4.5.1}  $\eta_t\overset{iid}{\sim}F$ with $F$ being continuous, $\EE\big[\eta_t^2\big]=1$ and $\eta_t$ is independent of $\{\epsilon_u:u<t\}$;

\item  \label{as:4.5.2} $\eta_t$ admits a density $f$ which is continuous and strictly positive around $\xi_\alpha<0$;
\item  \label{as:4.5.3}  $\EE\big[\eta_t^4\big]<\infty$.
\end{enumerate}
\end{assumption}

\begin{assumption}{(Interior)}
 \label{as:4.6}
$\theta_0$ belongs to the interior of $\Theta$ denoted by $\mathring{\Theta}$.
\end{assumption}

\begin{assumption}{(Non-degeneracy)}
\label{as:4.7}
There does not exist a non-zero $\lambda\in \R^r$ such that $\lambda'\frac{\partial \sigma_t(\theta_0)}{\partial \theta}=0$ a.s.
\end{assumption}

\begin{assumption}{(Monotonicity)}
\label{as:4.8}
For any real sequence $\{x_i\}$ and for any $\theta_1,\theta_2 \in \Theta$ satisfying $\theta_1\leq \theta_2$ componentwise, we have $\sigma(x_1,x_2,\dots;\theta_1)\leq \sigma(x_1,x_2,\dots;\theta_2)$.
\end{assumption}

\begin{assumption}{(Moments)}
\label{as:4.9}
There exists a neighborhood $\mathscr{V}(\theta_0)$ of $\theta_0$ such that the following variables have finite expectation
\begin{align*}
\text{(i)}\sup_{\theta \in \mathscr{V}(\theta_0)}\bigg|\frac{ \sigma_t(\theta_0)}{\sigma_t(\theta)}\bigg|^a, \qquad \;\;\: \text{(ii)} \sup_{\theta \in \mathscr{V}(\theta_0)}\bigg|\bigg|\frac{1}{\sigma_t(\theta)}\frac{\partial \sigma_t(\theta)}{\partial \theta}\bigg|\bigg|^{b}, \qquad \;\;\: \text{(iii)} \sup_{\theta \in \mathscr{V}(\theta_0)}\bigg|\bigg|\frac{1}{\sigma_t(\theta)}\frac{\partial^2 \sigma_t(\theta)}{\partial \theta \partial \theta'}\bigg|\bigg|^c
\end{align*}
for some $a$, $b$, $c$ (to be specified).\footnote{Note that the variables in (i)--(iii) are strictly stationary (\citeauthor{francq2011garch}, \citeyear{francq2011garch}, p.\ 181/406).}
\end{assumption}

\begin{assumption}{(Scaling Stability)}
\label{as:4.10}
 There exists a function $g$ such that for any $\theta \in \Theta$, for any $\lambda>0$, and any real sequence $\{x_i\}$
\begin{align*}
\lambda \sigma(x_1,x_2,\dots;\theta)=\sigma(x_1,x_2,\dots;\theta_\lambda),
\end{align*}
where $\theta_\lambda=g(\theta,\lambda)$ and $g$ is differentiable in $\lambda$.
\end{assumption}
The previous set of assumptions is comparable to the conditions imposed by \cite{francq2015risk}. Assumption \ref{as:4.3} calls for a correct specification of the volatility structure. If the researcher incorrectly specifies a volatility function $\varsigma(\dots;\vartheta)$ instead, the estimator of the misspecified conditional volatility model $\hat{\vartheta}_n$ will converge to a pseudo-true value, i.e.\ $\vartheta_0=\arg \min_\vartheta \EE\big[\frac{1}{2}\frac{\epsilon_t^2}{\varsigma_t^2(\vartheta)}+\log \varsigma_t(\vartheta)\big]$. The corresponding edf of the residuals $\frac{1}{n}\sum_{t=1}^n \mathbbm{1}_{\{\epsilon_t/\varsigma_t(\hat{\vartheta}_n)\leq x\}}$ converges to $\bar{F}(x)=\EE\big[F\big(x\frac{\varsigma_t(\vartheta_0)}{\sigma_t(\theta_0)}\big)\big]$ in view of Lemma \ref{lem:4.1} while the $\alpha$-quantile estimator converges to $\bar{F}^{-1}(\alpha)$, which is generally different from $F^{-1}(\alpha)$. Thus, the correct specification of the volatility function is crucial and one can test for it using the recently developed test by \cite{Maria2019}; for further recent results on goodness-of-fit testing for GARCH models see \cite{Bardet2020}. Regarding the bootstrap method to be developed below, it is plausible to expect that in the presence of a misspecified  conditional volatility model it will be consistent for the pseudo-true values, although we do not provide a rigorous proof.

Regarding the innovation process we do not need to assume $\EE[\eta_t]=0$ (cf.\ \citeauthor{francq2004maximum}, \citeyear{francq2004maximum}, Remark 2.5). 
The iid condition in Assumption \ref{as:4.5}\eqref{as:4.5.1} is vital for \eqref{eq:4.2.4} to hold and is the basis of the residual bootstrap in Section \ref{sec:4.4.1}.
Under correct specification of the volatility process the iid assumption imposed on the innovations can be tested for by considering the errors $\epsilon_t/ \sigma(\epsilon_{t-1},\epsilon_{t-2},\dots;\hat{\theta}_n)$, $t=1,..,n$, and applying the test of \cite{cho2011generalized}.

Whereas \cite{cavaliere2018fixed} assume the existence of the sixth moment of $\eta_t$ for the fixed-design bootstrap in ARCH($q$) models, we only require the fourth moment to be finite in Assumption \ref{as:4.5}(\ref{as:4.5.3}). 
We assume $\theta_0$ belongs to the interior of the parameter space in Assumption \ref{as:4.6}. Parameters on the boundary yield non-standard problems, which require special treatment. \cite{cavaliere2020bootstrap} provide bootstrap inference on the boundary of the parameter space with application to conditional volatility models. We return to this issue in our empirical application in Remark \ref{rem:5.1}.
In Assumption \ref{as:4.8} the function $\sigma(x_1,x_2,\dots;\theta)$ is presumed to be monotonically increasing in $\theta$, which is used to establish the strong consistency of the quantile estimator. While the monotonicity condition is a feature shared by various stochastic volatility models (cf.\ \citeauthor{berkes2003limit}, \citeyear{berkes2003limit}, Lemma 4.1), it excludes the exponential GARCH  \citep{nelson1991conditional} and the log-GARCH \citep{geweke1986comment,pantula1986modeling}. Further, we require higher order of moments in Assumption \ref{as:4.9} for the bootstrap, which does not seem to be restrictive for the classical GARCH-type models (cf.\ \citeauthor{francq2011garch}, \citeyear{francq2011garch}, p.\ 165; \citeauthor{hamadeh2011asymptotic}, \citeyear{hamadeh2011asymptotic}, p.\ 501). In particular, Assumption \ref{as:4.9} is presumed to hold with $a=\pm 12$, $b=12$ and $c=6$ for establishing the convergence of the bootstrap information matrix.

On the basis of the previous assumptions we extend the strong consistency result of \citeauthor{francq2015risk} (\citeyear{francq2015risk}, Theorem 1) to the quantile estimator.
\begin{theorem}\textit{(Strong Consistency)}\label{thm:4.1}
\begin{itemize}
    \item[(i)] \citep{francq2015risk}  Under Assumptions \ref{as:4.1}--\ref{as:4.3}, \ref{as:4.4}(\ref{as:4.4.1}) and \ref{as:4.5}(\ref{as:4.5.1}) the estimator  in \eqref{eq:4.3.3} is strongly consistent, i.e. $\hat{\theta}_n \overset{a.s.}{\to} \theta_0$.
\item[(ii)] If in addition Assumptions \ref{as:4.6} and \ref{as:4.9}(i) hold with $a=-1$, then the  estimator in \eqref{eq:4.3.4} satisfies $\hat{\xi}_{n,\alpha} \overset{a.s.}{\to} \xi_\alpha$.

\end{itemize}
\end{theorem}
To lighten notation, we henceforth write $D_t(\theta) =\frac{1}{\sigma_t(\theta)}\frac{\partial\sigma_t(\theta)}{\partial \theta}$  and drop the argument when evaluated at the true parameter, i.e.\ $D_t=D_t(\theta_0)$. The next result provides the joint asymptotic distribution of $\hat{\theta}_n$ and $\hat{\xi}_{n,\alpha}$ and is due to \cite{francq2015risk}. 
\begin{theorem} \label{thm:4.2} 
(Asymptotic Distribution, \citealp{francq2015risk}) Suppose Assumptions \ref{as:4.1}--\ref{as:4.7}, \ref{as:4.9} and \ref{as:4.10} hold with $a=b=4$ and $c=2$. Then, we have
\begin{align}
\label{eq:4.3.6}
      \begin{pmatrix}
      \sqrt{n}(\hat{\theta}_n-\theta_0)\\
    \sqrt{n}(\xi_{\alpha} - \hat{\xi}_{n,\alpha})
      \end{pmatrix}
\overset{d}{\to}N\big(0, \Sigma_\alpha\big) \qquad \mbox{with}\qquad
\Sigma_\alpha=
      \begin{pmatrix}
      \frac{\kappa-1}{4}J^{-1} & \lambda_\alpha J^{-1}\Omega\\
    \lambda_\alpha \Omega'J^{-1} & \zeta_\alpha
      \end{pmatrix},
\end{align}
where $\kappa = \EE[\eta_t^4]$, $\Omega = \EE[D_t]$, $J=\EE[D_tD_t']$, $\lambda_\alpha = \xi_\alpha\frac{\kappa-1}{4}+\frac{p_\alpha}{2f(\xi_{\alpha})}$, $\zeta_\alpha = \xi_\alpha^2\frac{\kappa-1}{4}+\frac{\xi_{\alpha} p_\alpha}{f(\xi_\alpha)}+\frac{\alpha(1-\alpha)}{f^2(\xi_\alpha)}$ and $p_\alpha =\EE[\eta_t^2 \mathbbm{1}_{\{\eta_t<\xi_\alpha\}}]-\alpha$.
\end{theorem}

\begin{remark}
It is worth mentioning that the asymptotics in this theorem for $\hat{\xi}_{n,\alpha}$ are for $\alpha$ fixed while $n$ goes to infinity. If, for instance, $\alpha$ is very small for moderate $n$ the distribution in the following theorem might not provide a good approximation. For such cases, approximations based on extreme value theory may provide better approximations to the unknown finite sample distribution. See, for example, \cite{McNeil2002} and the recent \cite{li_peng_song_2022} for an application of GARCH models with extreme value theory for VaR estimation. Note that the latter authors consider simultaneous inference for VaR and expected shortfall.   
\end{remark}

In applied work, to use Theorem \ref{thm:4.2} in order to conduct inference for $(\theta_0, \xi_{\alpha})$ one needs a consistent estimator for $\Sigma_{\alpha}$. It follows from Lemma \ref{lem:4.2} that $\kappa$, $\Omega$, and $J$ can be consistently estimated by
\begin{align}
\label{eq:4.3.7}
\begin{split}
\hat{\kappa}_n=\frac{1}{n}\sum_{t=1}^n\hat{\eta}_t^4,  \qquad 
\hat{\Omega}_n=\frac{1}{n}\sum_{t=1}^n\hat{D}_t,  \qquad \mbox{ and } \quad \hat{J}_n=\frac{1}{n}\sum_{t=1}^n\hat{D}_t\hat{D}_t',
\end{split}
\end{align}
respectively, with $\hat{D}_t = \tilde{D}_t(\hat{\theta}_n)$ and $\tilde{D}_t(\theta) = \frac{1}{\tilde{\sigma}_t(\theta)}\frac{\partial\tilde{\sigma}_t(\theta)}{\partial \theta}$. 
To estimate $\lambda_{\alpha}$ and $\zeta_{\alpha}$, which also appear in $\Sigma_{\alpha}$, one can use $\hat{\xi}_{n,\alpha}$ for $\xi_{\alpha}$ and 
\begin{align}\label{eq consistency matrix own} 
\hat{p}_{n,\alpha} = (1 \slash n) \sum_{t=1}^n \hat{\eta}_t^2 \mathbbm{1}_{\{\eta_t < \hat{\xi}_{n,\alpha}\}}-\alpha 
\end{align}
for $p_{\alpha}$ (see also Lemma \ref{alg:4.2} for the properties of $\hat{p}_{n,\alpha}$). To estimate the density $f$, which also appears in $\lambda_{\alpha}$ and $\zeta_{\alpha}$, kernel smoothing is commonly employed, i.e.\ 
\begin{align}\label{eq:4.3.8}
\hat{\mathbbm{f}}_n^S(x)  =\frac{1}{nh_n}\sum_{t=1}^n k\bigg(\frac{x-\hat{\eta}_t}{h_n}\bigg)
\end{align}
with kernel function $k$ and bandwidth $h_n>0$. Hence, whenever we use a consistent $\hat{\mathbbm{f}}_n^S$ for $f$ we obtain, combined with Equations \eqref{eq:4.3.7} and \eqref{eq:4.3.8}, a consistent estimator for $\Sigma_\alpha$ denoted by $\hat{\Sigma}_{n,\alpha}$. For the case that $\epsilon_t$ in Equation (\ref{eq:4.2.1}) follows a GARCH$(p,q)$ process, \cite{gao2008estimation} considered estimating $f$ by a kernel density estimator with Lipschitz-continuous kernels such as $k(x)=\phi(x)$, where $\phi$ is the standard normal density function. An alternative estimator is based on the uniform kernel $k(x)=\frac{1}{2}\mathbbm{1}_{\{|x|\leq 1\}}$ yielding $\hat{\mathbbm{f}}_n^S(\hat{\xi}_{n,\alpha})\overset{p}{\to}f(\xi_\alpha)$ whenever $h_n \sim n^{-\varrho}$ for some $\varrho \in (0,1/2]$.
 
At this point it is worth recalling that interest lies in $VaR_{\alpha}(\epsilon_{n+1}|\mathcal{F}_n)$ (cf.~Equation (\ref{eq:4.2.4})). More precisely, one is not so much interested in the distribution of this random variable or its moments but in inference for its possible realizations. Note also that $VaR_{\alpha}(\epsilon_{n+1}|\mathcal{F}_n)$ varies with $n$ and does not converge. This illustrates that inference for $VaR_{\alpha}(\epsilon_{n+1}|\mathcal{F}_n)$ is different from inference for  $(\xi_{\alpha},\theta_0)$ which is just a real-valued vector of dimension $r+1$. Nevertheless, it is intuitively clear how to conduct inference for $VaR_{\alpha}(\epsilon_{n+1}|\mathcal{F}_n)$, yet this difference results in two technical issues that need to be addressed to provide a thorough theoretical justification of inferential procedures for $VaR_{\alpha}(\epsilon_{n+1}|\mathcal{F}_n)$. These technical issues are known for quite some time; see, for instance, \cite{phillips1979sampling}, \citet{kreiss2015discussion} and the textbook \cite{pesaran2015time}[p.~389]. Several examples illustrating these two issues can be found in \citet{beutner2017justification}. In a nutshell the first issue stems from the fact that, as mentioned,  $VaR_{\alpha}(\epsilon_{n+1}|\mathcal{F}_n)$ varies over time (cf.~Eq.~(\ref{eq:4.2.4})) implying that a limiting distribution cannot exist.
The second issue is a result of $VaR_{\alpha}(\epsilon_{n+1}|\mathcal{F}_n)$ being a conditional quantity which on the one hand requires conditioning on the sample observed so far, yet on the other hand this conditioning eliminates all randomness, making it impossible to establish useful distributional results. We now briefly illustrate this here for the case that $VaR_{\alpha}(\epsilon_{n+1}|\mathcal{F}_n)$ is the object of interest; a very detailed illustration of this second issue by means of a GARCH(1,1) can be found in Example 2.1 in \cite{beutner2017justification}. Fixing some arbitrary starting values $\tilde{\epsilon}_0, \tilde{\epsilon}_{-1},\ldots$ and making the dependency of $\sigma_{n+1}$ on the realized values of the random variables $\epsilon_n,\epsilon_{n-1},\ldots,\epsilon_1$ and the starting values $\tilde{\epsilon}_0, \tilde{\epsilon}_{-1},\ldots$ explicit a feasible version of the quantity of interest is  
\begin{equation}\label{eq VaR for sample splitting}
VaR_{\alpha}(\epsilon_{n+1}|\mathcal{F}_n)= -\xi_{\alpha}\, \tilde{\sigma}_{n+1}(\epsilon_n^r,\epsilon_{n-1}^r,\ldots,\epsilon_1^r,\tilde{\epsilon}_0, \tilde{\epsilon}_{-1},\ldots;\theta_0),
\end{equation}
where $\epsilon_n^r,\epsilon_{n-1}^r,\ldots,\epsilon_{1}^r$ denote the realized values of $\epsilon_n,\epsilon_{n-1},\ldots,\epsilon_1$. Here feasible is in the sense that only parameter uncertainty remains. Now replacing, as at the beginning of Section 3, the unknown $\theta_0$ by the estimator $\hat{\theta}_{n}$ and doing the same with $\xi_{\alpha}$ we see that these estimators must enter Equation \eqref{eq VaR for sample splitting} \textit{given the realizations}, i.e.~as $\hat{\theta}_{n}(\epsilon_1^r,\epsilon_{2}^r,\ldots,\epsilon_n^r)$ and $\hat{\xi}_{\alpha,n}(\epsilon_1^r,\epsilon_{2}^r,\ldots,\epsilon_n^r)$, respectively, because otherwise we would end up with an inconsistency. Indeed, if they entered as $\hat{\theta}_n(\epsilon_1,\epsilon_{2},\ldots,\epsilon_n)$ and $\hat{\xi}_{n,\alpha}(\epsilon_1,\epsilon_{2},\ldots,\epsilon_n)$, respectively, we would treat $\epsilon_1,\epsilon_{2},\ldots,\epsilon_n$ as random and observed at the same time. Upon replacing $\theta_0$ and $\xi_{\alpha}$ in Equation \eqref{eq VaR for sample splitting} by $\hat{\theta}_{n}(\epsilon_1^r,\epsilon_{2}^r,\ldots,\epsilon_n^r)$ and $\hat{\xi}_{n,\alpha}(\epsilon_1^r,\epsilon_{2}^r,\ldots,\epsilon_n^r)$, respectively, no random variable appears in 
\begin{align*}
\reallywidehat{VaR}_{n,\alpha}= & -\hat{\xi}_{n,\alpha}(\epsilon_1^r,\epsilon_{2}^r,\ldots, \epsilon_n^r) \nonumber \\ & \times \sigma_{n+1}(\epsilon_n^r,\epsilon_{n-1}^r,\ldots,\epsilon_1^r,s_0,s_{-1},\ldots;\hat{\theta}_n(\epsilon_1^r,\epsilon_{2}^r,\ldots,\epsilon_n^r)),
\end{align*}
which is just the estimator of Equation (3.3) with all dependencies made explicit. As just observed no random variable appears in $\reallywidehat{VaR}_{n,\alpha}$ which therefore does not have a distribution which could be used to construct confidence intervals for $VaR_{\alpha}(\epsilon_{n+1}|\mathcal{F}_n)$. In the article \citet{beutner2017justification} merging is proposed as a solution to overcome the first issue and sample splitting as a means to solve the second. The results of this article can be used justify, through their Theorem 1, approximating the distribution of the VaR estimator, centered at $VaR_{n,\alpha}$ and inflated by $\sqrt{n}$, by
\begin{align}
\label{eq:4.3.9}
N\left(0, \begin{pmatrix}
    -\xi_\alpha \frac{\partial \sigma_{n+1}(\theta_0)}{\partial \theta}\\
    \sigma_{n+1}
      \end{pmatrix}' \Sigma_\alpha \begin{pmatrix}
    -\xi_\alpha \frac{\partial \sigma_{n+1}(\theta_0)}{\partial \theta}\\
         \sigma_{n+1}
      \end{pmatrix}\right),
\end{align}
and consequently to justify the intuitive (conditional) confidence interval 
\begin{align}\label{eq:4.3.10}
\reallywidehat{VaR}_{n,\alpha}\pm \frac{\Phi^{-1}(\gamma/2)}{\sqrt{n}}
\left\{\begin{pmatrix}
    -\hat{\xi}_{n,\alpha} \frac{\partial \tilde{\sigma}_{n+1}(\hat{\theta}_n)}{\partial \theta}\\
         \tilde{\sigma}_{n+1}(\hat{\theta}_n)
      \end{pmatrix}' \hat{\Sigma}_{n,\alpha} \begin{pmatrix}
    -\hat{\xi}_{n,\alpha} \frac{\partial \tilde{\sigma}_{n+1}(\hat{\theta}_n)}{\partial \theta}\\
    \tilde{\sigma}_{n+1}(\hat{\theta}_n)
      \end{pmatrix}\right\}^{1/2}
\end{align}
for ${VaR}_{n,\alpha}$. Here $\Phi$ denotes the standard normal cdf, and we used again the short-hand notation for $\tilde{\sigma}_{n+1}(\hat{\theta})$ and $\hat{\xi}_{n,\alpha}$, i.e.~did not make the dependency on the starting values etc.~explicit. Note that the interval \eqref{eq:4.3.10} lacks a proper theoretical justification as it treats the non-random $\reallywidehat{VaR}_{n,\alpha}$ as random because it assigns an asymptotic distribution to it. In order to use Theorem 3 and Corollary 2 of \citet{beutner2017justification} to provide a sound theoretical justification of \eqref{eq:4.3.10} let  $n_E: \N \to \N$ and $n_P: \N \to \N$ be such that for all $n \in \N$ we have $n_E(n) < n_P(n)$. First, in the sample split approach we will only use $\epsilon_1,\ldots,\epsilon_{n_E}$ to estimate $\xi_{\alpha}$ and $\theta_0$ which we denote by  $\hat{\xi}_{n_E,\alpha}^{SPL}(\epsilon_1,\epsilon_{2},\ldots,\epsilon_{n_E})$ and $\hat{\theta}_{n_E}^{SPL}(\epsilon_1,\epsilon_{2},\ldots,\epsilon_{n_E})$, respectively. To lighten the notation a bit we will also use the shorter notations   
$\hat{\xi}_{n_E,\alpha}^{SPL}$ and $\hat{\theta}_{n_E}^{SPL}$, respectively. Second, in the sample split approach conditioning on $\mathcal{F}_n$ on the left-hand side of  \eqref{eq VaR for sample splitting} is replaced by conditioning on $\epsilon_{n_P},\ldots,\epsilon_n$ only. The sample split version of (3.3) which we denote by $\reallywidehat{VaR}_{n,\alpha}^{SPL}$ becomes
\begin{align*}
\reallywidehat{VaR}_{n,\alpha}^{SPL}=
& -\hat{\xi}_{n_E,\alpha}^{SPL}(\epsilon_1,\epsilon_{2},\ldots,\epsilon_{n_E}) \tilde{\sigma}_{n+1}^{SPL}(\hat{\theta}^{SPL}_{n_E})
\nonumber \\ 
= & -\hat{\xi}_{n_E,\alpha}^{SPL}(\epsilon_1,\epsilon_{2},\ldots,\epsilon_{n_E}) \nonumber \\
& \times \tilde{\sigma}_{n+1}(\epsilon_n^r,\epsilon_{n-1}^r,\ldots,\epsilon_{n_P}^r,c_{n_P-1},\ldots,c_{1},\tilde{\epsilon}_{0},\tilde{\epsilon}_{-1}, \ldots;\hat{\theta}_{n_E}^{SPL}(\epsilon_1,\epsilon_{2},\ldots,\epsilon_{n_E})),
\end{align*}
where $\epsilon_n^r,\epsilon_{n-1}^r,\ldots,\epsilon_{n_P}^r$ and  $\tilde{\epsilon}_{0},\tilde{\epsilon}_{-1},\ldots$ are as before and $c_{n_P-1},\ldots,c_{1}$ are constants. These constants could be viewed as starting values but since they are not replacing unobserved value as $\tilde{\epsilon}_0, \tilde{\epsilon}_{-1},\ldots$ do we prefer to use a different letter for them. Note   
also that $\epsilon_1,\ldots,\epsilon_{n_E}$ enter the sample split version $\reallywidehat{VaR}_{n,\alpha}^{SPL}$ as random variables which is in contrast to  $\reallywidehat{VaR}_{n,\alpha}$. Therefore, $\reallywidehat{VaR}_{n,\alpha}^{SPL}$ does have a non-degenerate distribution which can be used to construct confidence intervals. The sample split (conditional) confidence interval is defined as 
\begin{align}\label{eq var spl interval}
\reallywidehat{VaR}_{n,\alpha}^{SPL}\pm \frac{\Phi^{-1}(\gamma/2)}{\sqrt{n_E}}
\left\{\begin{pmatrix}
    -\hat{\xi}_{n_E,\alpha}^{SPL} \frac{\partial \tilde{\sigma}_{n+1}^{SPL}(\hat{\theta}_{n_E}^{SPL})}{\partial \theta}\\
         \tilde{\sigma}_{n+1}^{SPL}(\hat{\theta}_{n_E}^{SPL})
      \end{pmatrix}' \hat{\Sigma}_{n_E,\alpha}^{SPL} \begin{pmatrix}
    -\hat{\xi}_{n_E,\alpha}^{SPL} \frac{\partial \tilde{\sigma}_{n+1}^{SPL}(\hat{\theta}_{n_E}^{SPL})}{\partial \theta}\\
    \tilde{\sigma}_{n+1}^{SPL}(\hat{\theta}_{n_E}^{SPL})
      \end{pmatrix}\right\}^{1/2}.
\end{align}
Here $\hat{\Sigma}_{n_E,\alpha}^{SPL}$ is as $\hat{\Sigma}_{n,\alpha}$ but based on the first $n_E$ observations only. Note that this interval is meaningful because $\reallywidehat{VaR}_{n,\alpha}^{SPL}$ is random and converges in distribution after centering and scaling. Intuitively, we would say that the statistically  meaningful interval in \eqref{eq var spl interval} provides a theoretical justification for the interval \eqref{eq:4.3.10} if 
\begin{equation}\label{eq convergence var's}
\reallywidehat{VaR}_{n,\alpha}^{SPL} \mbox{ converges to } \reallywidehat{VaR}_{n,\alpha}    
\end{equation}
 and if 
 \begin{equation}\label{eq convergence matrices I}
\left\{\begin{pmatrix}
    -\hat{\xi}_{n_E,\alpha}^{SPL} \frac{\partial \tilde{\sigma}_{n+1}^{SPL}(\hat{\theta}_{n_E}^{SPL})}{\partial \theta}\\
         \tilde{\sigma}_{n+1}^{SPL}(\hat{\theta}_{n_E}^{SPL})
      \end{pmatrix}' \hat{\Sigma}_{n_E,\alpha}^{SPL} \begin{pmatrix}
    -\hat{\xi}_{n_E,\alpha}^{SPL} \frac{\partial \tilde{\sigma}_{n+1}^{SPL}(\hat{\theta}_{n_E}^{SPL})}{\partial \theta}\\
    \tilde{\sigma}_{n+1}^{SPL}(\hat{\theta}_{n_E}^{SPL})
      \end{pmatrix}\right\}^{1/2}
 \end{equation}
converges to  
 \begin{equation}\label{eq convergence matrices II}
\left\{\begin{pmatrix}
    -\hat{\xi}_{n,\alpha} \frac{\partial \tilde{\sigma}_{n+1}(\hat{\theta}_n)}{\partial \theta}\\
         \tilde{\sigma}_{n+1}(\hat{\theta}_n)
      \end{pmatrix}' \hat{\Sigma}_{n,\alpha} \begin{pmatrix}
    -\hat{\xi}_{n,\alpha} \frac{\partial \tilde{\sigma}_{n+1}(\hat{\theta}_n)}{\partial \theta}\\
    \tilde{\sigma}_{n+1}(\hat{\theta}_n)
      \end{pmatrix}\right\}^{1/2}.
\end{equation}
This is also the concept employed in \citet{beutner2017justification}. Under the conditions of Theorem 3 of this article the convergence in \eqref{eq convergence var's} holds and \eqref{eq convergence matrices I} converges to 
\eqref{eq convergence matrices II} by Corollary 2 of this article. It is worth pointing out that convergence is in probability and that for this concept to be applicable $\reallywidehat{VaR}_{n,\alpha}$ and \eqref{eq convergence matrices II} are treated as random. Theorem 2 above ensures that one of the assumptions of Corollary 3 of \citet{beutner2017justification} holds so that it provides the basis for a sound theoretical justification of \eqref{eq:4.3.10} as a (conditional) confidence interval. Formally, we can state

\begin{corollary}\label{corollary supp} Assume that the assumptions of Theorem \ref{thm:4.1} are fulfilled. Define the function $\psi: \R^{\infty} \times \Upsilon$ with $\Upsilon=\R \times \Theta$ by 
$$\psi(x_1,x_2,\ldots;\upsilon)=\psi(x_1,x_2,\ldots;(\xi,\theta))=-\xi \sigma(x_1,x_2,\ldots;\theta),$$
and put $\upsilon_0=(\xi_{\alpha}, \theta_0)$. Assume further that 
\begin{enumerate}
\item $\Big|\Big|\frac{\partial \sigma (\epsilon_n, \epsilon_{n-1}, \ldots ;\theta_0)}{\partial \theta}\Big|\Big|=O_{p}(1)$;
\item $\sup_{\upsilon \in \mathscr{V}(\upsilon_0)}\Big|\Big|\frac{\partial^2 \psi (\epsilon_n, \epsilon_{n-1}, \ldots ;\upsilon)}{\partial \upsilon \partial \upsilon'}\Big|\Big|=O_{p}(1)$ for some open neighborhood $\mathscr{V}(\upsilon_0)$ around $\upsilon_0$;
\item Given sequences $\{\tilde{\epsilon}_t\}$ and $\{c_t\}$, we have
		\begin{align*}
	&	\sqrt{n}\big(\psi(\epsilon_n,\ldots,\epsilon_{t_1},c_{t_1-1},\ldots,c_1,\tilde{\epsilon}_0,\ldots; \upsilon_0) - \psi (\epsilon_n, \epsilon_{n-1}, \ldots; \upsilon_0)\big)=o_{p}(1),\\
&		\bigg|\bigg|\frac{\partial \psi(\epsilon_n,\ldots,\epsilon_{t_1},c_{t_1-1},\ldots,c_1,\tilde{\epsilon}_0,\ldots; \upsilon_0) }{\partial \upsilon} - \frac{\partial \psi (\epsilon_n, \epsilon_{n-1}, \ldots; \upsilon_0)}{\partial \upsilon}\bigg|\bigg|=o_{p}(1),\\
	& \sup_{\upsilon \in \mathscr{V}(\upsilon_0)} \bigg|\bigg|\frac{\partial^2 \psi(\epsilon_n,\ldots,\epsilon_{t_1},c_{t_1-1},\ldots,c_1,\tilde{\epsilon}_0,\ldots; \upsilon_0)}{\partial \upsilon \partial \upsilon'} - \frac{\partial^2 \psi (\epsilon_n, \epsilon_{n-1}, \ldots; \upsilon_0)}{\partial \upsilon \partial \upsilon'}\bigg|\bigg| \\
	& =o_{p}(1)
		\end{align*}
		for any $t_1 \geq 1$ such that $(n-t_1) / l_n \rightarrow \infty$ as $n \to \infty$ and for some model-specific $l_n$ with $l_n \rightarrow \infty$.
\end{enumerate}
Moreover, let $n_E$ and $n_P$ fulfill Assumption 3.a of \citet{beutner2017justification} and $\{\epsilon_t\}$ Assumption 3.c of the same article. Then the difference between $\reallywidehat{VaR}_{n,\alpha}^{SPL}$ and $ \reallywidehat{VaR}_{n,\alpha}$ converges to zero in probability and the same holds for the difference between \eqref{eq convergence matrices I} and \eqref{eq convergence matrices II}.
\end{corollary} 
\begin{proof}
The claims made are proved if the assumptions of Theorem 3 and Corollary 2 in \citet{beutner2017justification} are met. Assumption 1.a of \citet{beutner2017justification} holds by Theorem \ref{thm:4.2}. Note that the scaling sequence $m_T$ in \cite{beutner2017justification} equals $\sqrt{n}$ here. Moreover according to Theorem \ref{thm:4.2} $\hat{\upsilon}_n=(\hat{\xi}_{n,\alpha},\hat{\theta}_n)$ scaled by $\sqrt{n}$ and centered at $\upsilon_0$ converges to a multivariate normal distribution so that Assumption 5 of \citet{beutner2017justification} also holds. We now turn to Assumption 1.b of \citet{beutner2017justification}. Clearly $\upsilon \to \psi(\cdot;\upsilon$) is continuous because by Assumption 3 above $\theta \to \sigma(\cdot,;\theta)$ is continuous. Furthermore, because $\upsilon \to \psi(\cdot;\upsilon)$ is the product of the functions $\xi \to \xi$ and $\theta \to \sigma(\cdot;\theta)$ which are both twice differentiable (for $\theta \to \sigma(\cdot;\theta)$ this holds by Assumption 4 (ii) above)  the same holds for $\upsilon \to \psi(\cdot;\upsilon)$. Therefore, Assumption 3.b of \citet{beutner2017justification} holds. The gradient of $\psi(\epsilon_n,\epsilon_{n-1},\ldots;\upsilon_0)$ is given by 
$$
\left(-\sigma(\epsilon_n,\epsilon_{n-1},\ldots;\theta_0), -\xi_{\alpha}\frac{\partial \sigma (\epsilon_n, \epsilon_{n-1}, \ldots ;\theta_0)}{\partial \theta}\right).
$$
By Assumption 3 above and Markov's inequality $\sigma(\epsilon_n,\epsilon_{n-1},\ldots;\theta_0)$ is bounded in probability. Hence, together with Assumption 1.~in the statement of the corollary it follows that Assumption 1.c in \citet{beutner2017justification} holds.  
Assumption 1.d and 1.e in \citet{beutner2017justification} are just the Assumptions 2.~and 3.~of the corollary. Obviously, Assumptions 3.a and 3.c of \cite{beutner2017justification} hold under the assumptions of the corollary. Assumption 3.b of this article is met by Assumption 2 above which is nothing else than Assumption 3.b of \citet{beutner2017justification} for Value-at-Risk. As explained in the proof of Theorem 3 in \citet{beutner2017justification} the Assumption 2.b of that article is irrelevant for the quantity on the right-hand side of \eqref{eq convergence var's} and for \eqref{eq convergence matrices II}. Assumption 2.a of that article is obviously met. This finishes the proof.      
\end{proof}
\begin{remark} 
 Like some of the Assumptions 1-10 will hold or not depending on the specification of $\sigma$, the same is true for the assumptions of Corollary \ref{corollary supp} that ensure applicability of the results of \citet{beutner2017justification}. For popular time series models like a GARCH(1,1) they have been verified in \citet{beutner2019technical}.  
\end{remark}

Although the interval in \eqref{eq:4.3.10} is, as just outlined, well-justified, it may perform poorly since the density estimation appears rather sensitive regarding the choice of bandwidth (see \citeauthor{gao2008estimation}, \citeyear{gao2008estimation}, Section 4). Bootstrap methods offer an alternative way to quantify the uncertainty around the estimators.     

\section{Bootstrap}
\label{sec:4.4}

Bootstrap approximations frequently provide better insight into the actual distribution than the asymptotic approximation, yet they require a careful set-up. \cite{hall2003inference} show that conventional bootstrap methods are inconsistent in a GARCH model lacking finite fourth moment in the case of the squared innovations' distribution not being in the domain of attraction of the normal distribution. They consider a subsample bootstrap instead and study its asymptotic properties. In correspondence, an $m$-out-of-$n$ without-replacement bootstrap is proposed by \cite{spierdijk2016confidence} to construct confidence intervals for ARMA-GARCH VaR.

\cite{pascual2006bootstrap} present a residual bootstrap in a GARCH($1,1$) setting and assess its finite sample properties by means of simulation. Their bootstrap scheme follows a recursive design in which the bootstrap observations are generated iteratively using the estimated volatility dynamics. Building upon their results, \cite{christoffersen2005estimation} construct bootstrap confidence intervals for (conditional) VaR and Expected Shortfall and compare them to competitive methods within the GARCH($1,1$) model. Theoretical results on the recursive-design residual bootstrap are provided by \cite{hidalgo2007goodness} and  \cite{jeong2017residual} for the ARCH($\infty$) and  GARCH($p,q$) model, respectively. 

In contrast, \cite{shimizu2009bootstrapping} considers fixed-design variants of the wild and the residual bootstrap in which the ARMA-GARCH dynamics of the bootstrap samples are kept fixed at the values of the original series. The bootstrap estimators are based on a single Newton-Raphson iteration simplifying the proofs of first-order asymptotic validity. \citeauthor{shimizu2009bootstrapping}'s approach for the residual bootstrap is also employed in a multivariate GARCH setting by \cite{francq2016variance}. Recently, \cite{cavaliere2018fixed} study the fixed-design residual bootstrap in the context of ARCH($q$) models and propose a bootstrap Wald statistic based on a QML bootstrap estimator. While their theory has been developed independently to ours, their simulation study indicates that the fixed-design bootstrap performs as well as the recursive-design bootstrap.

\subsection{Fixed-design Residual Bootstrap}
\label{sec:4.4.1}

We propose a fixed-design residual bootstrap procedure, described in Algorithm \ref{alg:4.1}, to approximate the distribution of the estimators in \eqref{eq:4.3.3} -- \eqref{eq:4.3.5}.

\begin{algorithm}\textit{(Fixed-design residual bootstrap)}
\label{alg:4.1}
\begin{enumerate}

\item For $t=1,\dots, n$, generate  $\eta_t^* \overset{iid}{\sim} \hat{\mathbbm{F}}_n$ and the bootstrap observation $\epsilon_t^* = \tilde{\sigma}_t(\hat{\theta}_n) \eta_t^*$.
\item Calculate the bootstrap estimator 
\begin{align}
\label{eq:4.4.1}
\hat{\theta}_n^* = \arg \max_{\theta \in \Theta} \frac{1}{n}\sum_{t=1}^n \ell_t^*(\theta) \qquad \text{with} \qquad \ell_t^*(\theta)=-\frac{1}{2}\bigg(\frac{\epsilon_t^{*}}{\tilde{\sigma}_t(\theta)}\bigg)^2-\log \tilde{\sigma}_t(\theta).
\end{align}

\item For $t=1,\dots,n$ compute the bootstrap residual $\hat{\eta}_t^* = \epsilon_t^*/\tilde{\sigma}_t(\hat{\theta}_n^*)$
and obtain 
\begin{align}
\label{eq:4.4.2}
\hat{\xi}_{n,\alpha}^* = \arg\min_{z \in \R} \frac{1}{n}\sum_{t=1}^n\rho_\alpha(\hat{\eta}_t^*-z).
\end{align}

\item Obtain the bootstrap estimator of the conditional VaR
\begin{align}
\label{eq:4.4.3}
\reallywidehat{VaR}_{n,\alpha}^{*}=-\hat{\xi}_{n,\alpha}^{*}\: \tilde{\sigma}_{n+1}\big(\hat{\theta}_n^{*}\big).
\end{align}
\end{enumerate}
\end{algorithm}

\begin{remark}
\label{rem:4.1}
In contrast to the literature, the bootstrap errors are drawn with replacement from the residuals rather than the standardized residuals. In fact, re-centering would be inappropriate in the case of $\EE[\eta_t]\neq 0$. 
In addition, re-scaling of the residuals is typically redundant as $\frac{1}{n}\sum_{t=1}^n \hat{\eta}_t^2=1$ 
is implied by $\hat{\theta}_n \in \mathring{\Theta}$
under Assumption \ref{as:4.10}; see \citeauthor{francq2011garch}, \citeyear{francq2011garch}, p.\ 182/406 and note that the solution requires $\hat{\theta}_n$ belonging to the interior (\citeauthor{francq2011garch}, Oct.\ 2018, personal communication).
\end{remark}

\begin{remark}
\label{rem:4.2}
The term `fixed-design' refers to the fact that the bootstrap observations are generated using $\tilde{\sigma}_t(\hat{\theta}_n)=\sigma(\epsilon_{t-1},\dots,\epsilon_1,\tilde{\epsilon}_0,\tilde{\epsilon}_{-1},\dots;\hat{\theta}_n)$. In contrast, a recursive-design scheme replicates the model's dynamic structure, i.e.\  $\epsilon_t^\star = \sigma_t^\star \eta_t^\star$ with $\sigma_t^\star = \sigma(\epsilon_{t-1}^\star,\dots,\epsilon_1^\star,\tilde{\epsilon}_0,\tilde{\epsilon}_{-1},\dots;\hat{\theta}_n)$ and $\eta_t^\star \overset{iid}{\sim} \hat{\mathbbm{F}}_n$, which is computationally more demanding. We refer to Appendix \ref{app:4.B} for a complete description. See also \cite{cavaliere2018fixed} for more theoretical insights on the difference in the design in an ARCH($q$). 
\end{remark}

\begin{remark}
\label{rem:4.3}
Whereas \eqref{eq:4.4.1} involves a nonlinear optimization, \cite{shimizu2009bootstrapping} proposes a Newton-Raphson type bootstrap estimator instead. The Newton-Raphson bootstrap estimator corresponding to \eqref{eq:4.4.1} is given by
\begin{align*}
\hat{\theta}_n^{*NR} = \hat{\theta}_n+\hat{J}_n^{-1}\frac{1}{2n}\sum_{t=1}^n \hat{D}_t \big(\eta_t^{*2}-1\big),
\end{align*}
which can considerably speed up computations.
\end{remark}

Proposition \ref{prop:4.1} establishes the asymptotic validity of the bootstrap for the volatility parameters.
\begin{proposition}
\label{prop:4.1}
Suppose Assumptions \ref{as:4.1}--\ref{as:4.4}, \ref{as:4.5}(\ref{as:4.5.1}), \ref{as:4.5}(\ref{as:4.5.3}), \ref{as:4.6}, \ref{as:4.7}, \ref{as:4.9} and \ref{as:4.10} hold with $a=\pm 12$, $b=12$ and $c=6$. Then, we have
\begin{align*}
\sqrt{n}\big(\hat{\theta}_n^* - \hat{\theta}_n\big)
\overset{d^*}{\to}N\bigg(0,\frac{\kappa-1}{4}J^{-1}\bigg)
\end{align*}
almost surely.
\end{proposition}
Establishing the asymptotic validity of the bootstrap for the second part appears challenging since the bootstrap innovations are drawn from the discrete distribution $\hat{\mathbbm{F}}_n$. To overcome this issue we rely on arguments employed by \cite{bahadur1966note} and \cite{berkes2003limit}. The following theorem states the paper's main result.
\begin{theorem}\textit{(Bootstrap consistency)}
\label{thm:4.3}
Suppose Assumptions \ref{as:4.1}--\ref{as:4.10} hold with $a=\pm 12$, $b=12$ and $c=6$. Then, we have
\begin{align*}
      \begin{pmatrix}
      \sqrt{n}(\hat{\theta}_n^*-\hat{\theta}_n)\\
    \sqrt{n}(\hat{\xi}_{n,\alpha} - \hat{\xi}_{n,\alpha}^*)
      \end{pmatrix}
\overset{d^*}{\to}N\big(0, \Sigma_\alpha\big)
\end{align*}
in probability.
\end{theorem}
Theorem \ref{thm:4.3} is useful to validate the bootstrap for the conditional VaR estimator. For the asymptotic behavior of the conditional VaR estimator we refer to \eqref{eq:4.3.9} and the text around it. The following corollary is established.
\begin{corollary}
\label{cor:4.1}
Under the assumptions of Theorem \ref{thm:4.3} the conditional distribution of $\sqrt{n}\big(\reallywidehat{VaR}_{n,\alpha}^{*}-\reallywidehat{VaR}_{n,\alpha}\big)$ given $\mathcal{F}_n$ and \eqref{eq:4.3.9} given $\mathcal{F}_n$ merge in probability.
\end{corollary}

\subsection{Bootstrap Confidence Intervals for VaR}
\label{sec:4.4.3}

Clearly, the VaR evaluation in \eqref{eq:4.3.5} is subject to estimation risk that needs to be quantified. We propose the following algorithm to obtain approximately $100(1-\gamma)\%$ confidence intervals.
\begin{algorithm}{\textit{(Fixed-design Bootstrap Confidence Intervals for VaR)}}
\label{alg:4.2}
\begin{enumerate}
\item Acquire a set of $B$ bootstrap replicates, i.e.  $\reallywidehat{VaR}_{n,\alpha}^{*(b)}$ for $b=1,\dots,B$, by repeating Algorithm \ref{alg:4.1}.

\item[2.1.] Obtain the \textit{equal-tailed percentile} (EP) interval
\begin{align}
\label{eq:4.4.8}
\bigg[\reallywidehat{VaR}_{n,\alpha}-\frac{1}{\sqrt{n}}\hat{G}_{n,B}^{*-1}(1-\gamma/2),\:\reallywidehat{VaR}_{n,\alpha}-\frac{1}{\sqrt{n}}\hat{G}_{n,B}^{*-1}(\gamma/2)\bigg]
\end{align}
with $\hat{G}_{n,B}^{*-1}(\cdot)$ being the quantile function (generalized inverse) of $\hat{G}_{n,B}^*(x)=\frac{1}{B}\sum_{b=1}^B \mathbbm{1}_{\big\{\sqrt{n}\big(\reallywidehat{VaR}_{n,\alpha}^{*(b)}-\reallywidehat{VaR}_{n,\alpha}\big)\leq x\big\}}$. 
\item[2.2.] Calculate the \textit{reversed-tails} (RT) interval
\begin{align}
\label{eq:4.4.9}
\bigg[\reallywidehat{VaR}_{n,\alpha}+\frac{1}{\sqrt{n}}\hat{G}_{n,B}^{*-1}(\gamma/2),\reallywidehat{VaR}_{n,\alpha}+\frac{1}{\sqrt{n}}\hat{G}_{n,B}^{*-1}(1-\gamma/2)\bigg].
\end{align}
\item[2.3.] Compute the \textit{symmetric} (SY) interval
\begin{align}
\label{eq:4.4.10}
\bigg[\reallywidehat{VaR}_{n,\alpha}-\frac{1}{\sqrt{n}}\hat{H}_{n,B}^{*-1}(1-\gamma),\:\reallywidehat{VaR}_{n,\alpha}+\frac{1}{\sqrt{n}}\hat{H}_{n,B}^{*-1}(1-\gamma)\bigg]
\end{align}
with $\hat{H}_{n,B}^{*-1}(\cdot)$ being the quantile function (generalized inverse) of $\hat{H}_{n,B}^*(x)=\frac{1}{B}\sum_{b=1}^B \mathbbm{1}_{\big\{\sqrt{n}\big|\reallywidehat{VaR}_{n,\alpha}^{*(b)}-\reallywidehat{VaR}_{n,\alpha}\big|\leq x\big\}}$.
\end{enumerate}
\end{algorithm}
The interval in \eqref{eq:4.4.8} is obtained by the EP method, that is frequently encountered in the bootstrap literature. It is obtained from the (typically) infeasible equal-tailed confidence interval
\begin{align*}
\bigg[\reallywidehat{VaR}_{n,\alpha}-\frac{1}{\sqrt{n}}G_{n}^{-1}(1-\gamma/2),\:\reallywidehat{VaR}_{n,\alpha}-\frac{1}{\sqrt{n}}G_{n}^{-1}(\gamma/2)\bigg],
\end{align*}
where $G_{n}^{-1}$ is  the (unknown) quantile function of  $\sqrt{n}(\reallywidehat{VaR}_{n,\alpha}-{VaR}_{n,\alpha})$, which is replaced by its bootstrap analogue $\hat{G}_{n,B}^{*-1}$. The same reasoning leads to the SY interval but with test statistic $\sqrt{n}|(\reallywidehat{VaR}_{n,\alpha}-{VaR}_{n,\alpha})|$ instead of $\sqrt{n}(\reallywidehat{VaR}_{n,\alpha}-{VaR}_{n,\alpha})$ which makes it also clear that the interval in \eqref{eq:4.4.10} presumes symmetry for rationalizing its construction.
``Flipping around" the tails of the ET interval leads to the RT interval given in \eqref{eq:4.4.9}. Clearly, the RT and the EP have equal length. Whereas \eqref{eq:4.4.9} in its current form emphasizes the interval's name, RT type intervals are frequently reported in their reduced form, i.e.~the lower and upper bound of \eqref{eq:4.4.9} simplify to the $\gamma/2$ and $1-\gamma/2$ quantiles of $\frac{1}{B}\sum_{b=1}^B \mathbbm{1}_{\big\{\reallywidehat{VaR}_{n,\alpha}^{*(b)}\leq x\big\}}$, respectively. RT intervals can either be motivated by the results of \cite{falk1991coverage}\footnote{In a random sample setting \cite{falk1991coverage} prove that the RT bootstrap interval for quantiles has asymptotically greater coverage than the corresponding EP bootstrap interval. For additional insights we refer to \cite{hall1988bootstrap}.} or as the bootstrap analogue of 
the (uncentered) statistic $\reallywidehat{VaR}_{n,\alpha}$. 
It is worth mentioning that RT type bootstrap intervals for the VaR are also constructed in reduced form by \cite{christoffersen2005estimation}. Regardless of whether we use an EP, RT or SY interval the meaning is always the same: Given the past up to and including time $n$ the probability that the conditional VaR for period $n+1$ is contained in the intervals is approximately equal to $100(1-\gamma)\%$.

\subsection{Bootstrap Extensions} 
\label{sec:4.4.4}

The asymptotic normality result in Theorem \ref{thm:4.2} as well as the bootstrap consistency in Theorem \ref{thm:4.3} are derived, inter alia, under the assumption that the innovations are iid. In case this is not believed to be true -- e.g.~if the suggested specification tests mentioned in Section \ref{sec:4.3} indicate otherwise -- asymptotic normality of $\sqrt{n}(\hat{\theta}_n-\theta_0)$ can still be established under regularity assumptions.  \cite{Escanciano2009} studies the QML estimator under some dependence among the $\eta_t$'s while imposing  slightly stronger (moment) conditions, whereas the related paper of \cite{Linton2010} investigates estimators in a GARCH(1,1) with dependent errors but under weaker moment conditions. A multivariate version of the dependence condition in \cite{Escanciano2009} can be found in \cite{FrancqZakoian2016}.

The bootstrap method presented in Algorithm \ref{alg:4.1} is contingent on the iid assumption. In general, one cannot expect bootstrap procedures that are based on an iid assumption to be robust against deviations from this assumption; for a well-known example, one may refer to \cite{GONCALVES2004}. Alternative bootstrap techniques may be used if the iid condition is thought to be unrealistic. A variety of bootstrap methods exist that can capture dependence and non-identical random variables; see e.g.~\cite{Lahiri03} for a broad overview. The wild or multiplier bootstrap \citep{Mammen93, DavidsonFlachaire08} is particularly suited for dealing with non-identical variables, but does not capture dependence, unless it is properly modified \citep{Shao10, FSU20}. However, it remains an open question which bootstrap method combined with the fixed design approach leads to valid bootstrap procedures.

\section{Numerical Illustration}
\label{sec:4.5}

\subsection{Monte Carlo Experiment}
\label{sec:4.5.1}

In order to evaluate the finite sample performance of the proposed bootstrap procedure a Monte Carlo experiment is conducted. We confine ourselves to four conditional volatility specifications related to Examples \ref{ex:4.1} and \ref{ex:4.2} in Section \ref{sec:4.2}.
The first two are GARCH($1,1$) parameterizations with 
\begin{enumerate}[(i)]
   \item high persistence: $\theta_0 = (\omega_0,\alpha_0,\beta_0)’= \big(0.05\times 20^2/252,0.15,0.8\big)'$;
   \item low persistence: $\theta_0  = (\omega_0,\alpha_0,\beta_0)’ = \big(0.05\times 20^2/252,0.4,0.55\big)'$,
\end{enumerate}
which are similar to the specifications of \citeauthor{gao2008estimation} (\citeyear{gao2008estimation}, Section 4) and \citeauthor{spierdijk2016confidence} (\citeyear{spierdijk2016confidence}, Section 4.2). In addition, we study two T-GARCH($1,1$) scenarios likewise associated with high and low persistence:
\begin{enumerate}[(i)]
\setcounter{enumi}{2}
   \item high persistence: $\theta_0 = (\omega_0,\alpha_0^+, \alpha_0^-, \beta_0)’ = \big(0.05\times 20/\sqrt{252},0.05,0.10,0.8\big)'$;
   \item low persistence: $\theta_0 = (\omega_0,\alpha_0^+, \alpha_0^-, \beta_0)’ = \big(0.05\times 20/\sqrt{252},0.1,0.3,0.55\big)'$.
\end{enumerate}
Within the experiment the VaR level takes value $\alpha =0.05$ and there are two possible innovation distributions: a Student-$t$ distribution with $6$ degrees of freedom (df) and the standard normal distribution.\footnote{The Student-t innovations are appropriately standardized to satisfy $\EE\eta_t^2=1$.} We consider four estimation sample sizes, $n \in \{250; 500; 1{,}000; 5{,}000\}$, whereas the number of bootstrap replicates is fixed and equal to $B=999$. For each model version we simulate $S=10{,}000$ independent Monte Carlo trajectories. 

The numerical optimization of the log-likelihood function is carried out employing the build-in function \textit{fmincon} and running time is reduced by parallel computing using \textit{parfor}. The code is available on the website of the third author, and simulation results for a VaR level of $\alpha=0.01$ can be found in the working paper version.

\captionsetup[table]{labelformat=simple, labelsep=space}

\bgroup
\def\arraystretch{0.95}
\begin{table}[tbp] 
\begin{center}
\caption{\textbf{Fixed-design} bootstrap confidence intervals and asymptotic interval for \textbf{GARCH($1$,$1$)} with Student-t innovations} 
\label{tab:4.1}
\end{center}
\centering
\resizebox{\textwidth}{!}{\begin{tabular}{rc:ccc:ccc}
\hline \hline
\multicolumn{1}{c}{\begin{tabular}[c]{@{}c@{}}Sample\\ Size\end{tabular}} &                         & \begin{tabular}[c]{@{}c@{}}Average \\ coverage\end{tabular} & \begin{tabular}[c]{@{}c@{}}Av. coverage\\ below/above\end{tabular} & \begin{tabular}[c]{@{}c@{}}Average\\ length\end{tabular} & \begin{tabular}[c]{@{}c@{}}Average\\ coverage\end{tabular} & \begin{tabular}[c]{@{}c@{}}Av. coverage\\ below/above\end{tabular} & \begin{tabular}[c]{@{}c@{}}Average\\ length\end{tabular} \\ \hline
\multicolumn{1}{c}{} & & & & & & & \\[-9pt]
\multicolumn{1}{c}{} & & \multicolumn{6}{c}{} \\[-8pt]
\multicolumn{1}{c}{} & & \multicolumn{3}{c}{low persistence} & \multicolumn{3}{c}{high persistence} \\
250 & EP & 81.2 & 6.81/11.99 & 0.595 & 79.57 & 7.28/13.15 & 0.797 \\
& RT & 90.73 & 2.98/6.29 & 0.595 & 90.22 & 3.28/6.50 & 0.797 \\
& SY & 88.87 & 3.38/7.75 & 0.620 & 88.28 & 3.48/8.24 & 0.829 \\
& AS & 86.35 & 3.37/10.28 & 0.617 & 85.85 & 3.93/10.22 & 0.803 \\
\hdashline
500 & EP & 84.15 & 6.13/9.72 & 0.431 & 84.07 & 6.20/9.73 & 0.581 \\
& RT & 90.65 & 3.83/5.52 & 0.431 & 90.82 & 3.86/5.32 & 0.581 \\
& SY & 89.71 & 3.84/6.45 & 0.443 & 89.52 & 3.98/6.50 & 0.595 \\
& AS & 88.17 & 3.80/8.03 & 0.445 & 87.4 & 4.34/8.26 & 0.573 \\
\hdashline
1,000 & EP & 85.53 & 5.88/8.59 & 0.306 & 85.73 & 5.64/8.63 & 0.419 \\
& RT & 90.56 & 4.10/5.34 & 0.306 & 90.41 & 4.07/5.52 & 0.419 \\
& SY & 89.68 & 4.10/6.22 & 0.311 & 89.72 & 4.08/6.20 & 0.425 \\
& AS & 88.57 & 4.11/7.32 & 0.314 & 88.09 & 4.43/7.48 & 0.412 \\
\hdashline
5,000 & EP &87.59 & 5.72/6.69 & 0.144 & 87.97 & 5.35/6.68 & 0.191 \\
& RT & 90.43 & 4.58/4.99 & 0.144 & 89.77 & 4.82/5.41 & 0.191 \\
& SY & 89.69 & 4.74/5.57 & 0.145 & 89.69 & 4.56/5.75 & 0.192 \\
& AS & 89.62 & 4.57/5.81 & 0.146 & 88.86 & 4.84/6.30 & 0.188 \\
\hline \hline
\end{tabular}
}
\vspace{0.15cm}
\caption*{Table \ref{tab:4.1} reports distinct features of the \textbf{fixed-design} bootstrap confidence intervals and the asymptotic interval for the conditional VaR at \textbf{level} $\bm{\alpha=0.05}$ with \textbf{nominal coverage} $\bm{1-\gamma=90\%}$. For each interval type and different sample sizes ($n$), the interval's average coverage rates (in $\%$), the average rate of the conditional VaR being below/above the interval (in $\%$) and the interval's average length are tabulated. The bootstrap intervals are based on $B=999$ bootstrap replications and the averages are computed using $S=10{,}000$ simulations. The results are for the low (left part) and high (right part) persistence parametrization of a \textbf{GARCH(1,1)} with (normalized) Student-t innovations ($6$ df).
}
\end{table}
 \egroup

\captionsetup[table]{labelformat=simple, labelsep=space}

\bgroup
\def\arraystretch{0.95}
\begin{table}[tbp] 
\begin{center}
\caption{\textbf{Fixed-design} bootstrap confidence intervals and asymptotic interval for \textbf{T-GARCH($1$,$1$)} with Student-t innovations} 
\label{tab:4.1a}
\end{center}
\centering
\resizebox{\textwidth}{!}{\begin{tabular}{rc:ccc:ccc}
\hline \hline
\multicolumn{1}{c}{\begin{tabular}[c]{@{}c@{}}Sample\\ Size\end{tabular}} &                         & \begin{tabular}[c]{@{}c@{}}Average \\ coverage\end{tabular} & \begin{tabular}[c]{@{}c@{}}Av. coverage\\ below/above\end{tabular} & \begin{tabular}[c]{@{}c@{}}Average\\ length\end{tabular} & \begin{tabular}[c]{@{}c@{}}Average\\ coverage\end{tabular} & \begin{tabular}[c]{@{}c@{}}Av. coverage\\ below/above\end{tabular} & \begin{tabular}[c]{@{}c@{}}Average\\ length\end{tabular} \\ \hline
\multicolumn{1}{c}{} & & & & & & & \\[-9pt]
\multicolumn{1}{c}{} & & \multicolumn{6}{c}{} \\[-8pt]
& \multicolumn{1}{l}{} & \multicolumn{3}{c}{low persistence} & \multicolumn{3}{c}{high persistence} \\
250 & EP &
79.74 & 7.16/13.10 & 0.140 & 78.93 & 7.64/13.43 & 0.289 \\
& RT & 
90.22 & 3.42/6.36 & 0.140 & 90.34 & 2.98/6.68 & 0.289 \\
& SY & 
88.37 & 3.64/7.99 & 0.145 & 88.53 & 3.51/7.96 & 0.302 \\
& AS & 
87.59 & 3.27/9.14 & 0.146 & 87.87 & 3.14/8.99 & 0.304 \\
\hdashline
500 & EP &
82.41 & 6.51/11.08 & 0.104 & 82.12 & 6.66/11.22 & 0.211 \\
& RT & 
90.13 & 4.12/5.75 & 0.104 & 90.83 & 3.40/5.77 & 0.211 \\
& SY & 
88.83 & 4.26/6.91 & 0.106 & 89.35 & 3.50/7.15 & 0.218 \\
& AS & 
89.08 & 3.57/7.35 & 0.108 & 89.39 & 3.28/7.33 & 0.221 \\
\hdashline
1,000 & EP & 
84.98 & 6.09/8.93 & 0.076 & 83.46 & 6.84/9.70 & 0.155 \\
& RT & 
90.16 & 4.63/5.21 & 0.076 & 90.29 & 4.41/5.30 & 0.155 \\
& SY & 
89.37 & 4.33/6.30 & 0.077 & 89.22 & 4.34/6.44 & 0.159 \\
& AS & 
89.53 & 4.06/6.41 & 0.079 & 89.35 & 4.09/6.56 & 0.161 \\
\hdashline
5,000 & EP &
88.37 & 5.06/6.57 & 0.036 & 87.95 & 5.08/6.97 & 0.074 \\
& RT & 
90.72 & 4.72/4.56 & 0.036 & 90.16 & 4.84/5.0 & 0.074 \\
& SY & 
90.52 & 4.44/5.04 & 0.036 & 89.95 & 4.41/5.64 & 0.074 \\
& AS & 
91.23 & 4.03/4.74 & 0.037 & 90.54 & 4.11/5.35 & 0.076 \\
\hline \hline
\end{tabular}
}
\vspace{0.15cm}
\caption*{Table \ref{tab:4.1a} is exactly as Table \ref{tab:4.1} but for a \textbf{T-GARCH(1,1)} with (normalized) Student-t innovations ($6$ df) instead of a GARCH(1,1) with Student-t innovations ($6$ df).}
\end{table}
 \egroup

Table \ref{tab:4.1} and \ref{tab:4.1a} report the results of the three $90\%$--bootstrap intervals for the $5\%$--VaR when the innovation distribution is Student-t (henceforth referred to as baseline) and the model is a GARCH(1,1) and a T-GARCH(1,1), respectively. In both tables the results of the interval \eqref{eq:4.3.10} based on asymptotic (AS) theory are included for comparison, where a Gaussian kernel is utilized together with a bandwidth following \citeauthor{silverman1986density}'s (\citeyear{silverman1986density}) rule-of-thumb. In the GARCH($1,1$) high persistence case (right part of Table 1), we see that the average coverage varies around $90\%$ across all sample sizes for the RT and the SY interval. In contrast, the EP and the AS interval fall short of the nominal $90\%$ by $10.43$ and $4.15$ percentage points (pp), respectively, for small sample size ($n=250$). Nevertheless, their average coverage approaches the nominal value as the sample size increases. Remarkably, for all four intervals the average rate of the conditional VaR being below the interval is considerably less than the average rate of the conditional VaR being above the interval when the sample size is rather small ($n \leq 500$). Regarding the intervals' length, we observe that the SY interval is on average larger than the EP/RT interval. As the sample size increases this gap diminishes and the intervals' average lengths shrink. Considering the low persistent case (left part of Table \ref{tab:4.1}) we find similar results regarding the intervals' average coverage, yet their average lengths turn out to be smaller compared to the high persistent case. This is intuitive as the conditional volatility tends to vary less in the low persistent case. Regarding the T-GARCH($1,1$) in Table \ref{tab:4.1a}, the overall picture is similar as in the GARCH(1,1) case, however the under-coverage in small and medium-sized samples appears to be more extreme for the EP and reduced for the AS interval.

Simulation results for the scenario when the $\eta_t$'s follow a standard normal distribution and when the model is a GARCH(1,1) and a T-GARCH(1,1), respectively, are tabulated in Tables \ref{tab:4.2} and \ref{tab:4.2a} which are given in Appendix \ref{sec add simulation}. Here we only note that, although the error distribution underlying the QMLE is correctly specified in this case, the qualitative results stated above with regard to Table \ref{tab:4.1} persist: the RT and the SY intervals possess accurate coverage rates across sample sizes, whereas the EP and the AS interval exhibit under-coverage in samples of rather small size with different extent. 
Moreover, we observe that the intervals are on average shorter in the Gaussian case than in the baseline case. This seems partially driven by a smaller variance of $\hat{\xi}_{n,\alpha}$; for $\alpha=0.05$ the asymptotic variance $\zeta_\alpha$ in \eqref{eq:4.3.6} is equal to $3.11$ in the Gaussian case compared to $5.72$ in the Student-t case with $6$ degrees of freedom. 

While the small-sample-performance of the AS interval can be explained by its embodied density estimation, the question arises why the EP interval performs worse than the other bootstrap intervals, which seems counter-intuitive at first.  Howbeit the results are in line with the theoretical findings of \citeauthor{falk1991coverage} (\citeyear{falk1991coverage}, unnumbered Corollary, p.\ 488). In a random sample setting they prove that 
the RT bootstrap interval for quantiles has asymptotically greater coverage than the corresponding EP bootstrap interval. The emerging gap\footnote{We neglect their $o(n^{-1/2})$ term. Take note that the theoretical results of \cite{falk1991coverage} are not directly applicable in our setting due to GARCH-type effects.}
\begin{enumerate}
 \item[(i)] tends to be smaller for larger sample sizes,
 \item[(ii)] tends to be larger for more extreme quantiles, and 
 \item[(iii)] tends to vary with the nominal coverage rate in a non-monotonic way.
\end{enumerate}
Table \ref{tab:4.5} presents the average coverage gap between the EP and the RT bootstrap interval of the conditional VaR for the baseline specification as well as for  three deviations from the baseline (change in $F$, $\alpha$ and $\gamma$). For example, in the low persistence GARCH($1,1$) case of the baseline with $n=250$, the average coverage gap amounts to $90.73\%-81.20\%=9.53$pp (see also Table \ref{tab:4.1}). It is striking that all values are positive within Table \ref{tab:4.5}, which highlights the superiority of the RT bootstrap interval over the EP bootstrap interval. Further, it is eminent that average coverage gap tends to decrease with increasing sample size, which supports (i). Comparing columns (1) and (3) we also find that the average coverage gap tends to be larger for the $1\%$--VaR than for the $5\%$--VaR, which gives rise to (ii). Regarding (iii), the result of \citeauthor{falk1991coverage} (\citeyear{falk1991coverage}) suggests that the gap slightly decreases when increasing the nominal coverage from $90\%$ to $95\%$. Such tendency is precisely observed when comparing columns (1) and (4) of Table \ref{tab:4.5}.

\begin{table}[h]
\caption{\textbf{Average gap} between the \textbf{RT and the EP} fixed-design bootstrap intervals for different settings}\label{tab:4.5}
\centering
\begin{tabular}{rcccccccc}
\hline \hline
\begin{tabular}[c]{@{}c@{}}Sample\\ size\end{tabular} & 
(1) & (2) & (3) & (4) & (1) & (2) & (3) & (4) \\ 
\hline
\\
& \multicolumn{8}{c}{Panel I: GARCH($1,1$)} \\
& \multicolumn{4}{c}{low persistence} & \multicolumn{4}{c}{high persistence} \\
250 & 9.53 & 9.07 & 13.26 & 8.41 & 10.65 & 9.63 & 14.79 & 9.30 \\
500 & 6.50 & 6.86 & 10.09 & 5.55 & 6.75 & 6.88 & 9.87 & 5.68 \\
1,000 & 5.03 & 5.02 & 8.87 & 4.38 & 4.68 & 4.85 & 8.60 & 4.24 \\
5,000 & 2.84 & 2.67 & 5.67 & 2.32 & 1.80 & 1.48 & 5.07 & 1.74 \\
& \multicolumn{8}{c}{Panel II: T-GARCH($1,1$)} \\
& \multicolumn{4}{c}{low persistence} & \multicolumn{4}{c}{high persistence} \\
250 & 10.48 & 9.39 & 15.22 & 9.07 & 11.41 & 10.41 & 15.03 & 10.23 \\
500 & 7.72 & 6.11 & 10.84 & 6.73 & 8.71 & 7.94 & 11.55 & 7.86 \\
1,000 & 5.18 & 4.23 & 8.55 & 4.40 & 6.83 & 5.54 & 9.20 & 5.82 \\
5,000 & 2.35 & 1.63 & 4.63 & 1.89 & 2.21 & 2.02 & 5.10 & 1.73 \\
\hline \hline
\end{tabular}
\vspace{0.15cm}
\caption*{Table \ref{tab:4.5} reports the average coverage gap between the RT and the EP fixed-design bootstrap interval in percentage points for different settings and sample sizes. For varying sample sizes ($n$) Panel I presents the results for the low and high persistence parameterization of a GARCH($1,1$), whereas Panel II displays the results for the corresponding T-GARCH($1,1$) processes.\\ 
($1$)  $5\%$--VaR, Student-t innovations and $90\%$ nominal coverage (baseline)\\
($2$)  $5\%$--VaR, Gaussian innovations and $90\%$ nominal coverage\\
($3$) $1\%$--VaR, Student-t innovations and $90\%$ nominal coverage\\
($4$)  $5\%$--VaR, Student-t innovations and $95\%$ nominal coverage}
\end{table}

\bgroup
\def\arraystretch{0.95}
\begin{table}[tbp]
\caption{\textbf{Recursive-design} bootstrap confidence intervals for \textbf{GARCH(1,1)} with Student-t innovations}\label{tab:4.6}  
\centering
\resizebox{\textwidth}{!}{\begin{tabular}{rc:ccc:ccc}
\hline \hline
\multicolumn{1}{c}{\begin{tabular}[c]{@{}c@{}}Sample\\ Size\end{tabular}} &                         & \begin{tabular}[c]{@{}c@{}}Average \\ coverage\end{tabular} & \begin{tabular}[c]{@{}c@{}}Av. coverage\\ below/above\end{tabular} & \begin{tabular}[c]{@{}c@{}}Average\\ length\end{tabular} & \begin{tabular}[c]{@{}c@{}}Average\\ coverage\end{tabular} & \begin{tabular}[c]{@{}c@{}}Av. coverage\\ below/above\end{tabular} & \begin{tabular}[c]{@{}c@{}}Average\\ length\end{tabular} \\ \hline
\multicolumn{1}{c}{} & & & & & & & \\[-9pt]
\multicolumn{1}{c}{} & & \multicolumn{3}{c}{low persistence} & \multicolumn{3}{c}{high persistence} \\
250 & EP &
81.65 & 5.99/12.36 & 0.619 & 80.54 & 5.87/13.59 & 0.859 \\
& RT & 
90.49 & 3.78/5.73 & 0.619 & 90.09 & 4.08/5.83 & 0.859 \\
& SY & 
90.20 & 2.78/7.02 & 0.653 & 90.49 & 2.93/6.58 & 0.918 \\
\hdashline
500 & EP & 
84.66 & 5.55/9.79 & 0.441 & 84.50 & 5.38/10.12 & 0.604 \\
& RT & 
90.39 & 4.29/5.32 & 0.441 & 90.25 & 4.60/5.15 & 0.604 \\
& SY & 
90.62 & 3.36/6.02 & 0.459 & 90.88 & 3.27/5.85 & 0.628 \\
\hdashline
1,000 & EP &
85.78 & 5.40/8.82 & 0.310 & 86.19 & 5.08/8.73 & 0.427 \\
& RT & 
90.24 & 4.47/5.29 & 0.310 & 90.33 & 4.28/5.39 & 0.427 \\
& SY & 
90.30 & 3.76/5.94 & 0.318 & 90.69 & 3.57/5.74 & 0.438 \\
\hdashline
5,000 & EP &
87.68 & 5.60/6.72 & 0.144 & 87.75 & 5.35/6.90 & 0.191 \\
& RT & 
90.28 & 4.72/5.00 & 0.144 & 89.96 & 4.84/5.20 & 0.191 \\
& SY & 
90.14 & 4.49/5.37 & 0.146 & 89.69 & 4.59/5.72 & 0.193 \\
\hline \hline
\end{tabular}
}
\vspace{0.15cm}
\caption*{Table \ref{tab:4.6} reports distinct features of the \textbf{recursive-design} bootstrap confidence intervals for the conditional VaR at \textbf{level} $\bm{\alpha=0.05}$ with \textbf{nominal coverage} $\bm{1-\gamma=90\%}$. For each interval type and different sample sizes ($n$), the interval's average coverage rates (in $\%$), the average rate of the conditional VaR being below/above the interval (in $\%$) and the interval's average length are tabulated. The bootstrap intervals are based on $B=999$ bootstrap replications and the averages are computed using $S=10{,}000$ simulations. The results are for the low (left part) and the high (right part) persistence parametrization of a \textbf{GARCH(1,1)} with Student-$t$ innovations ($6$ df). 
}
\end{table}
 \egroup

\bgroup
\def\arraystretch{0.95}
\begin{table}[tbp]
\caption{\textbf{Recursive-design} bootstrap confidence intervals for \textbf{T-GARCH(1,1)} with Student-t innovations}\label{tab:4.6a}
\centering
\resizebox{\textwidth}{!}{\begin{tabular}{rc:ccc:ccc}
\hline \hline
\multicolumn{1}{c}{\begin{tabular}[c]{@{}c@{}}Sample\\ Size\end{tabular}} &                         & \begin{tabular}[c]{@{}c@{}}Average \\ coverage\end{tabular} & \begin{tabular}[c]{@{}c@{}}Av. coverage\\ below/above\end{tabular} & \begin{tabular}[c]{@{}c@{}}Average\\ length\end{tabular} & \begin{tabular}[c]{@{}c@{}}Average\\ coverage\end{tabular} & \begin{tabular}[c]{@{}c@{}}Av. coverage\\ below/above\end{tabular} & \begin{tabular}[c]{@{}c@{}}Average\\ length\end{tabular} \\ \hline
\multicolumn{1}{c}{} & & & & & & & \\[-9pt]
& \multicolumn{1}{l}{} & \multicolumn{3}{c}{low persistence} & \multicolumn{3}{c}{high persistence} \\
250 & EP &
79.62 & 6.79/13.59 & 0.143 & 79.12 & 7.35/13.53 & 0.294 \\
& RT & 
90.54 & 3.49/5.97 & 0.143 & 90.98 & 2.67/6.35 & 0.294 \\
& SY & 
89.31 & 3.33/7.36 & 0.150 & 89.17 & 3.10/7.73 & 0.309 \\
\hdashline
500 & EP & 
82.33 & 6.25/11.42 & 0.105 & 82.28 & 6.56/11.16 & 0.214 \\
& RT & 
90.27 & 4.26/5.47 & 0.105 & 91.29 & 3.37/5.34 & 0.214 \\
& SY & 
89.39 & 3.97/6.64 & 0.109 & 89.90 & 3.36/6.74 & 0.223 \\
\hdashline
1,000 & EP & 
85.22 & 5.81/8.97 & 0.077 & 83.79 & 6.58/9.63 & 0.157 \\
& RT & 
89.97 & 4.88/5.15 & 0.077 & 90.29 & 4.37/5.34 & 0.157 \\
& SY & 
90.09 & 4.12/5.79 & 0.079 & 89.69 & 4.13/6.18 & 0.162 \\
\hdashline
5,000 & EP & 
88.53 & 4.82/6.65 & 0.036 & 88.10 & 4.93/6.97 & 0.075 \\
& RT & 
90.40 & 4.95/4.65 & 0.036 & 90.12 & 4.99/4.89 & 0.075 \\
& SY & 
90.91 & 4.23/4.86 & 0.037 & 90.61 & 4.13/5.26 & 0.076 \\
\hline \hline
\end{tabular}
}
\vspace{0.15cm}
\caption*{Table \ref{tab:4.6a} is exactly as Table \ref{tab:4.6} but  for a \textbf{T-GARCH(1,1)} with Student-$t$ innovations ($6$ df).
}
\end{table}
 \egroup

With regard to Remark \ref{rem:4.2} in Section \ref{sec:4.4.1}, Tables \ref{tab:4.6} and \ref{tab:4.6a} report the simulation results for the recursive-design bootstrap for the DGPs of Tables \ref{tab:4.1} and \ref{tab:4.1a}, respectively. We refer to Appendix \ref{app:4.B} for computational details. In comparison to the fixed-design approach (see Tables \ref{tab:4.1} and \ref{tab:4.1a}) we find that the recursive-design method performs similarly in terms of average coverage for each interval type, which corresponds to the simulation results of \cite{cavaliere2018fixed}. It is striking, however, that the intervals' average lengths are larger in the recursive-design than in the fixed-design set-up. For example, in the high persistence GARCH(1,1) case (right part of Table \ref{tab:4.6}) for $n=500$ the average length in the recursive-design approach is $0.604$ for the EP/RT interval compared to $0.581$ in the fixed-design. As the sample size increases this difference disappears. 

In summary, the simulations suggest that the RT and the SY bootstrap interval work well for both bootstrap designs and that they outperform in smaller samples the AS interval in terms of average coverage even though their tails are unequally represented. In contrast, for both bootstrap designs the EP interval falls short of its nominal coverage, which is in line with the theoretical findings of \cite{falk1991coverage}. Since the fixed RT method leads on average to shorter intervals than the corresponding SY method and its recursive-design counterpart, this suggests to favor the fixed-design RT bootstrap interval in \eqref{eq:4.4.9}.

\subsection{Empirical Application}
\label{sec:4.5.2}

We analyze the French stock market index CAC 40 for the period  January 1, 2015 -- January 1, 2020. The index values for the period are retrieved from Yahoo Finance and daily (log-) returns (expressed in $\%$) are computed using $\epsilon_t = 100\log(p_t/p_{t-1})$, where $p_t$ denotes the closing value of the index at trading day $t$. 
\begin{figure}[tbp]
\centering
\begin{subfigure}[b]{0.45\textwidth}
\label{fig:2a}
	\centering
	\includegraphics[width=\textwidth]{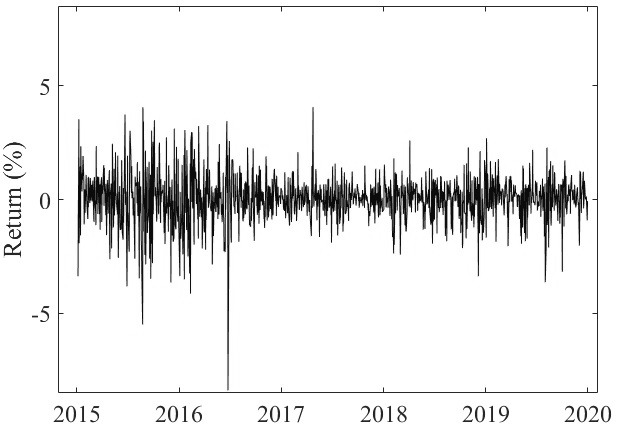}
            \caption{Returns of CAC 40 }    
\end{subfigure}
\quad
\begin{subfigure}[b]{0.45\textwidth}  
\label{fig:2b}
	\centering 
	\includegraphics[width=\textwidth]{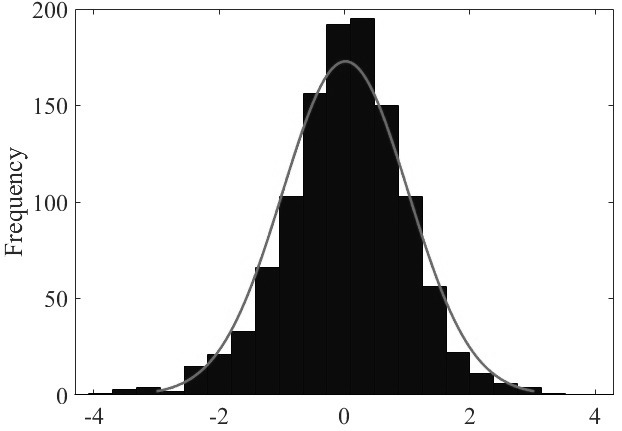}
	\caption{Histogram of the residuals $\hat{\eta}_t$'s} 
\end{subfigure}
\caption{The returns of the French stock market index CAC 40 are plotted in (a) for the period January 1, 2015 -- January 1, 2020. The histogram of the residuals is plotted in (b) after fitting a T-GARCH($1,1$) model to the subperiod January 1, 2015 -- July 1, 2019. A scaled normal density is superimposed.} 
\label{fig:4.2}
\end{figure}
Figure \ref{fig:4.2}(a) displays the resulting series of returns. We disregard the observations  from July 1, 2019 onwards, which we leave for the out-of-sample evaluation, yielding $n=1{,}146$ remaining observations (i.e.\ January 1, 2015 - July 1, 2019). For the volatility process we consider the T-GARCH($1,1$) model specified in Example \ref{ex:4.2}.\footnote{We also
consider an Asymmetric Power GARCH model \citep{ding1993long}, i.e.\ $\sigma_{t+1}^\delta= \omega_0 + \alpha_0^+ (\epsilon_{t}^+)^\delta +\alpha_0^-(\epsilon_{t}^-)^\delta + \beta_0 \sigma_{t}^\delta$ with $\delta>0$, which nests the GARCH($1,1$) model ($\delta=2$, $\alpha_0^+=\alpha_0^-$) and the T-GARCH($1,1$) model ($\delta=1$) of Examples \ref{ex:4.1} and \ref{ex:4.2}. In practice, the impact of the power $\delta$ on the volatility is minor and the QML approach of \cite{hamadeh2011asymptotic} suggests a $\delta$ close to $1$ in favor for the T-GARCH specification.} 
Table \ref{tab:4.7} reports the corresponding point estimates with standard errors obtained by bootstrapping based on Algorithm \ref{alg:4.1}. As documented in numerous studies we find that the volatility persistence is close to unity. Further, we observe that $\hat{\alpha}^-_n$ is considerably larger than $\hat{\alpha}^+_n$ indicating a strong leverage effect, i.e.\  negative returns tend to increase volatility by more than positive returns of the same magnitude.
\begin{table}[tbp]
\caption{T-GARCH($1,1$) estimates for CAC 40}\label{tab:4.7}
\centering
\begin{tabular}{lcccc}
\hline \hline
                     & $\hat{\omega}_n$ & $\hat{\alpha}^+_n$ & $\hat{\alpha}^-_n$ & $\hat{\beta}_n$ \\ \hline
point estimate       & $0.0292$         & $0.0046$           & $0.1798$           & $0.9026$        \\
std. error &  $0.0109$                &  $0.0215$                  & $0.0339$                   &  $0.0234$               \\ \hline \hline
\end{tabular}
\vspace{0.15cm}
\caption*{T-GARCH($1,1$) estimates for the subperiod January 1, 1998 -- December 31, 2017. The standard errors are obtained by applying the fixed-design residual bootstrap with $B=2{,}000$ bootstrap replications.}
\end{table}
Figure \ref{fig:4.2}(b) plots the histogram of the residuals with the normal distribution superimposed. Further, we test the condition that the innovations are iid (see Assumption \ref{as:4.5}\eqref{as:4.5.1}) with the generalized run tests of \cite{cho2011generalized}.\footnote{The implementation of the tests is available on the website of the first author.} These tests are particularly suitable in this case since they can be based on the residuals and are sensitive against a wide range of alternatives. The test statistic of the sup-norm based test is $0.40$, which corresponds to a p-value of $0.27$. consequently, one cannot reject the null hypothesis of iid innovations at any common significance level. Similarly, the generalized run test based on the $L_1$-norm cannot be rejected at a $10\%$ significance level.

Next, we perform a rolling window analysis starting with subperiod January 1, 2015 -- July 1, 2019 and ending with subperiod July 8, 2015 -- January 1, 2020. We have $130$ subperiods each consisting of $1{,}146$ observations. For each rolling window period we fit a T-GARCH($1,1$) model and estimate the one-period-ahead conditional VaR associated with level $\alpha=0.05$. 
For example, for the first window the T-GARCH($1,1$) estimates are reported in Table \ref{tab:4.7} and the conditional $5\%$-VaR of the one-period ahead (i.e.\ July 1, 2019) is estimated by $1.11$.
Further, we obtain the associated $95\%$-confidence intervals based on bootstrap and asymptotic normality. In addition to the RT intervals of the fixed- and residual-design bootstrap, we also computed an interval based on the asymptotic distribution. 
The corresponding intervals are $[0.850,1.136]$ (fixed-design), $[0.834,1.115]$ (recursive-design), and $[0.828,1.106]$ (asymp.\ normality). Although the intervals are fairly similar, the asymptotic and recursive bootstrap intervals are shorter than the fixed-design interval. Given its tendency to underestimate variability in finite samples, this result is unsurprising for the asymptotic interval, although for the recursive bootstrap this contrasts the simulation findings. Note that the fixed-design iid and block bootstraps produce very similar interval, which is not surprising as our conducted specification tests did not indicate any violation of the iid assumption on the innovations.

The results of the rolling window analysis are visualized in Figure \ref{fig:4.3}. It plots the realized return together with (the opposite of) the estimated conditional VaR. For clarity we only indicate the lower and upper bound of the $95\%$ RT fixed-design bootstrap interval.
\begin{figure}[tbp]
\centering
\includegraphics[width=\textwidth]{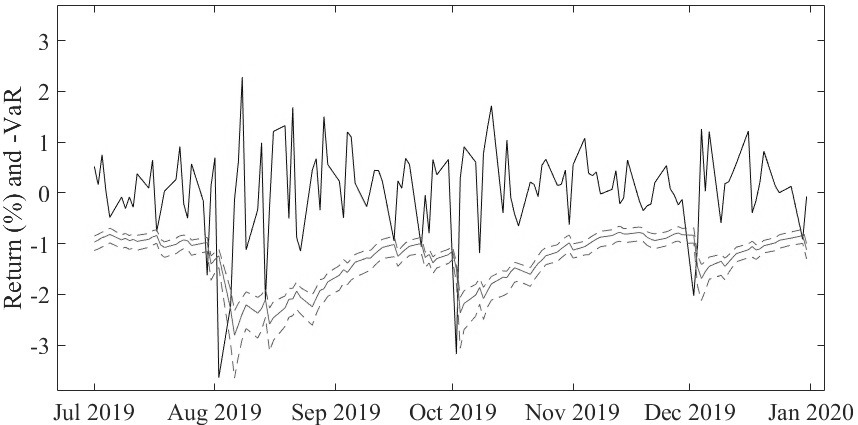}
\caption{Returns and the estimated conditional VaR (solid) for the period June 2, 2019 -- December 31, 2019. The estimation rests on the $1{,}146$ preceding observations. Lower and upper bounds for the conditional VaR (dashed) are based on the fixed-design bootstrap scheme using the RT method with $1-\gamma = 95\%$.} 
\label{fig:4.3}
\end{figure}
We observe that in more turbulent times (e.g.\ August, 2019), the estimated VaR amplifies. In such volatile periods we expect the estimation risk to increase and, accordingly, we find wider bootstrap confidence intervals. In this regard, although not directly related to the bootstrap procedures considered it is worth mentioning that a proposal for monitoring VaR estimates over time can be found in \citet{HogaDemetrescu}. 
\begin{remark}
\label{rem:5.1}
The point estimate $\hat{\alpha}^+_n$ in Table \ref{tab:4.7} is close to zero, indicating that the parameter $\alpha^+_0$ may lie on the boundary. In view of Assumption \ref{as:4.6}, the proposed bootstrap method becomes invalid when a parameter lies on the boundary of the parameter space. \cite{cavaliere2020bootstrap} modify the fixed volatility bootstrap to account for nuisance parameters on the boundary by shrinking the estimates for bootstrap sampling toward the boundary at an appropriate rate. Following their suggestion and using instead $\theta^*_n=(\hat{\omega}_n,0,\hat{\alpha}^-_n,\hat{\beta}_n)'$ to generate the bootstrap sample, does not notably alter the bootstrap results presented in this subsection.
\end{remark}

\section{Concluding Remarks}
\label{sec:4.6}

In this paper we study the two-step estimation procedure of \cite{francq2015risk} associated with the conditional VaR. In the first step, the conditional volatility parameters are estimated by QMLE, while the second step corresponds to  approximating the quantile of the innovations' distribution  by the empirical quantile of the residuals. A fixed-design residual bootstrap method is proposed to mimic the finite sample distribution of the two-step estimator and its consistency is proven under mild assumptions. In addition, an algorithm is provided for the construction of bootstrap intervals for the conditional VaR to take into account the uncertainty induced by estimation. Three interval types are suggested and a large-scale simulation study is conducted to investigate their performance in finite samples. We find that the equal-tailed percentile interval based on the fixed-design residual bootstrap tends to fall short of its nominal value, whereas the corresponding interval based on reversed tails yields accurate average coverage combined with the shortest average length. Although the result seems counter-intuitive at first, it is in line with the theoretical findings of \cite{falk1991coverage}. In the simulation study we also consider the recursive-design residual bootstrap. It turns out that the recursive-design and the fixed-design bootstrap perform similar in terms of average coverage. Yet in smaller samples the fixed-design scheme leads on average to shorter intervals. Further, the interval estimation by means of the fixed-design residual bootstrap is illustrated in an empirical application to daily returns of the French stock index CAC 40.

Natural extensions of this work are encompassing other risk measures such as Expected Shortfall \citep{heinemann2018residual} and developing a bootstrap procedure for the one-step estimator of \cite{francq2015risk}. Further, it is worthwhile to consider a smoothed bootstrap version in the spirit of \cite{hall1989smoothing}, which offers potential gains in accuracy. The latter two extensions are left for future research, as is the question how to extend the fixed-design bootstrap in order to give valid bootstrap inference when the iid assumption does not hold.

\section*{Acknowledgements}
The authors thank Franz Palm, Hanno Reuvers, Jean-Michel Zako\"ian and Christian Francq for useful comments and suggestions
as well as Dewi Peerlings and Benoit Duvocelle for computational support during the revision stage. In addition, the authors are grateful to the editors Oliver Linton and Torben Andersen, one of the associate editors and to two anonymous referees for their constructive remarks.

This research was financially supported by the Netherlands Organisation for Scientific Research (NWO).


\appendix

\section{Proofs of Main Results} \label{app:1}
\begin{proof}[Proof of Theorem \ref{thm:4.3}]
Following the steps of \citeauthor{francq2015risk} (\citeyear{francq2015risk}, Theorem 4), we standardize equation \eqref{eq:4.4.2} such that the bootstrap quantile estimator satisfies
\begin{align*}
\sqrt{n}(\hat{\xi}_{n,\alpha}^*-\hat{\xi}_{n,\alpha}) 
=&\arg\min_{z \in \R} \underbrace{\sum_{t=1}^n\rho_\alpha\Big(\hat{\eta}_t^*-\hat{\xi}_{n,\alpha}-\frac{z}{\sqrt{n}}\Big)-\sum_{t=1}^n\rho_\alpha(\eta_t^*-\hat{\xi}_{n,\alpha})}_{Q_n^*(z)}.
\end{align*}
Employing the identity of \citeauthor{koenker2006quantile} (\citeyear{koenker2006quantile}, Eq.\ (A.3)) we obtain\footnote{Note that the identity  holds not only for $u \neq 0$ but also for $u=0$.}
\begin{align}
\label{eq:4.4.5}
Q_n^*(z)=&zX_n^*+Y_n^*+I_n^*(z)+J_{n,1}^*(z)+J_{n,2}^*(z)
\end{align}
with 
\begin{align}
\nonumber
X_n^*&=\frac{1}{\sqrt{n}}\sum_{t=1}^n\big(\mathbbm{1}_{\{\eta_t^*< \hat{\xi}_{n,\alpha}\}}-\alpha\big),\\
\nonumber
Y_n^*&=\sum_{t=1}^n \big(\eta_t^*-\hat{\eta}_t^*\big)\big(\mathbbm{1}_{\{\eta_t^*< \hat{\xi}_{n,\alpha}\}}-\alpha\big),\\
\label{eq:219849862}
I_n^*(z)&=\sum_{t=1}^n \int_0^{\frac{z}{\sqrt{n}}}\big(\mathbbm{1}_{\{\eta_t^*\leq \hat{\xi}_{n,\alpha}+s\}}-\mathbbm{1}_{\{\eta_t^*< \hat{\xi}_{n,\alpha}\}}\big)ds,\\
\label{eq:4.4.6}
J_{n,1}^*(z)=&\sum_{t=1}^n \int_{0}^{\eta_t^*-\hat{\eta}_t^*}\big(\mathbbm{1}_{\{\eta_t^*\leq\hat{\xi}_{n,\alpha}+\frac{z}{\sqrt{n}}+ s\}}-\mathbbm{1}_{\{\eta_t^*-\hat{\xi}_{n,\alpha}-z/\sqrt{n}<0\}}\big)ds\\
\label{eq:4.4.7}
J_{n,2}^*(z)=&\sum_{t=1}^n \big(\eta_t^*-\hat{\eta}_t^*\big)\big(\mathbbm{1}_{\{\eta_t^*<\hat{\xi}_{n,\alpha}+\frac{z}{\sqrt{n}}\}}-\mathbbm{1}_{\{\eta_t^*<\hat{\xi}_{n,\alpha}\}}\big).
\end{align}
Lemma \ref{lem:4.7} yields $X_n^*\overset{d^*}{\to}N\big(0,\alpha(1-\alpha)\big)$ almost surely. Further, $Y_n^*$ neither depends on $z$ nor interacts with it; therefore it can be disregarded. The term $I_n^*(z)$ converges in conditional probability to $\frac{z^2}{2}f(\xi_\alpha)$ in probability by Lemma \ref{lem:4.8}.
Lemma \ref{lem:4.9} shows that $J_{n,1}^*(z)$ converges in conditional distribution to a random variable, which does not depend on $z$, in probability.  Moreover, we have $J_{n,2}^*(z)=z\xi_\alpha f(\xi_\alpha)\Omega'\sqrt{n}\big(\hat{\theta}_n^*-\hat{\theta}_n\big)+o_{p^*}(1)$ in probability by Lemma \ref{lem:4.10}. Thus, we obtain
\begin{align*}
Q_n^*(z) &= \frac{z^2}{2}f(\xi_\alpha)+z\Big(X_n^*+\xi_\alpha f(\xi_\alpha)\Omega'\sqrt{n}\big(\hat{\theta}_n^*-\hat{\theta}_n\big)\Big)+J_{n,1}^*(z)+Y_n^*+o_{p^*}(1)
\end{align*}
in probability. Employing \citeauthor{xiong2008some} (\citeyear{xiong2008some}, Theorem 3.3) and the basic corollary of \cite{hjort2011asymptotics}, we obtain\footnote{Matching notation, we take $A_n(z)=Q_n^*(z)$, which is convex, and set $B_n(z)=\frac{z^2}{2}V+zU_n+C_n$, where $V=f(\xi_\alpha)$, $U_n =X_n^*+\xi_\alpha f(\xi_\alpha)\Omega'\sqrt{n}\big(\hat{\theta}_n^*-\hat{\theta}_n\big)$ and $C_n+r_n(z) = J_{n,1}^*(z)+Y_n^*+o_{p^*}(1)$ with $r_n(z)\overset{p}{\to}0$ for each $z\in \R$. The minimizers of $A_n(z)$ and $B_n(z)$ are $\alpha_n=\sqrt{n}(\hat{\xi}_{n,\alpha}-\hat{\xi}_{n,\alpha}^*)$ and $\beta_n=-V^{-1}U_n$, respectively. The basic corollary of \cite{hjort2011asymptotics} states $\alpha_n-\beta_n=o_p(1)$, which implies $\alpha_n-\beta_n=o_{p^*}(1)$ in probability (\citeauthor{xiong2008some}, \citeyear{xiong2008some}, Theorem 3.3).}
\begin{align*}
 \sqrt{n}(\hat{\xi}_{n,\alpha}-\hat{\xi}_{n,\alpha}^*)= \xi_\alpha \Omega'\sqrt{n}\big(\hat{\theta}_n^*-\hat{\theta}_n\big)+\frac{1}{f(\xi_\alpha)}\frac{1}{\sqrt{n}}\sum_{t=1}^n(\mathbbm{1}_{\{\eta_t^*< \hat{\xi}_{n,\alpha}\}}-\alpha)+o_{p^*}(1)
\end{align*}
in probability. Together with \eqref{eq:4.4.4} we have 
\begin{align*}
      \begin{pmatrix}
      \sqrt{n}(\hat{\theta}_n^*-\hat{\theta}_n)\\
    \sqrt{n}(\hat{\xi}_{n,\alpha} - \hat{\xi}_{n,\alpha}^*)
      \end{pmatrix}
      = 
      \begin{pmatrix}
      \frac{1}{2}J^{-1}&O_{r\times 1}\\
    \frac{1}{2}\xi_\alpha\Omega'J^{-1} & \frac{1}{f(\xi_\alpha)}
      \end{pmatrix}
        \begin{pmatrix}
      \frac{1}{\sqrt{n}}\sum_{t=1}^n \hat{D}_t\big(\eta_t^{*2}-1\big)\\
    \frac{1}{\sqrt{n}}\sum_{t=1}^n(\mathbbm{1}_{\{\eta_t^*< \hat{\xi}_{n,\alpha}\}}-\alpha)
      \end{pmatrix}+o_{p^*}(1).
\end{align*}
Employing Lemma \ref{lem:4.7} completes the proof.
\end{proof}

\begin{proof}[Proof of Corollary \ref{cor:4.1}] 
The proof is similar to \citeauthor{beutner2017justification}  (\citeyear{beutner2017justification}, proof of Corollary 4) and given for completeness. A Taylor expansion yields
\begin{align}
\sqrt{n}\big(\reallywidehat{VaR}_{n,\alpha}^{*}-\reallywidehat{VaR}_{n,\alpha}\big) = \underbrace{\begin{pmatrix}
    -\xi_\alpha \frac{\partial \sigma_{n+1}(\theta_0)}{\partial \theta}\\
    \sigma_{n+1}
      \end{pmatrix}'}_{w_n} \underbrace{\begin{pmatrix}
      \sqrt{n}(\hat{\theta}_n^*-\hat{\theta}_n)\\
    \sqrt{n}(\hat{\xi}_{n,\alpha} - \hat{\xi}_{n,\alpha}^*)
      \end{pmatrix}}_{Z_n^*} +R_n^*
\end{align}
with
\begin{align*}
R_n^* =&  \bigg(\xi_\alpha  \frac{\partial \sigma_{n+1}(\theta_0)}{\partial \theta'}-\hat{\xi}_{n,\alpha} \frac{\partial \tilde{\sigma}_{n+1}(\hat{\theta}_n)}{\partial \theta'}-\frac{1}{2}\bar{\xi}_{n,\alpha}(\hat{\theta}_n^*-\hat{\theta}_n)' \frac{\partial^2 \tilde{\sigma}_{n+1}(\bar{\theta}_n)}{\partial \theta \partial \theta'}\bigg) \sqrt{n}(\hat{\theta}_n^*-\hat{\theta}_n)
\\
&+\bigg(\tilde{\sigma}_{n+1}(\hat{\theta}_n)-\sigma_{n+1}(\theta_0)+ \frac{\partial \tilde{\sigma}_{n+1}(\bar{\theta}_n)}{\partial \theta'}(\hat{\theta}_n^*-\hat{\theta}_n)\bigg) \sqrt{n}(\hat{\xi}_{n,\alpha} - \hat{\xi}_{n,\alpha}^*),
\end{align*}
where $\bar{\theta}_n$ lies between $\hat{\theta}_n^*$ and $\hat{\theta}_n$ while $\bar{\xi}_{n,\alpha}$ lies between   $\hat{\xi}_{n,\alpha}^*$ and $\hat{\xi}_{n,\alpha}$. Note that $R_n^* =o_{p^*}(1)$ in probability, which can easily be shown using Theorems \ref{thm:4.1} and \ref{thm:4.3} together with Assumptions \ref{as:4.4} and \ref{as:4.9}. Further, let $Z\sim N(0,\Sigma_\alpha)$ be generated independently of $\{\epsilon_t, -\infty< t <\infty\}$ such that $w_n Z$ given $\mathcal{F}_n$ follows the conditional distribution in \eqref{eq:4.3.9}. Take $\varepsilon>0$ arbitrarily small and $K\geq 1$ sufficiently large such that with probability close to one $||w_n||\leq K$. In that case
\begin{align*}
&\sup\limits_{||g||_{BL}\leq 1}\Big|\EE^* \big[g(w_n Z_n^*+R_n^*)\big]- \EE_Z \big[g(w_n Z)|\mathcal{F}_n\big]\Big|\\
\leq&
 \sup\limits_{||g||_{BL}\leq 1}\EE^*\Big[ \big|g(w_n Z_n^*+R_n^*)-g(w_n Z_n^*)\big|\Big]\\
 &\qquad + K \sup\limits_{||g||_{BL}\leq 1} \Big|\EE^* \big[g(w_n Z_n^*)/K\big]- \EE_Z \big[g(w_n Z)/K|\mathcal{F}_n\big]\Big|\\
\leq&
 \sup\limits_{||g||_{BL}\leq 1}\EE^*\Big[ \big|g(w_n Z_n^*+R_n^*)-g(w_n Z_n^*)\big|\big(\mathbbm{1}_{\{|R_n^*|\leq \varepsilon\}}+\mathbbm{1}_{\{|R_n^*|> \varepsilon\}}\big)\Big]\\
 &\qquad + K \sup\limits_{||h||_{BL}\leq 1}\Big|\EE^* \big[h( Z_n^*)\big]- \EE_Z \big[h(Z)|\mathcal{F}_n\big]\Big|\\
\leq& \varepsilon+ 2  \:\EE^*\big[\mathbbm{1}_{\{|R_n^*|> \varepsilon\}}\big]+K \sup\limits_{||h||_{BL}\leq 1}\Big|\EE^* \big[h( Z_n^*)\big]- \EE_Z \big[h(Z)\big]\Big|,
\end{align*}
with $||g||_{BL}=\sup_x \big|g(x)\big|+\sup_{x \neq y}\frac{|g(x)-g(y)|}{||x-y||}$ being the bounded Lipschitz norm and $\EE_Z$ denoting the expectation operator corresponding to $Z$. Together with Theorem \ref{thm:4.3} and $R_n^* =o_{p^*}(1)$ in probability, we obtain
\begin{align*}
&\sup\limits_{||g||_{BL}\leq 1}\Big|\EE^* \big[g(w_n Z_n^*+R_n^*)\big]- \EE_Z \big[g(w_n Z)|\mathcal{F}_n\big]\Big|\overset{p}{\to}0,
\end{align*}
which completes the proof.
\end{proof}

\section{Auxiliary Results and Proofs}
\label{app:4.A}

In analogy to $D_t(\theta)$ and $\hat{D}_t$ we write 
 $H_t(\theta) =\frac{1}{\sigma_t(\theta)}\frac{\partial^2\sigma_t(\theta)}{\partial \theta \partial \theta'}$ and $\hat{H}_t = \tilde{H}_t(\hat{\theta}_n)$ with $\tilde{H}_t(\theta) = \frac{1}{\tilde{\sigma}_t(\theta)}\frac{\partial^2\tilde{\sigma}_t(\theta)}{\partial \theta \partial \theta}$. Further, we introduce $S_t =\sup_{\theta \in \mathscr{V}(\theta_0)}\frac{\sigma_t(\theta_0)}{\sigma_t(\theta)}$, $T_t =\sup_{\theta \in \mathscr{V}(\theta_0)}\frac{\sigma_t(\theta)}{\sigma_t(\theta_0)}$, $U_t = \sup_{\theta \in \mathscr{V}(\theta_0)}||D_t(\theta)||$ and $V_t = \sup_{\theta \in \mathscr{V}(\theta_0)}||H_t(\theta)||$ and stress that $\{S_t\}$, $\{T_t\}$, $\{U_t\}$ and $\{V_t\}$ are strictly stationary and ergodic processes (cf.\ \citeauthor{francq2011garch}, \citeyear{francq2011garch}, p.\ 182/405).

\subsection{Non-bootstrap Lemmas}
\label{app:4.A.1}

\begin{lemma}
\label{lem:4.1}
 Suppose Assumptions \ref{as:4.1}, \ref{as:4.2}, \ref{as:4.3}, \ref{as:4.4}(\ref{as:4.4.1}), \ref{as:4.5}(\ref{as:4.5.1}), \ref{as:4.6} and \ref{as:4.9}(i) hold with $a=-1$. Then, we have $\sup_{x \in \R}|\hat{\mathbbm{F}}_n(x)-F(x)|\overset{a.s.}{\to}0$.
\end{lemma}

\begin{proof}
The proof  follows \citeauthor{berkes2003limit} (\citeyear{berkes2003limit}, Theorem 2.1 \& Lemma 5.1) and consists of three parts.  First, we show that for any $\varepsilon>0$ there is a $\tau>0$ such that
\begin{align}
\label{eq:4.A.2}
\begin{split}
&\limsup_{n \to \infty} \sup_{\theta \in \mathscr{V}_\tau(\theta_0)} \bigg|\frac{1}{n}\sum_{t=1}^n \mathbbm{1}_{\{\eta_t\leq x\tilde{\sigma}_t(\theta)/\sigma_t(\theta_0)\}}-F(x)\bigg|\\ 
&\qquad \qquad \qquad \qquad \qquad \leq 2\Big(F\big(x+\varepsilon |x|\big)-F\big(x-\varepsilon|x|\big)\Big)
\end{split}
\end{align}
almost surely for any $x \in \R$, where $\mathscr{V}_\tau(\theta_0) = \big\{ \theta \in \Theta: ||\theta-\theta_0||\leq \tau\big\}$. In the second step, we show $\hat{\mathbbm{F}}_n(x)\overset{a.s.}{\to}F(x)$ for any $x \in \R$ using \eqref{eq:4.A.2} and thereafter prove $\sup_{x \in \R}|\hat{\mathbbm{F}}_n(x)-F(x)|\overset{a.s.}{\to}0$.

Let $\varepsilon>0$ and note that $\sigma_t \geq \underline{\omega}$ by Assumption \ref{as:4.3}. Together with  Assumption \ref{as:4.4}(\ref{as:4.4.1}), there exists a random variable $n_0$ such that $C_1 \rho^t/\sigma_t(\theta_0)\leq \varepsilon$ for all $t>n_0$. Then 
\begin{align*}
 \frac{1}{n}\sum_{t=1}^n \mathbbm{1}_{\{\eta_t\leq x\tilde{\sigma}_t(\theta)/\sigma_t(\theta_0)\}}\leq &\frac{1}{n}\sum_{t=1}^n \mathbbm{1}_{\{\eta_t\leq x\sigma_t(\theta)/\sigma_t(\theta_0)+|x|C_1\rho^t/\sigma_t(\theta_0)\}}\\
 \leq & \frac{n_0}{n}+\frac{1}{n}\sum_{t=1}^n \mathbbm{1}_{\{\eta_t\leq x\sigma_t(\theta)/\sigma_t(\theta_0)+\varepsilon|x|\}}
\end{align*}
holds almost surely. Let $\tau>0$ (to be specified); for any $\theta \in \mathscr{V}_\tau(\theta_0)$ we get
\begin{align*}
\frac{1}{n}\sum_{t=1}^n \mathbbm{1}_{\{\eta_t\leq x\sigma_t(\theta)/\sigma_t(\theta_0)+\varepsilon|x|\}}\leq \frac{1}{n}\sum_{t=1}^n \mathbbm{1}_{\{\eta_t\leq \sup_{\theta \in \mathscr{V}_\tau(\theta_0)}x \sigma_t(\theta)/\sigma_t(\theta_0)+\varepsilon|x|\}}
\end{align*}
almost surely. The uniform ergodic theorem for strictly stationary sequences (cf.\ \citeauthor{francq2011garch}, \citeyear{francq2011garch}, p.\ 181), henceforth called the uniform ergodic theorem, and Assumptions \ref{as:4.2}, \ref{as:4.3} and \ref{as:4.5}(\ref{as:4.5.1})  yield
\begin{align*}
 \frac{1}{n}\sum_{t=1}^n \mathbbm{1}_{\{\eta_t\leq  \sup_{\theta \in \mathscr{V}_\tau(\theta_0)}x \sigma_t(\theta)/\sigma_t(\theta_0)+\varepsilon|x|\}} \overset{a.s.}{\to}&\EE \mathbbm{1}_{\{\eta_t\leq  \sup_{\theta \in \mathscr{V}_\tau(\theta_0)}x \sigma_t(\theta)/\sigma_t(\theta_0)+\varepsilon|x|\}}\\
 =& \EE F\Big( \sup_{\theta \in \mathscr{V}_\tau(\theta_0)}x \sigma_t(\theta)/\sigma_t(\theta_0)+\varepsilon|x|\Big).
\end{align*}
Further, Assumptions \ref{as:4.3} and \ref{as:4.9}(i) with $a=-1$ imply $\lim_{\tau\to 0}\sup_{\theta \in \mathscr{V}_\tau(\theta_0)}x \sigma_t(\theta)/\sigma_t(\theta_0)=x$ almost surely. Thus, the dominated convergence theorem entails
\begin{align*}
\lim_{\tau\to 0} \EE F\Big(\sup_{\theta \in \mathscr{V}_\tau(\theta_0)}x \sigma_t(\theta)/\sigma_t(\theta_0)+\varepsilon|x|\Big)=F( x+\varepsilon|x|).
\end{align*}
Putting the results together, we get that for every $\varepsilon>0$, there is a $\tau>0$ such that
\begin{align*}
&\limsup_{n \to \infty} \sup_{\theta \in \mathscr{V}_\tau(\theta_0)} \frac{1}{n}\sum_{t=1}^n \mathbbm{1}_{\{\eta_t\leq x\tilde{\sigma}_t(\theta)/\sigma_t(\theta_0)\}}\leq F(x)+2\Big(F\big(x+\varepsilon|x|\big)-F(x)\Big)
\end{align*}
almost surely for any $x\in\R$. Similarly it can be shown that for every $\varepsilon>0$, there is a $\tau>0$ such that
\begin{align*}
&\liminf_{n \to \infty} \inf_{\theta \in \mathscr{V}_\tau(\theta_0)} \frac{1}{n}\sum_{t=1}^n \mathbbm{1}_{\{\eta_t\leq x\tilde{\sigma}_t(\theta)/\sigma_t(\theta_0)\}}\geq F(x)-2\Big(F(x)-F\big(x-\varepsilon|x|\big)\Big).
\end{align*}
almost surely for any $x\in\R$. Combining both results, we establish \eqref{eq:4.A.2}.

Next, we show $\hat{\mathbbm{F}}_n(x)\overset{a.s.}{\to}F(x)$ for any $x \in \R$. Let $\delta>0$; by continuity of $F$ (see Assumption \ref{as:4.5}(\ref{as:4.5.1})), there is a $\varepsilon>0$ such that $\big|F\big(x+\varepsilon|x|\big)-F\big(x-\varepsilon|x|\big)\big|<\delta/2$. Employing equation \eqref{eq:4.A.2}, there are $\tau>0$ and a random variable $n_1$ such that 
\begin{align*}
 \sup_{\theta \in \mathscr{V}_\tau(\theta_0)} \bigg|\frac{1}{n}\sum_{t=1}^n \mathbbm{1}_{\{\eta_t\leq x \tilde{\sigma}_t(\theta)/\sigma_t(\theta_0)\}}-F(x)\bigg|<\delta
\end{align*}
for all $n\geq n_1$. Since $\hat{\theta}_n \overset{a.s.}{\to}\theta_0$ by Theorem \ref{thm:4.1} there is a random variable $n_2$ such that $\hat{\theta}_n  \in \mathscr{V}_\tau(\theta_0)$ for all $n\geq n_2$. Thus, 
\begin{align*}
\big|\hat{\mathbbm{F}}_n(x)-F(x)\big|\leq  \sup_{\theta \in \mathscr{V}_\tau(\theta_0)} \bigg|\frac{1}{n}\sum_{t=1}^n \mathbbm{1}_{\{\eta_t\leq x \tilde{\sigma}_t(\theta)/\sigma_t(\theta_0)\}}-F(x)\bigg|<\delta
\end{align*}
for all $n\geq \max\{n_1,n_2\}$, which establishes $\hat{\mathbbm{F}}_n(x)\overset{a.s.}{\to}F(x)$ for any $x \in \R$. Using P\'olya's lemma (cf.\ \citeauthor{roussas1997course}, \citeyear{roussas1997course}, p.\ 206), we establish $\sup_{x \in \R}|\hat{\mathbbm{F}}_n(x)-F(x)|\overset{a.s.}{\to}0$ completing the proof.
\end{proof}

\begin{lemma} 
\label{lem:4.2}
Suppose Assumptions \ref{as:4.1}--\ref{as:4.3}, \ref{as:4.4}(\ref{as:4.4.1}) and \ref{as:4.5}(\ref{as:4.5.1}) hold.
\begin{enumerate}[(i)]
\item If in addition Assumptions \ref{as:4.4}(\ref{as:4.4.2}) and \ref{as:4.9}(ii) hold with $b=1$, then $\hat{\Omega}_n\overset{a.s.}{\to}\Omega$.

\item If in addition Assumptions \ref{as:4.4}(\ref{as:4.4.2}) and \ref{as:4.9}(ii) hold with $b=2$, then $\hat{J}_n\overset{a.s.}{\to}J$.

\item If in addition Assumptions \ref{as:4.4}(\ref{as:4.4.2}) and \ref{as:4.9}(iii) hold with $c=1$, then $\frac{1}{n}\sum\limits_{t=1}^n \hat{H}_t \overset{a.s.}{\to}\EE [H_t]$.

\item If in addition Assumptions \ref{as:4.5}(\ref{as:4.5.3}) and \ref{as:4.9}(i) hold with $a=4$, then we have $ \frac{1}{n}\sum\limits_{t=1}^n \hat{\eta}_t^m \mathbbm{1}_{\{l\leq \hat{\eta}_t<u\}}\overset{a.s.}{\to}\EE\big[\eta_t^m\mathbbm{1}_{\{l\leq \eta_t<u\}}\big]$ for $ m \in \{0,1,2,3,4\}$ and $l< u$.

\item If in addition Assumptions \ref{as:4.4} and \ref{as:4.9}(i)-(ii) hold with $a=\pm2$ and $b=4$, then
\begin{align*}
    \frac{1}{n}\sum\limits_{t=1}^n \mathbbm{1}_{\{ l\leq \sqrt{n}(\tilde{\psi}_t-1)<u\}}  \left(\sqrt{n}\big(\tilde{\psi}_t-1\big)\right)^m\overset{a.s.}{\to}\EE\Big[ \mathbbm{1}_{\{ l \leq D_t'(v_1-v_2)<u\}}\big(D_t'(v_1-v_2)\big)^m\Big]
\end{align*}
for $v_1,v_2\in \R^r$, $ m \in \{0,1,2,3,4\}$ and $l< u$
with $\tilde{\psi}_t= \frac{\tilde{\sigma}_t(\hat{\theta}_n+n^{-1/2}v_1)}{\tilde{\sigma}_t(\hat{\theta}_n+n^{-1/2}v_2)}$.
\end{enumerate}
\end{lemma}


\begin{proof}
Consider the first statement and expand 
\begin{align*}
&\frac{1}{n}\sum_{t=1}^n \hat{D}_t
= \underbrace{\frac{1}{n}\sum_{t=1}^n D_t(\hat{\theta}_n)}_{I}+\underbrace{\frac{1}{n}\sum_{t=1}^n \Big(\tilde{D}_t(\hat{\theta}_n)-D_t(\hat{\theta}_n)\Big)}_{II}.
\end{align*}
Focusing on $I$, we take $\varepsilon>0$ and let $e_1,\dots,e_r$ denote the unit vectors spanning $\R^r$. 
\textcolor{black}{
Since $D_t(\theta)$ is continuous in $\theta$ we can take $\mathscr{V}_\varepsilon(\theta_0)\subseteq \mathscr{V}(\theta_0)$ such that
}
\begin{align*}
\EE\big[e_i'D_t \big]-\varepsilon<\EE\Big[\inf_{\theta \in \mathscr{V}_\varepsilon(\theta_0)}e_i' D_t(\theta)\Big]\leq \EE\Big[\sup_{\theta \in \mathscr{V}_\varepsilon(\theta_0)}e_i' D_t(\theta)\Big]<\EE\big[e_i'D_t\big]+ \varepsilon
\end{align*}
for all $i = 1,\dots,r$. Since $\hat{\theta}_n\overset{a.s.}{\to}\theta_0$ (Theorem \ref{thm:4.1}), we have
 $\hat{\theta}_n \in \mathscr{V}_\varepsilon(\theta_0)$ almost surely. Together with the uniform ergodic theorem we obtain
\begin{align*}
\frac{1}{n}\sum_{t=1}^n e_i' D_t(\hat{\theta}_n) \overset{a.s.}{\leq}& \frac{1}{n}\sum_{t=1}^n \sup_{\theta \in \mathscr{V}_\varepsilon(\theta_0)} e_i' D_t(\theta)  \overset{a.s.}{\to} \EE\Big[\sup_{\theta \in \mathscr{V}_\varepsilon(\theta_0)}e_i' D_t(\theta)\Big]<\EE\big[e_i'D_t \big]+ \varepsilon\\
\frac{1}{n}\sum_{t=1}^n e_i' D_t(\hat{\theta}_n) \overset{a.s.}{\geq}& \frac{1}{n}\sum_{t=1}^n \inf_{\theta \in \mathscr{V}_\varepsilon(\theta_0)} e_i' D_t(\theta) \overset{a.s.}{\to} \EE\Big[\inf_{\theta \in \mathscr{V}_\varepsilon(\theta_0)}e_i' D_t(\theta)\Big]>\EE\big[e_i'D_t \big]- \varepsilon.
\end{align*}
Taking $\varepsilon \searrow 0$ establishes $\frac{1}{n}\sum_{t=1}^n e_i' D_t(\hat{\theta}_n) \overset{a.s.}{\to}\EE[e_i'D_t]$ for all $i$ yielding $I \overset{a.s.}{\to}\EE[D_t ] = \Omega$. Regarding $II$, we note that for each $\theta \in \Theta$, Assumption \ref{as:4.4} implies
\begin{align}
\label{eq:4.A.3}
\begin{split}
&\big|\big|\tilde{D}_t(\theta)-D_t(\theta)\big|\big| = \bigg|\bigg|\frac{1}{\tilde{\sigma}_t(\theta)}\frac{\partial \tilde{\sigma}_t(\theta)}{\partial \theta}-\frac{1}{\sigma_t(\theta)}\frac{\partial \sigma_t(\theta)}{\partial \theta}\bigg|\bigg|\\
=& \bigg|\bigg|\frac{1}{\tilde{\sigma}_t(\theta)}\bigg(\frac{\partial \tilde{\sigma}_t(\theta)}{\partial \theta}-\frac{\partial \sigma_t(\theta)}{\partial \theta}\bigg)+\frac{\sigma_t(\theta)-\tilde{\sigma}_t(\theta)}{\tilde{\sigma}_t(\theta)}\frac{1}{\sigma_t(\theta)}\frac{\partial \sigma_t(\theta)}{\partial \theta}\bigg|\bigg|\\ 
\leq& \frac{1}{\tilde{\sigma}_t(\theta)}\bigg|\bigg|\frac{\partial \tilde{\sigma}_t(\theta)}{\partial \theta}-\frac{\partial \sigma_t(\theta)}{\partial \theta}\bigg|\bigg|+ \frac{|\sigma_t(\theta)-\tilde{\sigma}_t(\theta)|}{\tilde{\sigma}_t(\theta)}\:\bigg|\bigg|\frac{1}{\sigma_t(\theta)}\frac{\partial \sigma_t(\theta)}{\partial \theta}\bigg|\bigg|\\
\leq& \frac{C_1 \rho^t}{\underline{\omega}}+ \frac{C_1 \rho^t}{\underline{\omega}}\big|\big|D_t(\theta)\big|\big|=\frac{C_1 \rho^t}{\underline{\omega}}\Big(1+ \big|\big|D_t(\theta)\big|\big|\Big).
\end{split}
\end{align}
We obtain
\begin{align*}
||II|| \leq& \frac{1}{n}\sum_{t=1}^n \big|\big|\tilde{D}_t(\hat{\theta}_n)-D_t(\hat{\theta}_n)\big|\big|\leq  \frac{C_1}{\underline{\omega}} \frac{1}{n}\sum_{t=1}^n \rho^{t}\Big(1+ \big|\big|D_t(\hat{\theta}_n)\big|\big|\Big)
\overset{a.s.}{\leq}  \frac{C_1}{\underline{\omega}} \frac{1}{n}\sum_{t=1}^n \rho^{t}(1+U_t).
\end{align*}
For each $\varepsilon>0$, Markov's inequality entails
\begin{align*}
&\sum_{t=1}^\infty \PP\Big[\rho^{t}(1+U_t)>\varepsilon\Big] \leq \sum_{t=1}^\infty \rho^t \frac{1+\EE[U_t]}{\varepsilon} 
= \frac{1+\EE[U_t]}{\varepsilon(1-\rho)}<\infty
\end{align*}
since $\rho \in(0,1)$ and $\EE[U_t]<\infty$ by Assumption \ref{as:4.9}(ii). The Borel-Cantelli lemma implies
\begin{align}
\label{eq:4.A.4}
0 = \PP\bigg[\lim_{t\to \infty}\bigcup_{s=t}^\infty \Big\{\rho^{s}(1+U_s)>\varepsilon\Big\}\bigg]\geq \PP\bigg[\lim_{t\to \infty}\rho^{t}(1+U_t)>\varepsilon\bigg]
\end{align}
and hence $\rho^{t}(1+U_t)\to 0$ almost surely. Ces\'aro's lemma yields $\frac{1}{n}\sum_{t=1}^n \rho^{t}(1+U_t)\overset{a.s.}{\to}0$ and hence $||II|| \overset{a.s.}{\to}0$, which validates the first statement.

Consider the second statement and expand
\begin{align*}
\frac{1}{n}\sum_{t=1}^n \hat{D}_t\hat{D}_t'
=&  \underbrace{\frac{1}{n}\sum_{t=1}^n D_t(\hat{\theta}_n)D_t'(\hat{\theta}_n)}_{III}+\underbrace{\frac{1}{n}\sum_{t=1}^n \Big(\tilde{D}_t(\hat{\theta}_n)\tilde{D}_t'(\hat{\theta}_n)-D_t(\hat{\theta}_n)D_t'(\hat{\theta}_n)\Big)}_{IV}.
\end{align*}
We focus on $III$ and let $\varepsilon>0$. Since $D_t(\theta)D_t(\theta)'$ is continuous in $\theta$ we can take $\mathscr{V}_\varepsilon(\theta_0)\subseteq \mathscr{V}(\theta_0)$ such that 
\begin{align*}
\EE\big[e_i'D_tD_t' e_j \big]-\varepsilon<&\EE\Big[\inf_{\theta \in \mathscr{V}_\varepsilon(\theta_0)}e_i' D_t(\theta)D_t'(\theta) e_j\Big]\\
\leq& \EE\Big[\sup_{\theta \in \mathscr{V}_\varepsilon(\theta_0)}e_i' D_t(\theta)D_t'(\theta) e_j\Big]<\EE\big[e_i'D_t D_t' e_j\big]+ \varepsilon
\end{align*}
for all $i,j = 1,\dots,r$. Since $\hat{\theta}_n\overset{a.s.}{\to}\theta_0$ by Theorem \ref{thm:4.1}, we have
 $\hat{\theta}_n \in \mathscr{V}_\varepsilon(\theta_0)$ almost surely. Together with the uniform ergodic theorem we obtain
\begin{align*}
\frac{1}{n}\sum_{t=1}^n e_i' D_t(\hat{\theta}_n) D_t'(\hat{\theta}_n) e_j \overset{a.s.}{\leq}& \frac{1}{n}\sum_{t=1}^n \sup_{\theta \in \mathscr{V}_\varepsilon(\theta_0)} e_i' D_t(\theta) D_t'(\theta) e_j\\
\overset{a.s.}{\to}& \EE\Big[\sup_{\theta \in \mathscr{V}_\varepsilon(\theta_0)}e_i' D_t(\theta) D_t'(\theta) e_j \Big]<\EE\big[e_i'D_t D_t' e_j \big]+ \varepsilon\\
\frac{1}{n}\sum_{t=1}^n e_i' D_t(\hat{\theta}_n) D_t'(\hat{\theta}_n) e_j \overset{a.s.}{\geq}& \frac{1}{n}\sum_{t=1}^n \inf_{\theta \in \mathscr{V}_\varepsilon(\theta_0)} e_i' D_t(\theta) D_t'(\theta) e_j\\
\overset{a.s.}{\to}& \EE\Big[\inf_{\theta \in \mathscr{V}_\varepsilon(\theta_0)}e_i' D_t(\theta) D_t'(\theta) e_j \Big]>\EE\big[e_i'D_t D_t' e_j \big]- \varepsilon
\end{align*}
Taking $\varepsilon \searrow 0$ establishes  $\frac{1}{n}\sum_{t=1}^n e_i' D_t(\hat{\theta}_n) D_t'(\hat{\theta}_n) e_j \overset{a.s.}{\to}\EE[e_i'D_t D_t'e_j]$ for all pairs $(i,j)$ yielding $III \overset{a.s.}{\to}\EE[D_t D_t' ] = J$. Consider $IV$; using \eqref{eq:4.A.3} and the elementary inequality
\begin{align}
\label{eq:4.A.5}
||xx'-yy'||\leq ||x-y||^2+2||x-y||\:||y|| 
\end{align}
for all $x,y\in \R^m$ with $m \in \N$, we obtain for $\theta \in \Theta$
\begin{align}
\label{eq:4.A.6}
\begin{split}
&\Big|\Big|\tilde{D}_t(\theta)\tilde{D}_t'(\theta)-D_t(\theta)D_t'(\theta)\Big|\Big|\\
\leq&   \big|\big|\tilde{D}_t(\theta)-D_t(\theta)\big|\big|^2+ 2\big|\big|\tilde{D}_t(\theta)-D_t(\theta)\big|\big|\:\big|\big|D_t(\theta)\big|\big|\\
\leq&  \frac{C_1^2 }{\underline{\omega}^2}  \rho^{2t}\Big(1+ \big|\big|D_t(\theta)\big|\big|\Big)^2+ \frac{2C_1}{\underline{\omega}} \rho^t\Big(1+ \big|\big|D_t(\theta)\big|\big|\Big)\:\big|\big|D_t(\theta)\big|\big|\\
\leq & \frac{C_1^2 }{\underline{\omega}^2} \rho^{t}\Big(1+ \big|\big|D_t(\theta)\big|\big|\Big)^2+ \frac{2C_1}{\underline{\omega}} \rho^t\Big(1+ \big|\big|D_t(\theta)\big|\big|\Big)^2\\
=& \bigg(\frac{C_1^2}{\underline{\omega}^2}+\frac{2C_1}{\underline{\omega}}\bigg) \rho^{t}\Big(1+ \big|\big|D_t(\theta)\big|\big|\Big)^2.
\end{split}
\end{align}
Hence, we get
\begin{align}
\nonumber
||IV||\leq&\frac{1}{n}\sum_{t=1}^n \Big|\Big|\tilde{D}_t(\hat{\theta}_n)\tilde{D}_t'(\hat{\theta}_n)-D_t(\hat{\theta}_n)D_t'(\hat{\theta}_n)\Big|\Big|
\leq  \bigg(\frac{C_1^2}{\underline{\omega}^2}+\frac{2C_1}{\underline{\omega}}\bigg) \frac{1}{n}\sum_{t=1}^n \rho^{t}\Big(1+ \big|\big|D_t(\hat{\theta}_n)\big|\big|\Big)^2\\
\label{eq:4.A.7}
\overset{a.s.}{\leq}&  \bigg(\frac{C_1^2}{\underline{\omega}^2}+\frac{2C_1}{\underline{\omega}}\bigg) \frac{1}{n}\sum_{t=1}^n \rho^{t}(1+ U_t)^2.
\end{align}
For each $\varepsilon>0$, Markov's inequality yields
\begin{align*}
&\sum_{t=1}^\infty \PP\Big[\rho^{t}(1+U_t)^2>\varepsilon\Big]  \leq \sum_{t=1}^\infty \rho^{t/2} \frac{1+\EE[U_t]}{\sqrt{\varepsilon}} 
= \frac{1+\EE[U_t]}{\sqrt{\varepsilon}(1-\sqrt{\rho})}<\infty
\end{align*}
and $\frac{1}{n}\sum_{t=1}^n \rho^{t}(1+U_t)^2\overset{a.s.}{\to}0$ follows from combining the Borel-Cantelli lemma  with Ces\'aro's lemma. Hence, $||IV|| \overset{a.s.}{\to}0$, which validates the second statement.

Consider the third statement and expand
\begin{align*}
&\frac{1}{n}\sum_{t=1}^n \hat{H}_t
= \underbrace{\frac{1}{n}\sum_{t=1}^n H_t(\hat{\theta}_n)}_{V}+\underbrace{\frac{1}{n}\sum_{t=1}^n \Big(\tilde{H}_t(\hat{\theta}_n)-H_t(\hat{\theta}_n)\Big)}_{VI}
\end{align*}
We focus on $V$ and let $\varepsilon>0$. 
\textcolor{black}{
Since $H_t(\theta)$ is continuous in $\theta$ we can take $\mathscr{V}_\varepsilon(\theta_0)\subseteq \mathscr{V}(\theta_0)$ such that 
}
\begin{align*}
\EE\big[e_i'H_t e_j\big]-\varepsilon<\EE\Big[\inf_{\theta \in \mathscr{V}_\varepsilon(\theta_0)}e_i' H_t(\theta)e_j\Big]\leq \EE\Big[\sup_{\theta \in \mathscr{V}_\varepsilon(\theta_0)}e_i' H_t(\theta)e_j\Big]<\EE\big[e_i'H_t e_j\big]+ \varepsilon
\end{align*}
for all $i,j \in\{1,\dots,r\}$. As $\hat{\theta}_n\overset{a.s.}{\to}\theta_0$ by Theorem \ref{thm:4.1}, we have
 $\hat{\theta}_n \in \mathscr{V}_\varepsilon(\theta_0)$ almost surely. Together with the uniform ergodic theorem we obtain
\begin{align*}
\frac{1}{n}\sum_{t=1}^n e_i' H_t(\hat{\theta}_n) e_j \overset{a.s.}{\leq}& \frac{1}{n}\sum_{t=1}^n \sup_{\theta \in \mathscr{V}_\varepsilon(\theta_0)} e_i' H_t(\theta) e_j \overset{a.s.}{\to} \EE\Big[\sup_{\theta \in \mathscr{V}_\varepsilon(\theta_0)}e_i' H_t(\theta)e_j\Big]<\EE\big[e_i'H_t e_j\big]+ \varepsilon\\
\frac{1}{n}\sum_{t=1}^n e_i' H_t(\hat{\theta}_n) e_j \overset{a.s.}{\geq}& \frac{1}{n}\sum_{t=1}^n \inf_{\theta \in \mathscr{V}_\varepsilon(\theta_0)} e_i' H_t(\theta) e_j \overset{a.s.}{\to} \EE\Big[\inf_{\theta \in \mathscr{V}_\varepsilon(\theta_0)}e_i' H_t(\theta)e_j\Big]>\EE\big[e_i'H_t e_j\big]- \varepsilon
\end{align*}
Taking $\varepsilon \searrow 0$ establishes  $\frac{1}{n}\sum_{t=1}^n e_i' H_t(\hat{\theta}_n)  e_j \overset{a.s.}{\to}\EE[e_i'H_t e_j]$ for all pairs $(i,j)$ yielding $V \overset{a.s.}{\to}\EE[H_t]$. Regarding $VI$, we note that 
\begin{align}
\label{eq:4.A.8}
\begin{split}
&\big|\big|\tilde{H}_t(\theta)-H_t(\theta)\big|\big| = \bigg|\bigg|\frac{1}{\tilde{\sigma}_t(\theta)}\frac{\partial^2 \tilde{\sigma}_t(\theta)}{\partial \theta \partial \theta'}-\frac{1}{\sigma_t(\theta)}\frac{\partial^2 \sigma_t(\theta)}{\partial \theta \partial \theta'}\bigg|\bigg|\\
= & \bigg|\bigg|\frac{1}{\tilde{\sigma}_t(\theta)}\bigg(\frac{\partial^2 \tilde{\sigma}_t(\theta)}{\partial \theta \partial \theta'}-\frac{\partial^2 \sigma_t(\theta)}{\partial \theta \partial \theta'}\bigg)+\frac{\sigma_t(\theta)-\tilde{\sigma}_t(\theta)}{\tilde{\sigma}_t(\theta)}\frac{1}{\sigma_t(\theta)}\frac{\partial^2 \sigma_t(\theta)}{\partial \theta \partial \theta'}\bigg|\bigg|\\
\leq& \frac{1}{\tilde{\sigma}_t(\theta)}\bigg|\bigg|\frac{\partial^2 \tilde{\sigma}_t(\theta)}{\partial \theta \partial \theta'}-\frac{\partial^2 \sigma_t(\theta)}{\partial \theta \partial \theta'}\bigg|\bigg|+ \frac{|\sigma_t(\theta)-\tilde{\sigma}_t(\theta)|}{\tilde{\sigma}_t(\theta)}\bigg|\bigg|\frac{1}{\sigma_t(\theta)}\frac{\partial^2 \sigma_t(\theta)}{\partial \theta \partial \theta'}\bigg|\bigg|\\
\leq& \frac{C_1 \rho^t}{\underline{\omega}}+ \frac{C_1 \rho^t}{\underline{\omega}}\big|\big|H_t(\theta)\big|\big|=\frac{C_1 \rho^t}{\underline{\omega}}\Big(1+ \big|\big|H_t(\theta)\big|\big|\Big)
\end{split}
\end{align}
for each $\theta \in \Theta$. We obtain
\begin{align*}
||VI||\leq  \frac{1}{n}\sum_{t=1}^n \big|\big|\tilde{H}_t(\hat{\theta}_n)-H_t(\hat{\theta}_n)\big|\big| \leq  \frac{C_1 }{\underline{\omega}} \frac{1}{n}\sum_{t=1}^n \rho^{t}\Big(1+ \big|\big|H_t(\hat{\theta}_n)\big|\big|\Big)
\overset{a.s.}{\leq}  \frac{C_1 }{\underline{\omega}} \frac{1}{n}\sum_{t=1}^n \rho^{t}\big(1+ V_t\big).
\end{align*}
For each $\varepsilon>0$, Markov's inequality yields
\begin{align*}
&\sum_{t=1}^\infty \PP\Big[\rho^{t}(1+V_t)>\varepsilon\Big]  \leq \sum_{t=1}^\infty \rho^{t} \frac{1+\EE[V_t]}{\varepsilon} 
= \frac{1+\EE[V_t]}{\varepsilon(1-\rho)}<\infty
\end{align*}
and $\frac{1}{n}\sum_{t=1}^n \rho^{t}(1+V_t)\overset{a.s.}{\to}0$ follows from combining the Borel-Cantelli lemma  with Ces\'aro's lemma. Hence, $||VI|| \overset{a.s.}{\to}0$, which validates the third statement.

Consider the fourth statement; let $m \in \{0,1,2,3,4\}$ and take $l,u \in \R$ such that $l< u$. We employ the partial integration formula 
\begin{align}
\label{eq:4.A.9}
G(u-)H(u-)-G(l-)H(l-)  =  \int_{[l,u)}G(t-)\,dH(t)+ \int_{[l,u)} H(s)\,dG(s)
\end{align}
with $G$ and $H$ both right-continuous functions being locally of bounded variation to expand
\begin{align*}
&\frac{1}{n}\sum\limits_{t=1}^n \hat{\eta}_t^m \mathbbm{1}_{\{l\leq \hat{\eta}_t<u\}} - \EE\big[\eta_t^m \mathbbm{1}_{\{l\leq \eta_t<u\}}\big] = \int_{[l,u)} x^m d \hat{\mathbbm{F}}_{n}(x)- \int_{[l,u)} x^m d F(x)\\
=& u^m \Big(\hat{\mathbbm{F}}_{n}(u-) - F(u)\Big) - l^m \Big(\hat{\mathbbm{F}}_{n}(l-) - F(l)\Big)+ \int_{[l,u)} \Big(\hat{\mathbbm{F}}_{n}(x)-F(x) \Big) d x^m.
\end{align*}
Lemma \ref{lem:4.1} implies $\hat{\mathbbm{F}}_{n}(u-) \overset{a.s.}{\to} F(u)$ and $\hat{\mathbbm{F}}_{n}(l-) \overset{a.s.}{\to} F(l)$ and together with the dominated convergence theorem yields $\int_{[l,u)} \big(\hat{\mathbbm{F}}_{n}(x)-F(x) \big) d x^m\overset{a.s.}{\to}0$. Thus,
\begin{align*}
\frac{1}{n}\sum\limits_{t=1}^n \hat{\eta}_t^m \mathbbm{1}_{\{l\leq \hat{\eta}_t<u\}} \overset{a.s.}{\to} \EE\big[\eta_t^m \mathbbm{1}_{\{l\leq \eta_t<u\}}\big] 
\end{align*}
for $m \in \{0,1,2,3,4\}$ and $l,u \in \R$. Since $\EE\big[|\eta_t|^m\big]<\infty$ and $\EE\big[\eta_t^m \mathbbm{1}_{\{l\leq \eta_t<u\}}\big] = \int_l^u x^m f(x) dx$ is continuous in $l$ and $u$ it is easy to see that the result extends to $l= -\infty$ and $u= \infty$, which validates the fourth statement.

Consider the fifth statement, whose proof follows the general steps of the proof of Lemma \ref{lem:4.1} and the fourth statement. Define 
\begin{align*}
    \hat{\mathbbm{G}}_n(x)=\frac{1}{n}\sum_{t=1}^n \mathbbm{1}_{\{\sqrt{n}(\tilde{\psi}_t-1)\leq x\}} \qquad \text{and} \qquad G(x)=\PP[D_t'(v_1-v_2)\leq x].
\end{align*}
First, we show that for any $\varepsilon>0$ there is a $\tau>0$ such that almost surely
\begin{align}
\label{eq:235768935}
\begin{split}
&\limsup_{n \to \infty} \sup_{\theta_1,\theta_2 \in \mathscr{V}_\tau(\theta_0)} \bigg|\frac{1}{n}\sum_{t=1}^n \mathbbm{1}_{\left\{\frac{\tilde{\sigma}_t(\theta_1)}{\tilde{\sigma}_t(\theta_2)}\left(\tilde{D}_t'(\theta_1)v_1-\tilde{D}_t'(\theta_2)v_2\right)\leq x\right\}}-G(x)\bigg|\\ 
& \qquad \qquad \qquad \qquad \qquad \qquad \qquad \qquad \qquad\leq 2\big(G(x+ \Delta\varepsilon)-G(x- \Delta\varepsilon)\big)
\end{split}
\end{align}
 for any $x \in \R$, where  $\Delta=|x|+||v_1||+||v_2||$ and $\mathscr{V}_\tau(\theta_0) = \big\{ \theta \in \Theta: ||\theta-\theta_0||\leq \tau\big\}$. Then, we show $\hat{\mathbbm{G}}_n(x)\overset{a.s.}{\to}G(x)$ for any $x \in \R$ and  $\sup_{x \in \R}|\hat{\mathbbm{G}}_n(x)-G(x)|\overset{a.s.}{\to}0$. Last, we prove $\frac{1}{n}\sum_{t=1}^n \mathbbm{1}_{\{l \leq \sqrt{n}(\tilde{\psi}_t-1)<u\}}  \big(\sqrt{n}(\tilde{\psi}_t-1)\big)^m\overset{a.s.}{\to}\EE\left[ \mathbbm{1}_{\{l \leq D_t'(v_1-v_2)< u\}}\left(D_t'(v_1-v_2)\right)^m\right]$.

Let $\varepsilon>0$ and set $\tau>0$ sufficiently small such that $\mathscr{V}_\tau(\theta_0)\subset \mathscr{V}(\theta_0)$. Regarding the initial conditions Assumption \ref{as:4.4}(\ref{as:4.4.1}) implies
\begin{align}
\label{eq:4.A.51b}
\begin{split}
\bigg|\frac{\tilde{\sigma}_t(\theta_1)}{\tilde{\sigma}_t(\theta_2)}-\frac{\sigma_t(\theta_1)}{\sigma_t(\theta_2)}\bigg|=&\bigg|\frac{\tilde{\sigma}_t(\theta_1)-\sigma_t(\theta_1)}{\tilde{\sigma}_t(\theta_2)}+\frac{\sigma_t(\theta_1)}{\sigma_t(\theta_2)}\:\frac{\sigma_t(\theta_2)-\tilde{\sigma}_t(\theta_2)}{\tilde{\sigma}_t(\theta_2)}\bigg|\\
\leq & \frac{|\tilde{\sigma}_t(\theta_1)-\sigma_t(\theta_1)|}{\tilde{\sigma}_t(\theta_2)}+\frac{\sigma_t(\theta_1)}{\sigma_t(\theta_2)}\:\frac{|\sigma_t(\theta_2)-\tilde{\sigma}_t(\theta_2)|}{\tilde{\sigma}_t(\theta_2)}\\
\leq& \frac{C_1\rho^t}{\underline{\omega}}+ \frac{\sigma_t(\theta_1)}{\sigma_t(\theta_2)}\frac{C_1\rho^t}{\underline{\omega}}= \frac{C_1\rho^t}{\underline{\omega}}\bigg(1+\frac{\sigma_t(\theta_1)}{\sigma_t(\theta_2)}\bigg)
\end{split}
\end{align}
for any $\theta_1,\theta_2 \in \Theta$ and together with \eqref{eq:4.A.3} we find 
\begin{align*}
&\frac{1}{n}\sum_{t=1}^n \mathbbm{1}_{\left\{\frac{\tilde{\sigma}_t(\theta_1)}{\tilde{\sigma}_t(\theta_2)}\left(\tilde{D}_t'(\theta_1)v_1-\tilde{D}_t'(\theta_2)v_2\right)\leq x\right\}}\\
=&\frac{1}{n}\sum_{t=1}^n \mathbbm{1}_{\left\{D_t'(\theta_1)v_1-D_t'(\theta_2)v_2-x\frac{\sigma_t(\theta_2)}{\sigma_t(\theta_1)}\leq x\left(\frac{\tilde{\sigma}_t(\theta_2)}{\tilde{\sigma}_t(\theta_1)}-\frac{\sigma_t(\theta_2)}{\sigma_t(\theta_1)}\right)+\left(D_t(\theta_1)-\tilde{D}_t(\theta_1)\right)'v_1+\left(\tilde{D}_t(\theta_2)-D_t(\theta_2)\right)'v_2\right\}}\\
\leq&\frac{1}{n}\sum_{t=1}^n \mathbbm{1}_{\left\{D_t'(\theta_1)v_1-D_t'(\theta_2)v_2-x\frac{\sigma_t(\theta_2)}{\sigma_t(\theta_1)}\leq |x|\frac{C_1\rho^t}{\underline{\omega}}\left(1+\frac{\sigma_t(\theta_2)}{\sigma_t(\theta_1)}\right)+||v_1||\frac{C_1\rho^t}{\underline{\omega}}\left(1+||D_t(\theta_1)||\right)+||v_2||\frac{C_1\rho^t}{\underline{\omega}}\left(1+||D_t(\theta_2)||\right)\right\}}\\
\leq&\frac{1}{n}\sum_{t=1}^n \mathbbm{1}_{\left\{D_t'(\theta_1)u-D_t'(\theta_2)v-x\frac{\sigma_t(\theta_2)}{\sigma_t(\theta_1)}\leq |x|\frac{C_1\rho^t}{\underline{\omega}}\left(1+S_tT_t\right)+\left(||v_1||+||v_2||\right)\frac{C_1\rho^t}{\underline{\omega}}\left(1+U_t\right)\right\}}
\end{align*}
for all $\theta_1,\theta_2 \in \mathscr{V}_\tau(\theta_0)$.
We have $\rho^{t}(1+U_t)\overset{a.s.}{\to} 0$ by \eqref{eq:4.A.4}. Further, for each $\varepsilon>0$, Markov's and H\"older's inequality together with Assumption \ref{as:4.9}(i) entail 
\begin{align*}
&\sum_{t=1}^\infty \PP\Big[\rho^{t}(1+S_tT_t)>\varepsilon\Big] \leq \sum_{t=1}^\infty \rho^t \frac{1+\EE[S_tT_t]}{\varepsilon} 
\leq \frac{1+\EE\left[S_t^2\right]^{\frac{1}{2}}\EE\left[T_t^2\right]^{\frac{1}{2}}}{\varepsilon(1-\rho)}<\infty.
\end{align*}
The Borel-Cantelli lemma implies 
%
%
$\rho^{t}(1+S_tT_t)\overset{a.s.}{\to} 0$. Hence, there exists a random variable $n_0$ such that $\frac{C_1\rho^t}{\underline{\omega}}(1+U_t)\leq \varepsilon$ and $\frac{C_1\rho^t}{\underline{\omega}}(1+S_tT_t)\leq \varepsilon$ for all $t>n_0$. It follows that almost surely
\begin{align*}
&\frac{1}{n}\sum_{t=1}^n \mathbbm{1}_{\left\{\frac{\tilde{\sigma}_t(\theta_1)}{\tilde{\sigma}_t(\theta_2)}\left(\tilde{D}_t'(\theta_1)v_1-\tilde{D}_t'(\theta_2)v_2\right)\leq x\right\}} \leq \frac{n_0}{n}+\frac{1}{n}\sum_{t=1}^n \mathbbm{1}_{\Big\{\inf\limits_{\theta_1,\theta_2 \in \mathscr{V}_\tau(\theta_0)}\left(D_t'(\theta_1)v_1-D_t'(\theta_2)v_2-x\frac{\sigma_t(\theta_2)}{\sigma_t(\theta_1)}\right)\leq \Delta\varepsilon\Big\}}
\end{align*}
 for all $\theta_1,\theta_2 \in \mathscr{V}_\tau(\theta_0)$. The uniform ergodic theorem and Assumptions \ref{as:4.2} and \ref{as:4.3} yield
\begin{align*}
&\frac{1}{n}\sum_{t=1}^n \mathbbm{1}_{\left\{\inf\limits_{\theta_1,\theta_2 \in \mathscr{V}_\tau(\theta_0)}\left(D_t'(\theta_1)v_1-D_t'(\theta_2)v_2-x\frac{\sigma_t(\theta_2)}{\sigma_t(\theta_1)}\right)\leq \Delta\varepsilon\right\}}\\
&\qquad \qquad \qquad \qquad \qquad \overset{a.s.}{\to} \EE\bigg[ \mathbbm{1}_{\left\{\inf\limits_{\theta_1,\theta_2 \in \mathscr{V}_\tau(\theta_0)}\left(D_t'(\theta_1)v_1-D_t'(\theta_2)v_2-x\frac{\sigma_t(\theta_2)}{\sigma_t(\theta_1)}\right)\leq \Delta\varepsilon\right\}}\bigg].
\end{align*}
The dominated convergence theorem entails
\begin{align*}
& \lim_{\tau\to 0}\EE\bigg[ \mathbbm{1}_{\left\{\inf\limits_{\theta_1,\theta_2 \in \mathscr{V}_\tau(\theta_0)}\left(D_t'(\theta_1)v_1-D_t'(\theta_2)v_2-x\frac{\sigma_t(\theta_2)}{\sigma_t(\theta_1)}\right)\leq \Delta\varepsilon\right\}}\bigg] =\EE\big[ \mathbbm{1}_{\left\{D_t'(v_1-v_2)-x\leq \Delta\varepsilon\right\}}\big]=G\left(x+ \Delta\varepsilon\right).
\end{align*}
Putting the results together, we get that for every $\varepsilon>0$, there is a $\tau>0$ such that
\begin{align*}
&\limsup_{n \to \infty} \sup_{\theta_1,\theta_2 \in \mathscr{V}_\tau(\theta_0)} \frac{1}{n}\sum_{t=1}^n \mathbbm{1}_{\left\{\frac{\tilde{\sigma}_t(\theta_1)}{\tilde{\sigma}_t(\theta_2)}\left(\tilde{D}_t'(\theta_1)v_1-\tilde{D}_t'(\theta_2)v_2\right)\leq x\right\}} \leq G(x)+2\big(G(x+ \Delta\varepsilon)-G(x)\big)
\end{align*}
almost surely for any $x\in\R$. Similarly it can be shown that for every $\varepsilon>0$, there is a $\tau>0$ such that
\begin{align*}
&\liminf_{n \to \infty} \sup_{\theta_1,\theta_2 \in \mathscr{V}_\tau(\theta_0)} \frac{1}{n}\sum_{t=1}^n \mathbbm{1}_{\left\{\frac{\tilde{\sigma}_t(\theta_1)}{\tilde{\sigma}_t(\theta_2)}\left(\tilde{D}_t'(\theta_1)v_1-\tilde{D}_t'(\theta_2)v_2\right)\leq x\right\}} \geq G(x)-2\big(G(x)-G(x- \Delta\varepsilon)\big)
\end{align*}
almost surely for any $x\in\R$. Combining both results establishes \eqref{eq:235768935}.

Next, we show $\hat{\mathbbm{G}}_n(x)\overset{a.s.}{\to}G(x)$ for any $x \in \R$. Let $\delta>0$; by continuity of $G$, there is a $\varepsilon>0$ such that $\big|G(x+ \Delta\varepsilon)-G(x- \Delta\varepsilon)\big|<\delta/2$. Employing equation \eqref{eq:235768935}, there are $\tau>0$ and a random variable $n_1$ such that 
\begin{align*}
 \sup_{\theta \in \mathscr{V}_\tau(\theta_0)} \bigg|\frac{1}{n}\sum_{t=1}^n \mathbbm{1}_{\left\{\frac{\tilde{\sigma}_t(\theta_1)}{\tilde{\sigma}_t(\theta_2)}\left(\tilde{D}_t'(\theta_1)v_1-\tilde{D}_t'(\theta_2)v_2\right)\leq x\right\}}-G(x)\bigg|<\delta
\end{align*}
for all $n\geq n_1$. In addition, the mean value theorem implies
\begin{align}
\label{eq:2309247}
 \frac{1}{n}\sum_{t=1}^n \mathbbm{1}_{\{\sqrt{n}(\tilde{\psi}_t-1)\leq x\}} = \frac{1}{n}\sum_{t=1}^n \mathbbm{1}_{\left\{\frac{\tilde{\sigma}_t(\dot{\theta}_n)}{\tilde{\sigma}_t(\ddot{\theta}_n)}\left(\tilde{D}_t'(\dot{\theta}_n)v_1-\tilde{D}_t'(\ddot{\theta}_n)v_2\right)\leq x\right\}}
\end{align}
with $\dot{\theta}_n$ lying between $\hat{\theta}_n$ and $\hat{\theta}_n+n^{-1/2}v_1$ and $\ddot{\theta}_n$ lying between $\hat{\theta}_n$ and $\hat{\theta}_n+n^{-1/2}v_2$. Since $\hat{\theta}_n \overset{a.s.}{\to}\theta_0$ by Theorem \ref{thm:4.1} there is a random variable $n_2$ such that $\dot{\theta}_n  ,\ddot{\theta}_n  \in \mathscr{V}_\tau(\theta_0)$ for all $n\geq n_2$. Thus, 
\begin{align*}
\big|\hat{\mathbbm{G}}_n(x)-G(x)\big|\leq  \sup_{\theta \in \mathscr{V}_\tau(\theta_0)} \bigg|\frac{1}{n}\sum_{t=1}^n \mathbbm{1}_{\left\{\frac{\tilde{\sigma}_t(\theta_1)}{\tilde{\sigma}_t(\theta_2)}\left(\tilde{D}_t'(\theta_1)v_1-\tilde{D}_t'(\theta_2)v_2\right)\leq x\right\}}-G(x)\bigg|<\delta
\end{align*}
for all $n\geq \max\{n_1,n_2\}$, which establishes $\hat{\mathbbm{G}}_n(x)\overset{a.s.}{\to}G(x)$ for any $x \in \R$. Using P\'olya's lemma (cf.\ \citeauthor{roussas1997course}, \citeyear{roussas1997course}, p.\ 206), we establish $\sup_{x \in \R}|\hat{\mathbbm{G}}_n(x)-G(x)|\overset{a.s.}{\to}0$. Next, let $l, u \in \R$ with $l < u$.
We use the partial integration formula \eqref{eq:4.A.9} to expand
\begin{align*}
&\frac{1}{n}\sum\limits_{t=1}^n   \mathbbm{1}_{\{l\leq \sqrt{n}(\tilde{\psi}_t-1)<u\}} \Big(\sqrt{n}\big(\tilde{\psi}_t-1\big)\Big)^m  - \EE\left[ \mathbbm{1}_{\{l\leq D_t'(v_1-v_2)<u\}}\left(D_t'(v_1-v_2)\right)^m\right]\\
=& \int_{[l,u)} x^m d \hat{\mathbbm{G}}_{n}(x)- \int_{[l,u)} x^m d G(x)\\
=& u^m \big(\hat{\mathbbm{G}}_{n}(u-) - G(u)\big)  - l^m \big(\hat{\mathbbm{G}}_{n}(l-) - G(l)\big) + \int_{[l,u)} \big(\hat{\mathbbm{G}}_{n}(x)-G(x) \big) d x^m.
\end{align*}
We have $\hat{\mathbbm{G}}_{n}(u-) \overset{a.s.}{\to} G(u)$ and $\hat{\mathbbm{G}}_{n}(l-) \overset{a.s.}{\to} G(l)$ and together with the dominated convergence theorem yields $\int_{[l,u)} \big(\hat{\mathbbm{G}}_{n}(x)-G(x) \big) d x^m\overset{a.s.}{\to}0$. Thus, we establish
\begin{align*}
\frac{1}{n}\sum\limits_{t=1}^n   \mathbbm{1}_{\{l\leq \sqrt{n}(\tilde{\psi}_t-1)<u\}} \big(\sqrt{n}(\tilde{\psi}_t-1)\big)^m\overset{a.s.}{\to}\EE\left[ \mathbbm{1}_{\{l\leq D_t'(v_1-v_2)<u\}}\left(D_t'(v_1-v_2)\right)^m\right].
\end{align*}
Let $g(x)$ be the corresponding density of $G(x)$. As $\EE\big[|D_t'(v_1-v_2)|^m\big]\leq ||v_1-v_2||^m\EE\big[U_t^m\big]<\infty$ and $\EE\left[ \mathbbm{1}_{\{l\leq D_t'(v_1-v_2)<u\}}\left(D_t'(v_1-v_2)\right)^m\right] = \int_l^u x^m g(x) dx$ is continuous in $l$ and $u$ it is easy to see that the result extends to $l= -\infty$ and  $u= \infty$, which validates the fifth statement and completes the proof.
\end{proof}

\begin{lemma}
\label{lem:4.3}
Suppose Assumptions \ref{as:4.1}--\ref{as:4.9} hold with $a=\pm 6$, $b=6$ and $c=2$ and let $\mathcal{I}_n=(\xi_\alpha-a_n,\xi_\alpha+a_n)$ with $a_n \sim n^{-\varrho} \log n$ for some $\varrho\in(0,1)$. Then, we have
\begin{align*}
\sup_{x,y \in I_n} \Big| \sqrt{n} \big(\hat{\mathbbm{F}}_n(x)-\hat{\mathbbm{F}}_n(y)\big)- \sqrt{n} \big(F(x)-F(y)\big)\Big|\overset{p}{\to}0.
\end{align*}
Replacing any $\hat{\mathbbm{F}}_n(\cdot)$ by $\hat{\mathbbm{F}}_n(\cdot\:-)$  does not alter the result.
\end{lemma}

\begin{proof}
We follow \cite{berkes2003limit} and define
\begin{align*}
\tilde{\gamma}_t(u) =& \tilde{\sigma}_t(\theta_0+n^{-1/2}u)/\sigma_t(\theta_0)\\
\gamma_t(u) =& \sigma_t(\theta_0+n^{-1/2}u)/\sigma_t(\theta_0)\\
\zeta_t(x,u) =& \mathbbm{1}_{\{\eta_t\leq x\tilde{\gamma}_t(u)\}}- F\big(x\tilde{\gamma}_t(u)\big)-\big(\mathbbm{1}_{\{\eta_t\leq x\}}- F(x)\big)\\
S_n(x,u) =& \sum_{t=1}^n \zeta_t(x,u)\\
\mathbbm{F}_n(x)=&\frac{1}{n}\sum_{t=1}^n \mathbbm{1}_{\{\eta_t\leq x\}}.
\end{align*}
%
%
Let $A>0$ and write $\mathscr{V}(\xi_\alpha)$ to denote the neighborhood around $\xi_\alpha$ on which $f$ is continuous; see Assumption \ref{as:4.5}(\ref{as:4.5.2}). Since $\xi_\alpha<0$, we can take a compact neighborhood $\mathcal{X}=[\underline{x},\bar{x}]\subset \mathscr{V}(\xi_\alpha)$ such that $\xi_\alpha \in \mathcal{X}$ and $\bar{x}<0$. We establish the result in seven steps:

\begin{enumerate}
\item[]\textit{Step 1:} $\EE\big[|S_n(x,u)|^4\big]=O(n)$ for all $ x \in \mathcal{X}$ and for all $u \in \{u \in \R^r:||u||\leq A\}$; 

\item[]\textit{Step 2:} $\sup\limits_{x \in \mathcal{X}}|S_n(x,u)|=o_p(\sqrt{n})$ for all $u \in \{u \in \R^r:||u||\leq A\}$; 

\item[]\textit{Step 3:}  $\sup\limits_{||u||\leq A}\sup\limits_{x \in \mathcal{X}}|S_n(x,u)|=o_p(\sqrt{n})$;

\item[]\textit{Step 4:} $\sup\limits_{||u||\leq A}\sup\limits_{x \in \mathcal{X}}\Big|\frac{1}{\sqrt{n}}\sum_{t=1}^n\big(F(x\tilde{\gamma}_t(u))-F(x)\big)-xf(x)\Omega'u\Big|=o_p(1)$;

\item[] \textit{Step 5:} $\sup\limits_{x \in \mathcal{X}}\Big|\sqrt{n}\big(\hat{\mathbbm{F}}_n(x)-\mathbbm{F}_n(x)\big)-xf(x)\Omega'\sqrt{n}\big(\hat{\theta}_n-\theta_0\big)\Big|=o_p(1)$;

 \item[] \textit{Step 6:} $
\sup\limits_{x,y \in \mathcal{I}_n} \Big| \sqrt{n} \big(\mathbbm{F}_n(x)-\mathbbm{F}_n(y)\big)- \sqrt{n} \big(F(x)-F(y)\big)\Big|=O\big(n^{-\varrho/2} \log n\big)$ a.s.;

\item[] \textit{Step 7:} $\sup\limits_{x,y \in \mathcal{I}_n} \Big| \sqrt{n} \big(\hat{\mathbbm{F}}_n(x)-\hat{\mathbbm{F}}_n(y)\big)- \sqrt{n} \big(F(x)-F(y)\big)\Big|\overset{p}{\to}0$.
\end{enumerate}
\textit{Step 1} to \textit{Step 5} are similar to the proofs of \cite{berkes2003limit}, whereas \textit{Step 6} resembles \citeauthor{bahadur1966note} (\citeyear{bahadur1966note}, Lemma 1).

Throughout \textit{Step 1} to \textit{Step 4} we take $\delta \in (0,1/2)$ such that $ \mathcal{X}_\delta=[\underline{x}(1+2\delta),\bar{x}(1-2\delta)]$ satisfies $\mathcal{X}\subset \mathcal{X}_\delta\subset \mathscr{V}(\xi_\alpha)$. Because $f$ is continuous on $\mathcal{X}_\delta$ and $\mathcal{X}_\delta$ is compact, $f$ is uniformly continuous on $\mathcal{X}_\delta$ and there exists a finite $M>0$ such that 
\begin{align}
\label{eq:4.A.10}
\sup_{x \in \mathcal{X}_\delta}f(x)\leq M.
\end{align}

Consider \textit{Step 1}; let $\mathscr{F}_t$ be the $\sigma$-algebra generated by $\zeta_t,\zeta_{t-1},\dots$ and note that $\{S_t(x,u),\mathscr{F}_t\}$ is a martingale given $x$ and $u$. Theorem 2.11 of \cite{hall1980martingale} yields
\begin{align*}
\EE\Big[|S_n(x,u)|^4\Big]\leq C\Bigg(\EE\Big[\max_{1\leq t\leq n}\zeta_t^4(x,u)\Big]+\EE\bigg[\bigg(\sum_{t=1}^n\EE_{t-1}\big[\zeta_t^2(x,u)\big]\bigg)^2\bigg]\Bigg),
\end{align*}
for some absolute constant $C>0$ independent of $x$ and $u$, where $\EE_{t-1}=\EE[\:\cdot\:|\mathscr{F}_{t-1}]$ is the expectation given $\mathscr{F}_{t-1}$.
As $\big|\zeta_t(x,u)\big|\leq 2$ for all $t$ such that $\EE\big[\max_{ 1\leq t\leq n}\zeta_t^4(x,u)\big]\leq 16$, it suffices to show that
\begin{align}
\label{eq:4.A.11}
\EE\bigg[\bigg(\sum_{t=1}^n\EE_{t-1}\big[\zeta_t^2(x,u)\big]\bigg)^2\bigg]=O(n).
\end{align}
First, we focus on the inner part $\EE_{t-1}\big[\zeta_t^2(x,u)\big]$ and decompose $\zeta_t(x,u)$ into 
\begin{align*}
\zeta_t(x,u) =& \zeta_{t,1}(x,u) +\zeta_{t,2}(x,u)
\end{align*}
with
\begin{align*}
\zeta_{t,1}(x,u) =& \mathbbm{1}_{\{\eta_t\leq x\tilde{\gamma}_t(u)\}}- F\big(x\tilde{\gamma}_t(u)\big)-\mathbbm{1}_{\{\eta_t\leq x\gamma_t(u)\}}+F\big(x\gamma_t(u)\big)\\
\zeta_{t,2}(x,u) =& \mathbbm{1}_{\{\eta_t\leq x\gamma_t(u)\}}- F\big(x\gamma_t(u)\big)- \mathbbm{1}_{\{\eta_t\leq x\}}+ F(x).
\end{align*}
The elementary inequality
\begin{align}
\label{eq:4.A.12}
\Big(\sum_{i=1}^m x_i\Big)^2\leq m \sum_{i=1}^m x_i^2
\end{align}
for all $x_1,\dots,x_m \in \R$ with $m \in \N$ implies that
\begin{align*}
\EE_{t-1}\big[\zeta_t^2(x,u)\big]\leq 2\Big(\EE_{t-1}\big[\zeta_{t,1}^2(x,u)\big]+\EE_{t-1}\big[\zeta_{t,2}^2(x,u)\big]\Big).
\end{align*}
Moreover, the inequality $\Var[\mathbbm{1}_{\{X\leq y\}}- \mathbbm{1}_{\{X\leq z\}}]
\leq  |F_X(y)-F_X(z)|$ for $y,z \in \R$ and $X\sim F_X$ gives
\begin{align*}
\EE_{t-1}\big[\zeta_{t,1}^2(x,u)\big]=&\Var_{t-1}\big[\mathbbm{1}_{\{\eta_t\leq x\tilde{\gamma}_t(u)\}}-\mathbbm{1}_{\{\eta_t\leq x\gamma_t(u)\}}\big]
\leq \big|F\big(x\tilde{\gamma}_t(u)\big)-F\big(x\gamma_t(u)\big)\big|\\
\EE_{t-1}\big[\zeta_{t,2}^2(x,u)\big]=&\Var_{t-1}\big[\mathbbm{1}_{\{\eta_t\leq x\gamma_t(u)\}}-\mathbbm{1}_{\{\eta_t\leq x\}}\big] 
\leq \big|F\big(x\gamma_t(u)\big)-F(x)\big|.
\end{align*}
Combining results, it follows that 
\begin{align}
\label{eq:4.A.13}
\EE_{t-1}\big[\zeta_t^2(x,u)\big]\leq&2\Big(\big|F\big(x\gamma_t(u)\big)-F(x)\big|+\big|F\big(x\tilde{\gamma}_t(u)\big)-F\big(x\gamma_t(u)\big)\big|\Big).
\end{align}
Employing \eqref{eq:4.A.13}, we obtain that the left-hand side in \eqref{eq:4.A.11} is bounded by
\begin{align*}
&4\EE\bigg[\bigg(\sum_{t=1}^n\Big|F\big(x\gamma_t(u)\big)-F(x)\Big|+\sum_{t=1}^n\Big|F\big(x\tilde{\gamma}_t(u)\big)-F\big(x\gamma_t(u)\big)\Big|\bigg)^2\bigg]\\
\leq&8\Bigg(\underbrace{\EE\bigg[\bigg(\sum_{t=1}^n\Big|F\big(x\gamma_t(u)\big)-F(x)\Big|\bigg)^2\bigg]}_{I}+\underbrace{\EE\bigg[\bigg(\sum_{t=1}^n\Big|F\big(x\tilde{\gamma}_t(u)\big)-F\big(x\gamma_t(u)\big)\Big|\bigg)^2\bigg]}_{II}\Bigg),
\end{align*}
where the last inequality follows from applying \eqref{eq:4.A.12} once more. It suffices to show that both terms are $O(n)$. Consider $I$; The Cauchy-Schwarz inequality yields
\begin{align}
\label{eq:4.A.14}
\begin{split}
I=&\sum_{t=1}^n \sum_{\tau=1}^n\EE\bigg[\Big|F\big(x\gamma_t(u)\big)-F(x)\Big|\: \Big|F\big(x\gamma_\tau(u)\big)-F(x)\Big|\bigg]\\
\leq & \sum_{t=1}^n \sum_{\tau=1}^n\Bigg(\EE\bigg[\Big(F\big(x\gamma_t(u)\big)-F(x)\Big)^2\bigg]\Bigg)^{\frac{1}{2}}\Bigg(\EE\bigg[\Big(F\big(x\gamma_\tau(u)\big)-F(x)\Big)^2\bigg]\Bigg)^{\frac{1}{2}}.
\end{split}
\end{align}
Henceforth, we take $n$ sufficiently large such that $\big \{\theta:||\theta-\theta_0||\leq A/\sqrt{n}\big \} \subseteq \mathscr{V}(\theta_0)$. The mean value theorem implies
\begin{align}
\label{eq:4.A.15}
\begin{split}
&\sup_{||u||\leq A} \big|\gamma_t(u)-1\big| = \sup_{||u||\leq A}\bigg|\frac{\sigma_t(\theta_0+u/\sqrt{n})-\sigma_t(\theta_0)}{\sigma_t(\theta_0)}\bigg|\\
= & \sup_{||u||\leq A}\bigg|\frac{1}{\sigma_t(\theta_0)}\frac{\partial \sigma_t(\bar{\theta}_n)}{\partial \theta'}\frac{1}{\sqrt{n}}u\bigg|= \frac{1}{\sqrt{n}}  \sup_{||u||\leq A}\bigg|\frac{\sigma_t(\bar{\theta}_n)}{\sigma_t(\theta_0)}D_t'(\bar{\theta}_n)\:u\bigg|\\
\leq & \frac{1}{\sqrt{n}}  \sup_{||\theta-\theta_0||\leq An^{-1/2}}\frac{\sigma_t(\theta)}{\sigma_t(\theta_0)} \sup_{||\theta-\theta_0||\leq An^{-1/2}}\big|\big|D_t(\theta)\big|\big|  \sup_{||u||\leq A}||u||\leq \frac{A}{\sqrt{n}} T_t U_t ,
\end{split}
\end{align}
where $T_t$ and $U_t$ are defined on page 49 and $\bar{\theta}_n$ lies between $\theta_0$ and $\theta_0+u/\sqrt{n}$. Define the event
\begin{align}
\label{eq:4.A.16}
\mathscr{A}_{n,t}=\bigg\{\frac{ A}{\sqrt{n}}T_t U_t\leq \delta\bigg\},
\end{align}
where $\delta$ is given in the text preceding \eqref{eq:4.A.10}. The inner term of \eqref{eq:4.A.14} can be bounded by
\begin{align}
\label{eq:4.A.17}
\begin{split}
&\EE\bigg[\Big(F\big(x\gamma_t(u)\big)-F(x)\Big)^2\bigg]=\EE\bigg[\underbrace{\Big(F\big(x\gamma_t(u)\big)-F(x)\Big)^2}_{\leq 1}\big(\mathbbm{1}_{\{\mathscr{A}_{n,t}^c\}}+\mathbbm{1}_{\{\mathscr{A}_{n,t}\}}\big)\bigg]\\
\leq & \underbrace{\PP\big[\mathscr{A}_{n,t}^c\big]}_{I_1}+\underbrace{\EE\bigg[\Big(F\big(x\gamma_t(u)\big)-F(x)\Big)^2\mathbbm{1}_{\{\mathscr{A}_{n,t}\}}\bigg]}_{I_2},
\end{split}
\end{align}
where the superscript $c$ denotes the event's complement. Using Markov's inequality, the Cauchy-Schwarz inequality and Assumption \ref{as:4.9}, $I_1$ can be bounded by
\begin{align}
\label{eq:4.A.18}
I_1=&  \PP\bigg[ \frac{ A}{\sqrt{n}}T_t U_t> \delta \bigg]\leq \frac{ A^2}{n\delta^2} \EE\big[T_t^2 U_t^2\big] \leq \frac{ A^2}{n\delta^2} \Big(\underbrace{\EE\big[T_t^4\big]}_{<\infty}\Big)^{\frac{1}{2}}\Big(\underbrace{\EE\big[U_t^4\big]}_{<\infty}\Big)^{\frac{1}{2}}
\end{align}
and, thus, $I_1=O(n^{-1})$. Regarding $I_2$, the mean value theorem implies
\begin{align*}
I_2=&\EE\bigg[x^2 f^2\big(x\bar{\gamma}_t\big) \big(\gamma_t(u)-1\big)^2\mathbbm{1}_{\{\mathscr{A}_{n,t}\}}\bigg]
\end{align*}
with $\bar{\gamma}_t$ being between $\gamma_t(u)$ and $1$. Since $|\bar{\gamma}_t-1|\leq |\gamma_t(u)-1| \leq \delta$ in the event of $\mathscr{A}_{n,t}$, we have $x\bar{\gamma}_t \in \mathcal{X}_\delta$. Employing  \eqref{eq:4.A.10}, \eqref{eq:4.A.15}, the Cauchy-Schwarz inequality and Assumption \ref{as:4.9}, we establish
\begin{align}
\label{eq:4.A.19}
I_2\leq&\EE\bigg[\underline{x}^2 M^2 \frac{A^2 }{n}T_t^2 U_t^2 \mathbbm{1}_{\{\mathscr{A}_{n,t}\}}\bigg]\leq \frac{\underline{x}^2 M^2A^2}{n}\Big(\underbrace{\EE\big[T_t^4\big]}_{<\infty}\Big)^{\frac{1}{2}}\Big(\underbrace{\EE\big[U_t^4\big]}_{<\infty}\Big)^{\frac{1}{2}} =O(n^{-1}).
\end{align}
Combining \eqref{eq:4.A.17} to \eqref{eq:4.A.19} yields
\begin{align*}
\EE\bigg[\Big(F\big(x\gamma_t(u)\big)-F(x)\Big)^2\bigg]\leq I_1+I_2=O(n^{-1})
\end{align*}
and, together with \eqref{eq:4.A.14}, we get
\begin{align*}
I\leq & \sum_{t=1}^n \sum_{r=1}^n O(n^{-1/2})O(n^{-1/2})=O(n).
\end{align*}
Next, we consider $II$, which can be bounded analogously to \eqref{eq:4.A.14} by
\begin{align}
\label{eq:4.A.20}
II\leq & \sum_{t=1}^n \sum_{\tau=1}^n\Bigg(\EE\bigg[\Big(F\big(x\tilde{\gamma}_t(u)\big)-F\big(x\gamma_t(u)\big)\Big)^2\bigg]\Bigg)^{\frac{1}{2}} \\
\nonumber
&\qquad \qquad \qquad \qquad \times\Bigg(\EE\bigg[\Big(F\big(x\tilde{\gamma}_\tau(u)\big)-F\big(x\gamma_\tau(u)\big)\Big)^2\bigg]\Bigg)^{\frac{1}{2}} .
\end{align}
Assumption \ref{as:4.4}(\ref{as:4.4.1}) gives
\begin{align}
\label{eq:4.A.21}
\begin{split}
\sup_{||u||\leq A} \big|\tilde{\gamma}_t(u)-\gamma_t(u)\big|
=&\sup_{||u||\leq A}  \frac{|\tilde{\sigma}_t(\theta_0+n^{-1/2}u)-\sigma_t(\theta_0+n^{-1/2}u)|}{\sigma_t(\theta_0)}\leq  \rho^t \frac{C_1}{\underline{\omega}}.
\end{split}
\end{align} 
We define the events 
\begin{align}
\label{eq:4.A.22}
\mathscr{B}_{t}=\bigg\{\rho^{t}\frac{C_1}{\underline{\omega}}\leq \delta \rho^{t/2}\bigg\}\quad \text{and}\quad \mathscr{C}_{n,t}=\mathscr{A}_{n,t}\cap \mathscr{B}_{t}.
\end{align}
In analogy to \eqref{eq:4.A.17}, the inner part of \eqref{eq:4.A.20} can be bounded by 
\begin{align*}
\EE\bigg[\Big(F\big(x\tilde{\gamma}_t(u)\big)-F\big(x\gamma_t(u)\big)\Big)^2\bigg]\leq & \underbrace{\PP\big[\mathscr{C}_{n,t}^c\big]}_{II_1}+\underbrace{\EE\bigg[\Big(F\big(x\tilde{\gamma}_t(u)\big)-F\big(x\gamma_t(u)\big)\Big)^2\mathbbm{1}_{\{\mathscr{C}_{n,t}\}}\bigg]}_{II_2}.
\end{align*}
Employing \eqref{eq:4.A.18} and Markov's inequality yields
\begin{align}
\label{eq:4.A.23}
\begin{split}
II_1=& \PP\big[\mathscr{A}_{n,t}^c \cup \mathscr{B}_t^c \big]\leq \PP\big[\mathscr{A}_{n,t}^c\big] + \PP\big[\mathscr{B}_t^c \big]=\PP\big[\mathscr{A}_{n,t}^c\big]+\PP\bigg[\rho^{t/2}\frac{C_1}{\underline{\omega}}> \delta\bigg]\\
\leq& \frac{ A^2}{n\delta^2} \Big(\EE\big[T_t^4\big]\Big)^{\frac{1}{2}}\Big(\EE\big[U_t^4\big]\Big)^{\frac{1}{2}} + (\rho^{s/2})^t\frac{\EE[C_1^s]}{\delta^s\underline{\omega}^{s}}=O(n^{-1})+O\big((\rho^{s/2})^t\big).
\end{split}
\end{align}
Regarding $II_2$, the mean value theorem implies
\begin{align*}
II_2=&\EE\bigg[x^2 f^2\big(x\breve{\gamma}_t\big)\big(\tilde{\gamma}_t(u)-\gamma_t(u)\big)^2\mathbbm{1}_{\{\mathscr{C}_{n,t}\}}\bigg]
\end{align*}
with $\breve{\gamma}_t$ between $\tilde{\gamma}_t(u)$ and $\gamma_t(u)$.  Since 
\begin{align*}
|\breve{\gamma}_t-1|\leq |\breve{\gamma}_t-\gamma_t(u)|+ |\gamma_t(u)-1| \leq |\tilde{\gamma}_t(u)-\gamma_t(u)|+ |\gamma_t(u)-1|\leq 2\delta
\end{align*}
in the event of $\mathscr{C}_{n,t}=\mathscr{A}_{n,t}\cap \mathscr{B}_{t}$, we have $x\breve{\gamma}_t \in \mathcal{X}_\delta$. Employing  \eqref{eq:4.A.10} and \eqref{eq:4.A.21} we obtain
\begin{align}
\label{eq:4.A.24}
II_2\leq \EE\bigg[\underline{x}^2 M^2\underbrace{\bigg(\rho^t\frac{C_1}{\underline{\omega}}\bigg)^2\mathbbm{1}_{\{\mathscr{C}_{n,t}\}}}_{\leq \delta^2 \rho^t}\bigg] \leq   \underline{x}^2 M^2 \delta^2 \rho^t=O(\rho^t).
\end{align}
Equations \eqref{eq:4.A.23} and \eqref{eq:4.A.24} imply
\begin{align*}
&\EE\bigg[\Big(F\big(x\tilde{\gamma}_t(u)\big)-F\big(x\gamma_t(u)\big)\Big)^2\bigg]\leq C  \big(n^{-1}+\rho^t+(\rho^{s/2})^t\big)
\end{align*}
for some constant $C>0$. Inserting this result into \eqref{eq:4.A.20}, we conclude
\begin{align*}
II\leq & C\sum_{t=1}^n \sum_{\tau=1}^n \Big(n^{-1}+\rho^t+(\rho^{s/2})^t\Big)^{\frac{1}{2}}  \Big(n^{-1}+\rho^\tau+(\rho^{s/2})^\tau\big)\Big)^{\frac{1}{2}}  =O(n),
\end{align*}
which completes \textit{Step 1}.

In \textit{Step 2} we divide $\mathcal{X}$ into intervals with the points $\underline{x}=x_1<x_2<\dots<x_N<x_{N+1}=\bar{x}$ 
satisfying $0.5\:n^{-3/4}\leq x_{j+1}-x_j\leq n^{-3/4}$ for all $j=1,\dots,N$ and $N \in \N$. It follows that $N=O(n^{3/4})$. We obtain 
\begin{align}
\label{eq:4.A.25}
\begin{split}
&\sup_{x \in \mathcal{X}}\big|S_n(x,u)\big|
= \max_{1 \leq j\leq N} \sup_{x_j\leq x\leq x_{j+1}} \big|S_n(x,u)\big|\\
\leq & \max_{1 \leq j\leq N} \sup_{x_j\leq x\leq x_{j+1}} \Big(\big|S_n(x_{j+1},u)\big|+\big|S_n(x,u)-S_n(x_{j+1},u)\big|\Big)\\
\leq & \max_{1 \leq j\leq N}  \big|S_n(x_{j+1},u)\big|+\max_{1 \leq j\leq N} \sup_{x_j\leq x\leq x_{j+1}} \big|S_n(x,u)-S_n(x_{j+1},u)\big|.
\end{split}
\end{align}
We bound the second term using the elementary inequality
\begin{align}
\label{eq:4.A.26}
|x-y|\leq \max\{x,y\}
\end{align}
for all $x,y\geq 0$. For $j=1\dots,N$, we have
\begin{align}
\label{eq:4.A.27}
\begin{split}
&\sup_{x_j\leq x\leq x_{j+1}}\big|S_n(x,u)-S_n(x_{j+1},u)\big|\\
=&\sup_{x_j\leq x\leq x_{j+1}}\bigg|\sum_{t=1}^n\Big(\mathbbm{1}_{\{\eta_t\leq x_{j+1}\}}-\mathbbm{1}_{\{\eta_t\leq x\}}+F\big(x_{j+1}\tilde{\gamma}_t(u)\big)-F\big(x \tilde{\gamma}_t(u)\big)\Big)\\
& \qquad \qquad\qquad  -\sum_{t=1}^n\Big(\mathbbm{1}_{\{\eta_t\leq x_{j+1}\tilde{\gamma}_t(u)\}}-\mathbbm{1}_{\{\eta_t\leq x\tilde{\gamma}_t(u)\}}+ F(x_{j+1})-F(x)\Big)\bigg|\\
\leq&\sup_{x_j\leq x\leq x_{j+1}}\max\bigg\{\sum_{t=1}^n\Big(\mathbbm{1}_{\{\eta_t\leq x_{j+1}\}}-\mathbbm{1}_{\{\eta_t\leq x\}}+ F\big(x_{j+1} \tilde{\gamma}_t(u)\big)-F\big(x \tilde{\gamma}_t(u)\big)\Big),\\
& \qquad \qquad\qquad \qquad \sum_{t=1}^n\Big(\mathbbm{1}_{\{\eta_t\leq x_{j+1}\tilde{\gamma}_t(u)\}}-\mathbbm{1}_{\{\eta_t\leq x\tilde{\gamma}_t(u)\}}+ F(x_{j+1})-F(x)\Big)\bigg\}\\
\leq &  \max\bigg\{\underbrace{\sum_{t=1}^n\Big(\mathbbm{1}_{\{\eta_t\leq x_{j+1}\}}-\mathbbm{1}_{\{\eta_t\leq x_{j}\}}+ F\big(x_{j+1} \tilde{\gamma}_t(u)\big)-F\big(x_{j} \tilde{\gamma}_t(u)\big)\Big)}_{=A_n},\\
& \qquad \qquad\quad  \underbrace{\sum_{t=1}^n\Big(\mathbbm{1}_{\{\eta_t\leq x_{j+1}\tilde{\gamma}_t(u)\}}-\mathbbm{1}_{\{\eta_t\leq x_{j}\tilde{\gamma}_t(u)\}}+ F(x_{j+1})-F(x_{j} )\Big)}_{=B_n}\bigg\}.
\end{split}
\end{align}
Note that $A_n$ and $B_n$ are positive, where the later can be rewritten as
\begin{align}
\label{eq:4.A.28}
\begin{split}
B_n =&\sum_{t=1}^n\Big(\mathbbm{1}_{\{\eta_t\leq x_{j+1}\tilde{\gamma}_t(u)\}}- F\big(x_{j+1}\tilde{\gamma}_t(u)\big)-\mathbbm{1}_{\{\eta_t\leq x_{j+1}\}}+ F(x_{j+1})\Big)\\
&\quad -\sum_{t=1}^n\Big(\mathbbm{1}_{\{\eta_t\leq x_{j}\tilde{\gamma}_t(u)\}}- F\big(x_{j}\tilde{\gamma}_t(u)\big)-\mathbbm{1}_{\{\eta_t\leq x_{j}\}}+F(x_{j})\Big)\\
&\qquad +\sum_{t=1}^n\Big(\mathbbm{1}_{\{\eta_t\leq x_{j+1}\}}-\mathbbm{1}_{\{\eta_t\leq x_{j}\}}+F\big(x_{j+1}\tilde{\gamma}_t(u)\big)- F\big(x_{j}\tilde{\gamma}_t(u)\big)\Big)\\
 =&S_n(x_{j+1},u)-S_n(x_{j},u)+A_n.
 \end{split}
\end{align}
It follows from \eqref{eq:4.A.27} and \eqref{eq:4.A.28} that
\begin{align}
\label{eq:4.A.29}
\begin{split}
&\sup_{x_j\leq x\leq x_{j+1}}\big|S_n(x,u)-S_n(x_{j+1},u)\big|\leq  |S_n(x_{j+1},u)|+|S_n(x_{j},u)|+A_n.
\end{split}
\end{align}
Moreover, $A_n$ expands as follows: 
\begin{align}
\nonumber
A_n =&  \sum_{t=1}^n \Big(\mathbbm{1}_{\{\eta_t\leq x_{j+1}\}}-F(x_{j+1})-  \mathbbm{1}_{\{\eta_t\leq x_{j}\}}+F(x_{j})\Big) +n\big(F(x_{j+1})-F(x_j)\big)\\
\label{eq:4.A.30}
&\qquad + \sum_{t=1}^n \Big(F\big(x_{j+1} \tilde{\gamma}_t(u)\big)-F\big(x_{j} \tilde{\gamma}_t(u)\big)\Big)
\end{align}
Using equations \eqref{eq:4.A.25}, \eqref{eq:4.A.29} and \eqref{eq:4.A.30}, we establish
\begin{align}
\label{eq:4.A.31}
\sup_{x \in \mathcal{X}}\big|S_n(x,u)\big| \leq 3III+IV+V+VI+2VII,
\end{align}
where 
\begin{align*}
III=&\max_{1\leq j\leq N+1}\big|S_n(x_j,u)\big|\\
IV=&\max_{1\leq j\leq N} n \big(F(x_{j+1} )-F(x_{j} )\big)\\
V=&\max_{1\leq j\leq N} \bigg|\sum_{t=1}^n \big(\mathbbm{1}_{\{\eta_t\leq x_{j+1}\}}-F(x_{j+1})\big)- \sum_{t=1}^n  \big(\mathbbm{1}_{\{\eta_t\leq x_{j}\}}-F(x_{j})\big) \bigg|\\
VI=&\max_{1\leq j\leq N} \sum_{t=1}^n \Big(F\big(x_{j+1} \gamma_t(u)\big)-F\big(x_{j} \gamma_t(u)\big)\Big)\\
VII=&\max_{1\leq j\leq N+1} \sum_{t=1}^n \Big|F\big(x_{j} \tilde{\gamma}_t(u)\big)-F\big(x_{j} \gamma_t(u)\big)\Big|.
\end{align*}
We look at each term in turn. For each $\varepsilon>0$, Markov's inequality implies
\begin{align*}
\PP\big[III\geq \sqrt{n}\varepsilon\big]
=& \PP\Big[\max_{1\leq j\leq N+1}\big|S_n(x_j,u)\big|^4\geq n^2\varepsilon^4\Big]
\leq \frac{1}{n^2\varepsilon^4} \EE\Big[\max_{1\leq j\leq N+1}\big|S_n(x_j,u)\big|^4\Big]\\
\leq&  \sum_{j=1}^{N+1}\frac{1}{n^2\varepsilon^4}\EE\Big[\big|S_n(x_j,u)\big|^4\Big]\to 0
\end{align*}
as $N=O(n^{3/4})$ and $\EE\big[|S_n(x,u)|^4\big]=O(n)$ by \textit{Step 1}. Thus, we have $III=o_p(\sqrt{n})$. Regarding $IV$, the mean value theorem and \eqref{eq:4.A.10} yield
\begin{align}
\label{eq:4.A.32}
F(x_{j+1})-F(x_j) = f(\breve{x}_j)(x_{j+1}-x_j)\leq M n^{-3/4},
\end{align}
where $\breve{x}_j \in(x_j,x_{j+1})$. It follows that
\begin{align*}
IV\leq n M n^{-3/4}= M n^{1/4}
\end{align*}
yielding $IV=O(n^{1/4})$. Further, Theorem 4.3.1 of \cite{csorgo1981strong} implies that there exists a sequence of Brownian bridges $\{B_n(y):0\leq y\leq 1\}$ such that
\begin{align*}
&V/\sqrt{n}=  \max_{1\leq j\leq N} \Big|\sqrt{n}\big(\mathbbm{F}_n(x_{j+1})-F(x_{j+1})\big)-\sqrt{n}\big(\mathbbm{F}_n(x_{j})-F(x_{j})\big) \Big|\\
 \leq & \max_{1\leq j\leq N} \Big|B_n\big(F(x_{j+1})\big)-B_n\big(F(x_{j})\big) \Big|+\max_{1\leq j\leq N} \Big|\sqrt{n}\big(\mathbbm{F}_n(x_{j})-F(x_{j})\big)-B_n\big(F(x_{j})\big) \Big|\\
 &\qquad +\max_{1\leq j\leq N} \Big|\sqrt{n}\big(\mathbbm{F}_n(x_{j+1})-F(x_{j+1})\big)-B_n\big(F(x_{j+1})\big) \Big|\\
  \leq & \max_{1\leq j\leq N} \Big|B_n\big(F(x_{j+1})\big)-B_n\big(F(x_{j})\big) \Big|+2\sup_{x \in \R} \Big|\sqrt{n}\big(\mathbbm{F}_n(x)-F(x)\big)-B_n\big(F(x)\big) \Big|\\
 %
\overset{a.s.}{=} & \max_{1\leq j\leq N} \Big|\underbrace{B_n\big(F(x_{j+1})\big)-B_n\big(F(x_{j})\big)}_{Z_{n,j}} \Big|+o(1).
\end{align*}
Next, we show that $ \max_{1\leq j\leq N} \big|Z_{n,j}\big| =o_p(1)$. By the definition of a Brownian bridge (cf. \citeauthor{csorgo1981strong}, \citeyear{csorgo1981strong}, p. 41), $Z_{n,j}$ is Gaussian with mean $0$ and variance
\begin{align*}
 \Var[Z_{n,j}]
  =& \big(F(x_{j+1})-F(x_{j})\big)\Big(\underbrace{1-\big(F(x_{j+1})-F(x_{j})\big)}_{\leq 1}\Big)\leq M n^{-3/4}
\end{align*}
by \eqref{eq:4.A.32}. In addition, we have $\EE\big[Z_{n,j}^4\big]=3 \big(\Var[Z_{n,j}]\big)^2\leq 3 M^2 n^{-3/2}$.
Thus, for each $\varepsilon>0$, Markov's inequality implies
\begin{align*}
&\PP\Big[\max_{1\leq j\leq N} \big|Z_{n,j}\big|\geq \varepsilon\Big]= \PP\Big[\max_{1\leq j\leq N} Z_{n,j}^4\geq \varepsilon^4\Big]\leq  \frac{1}{\varepsilon^4} \EE\Big[\max_{1\leq j\leq N} Z_{n,j}^4\Big]\\
\leq&  \frac{1}{\varepsilon^4} \EE\bigg[\sum_{j=1}^N Z_{n,j}^4\bigg]\leq \frac{1}{\varepsilon^4} \sum_{j=1}^N 3 M^2 n^{-3/2}
=  \frac{3M^2}{\varepsilon^4} n^{-3/2}N\to 0
\end{align*}
as $N=O(n^{3/4})$ and we conclude $ \max_{1\leq j\leq N} |Z_{n,j}| =o_p(1)$. Thus, $V=o_p(\sqrt{n})$. In analogy to \eqref{eq:4.A.17}, we bound $VI$ by
\begin{align}
\label{eq:4.A.33}
\begin{split}
VI\leq  \underbrace{ \sum_{t=1}^n \mathbbm{1}_{\{\mathscr{A}_{n,t}^c\}}}_{VI_1}+\underbrace{\max_{1\leq j\leq N} \sum_{t=1}^n \Big(F\big(x_{j+1} \gamma_t(u)\big)-F\big(x_{j} \gamma_t(u)\big)\Big)\mathbbm{1}_{\{\mathscr{A}_{n,t}\}}}_{VI_2}.
\end{split}
\end{align}
Concerning the first subterm, for each $\varepsilon>0$, Markov's inequality and \eqref{eq:4.A.18} lead to
\begin{align}
\label{eq:4.A.34}
&\PP\big[VI_1\geq \sqrt{n} \varepsilon\big] \leq \frac{1}{\sqrt{n}\varepsilon}\EE\bigg[\sum_{t=1}^n \mathbbm{1}_{\{\mathscr{A}_{n,t}^c\}}\bigg] = \frac{1}{\sqrt{n}\varepsilon}\sum_{t=1}^n\PP[\mathscr{A}_{n,t}^c]\\
\nonumber
\leq& \frac{A^2}{\sqrt{n}\varepsilon \delta^2} \ \Big(\EE\big[T_t^4\big]\Big)^{\frac{1}{2}}\Big(\EE\big[U_t^4\big]\Big)^{\frac{1}{2}} = O(n^{-1/2}).
\end{align}
Thus, we have $VI_1=o_p(\sqrt{n})$. Regarding $VI_2$, the mean value theorem implies
\begin{align*}
VI_2=&\max_{1\leq j\leq N} \sum_{t=1}^n \gamma_t(u)f\big(\tilde{x}_j \gamma_t(u)\big)(x_{j+1}-x_{j})\mathbbm{1}_{\{\mathscr{A}_{n,t}\}},
\end{align*}
where $\tilde{x}_j$ lies between $x_j$ and $x_{j+1}$. Since $|\gamma_t(u)-1| \leq \delta$ in the event of $\mathscr{A}_{n,t}$, we have $\tilde{x}_j \gamma_t(u) \in \mathcal{X}_\delta$. Employing \eqref{eq:4.A.10} and \eqref{eq:4.A.15}, we get
\begin{align*}
VI_2\leq \sum_{t=1}^n \bigg(1+ \frac{A}{\sqrt{n}}T_t U_t\bigg) M n^{-3/4} = M n^{1/4}+ \frac{A}{n^{1/4}}\frac{1}{n}\sum_{t=1}^n T_t U_t
\end{align*}
Whereas the first term is of order $O(n^{1/4})$, the second term vanishes almost surely as
\begin{align}
\label{eq:4.A.35}
\frac{1}{n}\sum_{t=1}^n T_t U_t\leq \bigg(\underbrace{\frac{1}{n}\sum_{t=1}^n T_t^2}_{\overset{a.s.}{\to}\EE[T_t^2]<\infty} \bigg)^{\frac{1}{2}} \bigg(\underbrace{\frac{1}{n}\sum_{t=1}^n U_t^2}_{\overset{a.s.}{\to}\EE[U_t^2]<\infty} \bigg)^{\frac{1}{2}}
\end{align}
by Markov's inequality, the uniform ergodic theorem and Assumption \ref{as:4.9}.  Hence, $VI_2=O(n^{1/4})$ almost surely. Next, we show
\begin{align}
\label{eq:4.A.36}
VII^\diamond=\sup_{||u||\leq A}\sup_{x \in \mathcal{X}} \sum_{t=1}^n \Big|F\big(x \tilde{\gamma}_t(u)\big)-F\big(x \gamma_t(u)\big)\Big|=O_p(1),
\end{align}
which implies $VII=O_p(1)$. Similar to \eqref{eq:4.A.17}, we bound $VII^\diamond$ by
\begin{align*}
VII^\diamond \leq& \underbrace{\sum_{t=1}^n \mathbbm{1}_{\{\mathscr{C}_{n,t}^c\}}}_{VII_1^\diamond} + \underbrace{\sup_{||u||\leq A}\sup_{x \in \mathcal{X}} \sum_{t=1}^n \Big|F\big(x_{j} \tilde{\gamma}_t(u)\big)-F\big(x_{j} \gamma_t(u)\big)\Big|\mathbbm{1}_{\{\mathscr{C}_{n,t}\}}}_{VII_2^\diamond}
\end{align*}
where the event $\mathscr{C}_{n,t}=\mathscr{A}_{n,t}\cap\mathscr{B}_t$ is defined in \eqref{eq:4.A.22}. We show that both terms are $O_p(1)$.  Employing Markov's inequality and \eqref{eq:4.A.23}, we have for each $C>0$
\begin{align}
\label{eq:4.A.37}
\begin{split}
\PP\big[VII_1^\diamond\geq C\big]\leq& \frac{1}{C} \EE\big[VII_1^\diamond\big]=\frac{1}{C}\sum_{t=1}^n \PP\big[\mathscr{C}_{n,t}^c\big] \leq \frac{1}{C}\sum_{t=1}^n \Big(\PP\big[\mathscr{A}_{n,t}^c\big]+\PP\big[\mathscr{B}_t^c\big]\Big)\\
\leq &\frac{1}{C}\sum_{t=1}^n \bigg(\frac{ A^2}{n\delta^2} \Big(\EE\big[T_t^4\big]\Big)^{\frac{1}{2}}\Big(\EE\big[U_t^4\big]\Big)^{\frac{1}{2}} + (\rho^{s/2})^t\frac{\EE[C_1^s]}{\delta^s\underline{\omega}^{s}}\bigg)\\
\leq& \frac{1}{C} \bigg(\frac{ A^2}{\delta^2} \Big(\EE\big[T_t^4\big]\Big)^{\frac{1}{2}}\Big(\EE\big[U_t^4\big]\Big)^{\frac{1}{2}} + \frac{\EE[C_1^s]}{\underline{\omega}^{s}\delta^s(1-\rho^{s/2})}\bigg).
\end{split}
\end{align}
Choosing $C$ sufficiently large, $\PP[VII_1^\diamond\geq C]$ can be made sufficiently small and we conclude $V_1^\diamond=O_p(1)$. 
Analogously to \eqref{eq:4.A.24} we obtain
\begin{align}
\label{eq:4.A.38}
\begin{split}
VII_2^\diamond=&\sup_{||u||\leq A}\sup_{x \in \mathcal{X}}  \sum_{t=1}^n \Big|x f(x\breve{\gamma}_{t})\big(\tilde{\gamma}_t(u)-\gamma_t(u)\big)\Big|\mathbbm{1}_{\{\mathscr{C}_{n,t}\}}\\
\leq&  \sum_{t=1}^n |\underline{x}| M \underbrace{\frac{C_1\rho^t}{\underline{\omega}}\mathbbm{1}_{\{\mathscr{C}_{n,t}\}}}_{\leq \delta \rho^{t/2}}\leq \sum_{t=1}^n |\underline{x}| M \delta \rho^{t/2}\leq \frac{2|\underline{x}| M \delta }{(1-\sqrt{\rho})^2}=O(1)
\end{split}
\end{align}
and we conclude $VII^\diamond=O_p(1)$.  \textit{Step 2} is completed.

In \textit{Step 3} we divide the (hyper-)cube $[-A,A]^r$  into $L=(2N)^r$ cubes with side length $A/N$ and $N \in \N$. In case of a cube $\ell$, $u_\bullet(\ell)$ and $u^\bullet(\ell)$ denote the lower left and upper right vertex of $\ell$.\footnote{\textit{Lower left} (\textit{right}) vertex means that all coordinates of $u_\bullet(\ell)$ ($u^\bullet(\ell)$) are less (larger) than or equal to the corresponding coordinates of any elements of $\ell$.} Similar to \eqref{eq:4.A.25}, we obtain 
\begin{align}
\label{eq:4.A.39}
&\sup_{||u||\leq A}\sup_{x \in \mathcal{X}}\big|S_n(x,u)\big|
\leq   \max_{1 \leq \ell\leq L} \sup_{x \in \mathcal{X}}  \big|S_n\big(x,u^\bullet(\ell)\big)\big|\\ 
\nonumber
&\qquad \qquad \qquad\qquad \qquad  \quad  + \max_{1 \leq \ell\leq L} \sup_{u_\bullet(\ell)\leq u\leq u^\bullet(\ell)}\sup_{x \in \mathcal{X}} \big|S_n(x,u)-S_n\big(x,u^\bullet(\ell)\big)\big|.
\end{align}
We focus on the second term. Fix $\ell\in\{1\dots,L\}$ and consider $u$ satisfying $u_\bullet(\ell)\leq u\leq u^\bullet(\ell)$ (element-by-element comparison). Assumption \ref{as:4.8} implies $\tilde{\gamma}_t(u_\bullet(\ell))\leq \tilde{\gamma}_t(u)\leq \tilde{\gamma}_t(u^\bullet(\ell))$. Since $x<0$ for all $x \in \mathcal{X}$, the elementary inequality \eqref{eq:4.A.26} implies
\begin{align}
\nonumber
&  \big|S_n(x,u)-S_n\big(x,u^\bullet(\ell)\big)\big|\\
\nonumber
=&\bigg|\sum_{t=1}^n \Big(\mathbbm{1}_{\{\eta_t\leq x\tilde{\gamma}_t(u)\}}- F\big(x\tilde{\gamma}_t(u)\big)-\big(\mathbbm{1}_{\{\eta_t\leq x\}}- F(x)\big)\Big)\\
\nonumber
& \qquad \qquad\qquad  -\sum_{t=1}^n \Big(\mathbbm{1}_{\{\eta_t\leq x\tilde{\gamma}_t(u^\bullet(\ell))\}}- F\big(x\tilde{\gamma}_t(u^\bullet(\ell))\big)-\big(\mathbbm{1}_{\{\eta_t\leq x\}}- F(x)\big)\Big)\bigg|\\
\label{eq:4.A.40}
=&\bigg|\underbrace{\sum_{t=1}^n \Big(\mathbbm{1}_{\{\eta_t\leq x\tilde{\gamma}_t(u)\}}-\mathbbm{1}_{\{\eta_t\leq x\tilde{\gamma}_t(u^\bullet(\ell))\}}\Big)}_{\geq 0}-\underbrace{\sum_{t=1}^n \Big(F\big(x\tilde{\gamma}_t(u)\big)-F\big(x\tilde{\gamma}_t(u^\bullet(\ell))\big)\Big)}_{\geq 0}\bigg|\\
\nonumber
\leq& \max \bigg\{\sum_{t=1}^n \Big(\mathbbm{1}_{\{\eta_t\leq x\tilde{\gamma}_t(u)\}}-\mathbbm{1}_{\{\eta_t\leq x\tilde{\gamma}_t(u^\bullet(\ell))\}}\Big),\sum_{t=1}^n \Big(F\big(x\tilde{\gamma}_t(u)\big)-F\big(x\tilde{\gamma}_t(u^\bullet(\ell))\big)\Big)\bigg\}\\
\nonumber
\leq& \max \bigg\{\underbrace{\sum_{t=1}^n \Big(\mathbbm{1}_{\{\eta_t\leq x\tilde{\gamma}_t(u_\bullet(\ell))\}}-\mathbbm{1}_{\{\eta_t\leq x\tilde{\gamma}_t(u^\bullet(\ell))\}}\Big)}_{=C_n},\underbrace{\sum_{t=1}^n \Big(F\big(x\tilde{\gamma}_t(u_\bullet(\ell))\big)-F\big(x\tilde{\gamma}_t(u^\bullet(\ell))\big)\Big)}_{=D_n}\bigg\}.
\end{align}
Note that $C_n$ can be written as
\begin{align}
\label{eq:4.A.41}
\begin{split}
C_n =&\sum_{t=1}^n \Big(\mathbbm{1}_{\{\eta_t\leq x\tilde{\gamma}_t(u_\bullet(\ell))\}}- F\big(x\tilde{\gamma}_t(u_\bullet(\ell))\big)-\big(\mathbbm{1}_{\{\eta_t\leq x\}}- F(x)\big)\Big)\\
&\quad -\sum_{t=1}^n \Big(\mathbbm{1}_{\{\eta_t\leq x\tilde{\gamma}_t(u^\bullet(\ell))\}}- F\big(x\tilde{\gamma}_t(u^\bullet(\ell))\big)-\big(\mathbbm{1}_{\{\eta_t\leq x\}}- F(x)\big)\Big)\\
&\qquad +\sum_{t=1}^n \Big(F\big(x\tilde{\gamma}_t(u_\bullet(\ell))\big)-F\big(x\tilde{\gamma}_t(u^\bullet(\ell))\big)\Big)\\
=&S_n\big(x,u_\bullet(\ell)\big)-S_n\big(x,u^\bullet(\ell)\big)+D_n.
\end{split}
\end{align}
Combining \eqref{eq:4.A.40} and \eqref{eq:4.A.41}, we find 
\begin{align}
\label{eq:4.A.42}
\big|S_n(x,u)-S_n\big(x,u^\bullet(\ell)\big)\big| \leq \big|S_n\big(x,u_\bullet(\ell)\big)\big| + \big|S_n\big(x,u^\bullet(\ell)\big)\big|+\big|D_n\big|.
\end{align}
Moreover, $D_n$ expands as follows:
\begin{align}
\label{eq:4.A.43}
\begin{split}
D_n  = & \sum_{t=1}^n \Big(F\big(x\gamma_t(u_\bullet(\ell))\big)-F\big(x\gamma_t(u^\bullet(\ell))\big)\Big)\\
&\qquad +\sum_{t=1}^n \Big(F\big(x\tilde{\gamma}_t(u_\bullet(\ell))\big)-F\big(x\gamma_t(u_\bullet(\ell))\big)\Big)\\
&\qquad \qquad  -\sum_{t=1}^n \Big(F\big(x\tilde{\gamma}_t(u^\bullet(\ell))\big)-F\big(x\gamma_t(u^\bullet(\ell))\big)\Big)
\end{split}
\end{align}
Equations \eqref{eq:4.A.39} and \eqref{eq:4.A.43} lead to
\begin{align}
\label{eq:4.A.44}
&\sup_{||u||\leq A}\sup_{x \in \mathcal{X}}\big|S_n(x,u)\big| \leq 2VIII+IX+X+XI+XII
\end{align}
with
\begin{align*}
VIII =&  \max_{1 \leq \ell\leq L} \sup_{x \in \mathcal{X}} |S_n\big(x,u^\bullet(\ell)\big)|\\
IX =&  \max_{1 \leq \ell\leq L} \sup_{x \in \mathcal{X}} |S_n\big(x,u_\bullet(\ell)\big)|\\
X = &  \sup_{x \in \mathcal{X}} \sum_{t=1}^n \Big|F\big(x\tilde{\gamma}_t(u_\bullet(\ell))\big)-F\big(x\gamma_t(u_\bullet(\ell))\big)\Big|\\
XI = &  \sup_{x \in \mathcal{X}} \sum_{t=1}^n \Big|F\big(x\tilde{\gamma}_t(u^\bullet(\ell))\big)-F\big(x\gamma_t(u^\bullet(\ell))\big)\Big|\\
XII =&  \max_{1 \leq \ell\leq L} \sup_{x \in \mathcal{X}} \sum_{t=1}^n \Big(F\big(x\gamma_t(u_\bullet(\ell))\big)-F\big(x\gamma_t(u^\bullet(\ell))\big)\Big).
\end{align*}
$VIII$ and $IX$ are $o_p(\sqrt{n})$ for fixed $L$ by \textit{Step 2} whereas $X=O_p(1)$ and $XI=O_p(1)$ by \eqref{eq:4.A.36}. In analogy to \eqref{eq:4.A.17}, we bound $XII$ by 
\begin{align}
\label{eq:4.A.45}
XII\leq \underbrace{ \sum_{t=1}^n \mathbbm{1}_{\{\mathscr{A}_{n,t}^c\}}}_{XII_1}+\underbrace{\max_{1\leq j\leq N} \sup_{x \in \mathcal{X}} \sum_{t=1}^n \Big(F\big(x\gamma_t(u_\bullet(\ell))\big)-F\big(x\gamma_t(u^\bullet(\ell))\big)\Big)\mathbbm{1}_{\{\mathscr{A}_{n,t}\}}}_{XII_2}.
\end{align}
We have $XII_1=o_p(\sqrt{n}) $ by \eqref{eq:4.A.34}. Regarding $XII_2$, the mean value theorem implies
\begin{align*}
XII_2= & \max_{1 \leq \ell\leq L}\sup_{x \in \mathcal{X}} \sum_{t=1}^n x f(x\bar{\gamma}_t) \big(\gamma_t(u_\bullet(\ell))-\gamma_t(u^\bullet(\ell))\big)\mathbbm{1}_{\{\mathscr{A}_{n,t}\}}
\end{align*}
with $\bar{\gamma}_t$ lying between $\gamma_t(u_\bullet(\ell))$ and $\gamma_t(u^\bullet(\ell))$. Since $|\bar{\gamma}_t-1| \leq 2\delta$ in the event of $\mathscr{A}_{n,t}$, we have $x\bar{\gamma}_t \in \mathcal{X}_\delta$ for all $x \in \mathcal{X}$. Taking $n$ sufficiently large such that $\big \{\theta:||\theta-\theta_0||\leq A/\sqrt{n}\big \} \subseteq \mathscr{V}(\theta_0)$, \eqref{eq:4.A.10} and the mean value theorem imply
\begin{align*}
XII_2\leq & |\underline{x}| M \max_{1 \leq \ell\leq L}\sup_{x \in \mathcal{X}} \sum_{t=1}^n \big(\gamma_t(u^\bullet(\ell))-\gamma_t(u_\bullet(\ell))\big)\\
= & |\underline{x}| M \max_{1 \leq \ell\leq L}  \sum_{t=1}^n \frac{\sigma_t(\theta_0+n^{-1/2}u^\bullet(\ell))-\sigma_t(\theta_0+n^{-1/2}u_\bullet(\ell))}{\sigma_t(\theta_0)}\\
=& |\underline{x}| M \max_{1 \leq \ell\leq L}  \sum_{t=1}^n \frac{1}{\sigma_t(\theta_0)}\frac{\partial \sigma_t(\bar{\theta}_n)}{\partial \theta'}\frac{1}{\sqrt{n}}\big(u^\bullet(\ell)-u_\bullet(\ell)\big)\\
\leq & \frac{|\underline{x}|M}{\sqrt{n}}\max_{1 \leq \ell\leq L} \sum_{t=1}^n  \frac{\sigma_t(\bar{\theta}_n)}{\sigma_t(\theta_0)}\bigg|\bigg|\frac{1}{\sigma_t(\bar{\theta}_n)}\frac{\partial \sigma_t(\bar{\theta}_n)}{\partial \theta}\bigg|\bigg|\big|\big|u^\bullet(\ell)-u_\bullet(\ell)\big|\big|\\
\leq& \frac{rA |\underline{x}|M}{\sqrt{n}N} \sum_{t=1}^n  \sup_{||\theta -\theta_0||\leq A/\sqrt{n}}\frac{\sigma_t(\theta)}{\sigma_t(\theta_0)}\sup_{||\theta -\theta_0||\leq A/\sqrt{n}}\big|\big|D_t(\theta)\big|\big|\\
\leq& \frac{rA |\underline{x}|M}{\sqrt{n}N} \sum_{t=1}^n T_t U_t,
\end{align*}
where $\theta_0+n^{-1/2}u_\bullet(\ell)\leq \bar{\theta}_n\leq\theta_0+n^{-1/2}u^\bullet(\ell)$ (componentwise). Employing \eqref{eq:4.A.35}, we obtain $XII_2= O(\sqrt{n})/N$ almost surely, where the $O(\sqrt{n})$ term does not depend on $N$. Choosing $N$ large, we obtain $XII_2= o(\sqrt{n})$ almost surely and we conclude that $XII=o_p(\sqrt{n})$. \textit{Step 3} is completed.

Regarding \textit{Step 4} we establish the following bound:
\begin{align}
\label{eq:4.A.46}
\begin{split}
&\sup_{||u||\leq A}\sup_{x \in \mathcal{X}}\bigg|\frac{1}{\sqrt{n}}\sum_{t=1}^n\Big(F\big(\tilde{\gamma}_t(u)x\big)-F(x)\Big)-xf(x)\Omega'u\bigg|\\
\leq& \underbrace{\sup_{||u||\leq A}\sup_{x \in \mathcal{X}}\bigg|\frac{1}{\sqrt{n}}\sum_{t=1}^n\Big(F\big(x\tilde{\gamma}_t(u)\big)-F\big(x\gamma_t(u)\big)\Big)\bigg|}_{=XIII}\\
&\qquad + \underbrace{\sup_{||u||\leq A}\sup_{x \in \mathcal{X}}\bigg|xf(x)\frac{1}{n}\sum_{t=1}^n D_t' u-xf(x)\Omega'u\bigg|}_{=XIV}
\end{split}\\
\nonumber
& \qquad \qquad  +\underbrace{\sup_{||u||\leq A}\sup_{x \in \mathcal{X}}\bigg|\frac{1}{\sqrt{n}}\sum_{t=1}^n\Big(F\big(x\gamma_t(u)\big)-F(x)\Big)-xf(x)\frac{1}{n}\sum_{t=1}^nD_t' u\bigg|}_{=XV},
\end{align}
where $XIII=O_p(n^{-1/2})$ by \eqref{eq:4.A.36}. Further, \eqref{eq:4.A.10} and the ergodic theorem imply
\begin{align*}
XIV\leq\sup_{||u||\leq A} \sup_{x \in \mathcal{X}}|x|f(x)\bigg|\bigg|\frac{1}{n}\sum_{t=1}^n D_t -\Omega\bigg|\bigg|\: ||u||
\leq A|\underline{x}|M \bigg|\bigg|\frac{1}{n}\sum_{t=1}^n D_t -\Omega\bigg|\bigg| \overset{a.s.}{\to}0.
\end{align*}
Regarding the last term, we use the mean value theorem and \eqref{eq:4.A.10} to obtain
\begin{align*}
XV=& \sup_{||u||\leq A}\sup_{x \in \mathcal{X}}\bigg|\frac{1}{\sqrt{n}}\sum_{t=1}^n\Big(xf(x\bar{\gamma}_t)\big(\gamma_t(u)-1\big)-xf(x)\frac{1}{\sqrt{n}}D_t' u\Big)\bigg|\\
\leq&\sup_{||u||\leq A}\sup_{x \in \mathcal{X}}\bigg|\frac{1}{\sqrt{n}}\sum_{t=1}^n\Big(xf(x)\big(\gamma_t(u)-1\big)-xf(x)\frac{1}{\sqrt{n}}D_t' u\Big)\bigg|\\
&\qquad +\sup_{||u||\leq A}\sup_{x \in \mathcal{X}}\bigg|\frac{1}{\sqrt{n}}\sum_{t=1}^n\Big(xf(x\bar{\gamma}_t)\big(\gamma_t(u)-1\big)-xf(x)\big(\gamma_t(u)-1\big)\Big)\bigg|\\
\leq&\underbrace{\frac{|\underline{x}|M}{\sqrt{n}}\sum_{t=1}^n\sup_{||u||\leq A}\bigg|\big(\gamma_t(u)-1\big)-\frac{1}{\sqrt{n}}D_t' u\bigg|}_{XV_1}\\
&\qquad + \underbrace{\sup_{||u||\leq A}\sup_{x \in \mathcal{X}}\bigg|\frac{1}{\sqrt{n}}\sum_{t=1}^n x\big(f(x\bar{\gamma}_t)-f(x)\big)\big(\gamma_t(u)-1\big)\bigg|}_{XV_2}
\end{align*}
with $\bar{\gamma}_t$ being between $\gamma_t(u)$ and $1$. For $n$ sufficiently large such that $\big \{\theta:||\theta-\theta_0||\leq A/\sqrt{n}\big \} \subseteq \mathscr{V}(\theta_0)$, a second-order Taylor expansion gives
\begin{align*}
XV_1=& \frac{|\underline{x}| M}{\sqrt{n}}\sum_{t=1}^n \sup_{||u||\leq A}\frac{1}{\sigma_t(\theta_0)}\bigg|\sigma_t(\theta_0+n^{-1/2}u)-\sigma_t(\theta_0)-\frac{1}{\sqrt{n}}\frac{\partial \sigma_t(\theta_0)}{\partial \theta'}u\bigg|\\
=&\frac{|\underline{x}| M}{\sqrt{n}}\sum_{t=1}^n \sup_{||u||\leq A}\frac{1}{\sigma_t(\theta_0)}\bigg|\frac{1}{2n}u'\frac{\partial^2 \sigma_t(\bar{\theta}_n)}{\partial \theta \partial \theta'}u \bigg|
\leq  \frac{A^2|\underline{x}|M}{2n^{3/2}}\sum_{t=1}^n \frac{\sigma_t(\bar{\theta}_n)}{\sigma_t(\theta_0)}\bigg|\bigg|\frac{1}{\sigma_t(\bar{\theta}_n)}\frac{\partial^2 \sigma_t(\bar{\theta}_n)}{\partial \theta \partial \theta'} \bigg|\bigg|\\
\leq & \frac{A^2 |\underline{x}| M}{2n^{3/2}}\sum_{t=1}^n\sup_{||\theta-\theta_0||\leq A/\sqrt{n}}\frac{\sigma_t(\theta)}{\sigma_t(\theta_0)}\sup_{||\theta-\theta_0||\leq A/\sqrt{n}}\big|\big|H_t(\theta) \big|\big|
\leq \frac{A^2 |\underline{x}| M}{2n^{3/2}}\sum_{t=1}^n T_t V_t
\end{align*}
with $\bar{\theta}_n$ being between $\theta_0$ and $\theta_0+n^{-1/2}u$. The Cauchy-Schwarz inequality, the uniform ergodic theorem and Assumption \ref{as:4.9} yield
\begin{align*}
\frac{1}{n}\sum_{t=1}^n T_t V_t \leq \bigg(\underbrace{\frac{1}{n}\sum_{t=1}^n T_t^2}_{\overset{a.s.}{\to}\EE[T_t^2]<\infty}\bigg)^{\frac{1}{2}}\bigg(\underbrace{\frac{1}{n}\sum_{t=1}^n V_t^2}_{\overset{a.s.}{\to}\EE[V_t^2]<\infty}\bigg)^{\frac{1}{2}}
\end{align*}
and we conclude that $XV_1= O(n^{-1/2})$ almost surely. Before turning to $XV_2$, we establish two auxiliary results:
\begin{enumerate}[(i)]
	\item $\frac{1}{\sqrt{n}}\sum_{t=1}^n\sup_{||u||\leq A} \big|\gamma_t(u)-1\big|=O(1)$ almost surely; 
    \item $\sup_{||u||\leq A} \sup_{x \in \mathcal{X}} \max_{1\leq t\leq n} \big|f(x\bar{\gamma}_t)-f(x)\big|=o_p(1)$.
\end{enumerate}
Statement (i) follows from \eqref{eq:4.A.15} and \eqref{eq:4.A.35} as
\begin{align*}
\frac{1}{\sqrt{n}}\sum_{t=1}^n\sup_{||u||\leq A} \big|\gamma_t(u)-1\big|\leq \frac{A}{n} \sum_{t=1}^n T_t U_t\leq A \bigg(\underbrace{\frac{1}{n} \sum_{t=1}^n T_t^2}_{\overset{a.s.}{\to}\EE[ T_t^2]<\infty} \bigg)^{\frac{1}{2}}\bigg(\underbrace{\frac{1}{n} \sum_{t=1}^n U_t^2}_{\overset{a.s.}{\to}\EE[ U_t^2]<\infty}  \bigg)^{\frac{1}{2}}.
\end{align*}
To show (ii), we note that the Cauchy-Schwarz inequality and Assumption \ref{as:4.9} yield $\EE\big[(T_tU_t)^{3}\big]\leq \EE\big[T_t^6\big]^{\frac{1}{2}}\EE\big[U_t^6\big]^{\frac{1}{2}}<\infty$. For every $\varepsilon>0$ and for $n$ sufficiently large such that $\big \{\theta:||\theta-\theta_0||\leq A/\sqrt{n}\big \} \subseteq \mathscr{V}(\theta_0)$, we have
\begin{align*}
&\PP\bigg[\sup_{||u||\leq A} \max_{1\leq t\leq n} \big|\gamma_t(u)-1\big|\geq \varepsilon\bigg]\leq \PP\bigg[A\max_{1\leq t\leq n} T_t U_t\geq \varepsilon \sqrt{n}\bigg]\\
\leq &\PP\bigg[A^3\max_{1\leq t\leq n} (T_t U_t)^3\geq \varepsilon^3 n^{3/2}\bigg]
\leq  \frac{A^3}{n^{3/2}\varepsilon^3}\EE\big[\max_{1\leq t\leq n} (T_t U_t)^3\big]
\leq  \frac{A^3}{\sqrt{n}\varepsilon^3}\EE\big[(T_t U_t)^3\big],
\end{align*}
which converges to $0$, and thus we obtain $\sup_{||u||\leq A} \max_{1\leq t\leq n} \big|\gamma_t(u)-1\big|=o_p(1)$. Because $\bar{\gamma}_t$ lies between $\gamma_t(u)$ and $1$, it follows that $\sup_{||u||\leq A} \max_{1\leq t\leq n} \big|\bar{\gamma}_t-1\big|=o_p(1)$. Thus, for sufficiently large $n$, we have $x\bar{\gamma}_t \in \mathcal{X}_\delta$ with probability close to one. Then, statement (ii) follows from the fact that $f$ is uniformly continuous on $\mathcal{X}_\delta$. Employing both auxiliary results, we obtain
\begin{align*}
XV_2\leq&\sup_{||u||\leq A} \sup_{x \in \mathcal{X}}\frac{1}{\sqrt{n}}\sum_{t=1}^n |x|\:\big|f(x\bar{\gamma}_t)-f(x)\big|\:\big|\gamma_t(u)-1\big|\\
\leq& |\underline{x}|\sup_{||u||\leq A} \sup_{x \in \mathcal{X}}\max_{1\leq t\leq n} \big|f(x\bar{\gamma}_t)-f(x)\big| \frac{1}{\sqrt{n}}\sum_{t=1}^n\sup_{||u||\leq A} \big|\gamma_t(u)-1\big|=o_p(1).
\end{align*}
Thus $XV$ is $o_p(1)$, which completes \textit{Step 4}.

Concerning \textit{Step 5} we obtain for each $\varepsilon>0$
\begin{align*}
&\PP\bigg[\sup_{x \in \mathcal{X}}\bigg|\frac{1}{\sqrt{n}}\sum_{t=1}^n\mathbbm{1}_{\{\hat{\eta}_t\leq x\}}-\frac{1}{\sqrt{n}}\sum_{t=1}^n\mathbbm{1}_{\{\eta_t\leq x\}}-xf(x)\Omega'\sqrt{n}\big(\hat{\theta}_n-\theta_0\big)\bigg|\geq \varepsilon\bigg]\\
\leq & \PP\bigg[\sup_{||u||\leq A}\sup_{x \in \mathcal{X}}\bigg|\frac{1}{\sqrt{n}}\sum_{t=1}^n\mathbbm{1}_{\{\eta_t\leq \tilde{\gamma}_t(u) x\}}-\frac{1}{\sqrt{n}}\sum_{t=1}^n\mathbbm{1}_{\{\eta_t\leq x\}}-xf(x)\Omega'u\bigg|\geq \varepsilon\bigg]\\
 &\qquad +\PP\Big[\sqrt{n}||\hat{\theta}_n-\theta_0||> A\Big]\\
 \leq & \PP\bigg[\sup_{||u||\leq A}\sup_{x \in \mathcal{X}}\bigg|\frac{1}{\sqrt{n}}\sum_{t=1}^n\Big(F\big(\tilde{\gamma}_t(u)x\big)-F(x)\Big)-xf(x)\Omega'u\bigg|\geq \frac{\varepsilon}{2}\bigg]\\
 &\qquad + \PP\bigg[\sup_{||u||\leq A}\sup_{x \in \mathcal{X}}\big|S_n(x,u)/\sqrt{n}\big|\geq \frac{\varepsilon}{2}\bigg] +\PP\Big[\sqrt{n}||\hat{\theta}_n-\theta_0||> A\Big].
\end{align*}
Since $\sqrt{n}||\hat{\theta}_n-\theta_0||=O_p(1)$ by Theorem \ref{thm:4.2}, the third term can be made arbitrarily small for large $n$ by choosing $A$ sufficiently large. Given $A$, the first two terms converge to zero by \textit{Step 3} and \textit{Step 4}, which completes \textit{Step 5}.

Regarding \textit{Step 6} we refer to \citeauthor{bahadur1966note} (\citeyear{bahadur1966note}, Lemma 1). Replacing $\xi$ by $\xi_\alpha$ in the proof and choosing the sequences $a_n$ and $b_n$ to satisfy $a_n \sim n^{-\varrho} \log n$ and $b_n \sim n^\psi$ as $n\to \infty$, where $\psi =(1-\varrho)/2$, it follows that
\begin{align*}
\mathbbm{H}_{n,\alpha}=\sup\limits_{x\in \mathcal{I}_n} \Big|  \big(\mathbbm{F}_n(x)-\mathbbm{F}_n(\xi_\alpha)\big)-  \big(F(x)-F(\xi_\alpha)\big)\Big|=O\big(n^{-(\varrho+\psi)}\log n\big)
\end{align*}
almost surely as $n \to \infty$. Inserting the definition of $\psi$ and inflating the term by $\sqrt{n}$ leads to $\sqrt{n}\:\mathbbm{H}_{n,\alpha}=O\big(n^{-\varrho/2}\log n\big)$ almost surely as $n \to \infty$. Together with the triangle inequality, we establish 
\begin{align*}
\sup\limits_{x,y\in \mathcal{I}_n} \Big| \sqrt{n} \big(\mathbbm{F}_n(x)-\mathbbm{F}_n(y)\big)- \sqrt{n} \big(F(x)-F(y)\big)\Big|\leq 2 \sqrt{n} \: \mathbbm{H}_{n,\alpha}=O\big(n^{-\varrho/2}\log n\big),
\end{align*}
which completes \textit{Step 6}.

Regarding \textit{Step 7} we bound
\begin{align*}
&\sup_{x,y \in \mathcal{I}_n} \Big| \sqrt{n} \big(\hat{\mathbbm{F}}_n(x)-\hat{\mathbbm{F}}_n(y)\big)- \sqrt{n} \big(F(x)-F(y)\big)\Big|\\
\leq & 2\sup_{x \in \mathcal{I}_n}\Big|\sqrt{n}\big(\hat{\mathbbm{F}}_n(x)-\mathbbm{F}_n(x)\big)-xf(x)\Omega'\sqrt{n}\big(\hat{\theta}_n-\theta_0\big)\Big|\\
& \quad  + \sup_{x,y \in \mathcal{I}_n} \Big| \sqrt{n} \big(\mathbbm{F}_n(x)-\mathbbm{F}_n(y)\big)-\sqrt{n} \big(F(x)-F(y)\big)\Big|\\
 & \qquad  + \sup_{x,y \in \mathcal{I}_n} \Big| \big(xf(x)-yf(y)\big)\Omega'\sqrt{n}\big(\hat{\theta}_n-\theta_0\big)\Big|.
\end{align*}
Taking $n$ sufficiently large such that $\mathcal{I}_n \subset \mathcal{X}$, the first term on the right-hand side vanishes in probability by \textit{Step 5}. The second term vanishes almost surely by \textit{Step 6}. The last term can be bounded as follows:
\begin{align*}
\sup_{x,y \in \mathcal{I}_n} \Big| \big(xf(x)-yf(y)\big)\Omega'\sqrt{n}\big(\hat{\theta}_n-\theta_0\big)\Big|\leq \sup_{x,y \in \mathcal{I}_n} \big|xf(x)-yf(y)\big|\:||\Omega||\:\sqrt{n}\big|\big|\hat{\theta}_n-\theta_0\big|\big|.
\end{align*}
Since $f(x)$, and hence $xf(x)$, is continuous in a neighborhood of $\xi_\alpha$ by Assumption \ref{as:4.5}(\ref{as:4.5.2}) and $\mathcal{I}_n$ shrinks to $\xi_\alpha$ we have $\sup_{x,y \in \mathcal{I}_n} \big|xf(x)-yf(y)\big|\to 0$. Together with $\sqrt{n}\big|\big|\hat{\theta}_n-\theta_0\big|\big|=O_p(1)$  (Theorem \ref{thm:4.2}) we find that the last term converges in probability to $0$, which completes \textit{Step 7}.

To verify that replacing any $\hat{\mathbbm{F}}_n(\cdot)$ by $\hat{\mathbbm{F}}_n(\cdot\:-)$  does not alter the result, we note that
$\hat{\mathbbm{F}}_n\big(x-n^{-1}\big)\leq \hat{\mathbbm{F}}_n(x-)\leq \hat{\mathbbm{F}}_n(x)\leq \hat{\mathbbm{F}}_n\big(x+n^{-1}\big)$ for all $x \in \mathcal{I}_n$ (similarly for $y$). Setting $\bar{\mathcal{I}}_n=(\xi_\alpha - \bar{a}_n,\xi_\alpha + \bar{a}_n)$ with $\bar{a}_n =a_n +n^{-1}$, we can bound $\sup\limits_{x,y \in \mathcal{I}_n} \big| \sqrt{n} \big(\hat{\mathbbm{F}}_n(x-)-\hat{\mathbbm{F}}_n(y)\big)- \sqrt{n} \big(F(x)-F(y)\big)\big|$ and $\sup\limits_{x,y \in \mathcal{I}_n} \big| \sqrt{n} \big(\hat{\mathbbm{F}}_n(x-)-\hat{\mathbbm{F}}_n(y-)\big)- \sqrt{n} \big(F(x)-F(y)\big)\big|$ by
\begin{align}
\label{eq:4.A.47}
\begin{split}
&\sup_{x,y \in \bar{\mathcal{I}}_n} \Big| \sqrt{n} \big(\hat{\mathbbm{F}}_n(x)-\hat{\mathbbm{F}}_n(y)\big)- \sqrt{n} \big(F(x)-F(y)\big)\Big|\\
&\qquad \qquad \qquad \qquad+2 \sup_{y \in \mathcal{I}_n} \sqrt{n} \Big(F\big(y+n^{-1}\big)-F\big(y-n^{-1}\big)\Big).
\end{split}
\end{align}
The first term in \eqref{eq:4.A.47} vanishes in probability by \textit{Step 7} as $\bar{a}_n \sim a_n$. Regarding the second term, the mean value theorem implies
\begin{align*}
2\sup_{y \in \mathcal{I}_n} \sqrt{n} \bigg(F\Big(y+\frac{1}{n}\Big)-F\Big(y-\frac{1}{n}\Big)\bigg) = \frac{4}{\sqrt{n}} \sup_{y \in \mathcal{I}_n}  f\big(y+\varepsilon_n\big),
\end{align*}
where $|\varepsilon_n|\leq n^{-1}$. Since $\frac{4}{\sqrt{n}}\to 0$ and $\sup_{y \in \mathcal{I}_n}  f(y+\varepsilon_n)\to f(\xi_\alpha)$ the term vanishes, which completes the proof.
\end{proof}

\subsection{Bootstrap Lemmas}
\label{app:4.A.2}

Henceforth we use $\PP^*$, $\EE^*$, $\Var^*$ and $\Cov^*$ to denote the probability, expectation, variance and covariance conditional on $\mathcal{F}_n$.

\begin{lemma}
\label{lem:4.4}
Suppose Assumptions \ref{as:4.1}--\ref{as:4.3}, \ref{as:4.4}(\ref{as:4.4.1}), \ref{as:4.5}(\ref{as:4.5.1}) and \ref{as:4.5}(\ref{as:4.5.3}) hold.
\begin{enumerate}[(i)]

\item If in addition Assumption \ref{as:4.9}(i) holds with $a=4$, then  $\EE^*[\eta_t^{*m}]\overset{a.s.}{\to}\EE[\eta_t^m]$ for  $ m \in \{1,2,3,4\}$.

    \item If in addition Assumptions \ref{as:4.6}, \ref{as:4.7} and \ref{as:4.9}(i) hold with $a=-1,4$, then we have $\EE^*[\eta_t^{*m}\mathbbm{1}_{\{\eta_t^*< \hat{\xi}_{n,\alpha}\}}]\overset{a.s.}{\to}\EE[\eta_t^{m}\mathbbm{1}_{\{\eta_t< \xi_\alpha\}}]$ for  $ m \in \{0,1,2,3,4\}$.
\end{enumerate}
\end{lemma}

\begin{proof}
Lemma \ref{lem:4.2} gives 
$\EE^*[\eta_t^{*m}\mathbbm{1}_{\{\eta_t^*< u\}}]=\frac{1}{n}\sum_{t=1}^n\hat{\eta}_t^{m}\mathbbm{1}_{\{\hat{\eta}_t< u\}} \overset{a.s.}{\to}\EE[\eta_t^{m}\mathbbm{1}_{\{\eta_t< u\}}]$. Taking $u= \infty$ proves the first claim, whereas the second claim follows from $\EE[\eta_t^{m}\mathbbm{1}_{\{\eta_t< u\}}]$ being continuous in $u$ and $\hat{\xi}_{n,\alpha} \overset{a.s.}{\to}\xi_\alpha$ by Theorem \ref{thm:4.1}.
\end{proof}

\begin{lemma}
\label{lem:4.5}
Suppose Assumptions \ref{as:4.1}--\ref{as:4.3}, \ref{as:4.4}(\ref{as:4.4.1}), \ref{as:4.5}(\ref{as:4.5.1}), \ref{as:4.5}(\ref{as:4.5.3}), \ref{as:4.6} and \ref{as:4.9}(i)--(ii) hold with $a=\pm 4$. Then, we have $\hat{\theta}_n^* \overset{p^*}{\to}\theta_0$ almost surely.
\end{lemma}

\begin{proof}
The proof is inspired by \citeauthor{francq2004maximum} (\citeyear{francq2004maximum},  Theorem 2.1). 
Let $\nu>0$ and set $\mathscr{B} =\{\theta \in \Theta: ||\theta-\theta_0||\geq \nu\}$;
We establish the result in three steps:

\begin{enumerate}
	\item[] \textit{Step 1:} we obtain $L_n^*(\theta)-L_n^*(\hat{\theta}_n)= \frac{1}{2n}\sum_{t=1}^n\Big(1-\frac{\sigma_t^2(\hat{\theta}_n)}{\sigma_t^2(\theta)}\eta_t^{*2}+\log \frac{\sigma_t^2(\hat{\theta}_n)}{\sigma_t^2(\theta)}\Big)+R_n^*(\theta)$ with $\sup_{\theta \in \Theta} \big|R_n^*(\theta)\big|\overset{p^*}{\to}0 $ almost surely;

    \item[] \textit{Step 2:} There exists a $\zeta<0$ such that $\sup_{\theta \in \mathscr{B}} L_n^*(\theta) -L_n^*(\hat{\theta}_n)<\zeta/2+S_n^*$ with $S_n^*\overset{p^*}{\to}0 $ almost surely;
    
    \item[] \textit{Step 3:} we show $\PP^*\big[\hat{\theta}_n^* \in \mathscr{B}\big]\overset{a.s.}{\to}0$. 
\end{enumerate}
Regarding \textit{Step 1} we  find
\begin{align*}
 L_n^*(\theta)-L_n^*(\hat{\theta}_n)=\frac{1}{2n}\sum_{t=1}^n\bigg\{\eta_t^{*2}-\frac{\tilde{\sigma}_t^2(\hat{\theta}_n)}{\tilde{\sigma}_t^2(\theta)}\eta_t^{*2}+\log \frac{\tilde{\sigma}_t^2(\hat{\theta}_n)}{\tilde{\sigma}_t^2(\theta)}\bigg\},
\end{align*}
where $ \frac{1}{n}\sum_{t=1}^n \eta_t^{*2}\overset{p^*}{\to}1$ almost surely since
\begin{align*}
\EE^*\bigg[\frac{1}{n}\sum_{t=1}^n \eta_t^{*2}\bigg] = \EE^*\big[ \eta_t^{*2}\big]\overset{a.s.}{\to} 1 \qquad \text{and} \qquad  \Var^*\bigg[\frac{1}{n}\sum_{t=1}^n \eta_t^{*2}\bigg] = \frac{1}{n}\Var^*\big[ \eta_t^{*2}\big]\overset{a.s.}{\to} 0
\end{align*}
by Lemma \ref{lem:4.4}. It remains to show the negligibility of the initial conditions, i.e.\ 

\begin{align}
\label{eq:4.A.48}
\sup_{\theta \in  \Theta}\bigg|\frac{1}{n}\sum_{t=1}^n\bigg\{\log \frac{\tilde{\sigma}_t^2(\hat{\theta}_n)}{\tilde{\sigma}_t^2(\theta)}-\log \frac{\sigma_t^2(\hat{\theta}_n)}{\sigma_t^2(\theta)}\bigg\}\bigg|\overset{a.s}{\to}0 
\end{align}
and
\begin{align}
\label{eq:4.A.49}
\sup_{\theta \in  \Theta}\bigg|\frac{1}{n}\sum_{t=1}^n \bigg(\frac{\sigma_t^2(\hat{\theta}_n)}{\sigma_t^2(\theta)}-\frac{\tilde{\sigma}_t^2(\hat{\theta}_n)}{\tilde{\sigma}_t^2(\theta)}\bigg)\eta_t^{*2}\bigg|\overset{p^*}{\to}0 
\end{align}
almost surely. The  inequality $\log(1+x)\leq x$ for all $x>-1$ and Assumption \ref{as:4.4}(\ref{as:4.4.1}) yield
\begin{align*}
&\sup_{\theta \in \Theta}\bigg|\frac{1}{n} \sum_{t=1}^n  \bigg(\log \frac{\sigma_t^2(\hat{\theta}_n)}{\sigma_t^2(\theta)}-\log \frac{\tilde{\sigma}_t^2(\hat{\theta}_n)}{\tilde{\sigma}_t^2(\theta)}\bigg)\bigg|= \sup_{\theta \in \Theta} \bigg| \frac{1}{n} \sum_{t=1}^n \bigg( \log \frac{\tilde{\sigma}_t^2(\theta)}{\sigma_t^2(\theta)}-\log \frac{\tilde{\sigma}_t^2(\hat{\theta}_n)}{\sigma_t^2(\hat{\theta}_n)}\bigg)\bigg|\\
\leq& \sup_{\theta \in \Theta} \frac{2}{n} \sum_{t=1}^n \bigg| \log \frac{\tilde{\sigma}_t^2(\theta)}{\sigma_t^2(\theta)}\bigg|= \sup_{\theta \in \Theta} \frac{4}{n} \sum_{t=1}^n \bigg| \log \frac{\tilde{\sigma}_t(\theta)}{\sigma_t(\theta)}\bigg|= \sup_{\theta \in \Theta} \frac{4}{n} \sum_{t=1}^n \bigg| \log \bigg(1+ \frac{\tilde{\sigma}_t(\theta)-\sigma_t(\theta)}{\sigma_t(\theta)}\bigg)\bigg|\\
\leq&  \frac{4}{n} \sum_{t=1}^n  \log \bigg(1+ \frac{C_1\rho^t}{\underline{\omega}}\bigg)\leq \frac{4}{n} \sum_{t=1}^n  \frac{C_1\rho^t}{\underline{\omega}}\leq  \frac{4C_1}{\underline{\omega}(1-\rho)n}\overset{a.s.}{\to}0
\end{align*}
verifying \eqref{eq:4.A.48}. Further, Assumption \ref{as:4.4}(\ref{as:4.4.1}) and \eqref{eq:4.A.5} imply
\begin{align*}
&\sup_{\theta \in \Theta}\bigg|\frac{1}{n} \sum_{t=1}^n  \bigg(\frac{\tilde{\sigma}_t^2(\hat{\theta}_n)}{\tilde{\sigma}_t^2(\theta)}-\frac{\sigma_t^2(\hat{\theta}_n)}{\sigma_t^2(\theta)}\bigg)\eta_t^{*2}\bigg|\leq \sup_{\theta \in \Theta}\frac{1}{n}\sum_{t=1}^n\bigg|\frac{\tilde{\sigma}_t^2(\hat{\theta}_n)}{\tilde{\sigma}_t^2(\theta)}-\frac{\sigma_t^2(\hat{\theta}_n)}{\sigma_t^2(\theta)}\bigg|\eta_t^{*2}\\
=&\sup_{\theta \in \Theta} \frac{1}{n}\sum_{t=1}^n \frac{\sigma_t^2(\hat{\theta}_n)}{\tilde{\sigma}_t^2(\theta)}\bigg|\frac{\tilde{\sigma}_t^2(\hat{\theta}_n)-\sigma_t^2(\hat{\theta}_n)}{\sigma_t^2(\hat{\theta}_n)}+\frac{\sigma_t^2(\theta)-\tilde{\sigma}_t^2(\theta)}{\sigma_t^2(\theta)}\bigg|\eta_t^{*2}\\
\leq& \sup_{\theta \in \Theta}\frac{1}{n}\sum_{t=1}^n \frac{\sigma_t^2(\hat{\theta}_n)}{\tilde{\sigma}_t^2(\theta)}\bigg(\frac{|\tilde{\sigma}_t^2(\hat{\theta}_n)-\sigma_t^2(\hat{\theta}_n)|}{\sigma_t^2(\hat{\theta}_n)}+\frac{|\sigma_t^2(\theta)-\tilde{\sigma}_t^2(\theta)|}{\sigma_t^2(\theta)}\bigg)\eta_t^{*2}\\
\leq& \sup_{\theta \in \Theta}\frac{1}{n}\sum_{t=1}^n \frac{\sigma_t^2(\hat{\theta}_n)}{\tilde{\sigma}_t^2(\theta)}\bigg(\frac{|\tilde{\sigma}_t(\hat{\theta}_n)-\sigma_t(\hat{\theta}_n)|^2}{\sigma_t^2(\hat{\theta}_n)}+2\frac{|\tilde{\sigma}_t(\hat{\theta}_n)-\sigma_t(\hat{\theta}_n)|}{\sigma_t(\hat{\theta}_n)}\\
&\qquad \qquad \qquad \qquad \qquad +\frac{|\sigma_t(\theta)-\tilde{\sigma}_t(\theta)|^2}{\sigma_t^2(\theta)}+2\frac{|\sigma_t(\theta)-\tilde{\sigma}_t(\theta)|}{\sigma_t(\theta)}\bigg)\eta_t^{*2}
\\
\leq& \frac{1}{n}\sum_{t=1}^n \frac{\sigma_t^2(\hat{\theta}_n)}{\underline{\omega}^2}\bigg(\frac{C_1^2 \rho^{2t}}{\underline{\omega}^2}+2\frac{C_1 \rho^t}{\underline{\omega}}+\frac{C_1^2 \rho^{2t}}{\underline{\omega}^2}+2\frac{C_1 \rho^t}{\underline{\omega}}\bigg)\eta_t^{*2}\\
\leq& \bigg(\frac{2C_1^2}{\underline{\omega}^4} + \frac{4C_1}{\underline{\omega}^{3}}\bigg)\frac{1}{n}\sum_{t=1}^n \rho^t \sigma_t^2 (\hat{\theta}_n)\eta_t^{*2}.
\end{align*}
To verify \eqref{eq:4.A.49} we are left to show that $\frac{1}{n}\sum_{t=1}^n \rho^t \sigma_t^2 (\hat{\theta}_n)\eta_t^{*2}\overset{p^*}{\to}0$ almost surely. For every $\varepsilon>0$, Markov's inequality implies
\begin{align*}
\PP^*\bigg[\frac{1}{n}\sum_{t=1}^n \rho^t \sigma_t^2 (\hat{\theta}_n)\eta_t^{*2} \geq \varepsilon\bigg] \leq \frac{1}{\varepsilon} \frac{1}{n}\sum_{t=1}^n \rho^t \sigma_t^2 (\hat{\theta}_n)\EE^*\big[\eta_t^{*2}\big]
\end{align*}
As $\EE^*\big[\eta_t^{*2}\big]\overset{a.s.}{\to}1$ (Lemma \ref{lem:4.4}), it remains to show that $\frac{1}{n}\sum_{t=1}^n \rho^{t} \sigma_t^{2} (\hat{\theta}_n)\overset{a.s.}{\to}0$. We have
\begin{align*}
\frac{1}{n}\sum_{t=1}^n \rho^{t} \sigma_t^{2} (\hat{\theta}_n)=\frac{1}{n}\sum_{t=1}^n \rho^{t} \sigma_t^{2} (\theta_0)\frac{\sigma_t^{2} (\hat{\theta}_n)}{\sigma_t^{2}(\theta_0)}\leq \bigg(\frac{1}{n}\sum_{t=1}^n \rho^{2 t} \sigma_t^{4} (\theta_0)\bigg)^{\frac{1}{2}} \bigg(\frac{1}{n}\sum_{t=1}^n  \frac{\sigma_t^{4} (\hat{\theta}_n)}{\sigma_t^{4}(\theta_0)}\bigg)^{\frac{1}{2}}
\end{align*}
by the Cauchy-Schwarz inequality. Since $\hat{\theta}_n\overset{a.s.}{\to} \theta_0$ (Theorem \ref{thm:4.1}) such that $\hat{\theta}_n \in  \mathscr{V}(\theta_0)$ almost surely, the uniform ergodic theorem and  Assumption \ref{as:4.9}(i) result in
\begin{align*}
\frac{1}{n}\sum_{t=1}^n \frac{\sigma_t^{4} (\hat{\theta}_n)}{\sigma_t^{4} (\theta_0)} \overset{a.s.}{\leq} \frac{1}{n}\sum_{t=1}^n T_t^{4}  \overset{a.s.}{\to} \EE \big[T_t^{4}\big]<\infty.
\end{align*}
In addition, we have for $\delta>0$
\begin{align*}
\sum_{t=1}^\infty \PP\big[\rho^{2 t} \sigma_t^{4} (\theta_0)>\delta \big]  \leq \sum_{t=1}^\infty \frac{\rho^{st/2}\EE[\sigma_t^s (\theta_0)]}{\delta^{s/(4)}}=   \frac{\EE[\sigma_t^s (\theta_0)]}{\delta^{s/(4)}(1-\rho^{s/2})}<\infty
\end{align*}
such that the Borel-Cantelli Lemma implies $\rho^{2 t} \sigma_t^{4}(\theta_0)\overset{a.s.}{\to} 0$ as $t\to \infty$. Therefore, $\frac{1}{n}\sum_{t=1}^n  \rho^{2 t} \sigma_t^{4} (\theta_0)\overset{a.s.}{\to}0$ follows by Ces\'aro's lemma. Combining results, we establish  $\frac{1}{n}\sum_{t=1}^n \rho^{t} \sigma_t^{2} (\hat{\theta}_n)\overset{a.s.}{\to}0$, which verifies  \eqref{eq:4.A.49} and completes \textit{Step 1}.

Consider \textit{Step 2}; by compactness of $\mathscr{B}$ the Heine-Borel theorem entails that there exists a finite number of neighborhoods of size smaller than $1/k$, i.e.\ $\mathscr{V}_k(\theta_1),\dots,\mathscr{V}_k(\theta_K)$ with $K=K(k) \in \N$, covering $\mathscr{B}$. We have
\begin{align*}
\sup_{\theta \in  \mathscr{B}} L_n^*(\theta)-L_n^*(\hat{\theta}_n)
=& \max_{i=1,\dots,K} \sup_{\theta \in \mathscr{V}_k(\theta_i)\cap \mathscr{B}} L_n^*(\theta)-L_n^*(\hat{\theta}_n).
\end{align*}
Next, we fix $i\in \{1,\dots,K\}$. With regard to \textit{Step 1}, we obtain for each $M>1$
\begin{align*}
& L_n^*(\theta)-L_n^*(\hat{\theta}_n)\\
=& \frac{1}{2n}\sum_{t=1}^n\mathbbm{1}_{\big\{\frac{\sigma_t^2(\hat{\theta}_n)}{\sigma_t^2(\theta)}> M\big\}} \bigg(1-\underbrace{\frac{\sigma_t^2(\hat{\theta}_n)}{\sigma_t^2(\theta)}\eta_t^{*2}}_{\geq 0}+\log \frac{\sigma_t^2(\hat{\theta}_n)}{\sigma_t^2(\theta)}\bigg)\\
&\qquad +\frac{1}{2n}\sum_{t=1}^n\mathbbm{1}_{\big\{\frac{\sigma_t^2(\hat{\theta}_n)}{\sigma_t^2(\theta)}\leq M\big\}} \bigg(1-\frac{\sigma_t^2(\hat{\theta}_n)}{\sigma_t^2(\theta)}\eta_t^{*2}+\log \frac{\sigma_t^2(\hat{\theta}_n)}{\sigma_t^2(\theta)}\bigg)+R_n^*(\theta)\\ 
\leq& \frac{1}{2n}\sum_{t=1}^n\mathbbm{1}_{\big\{\frac{\sigma_t^2(\hat{\theta}_n)}{\sigma_t^2(\theta)}> M\big\}} \bigg(1+\log \frac{\sigma_t^2(\hat{\theta}_n)}{\sigma_t^2(\theta)}\bigg)+\frac{1}{2n}\sum_{t=1}^n\mathbbm{1}_{\big\{\frac{\sigma_t^2(\hat{\theta}_n)}{\sigma_t^2(\theta)}\leq M\big\}} \frac{\sigma_t^2(\hat{\theta}_n)}{\sigma_t^2(\theta)}\big(1-\eta_t^{*2}\big)\\
&\qquad +\frac{1}{2n}\sum_{t=1}^n\mathbbm{1}_{\big\{\frac{\sigma_t^2(\hat{\theta}_n)}{\sigma_t^2(\theta)}\leq M\big\}} \bigg(1-\frac{\sigma_t^2(\hat{\theta}_n)}{\sigma_t^2(\theta)}+\log \frac{\sigma_t^2(\hat{\theta}_n)}{\sigma_t^2(\theta)}\bigg)+R_n^*(\theta)
\end{align*}
such that
\begin{align*}
&\sup_{\theta \in \mathscr{V}_k(\theta_i)\cap \mathscr{B}} L_n^*(\theta)-L_n^*(\hat{\theta}_n)\\
\overset{a.s.}{\leq} &\frac{1}{2}\underbrace{\frac{1}{n}\sum_{t=1}^n\sup_{\substack{||\dot{\theta}-\theta_0||\leq 1/k \\ ||\theta-\theta_i||\leq 1/k}}\mathbbm{1}_{\big\{\frac{\sigma_t^2(\dot{\theta})}{\sigma_t^2(\theta)}> M\big\}} \bigg(1+\log \frac{\sigma_t^2(\dot{\theta})}{\sigma_t^2(\theta)}\bigg)}_{I}\\
&\quad +\frac{1}{2}\underbrace{\frac{1}{n}\sum_{t=1}^n \sup_{\substack{||\dot{\theta}-\theta_0||\leq 1/k \\ ||\theta-\theta_i||\leq 1/k}}\mathbbm{1}_{\big\{\frac{\sigma_t^2(\dot{\theta})}{\sigma_t^2(\theta)}\leq M\big\}} \frac{\sigma_t^2(\dot{\theta})}{\sigma_t^2(\theta)}\big(1-\eta_t^{*2}\big)}_{II}\\
&\qquad +\frac{1}{2}\underbrace{\frac{1}{n}\sum_{t=1}^n\sup_{\substack{||\dot{\theta}-\theta_0||\leq 1/k \\ ||\theta-\theta_i||\leq 1/k}}\mathbbm{1}_{\big\{\frac{\sigma_t^2(\dot{\theta})}{\sigma_t^2(\theta)}\leq M\big\}} \bigg(1-\frac{\sigma_t^2(\dot{\theta})}{\sigma_t^2(\theta)}+\log \frac{\sigma_t^2(\dot{\theta})}{\sigma_t^2(\theta)}\bigg)}_{III}+\underbrace{\sup_{\theta \in \Theta}\big|R_n^*(\theta)\big|}_{IV}.
\end{align*}
Subsequently, we consider each term in turn. 
Regarding $I$, take $k$ sufficiently large such that $\dot{\theta}$ satisfying $||\dot{\theta}-\theta_0||\leq 1/k$ yields $\dot{\theta} \in \mathscr{V}(\theta_0)$. The uniform ergodic theorem, the inequality $\log(x) \leq x$ for all $x>0$ and the Cauchy-Schwarz inequality imply
\begin{align*}
I\overset{a.s.}{\to} & \EE\Bigg[\sup_{\substack{||\dot{\theta}-\theta_0||\leq 1/k \\ ||\theta-\theta_i||\leq 1/k}}\mathbbm{1}_{\big\{\frac{\sigma_t^2(\dot{\theta})}{\sigma_t^2(\theta)}> M\big\}} \bigg(1+\log \frac{\sigma_t^2(\dot{\theta})}{\sigma_t^2(\theta)}\bigg)\Bigg] \leq \EE\bigg[  \mathbbm{1}_{\{\sigma_t^2T_t^2>M \underline{\omega}^2\}}\bigg(1+\log \frac{\sigma_t^2T_t^2}{\underline{\omega}^2}\bigg)\bigg]\\
=& \EE\bigg[  \mathbbm{1}_{\{\sigma_t^2T_t^2>M \underline{\omega}^2\}}\Big(1-2\log \underline{\omega} +\frac{4}{s}\log \sigma_t^{s/2} +  2\log T_t \Big)\bigg]\\
\leq & \EE\bigg[  \mathbbm{1}_{\{\sigma_t^2T_t^2>M \underline{\omega}^2\}}\Big(1-2\log \underline{\omega} +\frac{4}{s} \sigma_t^{s/2} +  2 T_t \Big)\bigg]\\
\leq & \Bigg(\underbrace{\EE\bigg[\Big(1-2\log \underline{\omega} +\frac{4}{s} \sigma_t^{s/2} +  2 T_t \Big)^2\bigg]}_{I_1}\Bigg)^{\frac{1}{2}} \bigg(\underbrace{\PP\Big[ \sigma_t^2T_t^2>M \underline{\omega}^2\Big]}_{I_2}\bigg)^{\frac{1}{2}}
\end{align*}
with $\sigma_t = \sigma_t (\theta_0)$. Employing  \eqref{eq:4.A.12}  we find that 
\begin{align*}
I_1 \leq 4 \bigg(1+\big(2\log \underline{\omega}\big)^2 +\frac{16}{s^2} \EE\big[\sigma_t^{s}\big] +  4 \EE\big[T_t^2\big] \bigg)<\infty
\end{align*}
and using Markov's inequality the second subterm can be bounded by
\begin{align*}
I_2 \leq \PP\Big[ T_t^2>M \underline{\omega}^2/2\Big] + \PP\Big[ \sigma_t^2>M \underline{\omega}^2/2\Big] \leq \frac{2}{M\underline{\omega}^2}\EE\big[T_t^2\big]+ \bigg(\frac{2}{\sqrt{M}\underline{\omega}}\bigg)^s\EE\big[\sigma_t^s\big].
\end{align*}
Since $I_1$ can be made arbitrarily small by the choice of $M$ we get $I=o(1)$ almost surely. Further, for given $M$, Lemma \ref{lem:4.4} entails
\begin{align*}
\Big|\EE^*\big[II\big]\Big|\leq M \Big|1-\EE^*\big[\eta_t^{*2}\big]\Big|\overset{a.s.}{\to}0 \quad \text{and} \quad 
\Var^*\big[II\big]\leq \frac{M^2}{n}\Var^*\big[\eta_t^{*2}\big]\overset{a.s.}{\to}0
\end{align*}
such that $II\overset{p^*}{\to}0$ almost surely. Consider $III$; the uniform ergodic theorem yields
\begin{align*}
III\overset{a.s.}{\to} \EE\Bigg[  \sup_{\substack{||\dot{\theta}-\theta_0||\leq 1/k \\ ||\theta-\theta_i||\leq 1/k}}\mathbbm{1}_{\big\{\frac{\sigma_t^2(\dot{\theta})}{\sigma_t^2(\theta)}\leq M\big\}} \bigg(1-\frac{\sigma_t^2(\dot{\theta})}{\sigma_t^2(\theta)}+\log \frac{\sigma_t^2(\dot{\theta})}{\sigma_t^2(\theta)}\bigg)\Bigg]
\end{align*}
and the right-hand side approaches
\begin{align}
\label{eq:4.A.50}
\EE\bigg[1-\frac{\sigma_t^2(\theta_0)}{\sigma_t^2(\theta_i)}+\log \frac{\sigma_t^2(\theta_0)}{\sigma_t^2(\theta_i)}\bigg]
\end{align}
as $M$ and $k$ grow large.  Thus, almost surely, $III$ can be made arbitrarily close to \eqref{eq:4.A.50} by choosing $M$ and $k$ sufficiently large. Further, since $\theta_i \in \mathscr{B}$, we have $\theta_i \neq \theta_0$ and  Assumption \ref{as:4.3} implies $\frac{\sigma_t^2(\theta_0)}{\sigma_t^2(\theta_i)}\neq 1$ almost surely. The elementary inequality $1-x+\log x \leq 0$ for $x>0$, which  holds with equality if and only if $x=1$, implies that \eqref{eq:4.A.50} is strictly smaller than $0$. We conclude that there exists a $\zeta_i<0$ such that $III<\zeta_i$ holds for sufficiently large $M$ and $k$ and $n$ almost surely. Set $\zeta = \max_{i=1,\dots,K} \zeta_i$, which satisfies $\zeta<0$. Combining results we complete \textit{Step 2}. 

Consider \textit{Step 3}; if $\hat{\theta}_n^* \in \mathscr{B}$, then \eqref{eq:4.4.1}  yields
\begin{align*}
\sup_{\theta \in \mathscr{B}} L_n^*(\theta) = L_n^*(\hat{\theta}_n^*) \geq L_n^*(\hat{\theta}_n).
\end{align*}
and by monotonicity of the probability measure $\PP^*$ we obtain
\begin{align*}
\PP^*\big[\hat{\theta}_n^* \in \mathscr{B}\big]\leq& \PP^*\bigg[ \sup_{\theta \in \mathscr{B}} L_n^*(\theta) - L_n^*(\hat{\theta}_n)\geq 0  \bigg].
\end{align*}
Together with \textit{Step 2} we obtain
\begin{align*}
 \PP^*\big[\hat{\theta}_n^* \in \mathscr{B}\big]\leq \PP^*\big[\zeta/2+S_n^*> 0 \big]+o(1)\leq  \PP^*\big[|S_n^*|> -\zeta/2\big]+o(1)=o(1)
\end{align*}
almost surely, which completes  \textit{Step 3} and establishes the lemma's claim.
\end{proof}

\begin{lemma}
\label{lem:4.6}
If Assumptions \ref{as:4.1}--\ref{as:4.4}, \ref{as:4.5}(\ref{as:4.5.1}), \ref{as:4.5}(\ref{as:4.5.3}), \ref{as:4.6} and \ref{as:4.9} hold with $a=\pm 12$, $b=12$ and $c=6$, then $\frac{1}{n}\sum_{t=1}^n \frac{\partial^2}{\partial \theta \partial \theta'} \ell_t^*(\breve{\theta}_n)\overset{p^*}{\to}-2J$ almost surely for $\breve{\theta}_n$ between $\hat{\theta}_n^*$ and  $\hat{\theta}_n$.
\end{lemma}
\begin{proof}
We have
\begin{align*}
\frac{1}{n}\sum_{t=1}^n \frac{\partial^2}{\partial \theta \partial \theta'} \ell_t^*(\breve{\theta}_n)= \underbrace{\frac{1}{n}\sum_{t=1}^n \bigg(\frac{\epsilon_t^{*2}}{\tilde{\sigma}_t^2(\breve{\theta}_n)}-1\bigg)\tilde{H}_t(\breve{\theta}_n)}_{I} - \underbrace{\frac{1}{n}\sum_{t=1}^n\bigg(3\frac{\epsilon_t^{*2}}{\tilde{\sigma}_t^2(\breve{\theta}_n)}-1\bigg)\tilde{D}_t(\breve{\theta}_n)\tilde{D}_t'(\breve{\theta}_n)}_{II}.
\end{align*}
Employing $\epsilon_t^*=\tilde{\sigma}_t(\hat{\theta}_n)\eta_t^*$ the first term can be expanded as follows:
\begin{align*}
I = \underbrace{\frac{1}{n}\sum_{t=1}^n \frac{\sigma_t^2(\hat{\theta}_n)}{\sigma_t^2(\breve{\theta}_n)}H_t(\breve{\theta}_n)\eta_t^{*2}}_{I_1}+\underbrace{\frac{1}{n}\sum_{t=1}^n \bigg(\frac{\tilde{\sigma}_t^2(\hat{\theta}_n)}{\tilde{\sigma}_t^2(\breve{\theta}_n)}\tilde{H}_t(\breve{\theta}_n)-\frac{\sigma_t^2(\hat{\theta}_n)}{\sigma_t^2(\breve{\theta}_n)}H_t(\breve{\theta}_n)\bigg)\eta_t^{*2}}_{I_2}-\underbrace{\frac{1}{n}\sum_{t=1}^n \tilde{H}_t(\breve{\theta}_n)}_{I_3}.
\end{align*}
Consider $I_1$; we take $\varepsilon>0$ and denote the unit vectors spanning $\R^r$ by $e_1,\dots,e_r$. 
\textcolor{black}{
Since $\frac{\sigma_t^2(\theta_1)}{\sigma_t^2(\theta_2)}H_t(\theta_2)$ is continuous in $\theta_1$ and $\theta_2$ we can take $\mathscr{V}_{\varepsilon}(\theta_0)\subseteq \mathscr{V}(\theta_0)$ such that 
}
\begin{align*}
\EE\big[e_i'H_t e_j\big]-\varepsilon <&\EE\bigg[\inf_{\theta_1,\theta_2 \in \mathscr{V}_\varepsilon(\theta_0)}\frac{\sigma_t^2(\theta_1)}{\sigma_t^2(\theta_2)}e_i'H_t(\theta_2)e_j\bigg]\\
\leq& \EE\bigg[\sup_{\theta_1,\theta_2 \in \mathscr{V}_\varepsilon(\theta_0)}\frac{\sigma_t^2(\theta_1)}{\sigma_t^2(\theta_2)}e_i'H_t(\theta_2)e_j\bigg]<\EE\big[e_i'H_t e_j\big]+\varepsilon
\end{align*}
for all $i,j = 1,\dots,r$. Since $\breve{\theta}_n$ lies between $\hat{\theta}_n^*$ and $\hat{\theta}_n$, Theorem \ref{thm:4.1} and Lemma \ref{lem:4.5} imply $\breve{\theta}_n\overset{p^*}{\to}\theta_0$ almost surely.
Since $\hat{\theta}_n \overset{a.s.}{\to}\theta_0$ and $\breve{\theta}_n\overset{p^*}{\to}\theta_0$ almost surely, we have $\hat{\theta}_n \in \mathscr{V}_\varepsilon(\theta_0)$ almost surely and $\breve{\theta}_n \in \mathscr{V}_\varepsilon(\theta_0)$ with conditional probability close to one almost surely. In such case, we have for all pairs $(i,j)$ 
\begin{align*}
 L_n^*(i,j)\leq& \frac{1}{n}\sum_{t=1}^n \frac{\sigma_t^2(\hat{\theta}_n)}{\sigma_t^2(\breve{\theta}_n)}e_i'H_t(\breve{\theta}_n)e_j \eta_t^{*2} \leq U_n^*(i,j)
\end{align*}
with
\begin{align*}
 L_n^*(i,j) =& \frac{1}{n}\sum_{t=1}^n \inf_{\theta_1,\theta_2 \in \mathscr{V}_\varepsilon(\theta_0)}\frac{\sigma_t^2(\theta_1)}{\sigma_t^2(\theta_2)}e_i'H_t(\theta_2)e_j \eta_t^{*2}\\
 U_n^*(i,j) =& \frac{1}{n}\sum_{t=1}^n \sup_{\theta_1,\theta_2 \in \mathscr{V}_\varepsilon(\theta_0)}\frac{\sigma_t^2(\theta_1)}{\sigma_t^2(\theta_2)}e_i'H_t(\theta_2)e_j \eta_t^{*2}.
\end{align*}
Using the uniform ergodic theorem, the conditional mean of the upper bound satisfies
\begin{align*}
\EE^*\big[U_n^*(i,j)\big]=&\EE^*\big[\eta_t^{*2}\big]\frac{1}{n}\sum_{t=1}^n \sup_{\theta_1,\theta_2 \in \mathscr{V}_\varepsilon(\theta_0)}\frac{\sigma_t^2(\theta_1)}{\sigma_t^2(\theta_2)}e_i'H_t(\theta_2)e_j\\
\overset{a.s.}{\to}& \EE\bigg[\sup_{\theta_1,\theta_2 \in \mathscr{V}_\varepsilon(\theta_0)}\frac{\sigma_t^2(\theta_1)}{\sigma_t^2(\theta_2)}e_i'H_t(\theta_2)e_j\bigg]<\EE\big[e_i'H_t e_j\big]+\varepsilon.
\end{align*}
whereas its conditional variance vanishes:
\begin{align*}
\Var^*\big[U_n^*(i,j)\big]=&\Var^*\big[\eta_t^{*2}\big]\frac{1}{n^2}\sum_{t=1}^n \Big(\sup_{\theta_1,\theta_2 \in \mathscr{V}_\varepsilon(\theta_0)}\frac{\sigma_t^2(\theta_1)}{\sigma_t^2(\theta_2)}e_i'H_t(\theta_2)e_j\Big)^2\leq \Var^*\big[\eta_t^{*2}\big]\frac{1}{n^2}\sum_{t=1}^n S_t^4T_t^4 V_t^2\\
\leq& \Var^*\big[\eta_t^{*2}\big]\frac{1}{n}\bigg(\underbrace{\frac{1}{n}\sum_{t=1}^n S_t^{12}}_{\overset{a.s.}{\to}\EE[S_t^{12}]<\infty}\bigg)^{\frac{1}{3}}\bigg(\underbrace{\frac{1}{n}\sum_{t=1}^n T_t^{12}}_{\overset{a.s.}{\to}\EE[T_t^{12}]<\infty}\bigg)^{\frac{1}{3}} \bigg(\underbrace{\frac{1}{n}\sum_{t=1}^n V_t^{6}}_{\overset{a.s.}{\to}\EE[V_t^{6}]<\infty}\bigg)^{\frac{1}{3}}\overset{a.s.}{\to}0.
\end{align*}
Similarly, we obtain for the lower bound 
\begin{align*}
\EE^*\big[L_n^*(i,j)\big]\overset{a.s.}{\to} \EE\bigg[\inf_{\theta_1,\theta_2 \in \mathscr{V}_\varepsilon(\theta_0)}\frac{\sigma_t^2(\theta_1)}{\sigma_t^2(\theta_2)}e_i'H_t(\theta_2)e_j\bigg]>\EE\big[e_i'H_t e_j\big]-\varepsilon
\end{align*}
and $\Var^*\big[L_n^*(i,j)\big]\overset{a.s.}{\to}0$. Taking $\varepsilon \searrow 0$ subsequently, we get $\frac{1}{n}\sum_{t=1}^n \frac{\sigma_t^2(\hat{\theta}_n)}{\sigma_t^2(\breve{\theta}_n)}e_i'H_t(\breve{\theta}_n)e_j'\eta_t^{*2}\overset{p^*}{\to}\EE\big[e_i'H_t e_j\big]$ almost surely for all pairs ($i,j$), which in turn yields $I_1 \overset{p^*}{\to}\EE[H_t]$ almost surely. Regarding $I_2$, we combine \eqref{eq:4.A.51b} and the elementary inequalities \eqref{eq:4.A.5} with $m=1$, which yields
\begin{align}
\label{eq:4.A.52}
\begin{split}
&\bigg|\frac{\tilde{\sigma}_t^2(\theta_1)}{\tilde{\sigma}_t^2(\theta_2)}-\frac{\sigma_t^2(\theta_1)}{\sigma_t^2(\theta_2)}\bigg|\leq \bigg|\frac{\tilde{\sigma}_t(\theta_1)}{\tilde{\sigma}_t(\theta_2)}-\frac{\sigma_t(\theta_1)}{\sigma_t(\theta_2)}\bigg|^2 + 2 \bigg|\frac{\tilde{\sigma}_t(\theta_1)}{\tilde{\sigma}_t(\theta_2)}-\frac{\sigma_t(\theta_1)}{\sigma_t(\theta_2)}\bigg| \frac{\sigma_t(\theta_1)}{\sigma_t(\theta_2)}\\
\leq& \frac{C_1^2\rho^{2t}}{\underline{\omega}^2}\bigg(1+\frac{\sigma_t(\theta_1)}{\sigma_t(\theta_2)}\bigg)^2 +  \frac{2 C_1\rho^t}{\underline{\omega}}\bigg(1+\frac{\sigma_t(\theta_1)}{\sigma_t(\theta_2)}\bigg) \frac{\sigma_t(\theta_1)}{\sigma_t(\theta_2)}\\
\leq& \bigg(\frac{C_1^2}{\underline{\omega}^2}+ \frac{2 C_1}{\underline{\omega}}\bigg)\rho^t\bigg(1+\frac{\sigma_t(\theta_1)}{\sigma_t(\theta_2)}\bigg)^2\leq \bigg(\frac{2C_1^2}{\underline{\omega}^2}+ \frac{4 C_1}{\underline{\omega}}\bigg)\rho^t\bigg(1+\frac{\sigma_t^2(\theta_1)}{\sigma_t^2(\theta_2)}\bigg)
\end{split}
\end{align}
for any $\theta_1,\theta_2 \in \Theta$. It follows that
\begin{align*}
||I_2||\leq& \frac{1}{n}\sum_{t=1}^n \bigg|\bigg|\frac{\tilde{\sigma}_t^2(\hat{\theta}_n)}{\tilde{\sigma}_t^2(\breve{\theta}_n)}\tilde{H}_t(\breve{\theta}_n)-\frac{\sigma_t^2(\hat{\theta}_n)}{\sigma_t^2(\breve{\theta}_n)}H_t(\breve{\theta}_n)\bigg|\bigg|\eta_t^{*2}\\
=& \frac{1}{n}\sum_{t=1}^n \bigg|\bigg|\frac{\tilde{\sigma}_t^2(\hat{\theta}_n)}{\tilde{\sigma}_t^2(\breve{\theta}_n)}\Big(\tilde{H}_t(\breve{\theta}_n)-H_t(\breve{\theta}_n)\Big)+\bigg(\frac{\tilde{\sigma}_t^2(\hat{\theta}_n)}{\tilde{\sigma}_t^2(\breve{\theta}_n)}-\frac{\sigma_t^2(\hat{\theta}_n)}{\sigma_t^2(\breve{\theta}_n)}\bigg)H_t(\breve{\theta}_n)\bigg|\bigg|\eta_t^{*2}\\
\leq& \frac{1}{n}\sum_{t=1}^n\bigg\{\frac{\tilde{\sigma}_t^2(\hat{\theta}_n)}{\tilde{\sigma}_t^2(\breve{\theta}_n)}\Big|\Big|\tilde{H}_t(\breve{\theta}_n)-H_t(\breve{\theta}_n)\Big|\Big|+\bigg|\frac{\tilde{\sigma}_t^2(\hat{\theta}_n)}{\tilde{\sigma}_t^2(\breve{\theta}_n)}-\frac{\sigma_t^2(\hat{\theta}_n)}{\sigma_t^2(\breve{\theta}_n)}\bigg| \:\big|\big|H_t(\breve{\theta}_n)\big|\big|\bigg\}\eta_t^{*2}\\
\leq& \frac{1}{n}\sum_{t=1}^n\Bigg\{\Bigg(\frac{\sigma_t^2(\hat{\theta}_n)}{\sigma_t^2(\breve{\theta}_n)}+\bigg(\frac{2C_1^2}{\underline{\omega}^2}+ \frac{4 C_1}{\underline{\omega}}\bigg)\rho^t\bigg(1+\frac{\sigma_t^2(\hat{\theta}_n)}{\sigma_t^2(\breve{\theta}_n)}\bigg)\Bigg)\:\frac{C_1\rho^t}{\underline{\omega}}\Big(1+\big|\big|H_t(\breve{\theta}_n)\big|\big|\Big) \\
&\qquad +\bigg(\frac{2C_1^2}{\underline{\omega}^2}+ \frac{4 C_1}{\underline{\omega}}\bigg)\rho^t\bigg(1+\frac{\sigma_t^2(\hat{\theta}_n)}{\sigma_t^2(\breve{\theta}_n)}\bigg) \:\big|\big|H_t(\breve{\theta}_n)\big|\big|\Bigg\}\eta_t^{*2}\\
\leq & \bigg(\frac{5 C_1}{\underline{\omega}}+\frac{6 C_1^2}{\underline{\omega}^2}+ \frac{2 C_1^3}{\underline{\omega}^3}\bigg)\frac{1}{n}\sum_{t=1}^n\rho^t\bigg(1+\frac{\sigma_t^2(\hat{\theta}_n)}{\sigma_t^2(\breve{\theta}_n)}\bigg) \Big(1+\big|\big|H_t(\breve{\theta}_n)\big|\big|\Big)\eta_t^{*2},
\end{align*}
where the third inequality comes from \eqref{eq:4.A.8} and \eqref{eq:4.A.52}. In the case of $\hat{\theta}_n \in \mathscr{V}(\theta_0)$ and $\breve{\theta}_n \in \mathscr{V}(\theta_0)$, we get
\begin{align*}
\frac{1}{n}\sum_{t=1}^n\rho^t\bigg(1+\frac{\sigma_t^2(\hat{\theta}_n)}{\sigma_t^2(\breve{\theta}_n)}\bigg) \:\Big(1+||H_t(\breve{\theta}_n)||\Big)\eta_t^{*2} \leq \frac{1}{n}\sum_{t=1}^n\rho^t\big(1+S_t^2 T_t^2\big) \:(1+V_t)\eta_t^{*2}.
\end{align*}
For any $\delta>0$ we find
\begin{align*}
\PP^*\bigg[\frac{1}{n}\sum_{t=1}^n\rho^t\big(1+S_t^2 T_t^2\big)(1+V_t)\eta_t^{*2}\geq \delta\bigg]  =& \frac{\EE^*[\eta_t^{*2}]}{\delta}\frac{1}{n}\sum_{t=1}^n\rho^t\big(1+S_t^2 T_t^2\big) (1+V_t).
\end{align*}
using Markov's inequality. Moreover, for $\varepsilon>0$ we have
\begin{align*}
\sum_{t=1}^\infty \PP\Big[\rho^{t}\big(1+S_t^2 T_t^2\big) (1+V_t)>\varepsilon\Big]\leq& \sum_{t=1}^\infty \rho^{t}\frac{\EE\big[(1+S_t^2 T_t^2) (1+V_t)\big]}{\varepsilon}\\
=& \frac{\EE\big[(1+S_t^2 T_t^2) (1+V_t)\big]}{\varepsilon(1-\rho)}<\infty
\end{align*}
such that the Borel-Cantelli Lemma implies $\rho^{t}\big(1+S_t^2 T_t^2\big) (1+V_t)\overset{a.s.}{\to}0$ as $t\to \infty$. Therefore, $\frac{1}{n}\sum_{t=1}^n\rho^{t}\big(1+S_t^2 T_t^2\big) (1+V_t)\overset{a.s.}{\to}0$ follows by C\'esaro's lemma and we get $\frac{1}{n}\sum_{t=1}^n\rho^t\big(1+S_t^2 T_t^2\big) \:(1+V_t)\eta_t^{*2}\overset{p^*}{\to}0$ almost surely. Combining results gives $||I_2||\overset{p^*}{\to}0$ almost surely. Similar to the proof of Lemma \ref{lem:4.2}(iii), we establish  $I_3\overset{p^*}{\to}\EE[H_t]$ almost surely using $\breve{\theta}_n\overset{p^*}{\to}\theta_0$ almost surely. Combining results we establish that $I=I_1+I_2-I_3\overset{p^*}{\to}0$ almost surely. Consider the second term and expand 
\begin{align*}
II =& 3\underbrace{\frac{1}{n}\sum_{t=1}^n \frac{\sigma_t^2(\hat{\theta}_n)}{\sigma_t^2(\breve{\theta}_n)}D_t(\breve{\theta}_n)D_t'(\breve{\theta}_n)\eta_t^{*2}}_{II_1}+3\underbrace{\frac{1}{n}\sum_{t=1}^n \bigg(\frac{\tilde{\sigma}_t^2(\hat{\theta}_n)}{\tilde{\sigma}_t^2(\breve{\theta}_n)}D_t(\breve{\theta}_n)D_t'(\breve{\theta}_n)-\frac{\sigma_t^2(\hat{\theta}_n)}{\sigma_t^2(\breve{\theta}_n)}D_t(\breve{\theta}_n)D_t'(\breve{\theta}_n)\bigg)\eta_t^{*2}}_{II_2}\\
&\qquad -\underbrace{\frac{1}{n}\sum_{t=1}^n D_t(\breve{\theta}_n)D_t'(\breve{\theta}_n)}_{II_3}.
\end{align*}
We treat the subterms of $II$ analogously to the subterms of $I$. We begin with $II_1$ and take $\varepsilon>0$. 
\textcolor{black}{
Since $\frac{\sigma_t^2(\theta_1)}{\sigma_t^2(\theta_2)}D_t(\theta_2)D_t'(\theta_2)$ is continuous in $\theta_1$ and $\theta_2$ we can take $\mathscr{V}_{\varepsilon}(\theta_0)\subseteq \mathscr{V}(\theta_0)$ such that 
}
\begin{align*}
\EE\big[e_i'D_t D_t' e_j\big]-\varepsilon <&\EE\bigg[\inf_{\theta_1,\theta_2 \in \mathscr{V}_\varepsilon(\theta_0)}\frac{\sigma_t^2(\theta_1)}{\sigma_t^2(\theta_2)}e_i'D_t(\theta_2) D_t'(\theta_2)e_j\bigg]\\
\leq& \EE\bigg[\sup_{\theta_1,\theta_2 \in \mathscr{V}_\varepsilon(\theta_0)}\frac{\sigma_t^2(\theta_1)}{\sigma_t^2(\theta_2)}e_i'D_t(\theta_2) D_t'(\theta_2)e_j\bigg]<\EE\big[e_i'D_t D_t'e_j\big]+\varepsilon
\end{align*}
for all $i,j = 1,\dots,r$. Since $\hat{\theta}_n \overset{a.s.}{\to}\theta_0$ and $\breve{\theta}_n\overset{p^*}{\to}\theta_0$ almost surely, we have $\hat{\theta}_n \in \mathscr{V}_\varepsilon(\theta_0)$ almost surely and $\breve{\theta}_n \in \mathscr{V}_\varepsilon(\theta_0)$ with conditional probability close to one almost surely. In such case, we have for all pairs $(i,j)$ 
\begin{align*}
 \bar{L}_n^*(i,j)\leq& \frac{1}{n}\sum_{t=1}^n \frac{\sigma_t^2(\hat{\theta}_n)}{\sigma_t^2(\breve{\theta}_n)}e_i'D_t(\breve{\theta}_n) D_t'(\breve{\theta}_n) e_j'\eta_t^{*2} \leq \bar{U}_n^*(i,j)
\end{align*}
with
\begin{align*}
 \bar{L}_n^*(i,j) =& \frac{1}{n}\sum_{t=1}^n \inf_{\theta_1,\theta_2 \in \mathscr{V}_\varepsilon(\theta_0)}\frac{\sigma_t^2(\theta_1)}{\sigma_t^2(\theta_2)}e_i'D_t(\theta_2) D_t'(\theta_2) e_j \eta_t^{*2}\\
 \bar{U}_n^*(i,j) =& \frac{1}{n}\sum_{t=1}^n \sup_{\theta_1,\theta_2 \in \mathscr{V}_\varepsilon(\theta_0)}\frac{\sigma_t^2(\theta_1)}{\sigma_t^2(\theta_2)}e_i'D_t(\theta_2) D_t'(\theta_2) e_j \eta_t^{*2}.
\end{align*}
Using the uniform ergodic theorem, the conditional mean of the upper bound satisfies
\begin{align*}
\EE^*\big[\bar{U}_n^*(i,j)\big]=&\EE^*\big[\eta_t^{*2}\big]\frac{1}{n}\sum_{t=1}^n \sup_{\theta_1,\theta_2 \in \mathscr{V}_\varepsilon(\theta_0)}\frac{\sigma_t^2(\theta_1)}{\sigma_t^2(\theta_2)}e_i'D_t(\theta_2) D_t'(\theta_2) e_j\\
\overset{a.s.}{\to}& \EE\bigg[\sup_{\theta_1,\theta_2 \in \mathscr{V}_\varepsilon(\theta_0)}\frac{\sigma_t^2(\theta_1)}{\sigma_t^2(\theta_2)}e_i'D_t(\theta_2) D_t'(\theta_2) e_j\bigg]<\EE\big[e_i'D_t D_t' e_j\big]+\varepsilon
\end{align*}
whereas its conditional variance vanishes:
\begin{align*}
\Var^*\big[\bar{U}_n^*(i,j)\big]=&\Var^*\big[\eta_t^{*2}\big]\frac{1}{n^2}\sum_{t=1}^n \Big(\sup_{\theta_1,\theta_2 \in \mathscr{V}_\varepsilon(\theta_0)}\frac{\sigma_t^2(\theta_1)}{\sigma_t^2(\theta_2)}e_i'D_t(\theta_2) D_t'(\theta_2)e_j\Big)^2\\
\leq& \Var^*\big[\eta_t^{*2}\big]\frac{1}{n^2}\sum_{t=1}^n S_t^4T_t^4 U_t^4\\
\leq& \Var^*\big[\eta_t^{*2}\big]\frac{1}{n}\bigg(\underbrace{\frac{1}{n}\sum_{t=1}^n S_t^{12}}_{\overset{a.s.}{\to}\EE[S_t^{12}]<\infty}\bigg)^{\frac{1}{3}}\bigg(\underbrace{\frac{1}{n}\sum_{t=1}^n T_t^{12}}_{\overset{a.s.}{\to}\EE[T_t^{12}]<\infty}\bigg)^{\frac{1}{3}} \bigg(\underbrace{\frac{1}{n}\sum_{t=1}^n U_t^{12}}_{\overset{a.s.}{\to}\EE[U_t^{12}]<\infty}\bigg)^{\frac{1}{3}}\overset{a.s.}{\to}0.
\end{align*}
Similarly, we obtain for the lower bound  
\begin{align*}
\EE^*\big[\bar{L}_n^*(i,j)\big]\overset{a.s.}{\to} \EE\bigg[\inf_{\theta_1,\theta_2 \in \mathscr{V}_\varepsilon(\theta_0)}\frac{\sigma_t^2(\theta_1)}{\sigma_t^2(\theta_2)}e_i'D_t(\theta_2) D_t'(\theta_2)e_j\bigg]>\EE\big[e_i'D_t D_t' e_j\big]-\varepsilon
\end{align*}
and $\Var^*\big[\bar{L}_n^*(i,j)\big]\overset{a.s.}{\to}0$.  Next, we take $\varepsilon \searrow 0$ and get $\frac{1}{n}\sum_{t=1}^n \frac{\sigma_t^2(\hat{\theta}_n)}{\sigma_t^2(\breve{\theta}_n)}e_i'D_t(\breve{\theta}_n) D_t'(\breve{\theta}_n)e_j'\eta_t^{*2}\overset{p^*}{\to}\EE\big[e_i'D_t D_t' e_j\big]$ almost surely for all pairs ($i,j$), which in turn yields $II_1 \overset{p^*}{\to}\EE[D_tD_t']=J$ almost surely. Regarding $II_2$, we find
\begin{align*}
&||II_2||\leq \frac{1}{n}\sum_{t=1}^n \bigg|\bigg|\frac{\tilde{\sigma}_t^2(\hat{\theta}_n)}{\tilde{\sigma}_t^2(\breve{\theta}_n)}\tilde{D}_t(\breve{\theta}_n)\tilde{D}_t'(\breve{\theta}_n)-\frac{\sigma_t^2(\hat{\theta}_n)}{\sigma_t^2(\breve{\theta}_n)}D_t(\breve{\theta}_n)D_t'(\breve{\theta}_n)\bigg|\bigg|\eta_t^{*2}\\
=& \frac{1}{n}\sum_{t=1}^n \bigg|\bigg|\frac{\tilde{\sigma}_t^2(\hat{\theta}_n)}{\tilde{\sigma}_t^2(\breve{\theta}_n)}\Big(\tilde{D}_t(\breve{\theta}_n)\tilde{D}_t'(\breve{\theta}_n)-D_t(\breve{\theta}_n)D_t'(\breve{\theta}_n)\Big)+\bigg(\frac{\tilde{\sigma}_t^2(\hat{\theta}_n)}{\tilde{\sigma}_t^2(\breve{\theta}_n)}-\frac{\sigma_t^2(\hat{\theta}_n)}{\sigma_t^2(\breve{\theta}_n)}\bigg)D_t(\breve{\theta}_n)D_t'(\breve{\theta}_n)\bigg|\bigg|\eta_t^{*2}\\
\leq& \frac{1}{n}\sum_{t=1}^n\bigg\{\frac{\tilde{\sigma}_t^2(\hat{\theta}_n)}{\tilde{\sigma}_t^2(\breve{\theta}_n)}\Big|\Big|\tilde{D}_t(\breve{\theta}_n)\tilde{D}_t'(\breve{\theta}_n)-D_t(\breve{\theta}_n) D_t'(\breve{\theta}_n)\Big|\Big|+\bigg|\frac{\tilde{\sigma}_t^2(\hat{\theta}_n)}{\tilde{\sigma}_t^2(\breve{\theta}_n)}-\frac{\sigma_t^2(\hat{\theta}_n)}{\sigma_t^2(\breve{\theta}_n)}\bigg| \:\big|\big|D_t(\breve{\theta}_n)\big|\big|^2\bigg\}\eta_t^{*2}\\
\leq& \frac{1}{n}\sum_{t=1}^n\Bigg\{\Bigg(\frac{\sigma_t^2(\hat{\theta}_n)}{\sigma_t^2(\breve{\theta}_n)}+\bigg(\frac{2C_1^2}{\underline{\omega}^2}+ \frac{4 C_1}{\underline{\omega}}\bigg)\rho^t\bigg(1+\frac{\sigma_t^2(\hat{\theta}_n)}{\sigma_t^2(\breve{\theta}_n)}\bigg)\Bigg)\:\bigg(\frac{C_1^2}{\underline{\omega}^2}+\frac{2C_1}{\underline{\omega}}\bigg) \rho^{t}\Big(1+\big|\big|D_t(\breve{\theta}_n)\big|\big|^2\Big) \\
&\qquad +\bigg(\frac{2C_1^2}{\underline{\omega}^2}+ \frac{4 C_1}{\underline{\omega}}\bigg)\rho^t\bigg(1+\frac{\sigma_t^2(\hat{\theta}_n)}{\sigma_t^2(\breve{\theta}_n)}\bigg) \:\big|\big|D_t(\breve{\theta}_n)\big|\big|^2\Bigg\}\eta_t^{*2}\\
\leq & \bigg(\frac{6 C_1}{\underline{\omega}}+\frac{11 C_1^2}{\underline{\omega}^2}+ \frac{8 C_1^3}{\underline{\omega}^3}+\frac{2 C_1^4}{\underline{\omega}^4}\bigg)\frac{1}{n}\sum_{t=1}^n\rho^t\bigg(1+\frac{\sigma_t^2(\hat{\theta}_n)}{\sigma_t^2(\breve{\theta}_n)}\bigg) \Big(1+\big|\big|D_t(\breve{\theta}_n)\big|\big|^2\Big)\eta_t^{*2},
\end{align*}
where the third inequality follows from \eqref{eq:4.A.6} and \eqref{eq:4.A.52}. In the case of $\hat{\theta}_n \in \mathscr{V}(\theta_0)$ and $\breve{\theta}_n \in \mathscr{V}(\theta_0)$, we get
\begin{align*}
\frac{1}{n}\sum_{t=1}^n\rho^t\bigg(1+\frac{\sigma_t^2(\hat{\theta}_n)}{\sigma_t^2(\breve{\theta}_n)}\bigg) \:\Big(1+||D_t(\breve{\theta}_n)||^2\Big)\eta_t^{*2} \leq \frac{1}{n}\sum_{t=1}^n\rho^t\big(1+S_t^2 T_t^2\big) \big(1+U_t^2\big)\eta_t^{*2}.
\end{align*}
For any $\delta>0$ we find
\begin{align*}
\PP^*\bigg[\frac{1}{n}\sum_{t=1}^n\rho^t\big(1+S_t^2 T_t^2\big)\big(1+U_t^2\big)\eta_t^{*2}\geq \delta\bigg]  =& \frac{\EE^*[\eta_t^{*2}]}{\delta}\frac{1}{n}\sum_{t=1}^n\rho^t\big(1+S_t^2 T_t^2\big) \big(1+U_t^2\big).
\end{align*}
using Markov's inequality. Moreover, for $\varepsilon>0$ we have
\begin{align*}
\sum_{t=1}^\infty \PP\Big[\rho^{t}\big(1+S_t^2 T_t^2\big) \big(1+U_t^2\big)>\varepsilon\Big]\leq& \sum_{t=1}^\infty \rho^{t}\frac{\EE\big[(1+S_t^2 T_t^2) (1+U_t^2)\big]}{\varepsilon}\\
=& \frac{\EE\big[(1+S_t^2 T_t^2) (1+U_t^2)\big]}{\varepsilon(1-\rho)}<\infty
\end{align*}
such that the Borel-Cantelli Lemma implies $\rho^{t}\big(1+S_t^2 T_t^2\big) \big(1+U_t^2\big)\overset{a.s.}{\to}0$ as $t\to \infty$. Therefore, $\frac{1}{n^2}\sum_{t=1}^n\rho^{t}\big(1+S_t^2 T_t^2\big) \big(1+U_t^2\big)\overset{a.s.}{\to}0$ follows by C\'esaro's lemma and we get $\frac{1}{n}\sum_{t=1}^n\rho^t\big(1+S_t^2 T_t^2\big) \:\big(1+U_t^2\big)\eta_t^{*2}\overset{p^*}{\to}0$ almost surely. Combining results gives $||II_2||\overset{p^*}{\to}0$ almost surely. Similar to the proof of Lemma \ref{lem:4.2}(ii), we establish  $II_3\overset{p^*}{\to}\EE\big[D_tD_t'\big]=J$ almost surely using $\breve{\theta}_n\overset{p^*}{\to}\theta_0$ almost surely. Combining results we find $II=3II_1+3II_2-II_3\overset{p^*}{\to}3J+0-J=2J$ almost surely. In conclusion, we have
\begin{align*}
\frac{1}{n}\sum_{t=1}^n \frac{\partial^2}{\partial \theta \partial \theta'} \ell_t^*(\breve{\theta}_n)
=& I-II\overset{p^*}{\to}-2J
\end{align*}
almost surely, which completes the proof.
\end{proof}

\begin{lemma}
\label{lem:4.7}
Suppose Assumptions \ref{as:4.1}--\ref{as:4.4}, \ref{as:4.5}(\ref{as:4.5.1}), \ref{as:4.5}(\ref{as:4.5.3}), \ref{as:4.6}, \ref{as:4.9} and \ref{as:4.10} hold with $a=-1,4$, $b=4$ and $c=2$. Then, we have
\begin{align*}
  \frac{1}{\sqrt{n}}\sum_{t=1}^n    \begin{pmatrix}
      \hat{D}_t \big(\eta_t^{*2}-1\big)\\
  \mathbbm{1}_{\{\eta_t^*< \hat{\xi}_{n,\alpha}\}}-\alpha
      \end{pmatrix}
\overset{d^*}{\to}N (0,\Upsilon_\alpha) \quad \text{with}\quad \Upsilon_\alpha =\begin{pmatrix}
      (\kappa-1)J& p_\alpha \Omega\\
    p_\alpha \Omega' & \alpha(1-\alpha)
      \end{pmatrix}
\end{align*}
almost surely.
\end{lemma}
\begin{proof}
Set $\alpha_n = \frac{1}{n}\sum_{t=1}^n\mathbbm{1}_{\{\hat{\eta}_t< \hat{\xi}_{n,\alpha}\}}$ and expand 
\begin{align*}
\frac{1}{\sqrt{n}} \sum_{t=1}^n  \!  \begin{pmatrix}
      \hat{D}_t \big(\eta_t^{*2}-1\big)\\
  \mathbbm{1}_{\{\eta_t^*< \hat{\xi}_{n,\alpha}\}}-\alpha
      \end{pmatrix}\! =\! \frac{1}{\sqrt{n}} \sum_{t=1}^n \!  \begin{pmatrix}
      \hat{D}_t \big(\eta_t^{*2}-\EE^*[\eta_t^{*2}]\big)\\
  \mathbbm{1}_{\{\eta_t^*< \hat{\xi}_{n,\alpha}\}}-\alpha_n
      \end{pmatrix} \!+ \! \frac{1}{\sqrt{n}} \sum_{t=1}^n  \!  \begin{pmatrix}
      \hat{D}_t \big(\EE^*[\eta_t^{*2}]-1\big)\\
  \alpha_n-\alpha
      \end{pmatrix}.
\end{align*}
Consider the second term; with regard to Remark \ref{rem:4.1} we have $\EE^*\big[\eta_t^{*2}\big]=1$ whenever $\hat{\theta}_n \in \mathring{\Theta}$ under Assumption \ref{as:4.10}. Since $\hat{\theta}_n\overset{a.s.}{\to}\theta_0 \in \mathring{\Theta}$ by Theorem \ref{thm:4.1} and Assumption \ref{as:4.6}, we have $\frac{1}{\sqrt{n}} \sum_{t=1}^n \hat{D}_t \big(\EE^*[\eta_t^{*2}]-1\big)=0$ for sufficiently large $n$ almost surely. Further, $\alpha_n \overset{a.s.}{=} \frac{\floor{n \alpha}+1}{n}=\alpha+O(n^{-1})$
and hence $\frac{1}{\sqrt{n}} \sum_{t=1}^n (\alpha_n-\alpha)\overset{a.s.}{\to}0$. Using the Cram\'er-Wold device it remains to show
that for each  $\lambda=(\lambda_1',\lambda_2)'  \in \R^{r+1}$ with $||\lambda||\neq 0$ 
\begin{align*}
 \sum_{t=1}^n \underbrace{\frac{1}{\sqrt{n}}\lambda'   \begin{pmatrix}
      \hat{D}_t \big(\eta_t^{*2}-\EE^*[\eta_t^{*2}]\big)\\
  \mathbbm{1}_{\{\eta_t^*< \hat{\xi}_{n,\alpha}\}}-\alpha_n
      \end{pmatrix}}_{Z_{n,t}^*}
    \overset{d^*}{\to}N\big(0, \lambda'\Upsilon_\alpha\lambda \big)
\end{align*}
almost surely.
By construction, we have $\EE\big[Z_{n,t}^*\big]=0$. Further, we obtain
\begin{align}
\nonumber
s_n^2 =&\sum_{t=1}^n\Var^*\big[Z_{n,t}^{*}\big] 
=\lambda'\begin{pmatrix}
      \Var^*[\eta_t^{*2}]\hat{J}_n& \Cov^*[\eta_t^{*2},\mathbbm{1}_{\{\eta_t^*< \hat{\xi}_{n,\alpha}\}}] \hat{\Omega}_n\\
    \Cov^*[\eta_t^{*2},\mathbbm{1}_{\{\eta_t^*< \hat{\xi}_{n,\alpha}\}}] \hat{\Omega}_n' & \Var^*[\mathbbm{1}_{\{\eta_t^*< \hat{\xi}_{n,\alpha}\}}] 
      \end{pmatrix}\lambda.
\end{align}
Lemma \ref{lem:4.2} states $\hat{J}_n \overset{a.s.}{\to}J$ and $\hat{\Omega}_n\overset{a.s.}{\to}\Omega$. Employing Lemma  \ref{lem:4.4} yields $\Var^*\big[\eta_t^{*2}\big]\overset{a.s.}{\to}\kappa-1$, $\Var^*[\mathbbm{1}_{\{\eta_t^*< \hat{\xi}_{n,\alpha}\}}] = \alpha_n (1-\alpha_n) \overset{a.s.}{\to} \alpha(1-\alpha)$ and $\Cov^*\big[\eta_t^{*2},\mathbbm{1}_{\{\eta_t^*< \hat{\xi}_{n,\alpha}\}}\big] = \EE^*\big[\eta_t^{*2}\mathbbm{1}_{\{\eta_t^*< \hat{\xi}_{n,\alpha}\}}\big]-\EE^*\big[\eta_t^{*2}\big]\alpha_n \overset{a.s.}{\to} p_\alpha$. It follows that
$s_n^2\overset{a.s.}{\to} \lambda'\Upsilon_\alpha \lambda$. Next, we verify Lindeberg condition. 
For an arbitrary $\varepsilon>0$ 
\begin{align*}
&  \sum_{t=1}^n  \EE^*\big[Z_{n,t}^{*2}\mathbbm{1}_{\{|Z_{n,t}^{*}|\geq s_n \varepsilon \}}\big] \leq  \underbrace{\sum_{t=1}^n\EE^*\big[Z_{n,t}^{*2}\mathbbm{1}_{\{|\eta_t^{*}|> C \}}\big]}_{I}+ \underbrace{\sum_{t=1}^n  \EE^*\big[Z_{n,t}^{*2}\mathbbm{1}_{\{|Z_{n,t}^{*}|\geq s_n \varepsilon \}}\mathbbm{1}_{\{|\eta_t^{*}|\leq C \}}\big]}_{II}
\end{align*}
holds, where $C>0$. Employing the elementary inequalities 
\begin{align}
\label{eq:4.A.53}
(x+y)^z\leq 2^z(x^z+y^z) 
\end{align}
and $|x-y|^z\leq x^z+y^z$ for all $x,y,z\geq 0$ we find that
\begin{align*}
Z_{n,t}^{*2}\leq  \frac{4}{n} \Big(\big(\lambda_1'\hat{D}_t\big)^2 \big(\eta_t^{*4}+\EE^*[\eta_t^{*2}]^2\big)+
    \lambda_2^2\Big).
\end{align*}
Hence, we obtain
\begin{align*}
I \leq &  \frac{4}{n}\sum_{t=1}^n\EE^*\bigg[ \Big(\big(\lambda_1'\hat{D}_t\big)^2 \big(\eta_t^{*4}+\EE^*[\eta_t^{*2}]^2\big)+
    \lambda_2^2\Big)\mathbbm{1}_{\{|\eta_t^{*}|> C \}}\bigg]\\
= & 4 \Big( \lambda_1' \hat{J}_n \lambda_1 \EE^*\big[\eta_t^{*4}\mathbbm{1}_{\{|\eta_t^{*}|> C \}}\big]+\big(\lambda_1' \hat{J}_n \lambda_1 \EE^*[\eta_t^{*2}]^2+
    \lambda_2^2\big)\EE^*\big[\mathbbm{1}_{\{|\eta_t^{*}|> C \}}\big]\Big)\\
    \overset{a.s.}{\to}&4 \Big( \lambda_1' J \lambda_1 \EE\big[\eta_t^{4}\mathbbm{1}_{\{|\eta_t|> C \}}\big]+\big(\lambda_1' J \lambda_1 \EE[\eta_t^{2}]^2+
    \lambda_2^2\big)\EE\big[\mathbbm{1}_{\{|\eta_t|> C \}}\big]\Big)
\end{align*}
and choosing $C$ sufficiently large yields $I\overset{a.s.}{\to}0$. Given a value of $C$, we have
\begin{align*}
II\leq  &  \frac{4}{n}  \sum_{t=1}^n \EE^*\bigg[ \Big(\big(\lambda_1'\hat{D}_t\big)^2 \big(\eta_t^{*4}+\EE^*[\eta_t^{*2}]^2\big)+
    \lambda_2^2\Big)\mathbbm{1}_{ \{ ||\lambda_1|| (\eta_t^{*2}+\EE^*[\eta_t^{*2}]) \max_{t} ||\hat{D}_t||+
    |\lambda_2|\geq \sqrt{n} s_n \varepsilon \}}\mathbbm{1}_{\{|\eta_t^{*}|\leq C \}}\bigg]\\
\leq  &  \frac{4}{n}  \sum_{t=1}^n  \Big(\big(\lambda_1'\hat{D}_t\big)^2 \big(C^{4}+\EE^*[\eta_t^{*2}]^2\big)+
    \lambda_2^2\Big)\mathbbm{1}_{ \{ ||\lambda_1|| (C^{2}+\EE^*[\eta_t^{*2}]) \max_{t} ||\hat{D}_t||+
    |\lambda_2|\geq \sqrt{n} s_n \varepsilon \}}\\
=  &   4\Big(\lambda_1' \hat{J}_n \lambda_1  \big(C^{4}+\EE^*[\eta_t^{*2}]^2\big)+
    \lambda_2^2\Big)\mathbbm{1}_{ \{ ||\lambda_1|| (C^{2}+\EE^*[\eta_t^{*2}]) \max_{t} ||\hat{D}_t||+
    |\lambda_2|\geq \sqrt{n} s_n \varepsilon \}}\\
    \overset{a.s.}{\to}& 4\Big(\lambda_1' J \lambda_1  \big(C^{4}+\EE[\eta_t^{2}]^2\big)+
    \lambda_2^2\Big)\times 0 = 0
\end{align*}
To appreciate why the indicator function converges to $0$ almost surely we employ \eqref{eq:4.A.3} as well as \eqref{eq:4.A.53} and note $\hat{\theta}_n \in \mathscr{V}(\theta_0)$ almost surely to get
\begin{align}
\label{eq:4.A.54}
\begin{split}
&\frac{1}{n} \sum_{t=1}^n \big|\big|\hat{D}_t\big|\big|^{4}\leq \frac{1}{n} \sum_{t=1}^n \bigg(\big|\big|D_t(\hat{\theta}_n)\big|\big|+\frac{C_1 \rho^t}{\underline{\omega}}\Big(1+ \big|\big|D_t(\hat{\theta}_n)\big|\big|\Big)\bigg)^{4}\\
\overset{a.s.}{\leq}& \frac{1}{n} \sum_{t=1}^n \bigg(U_t+\frac{C_1 \rho^t}{\underline{\omega}}(1+ U_t)\bigg)^{4} \leq  2^{4} \bigg(\frac{1}{n} \sum_{t=1}^n  U_t^4+\frac{C_1^4}{\underline{\omega}^4}\frac{1}{n} \sum_{t=1}^n \big\{\rho^{t}(1+ U_t)\big\}^{4}\bigg).
\end{split}
\end{align}
The uniform ergodic theorem and Assumption \ref{as:4.9}(ii) imply $\frac{1}{n} \sum_{t=1}^n  U_t^4 \overset{a.s.}{\to}\EE\big[U_t^4 \big]<\infty$.  Further, \eqref{eq:4.A.4} leads to $\rho^{t}(1+U_t)\overset{a.s.}{\to}0$ as $t\to \infty$, which in turn implies $\big\{\rho^{t}(1+ U_t)\big\}^{4}\overset{a.s.}{\to}0$ as $t\to \infty$. Ces\'aro's lemma yields $\frac{1}{n}\sum_{t=1}^n \big\{\rho^{t}(1+ U_t)\big\}^4 \overset{a.s.}{\to}0$ and we have $\lim_{n \to \infty}\frac{1}{n} \sum_{t=1}^n \big|\big|\hat{D}_t\big|\big|^4<\infty$ almost surely. Thus, $\max_t||\hat{D}_t||/\sqrt{n}\overset{a.s.}{\to}0$ as
\begin{align*}
\bigg(\frac{\max_t||\hat{D}_t||}{\sqrt{n}}\bigg)^4 \leq   \frac{1}{n^2}\sum_{t=1}^n||\hat{D}_t||^4 \overset{a.s.}{\to}0.
\end{align*}
and $\mathbbm{1}_{ \{ ||\lambda_1|| (C^{2}+\EE^*[\eta_t^{*2}]) \max_{t} ||\hat{D}_t||+|\lambda_2|\geq \sqrt{n} s_n \varepsilon \}}\overset{a.s.}{\to}0$ follows.
Combining results, establishes $\frac{1}{s_n^2}\sum_{t=1}^n  \EE^*\big[Z_{n,t}^{*2}\mathbbm{1}_{\{|Z_{n,t}^{*}|\geq s_n \epsilon \}}\big] \overset{a.s.}{\to}0$.
 The Central Limit Theorem for triangular arrays (cf.\ \citeauthor{billingsley1986probability}, \citeyear{billingsley1986probability}, Theorem 27.3) implies that $\sum_{t=1}^n Z_{n,t}^*$ converges in conditional distribution to $N\big(0,\lambda'\Upsilon_\alpha\lambda \big)$ almost surely.
\end{proof}

\noindent
\textit{Proof of Proposition \ref{prop:4.1}.}
Since $L_n^*$ is maximized at $\hat{\theta}_n^*$  its derivative is equal to zero: $\frac{\partial L_n^*(\hat{\theta}_n^*)}{\partial \theta}=0$. A Taylor expansion around $\hat{\theta}_n$ yields
\begin{align*}
0 = \sqrt{n}\frac{\partial L_n^*(\hat{\theta}_n^*)}{\partial \theta}= \frac{1}{\sqrt{n}}\sum_{t=1}^n \frac{\partial}{\partial \theta} \ell_t^*(\hat{\theta}_n) + \bigg( \frac{1}{n}\sum_{t=1}^n \frac{\partial^2}{\partial \theta \partial \theta'} \ell_t^*(\breve{\theta}_n)\bigg) \sqrt{n}\big(\hat{\theta}_n^*-\hat{\theta}_n\big)
\end{align*}
with $\breve{\theta}_n$ between $\hat{\theta}_n^*$ and $\hat{\theta}_n$. Lemma \ref{lem:4.6} establishes $\frac{1}{n}\sum_{t=1}^n \frac{\partial^2}{\partial \theta \partial \theta'} \ell_t^*(\breve{\theta}_n)\overset{p^*}{\to}-2J$ almost surely. Since $\frac{\partial}{\partial \theta} \ell_t^*(\theta)=\tilde{D}_t(\theta)\big(\frac{\epsilon_t^{*2}}{\tilde{\sigma}_t^2(\theta)}-1\big)$, the first term on the right hand side reduces to $\frac{1}{\sqrt{n}}\sum_{t=1}^n \hat{D}_t\big(\eta_t^{*2}-1\big)$.  Hence, we obtain
\begin{align}
\label{eq:4.4.4}
\sqrt{n}\big(\hat{\theta}_n^*-\hat{\theta}_n\big)=\frac{1}{2}J^{-1} \frac{1}{\sqrt{n}}\sum_{t=1}^n \hat{D}_t\big(\eta_t^{*2}-1\big)+o_{p^*}(1)
\end{align}
almost surely with $\frac{1}{\sqrt{n}}\sum_{t=1}^n \hat{D}_t\big(\eta_t^{*2}-1\big)$ converging in conditional distribution to $N\big(0,(\kappa-1)J\big)$ almost surely by Lemma \ref{lem:4.7}. The claim follows. \qed \\

\begin{lemma}
\label{lem:4.8}
Suppose Assumptions \ref{as:4.1}--\ref{as:4.9} hold with $a=\pm 6$, $b=6$ and $c=2$. Then, $I_n^*(z)$ given in \eqref{eq:219849862} satisfies $I_n^*(z)\overset{p^*}{\to} \frac{z^2}{2}f(\xi_\alpha)$ in probability.
\end{lemma}
\begin{proof}
Using Fubini's theorem, the conditional expectation is equal to
\begin{align*}
\EE^*\big[I_n^*(z)\big]
=& \sum_{t=1}^n \int_0^{z/\sqrt{n}}\EE^*\Big[\mathbbm{1}_{\{\eta_t^*\leq \hat{\xi}_{n,\alpha}+s\}}-\mathbbm{1}_{\{\eta_t^*< \hat{\xi}_{n,\alpha}\}}\Big]ds\\
=&  \int_0^{z}\sqrt{n}\bigg(\hat{\mathbbm{F}}_n\Big(\hat{\xi}_{n,\alpha}+\frac{u}{\sqrt{n}}\Big)-\hat{\mathbbm{F}}_n(\hat{\xi}_{n,\alpha}-)\bigg)\:du\\
=&  \underbrace{\int_0^{z}\sqrt{n}\bigg(\hat{\mathbbm{F}}_n\Big(\hat{\xi}_{n,\alpha}+\frac{u}{\sqrt{n}}\Big)-\hat{\mathbbm{F}}_n(\hat{\xi}_{n,\alpha}-)-F\Big(\hat{\xi}_{n,\alpha}+\frac{u}{\sqrt{n}}\Big)+F(\hat{\xi}_{n,\alpha})\bigg)\:du}_{I}\\
&\qquad +\underbrace{\int_0^{z}\sqrt{n}\bigg(F\Big(\hat{\xi}_{n,\alpha}+\frac{u}{\sqrt{n}}\Big)-F(\hat{\xi}_{n,\alpha})\bigg)\:du}_{II}.
\end{align*}
Regarding $I$, take $\varrho \in (0,1/2)$ and set $\bar{\mathcal{I}}_n=\big[\xi_\alpha -0.5 n^{-\varrho},\xi_\alpha +0.5 n^{-\varrho}\big]$. Since $\sqrt{n}(\hat{\xi}_{n,\alpha}-\xi_\alpha)=O_p(1)$, the probabilities of the events $\big \{\hat{\xi}_{n,\alpha}+\frac{|z|}{\sqrt{n}}\notin \bar{\mathcal{I}}_n \big\}$ and $\big \{\hat{\xi}_{n,\alpha}-\frac{|z|}{\sqrt{n}}\notin \bar{\mathcal{I}}_n \big \}$ can be made arbitrarily small for large $n$. If $\hat{\xi}_{n,\alpha}+\frac{|z|}{\sqrt{n}}\in \bar{\mathcal{I}}_n$ and $\hat{\xi}_{n,\alpha}-\frac{|z|}{\sqrt{n}}\in \bar{\mathcal{I}}_n$, then $\hat{\xi}_{n,\alpha} \in \bar{\mathcal{I}}_n$ and $\hat{\xi}_{n,\alpha}+\frac{u}{\sqrt{n}} \in \bar{\mathcal{I}}_n$ belong to $\bar{\mathcal{I}}_n$ for all $u$ between $0$ and $z$. In that case 
\begin{align*}
|I|\leq |z|\sup_{x,y \in \bar{\mathcal{I}}_n} \Big| \sqrt{n} \big(\hat{\mathbbm{F}}_n(x)-\hat{\mathbbm{F}}_n(y-)\big)- \sqrt{n} \big(F(x)-F(y)\big)\Big|\overset{p}{\to}0
\end{align*}
by Lemma \ref{lem:4.3}. Focusing on $II$, the mean value theorem implies that 
\begin{align*}
II = \int_0^{z}u f\big(\hat{\xi}_{n,\alpha}+\varepsilon_n\big)\:du = \underbrace{\int_0^{z}u \Big(f\big(\hat{\xi}_{n,\alpha}+\varepsilon_n\big)-f(\xi_\alpha)\Big)\:du}_{II_1}+ \underbrace{\int_0^{z}u f(\xi_\alpha)\:du}_{II_2}
\end{align*}
with $\varepsilon_n$ lying between $0$ and $u/\sqrt{n}$. Since $|\varepsilon_n| \leq |z|/\sqrt{n}$ and $\hat{\xi}_{n,\alpha}\overset{a.s.}{\to} \xi_\alpha$ we have
\begin{align*}
|II_1| \leq \frac{z^2}{2} \sup_{|v|\leq |z| } \Big|f\Big(\hat{\xi}_{n,\alpha}+\frac{v}{n}\Big)-f(\xi_\alpha)\Big|\overset{a.s.}{\to}0.
\end{align*}
Further, $II_2$ simplifies to $II_2=\frac{z^2}{2}f(\xi_\alpha)$ and combining results establishes
%
$\EE^*\big[I_n^*(z)\big]\overset{p}{\to}\frac{z^2}{2}f(\xi_\alpha)$.
%
The conditional variance vanishes in probability as
\begin{align*}
&\Var^*\big[I_n^*(z)\big] = \sum_{t=1}^n \Var^*\bigg[\int_0^{z/\sqrt{n}}(\mathbbm{1}_{\{\eta_t^*\leq \hat{\xi}_{n,\alpha}+s\}}-\mathbbm{1}_{\{\eta_t^*< \hat{\xi}_{n,\alpha}\}})ds\bigg]\\
\leq& \sum_{t=1}^n \frac{|z|}{\sqrt{n}}\EE^*\bigg[\int_0^{z/\sqrt{n}}(\mathbbm{1}_{\{\eta_t^*\leq \hat{\xi}_{n,\alpha}+s\}}-\mathbbm{1}_{\{\eta_t^*< \hat{\xi}_{n,\alpha}\}})ds\bigg]
=  \frac{|z|}{\sqrt{n}} \EE^*\big[I_n^*(z)\big]\overset{p}{\to}0,
\end{align*}
where the inequality follows from the fact that
\begin{align}
\label{eq:4.A.55}
\Var(Y)\leq |c|\: \EE[Y]\qquad 
\end{align}
with  $Y=\int_0^c (\mathbbm{1}_{\{X\leq s\}}-\mathbbm{1}_{\{X<0\}})ds$, $X$ is a real-valued integrable random variable and $c \in \R$ (cf.\  \citeauthor{francq2015risk}, \citeyear{francq2015risk}, p.\ 171).
\end{proof}

\begin{lemma}
\label{lem:4.9}
Suppose Assumptions \ref{as:4.1}--\ref{as:4.10} hold with $a=\pm 12$, $b=12$ and $c=6$. Then,  $J_{n,1}^*(z)$ given in \eqref{eq:4.4.6}
satisfies $J_{n,1}^*(z) \overset{d^*}{\to}\Gamma \big(\frac{r}{2},\frac{\kappa-1}{4}\xi_\alpha^2f(\xi_\alpha)\big)$ in probability, i.e.\ a Gamma distribution with shape parameter $\frac{r}{2}$ and scale parameter $\frac{\kappa-1}{4}\xi_\alpha^2f(\xi_\alpha)$.
\end{lemma}


\begin{proof}
We set $\bar{\xi}_{n,\alpha}=\hat{\xi}_{n,\alpha}+\frac{z}{\sqrt{n}}$ and define for $z \in \R$ and $u \in \R^r$
\vspace{-0.5cm}
\begin{align*}
T_{n}^*=& T_{n}^*(z,u)=\sum_{t=1}^n\tau_t^{*}\\
\tau_{t}^{*}=&\tau_{t}^{*}(z,u) =\int_{0}^{(1-\tilde{\lambda}_t^{-1}(u))\eta_t^*}(\mathbbm{1}_{\{\eta_t^*-\bar{\xi}_{n,\alpha}\leq s\}}-\mathbbm{1}_{\{\eta_t^*-\bar{\xi}_{n,\alpha}<0\}})ds\\
\tilde{\lambda}_t=&\tilde{\lambda}_t(u) = \frac{\tilde{\sigma}_t(\hat{\theta}_n+n^{-1/2}u)}{\tilde{\sigma}_t(\hat{\theta}_n)},
\end{align*}
where we suppress the dependence of $\tau_{t}^{*}$ and $\tilde{\lambda}_t$ on $n$ and drop the arguments $z$ and $u$ at times for notational simplicity. Further, we split $T_{n}^*$ into $T_{n,1}^*=\sum_{t=1}^n\mathbbm{1}_{\{\tilde{\lambda}_t>1\}}\tau_t^{*}$ and $T_{n,2}^*=\sum_{t=1}^n\mathbbm{1}_{\{\tilde{\lambda}_t<1\}}\tau_t^{*}$. Let $A>0$; We establish the lemma's claim in three steps:
\begin{enumerate}

	\item[] \textit{Step 1:} 
\vspace{-1.2cm}
\begin{align*}
T_{n,k}^*(z,u)\overset{p^*}{\to}\begin{cases}\frac{1}{2}\xi_{\alpha}^2f(\xi_\alpha)\EE\big[\mathbbm{1}_{\{D_t'u>0\}}u'D_tD_t'u\big]\quad \text{if } k=1\\
\frac{1}{2}\xi_{\alpha}^2f(\xi_\alpha)\EE\big[\mathbbm{1}_{\{D_t'u<0\}}u'D_tD_t'u\big]\quad \text{if } k=2
\end{cases}
\end{align*}
\qquad \quad \;\; in probability for all $z \in \R$ and for all $u\in \{u\in \R^r:||u||\leq A\}$; 
    
    \item[] \textit{Step 2:} $\sup_{||u||\leq A}\big|T_n^*(z,u)-\frac{1}{2}\xi_\alpha^2f(\xi_\alpha) u'Ju\big|\overset{p^*}{\to}0$ in probability for all $z \in \R$;
    
    \item[] \textit{Step 3:} $J_{n,1}^*(z) \overset{d^*}{\to}\Gamma \big(\frac{r}{2},\frac{\kappa-1}{4}\xi_\alpha^2f(\xi_\alpha)\big)$ in probability.

\end{enumerate}
Consider \textit{Step 1}; using the identity 
$\int_0^c (\mathbbm{1}_{\{x\leq s\}}-\mathbbm{1}_{\{x<0\}})ds= (x-c)(\mathbbm{1}_{\{c\leq x<0\}}-\mathbbm{1}_{\{0\leq x<c\}})$
for $c,s,x \in \R$ we rewrite $\tau_t^*$ yielding
\begin{align*}
T_{n,1}^*=\sum_{t=1}^n\mathbbm{1}_{\{\tilde{\lambda}_t>1\}}\underbrace{\tilde{\lambda}_t^{-1} \!\!\int_{0}^{(1-\tilde{\lambda}_t)\bar{\xi}_{n,\alpha}}\!\!\!\big(\mathbbm{1}_{\{\eta_t^*-\bar{\xi}_{n,\alpha}\leq s\}}-\mathbbm{1}_{\{\eta_t^*-\bar{\xi}_{n,\alpha}<0\}}\big)ds}_{=\tau_t^*}.
\end{align*}
Using Fubini's theorem and expanding, the bootstrap mean of $T_{n,1}^*$ is equal to
\begin{align}
\nonumber
&\EE^*\big[T_{n,1}^*\big]=
\sum_{t=1}^n\mathbbm{1}_{\{\tilde{\lambda}_t>1\}} \tilde{\lambda}_t^{-1}  \int_{0}^{(1-\tilde{\lambda}_t)\bar{\xi}_{n,\alpha}}\!\!\big( \hat{\mathbbm{F}}_n(\bar{\xi}_{n,\alpha}+ s)-\hat{\mathbbm{F}}_n(\bar{\xi}_{n,\alpha}-)\big)ds\\
\nonumber
=& \underbrace{\frac{1}{2}\bar{\xi}_{n,\alpha}^2f(\xi_\alpha)\frac{1}{n}\sum_{t=1}^n\mathbbm{1}_{\{\tilde{\lambda}_t>1\}} \tilde{\lambda}_t^{-1} n (\tilde{\lambda}_t-1)^2}_{I}\\
\label{eq:7821821}
 &+ \underbrace{\sum_{t=1}^n\mathbbm{1}_{\{\tilde{\lambda}_t>1\}} \tilde{\lambda}_t^{-1}  \int_{0}^{(1-\tilde{\lambda}_t)\bar{\xi}_{n,\alpha}}\!\!\big( F(\bar{\xi}_{n,\alpha}+ s)-F(\bar{\xi}_{n,\alpha})-sf(\xi_\alpha)\big)ds}_{II}\\
\nonumber
&+ \underbrace{\sum_{t=1}^n\mathbbm{1}_{\{\tilde{\lambda}_t>1\}} \tilde{\lambda}_t^{-1}\!\!\!  \int_{0}^{(1-\tilde{\lambda}_t)\bar{\xi}_{n,\alpha}}\!\!\!\!\big( \hat{\mathbbm{F}}_n(\bar{\xi}_{n,\alpha}+ s)-\hat{\mathbbm{F}}_n(\bar{\xi}_{n,\alpha}-) - F(\bar{\xi}_{n,\alpha}+ s)+F(\bar{\xi}_{n,\alpha})\big)ds}_{III}.
\end{align}
We consider each term in turn. Expanding $I$ we obtain
\begin{align*}
    I=\underbrace{\frac{1}{2}\bar{\xi}_{n,\alpha}^2f(\xi_\alpha)}_{I_1}\bigg(\underbrace{\frac{1}{n}\sum_{t=1}^n\mathbbm{1}_{\{\tilde{\lambda}_t>1\}} n (\tilde{\lambda}_t-1)^2}_{I_2} + \underbrace{\frac{1}{n}\sum_{t=1}^n\mathbbm{1}_{\{\tilde{\lambda}_t>1\}}( \tilde{\lambda}_t^{-1}-1) n (\tilde{\lambda}_t-1)^2}_{I_3}\bigg).
\end{align*}
Theorem \ref{thm:4.1} yields $\bar{\xi}_{n,\alpha}\overset{a.s.}{\to}\xi_\alpha$ such that $I_1\overset{a.s.}{\to}\frac{1}{2}\xi_{\alpha}^2f(\xi_\alpha)$. Lemma \ref{lem:4.2} implies $I_2\overset{a.s.}{\to}\EE\big[\mathbbm{1}_{\{D_t'u>0\}}u'D_t D_t'u\big]$. Further, the lemma entails
$n^{1/8}\max_{t=1,\dots,n}\big|\tilde{\lambda}_t^{-1}-1\big|\overset{a.s.}{\to}0$ as
\begin{align}
\label{eq:12479124}
    \Big(n^{1/8}\max_{t=1,\dots,n}\big|\tilde{\lambda}_t^{-1}-1\big|\Big)^3\leq \frac{1}{n^{1/8}}\underbrace{\frac{1}{n}\sum_{t=1}^{n}\big(\sqrt{n}\big|\tilde{\lambda}_t^{-1}-1\big|\big)^3}_{\overset{a.s.}{\to}\EE[|D_t'u|^3]}\overset{a.s.}{\to}0.
\end{align}
It follows that
\begin{align*}
  |I_3|\leq \max_{t=1,\dots,n}\big|\tilde{\lambda}_t^{-1}-1\big| \underbrace{\frac{1}{n}\sum_{t=1}^n\mathbbm{1}_{\{\tilde{\lambda}_t>1\}} n (\tilde{\lambda}_t-1)^2}_{=I_2}\overset{a.s.}{\to}0,
\end{align*}
which establishes $I\overset{a.s.}{\to}\frac{1}{2}\xi_{\alpha}^2f(\xi_\alpha)\EE\big[\mathbbm{1}_{\{D_t'u>0\}}u'D_t D_t'u\big]$. Consider $II$ in \eqref{eq:7821821}; we define 
\begin{align*}
\bar{\xi}_{n,\alpha}^+= \bar{\xi}_{n,\alpha}+ \max_{t=1,\dots,n}|\tilde{\lambda}_t-1|\:|\bar{\xi}_{n,\alpha}|\\
\bar{\xi}_{n,\alpha}^-= \bar{\xi}_{n,\alpha}- \max_{t=1,\dots,n}|\tilde{\lambda}_t-1|\:|\bar{\xi}_{n,\alpha}|
\end{align*}
and set $\mathcal{I}_n=[\xi_\alpha -a_n,\xi_\alpha +a_n]$ with $a_n \sim n^{-1/8} \log n$. Similar to \eqref{eq:12479124} we obtain
\begin{align}
\label{eq:4.A.64}
n^{1/8}\max_{t=1,\dots,n}\big|\tilde{\lambda}_t-1\big| \overset{a.s.}{\to}0
\end{align}
and together with $\sqrt{n}(\hat{\xi}_{n,\alpha}-\xi_\alpha)=O_p(1)$ we find that $n^{1/8}\big(\bar{\xi}_{n,\alpha}^+ -\xi_\alpha\big)\overset{p}{\to}0$ and $n^{1/8}\big(\bar{\xi}_{n,\alpha}^--\xi_\alpha\big)\overset{p}{\to}0$. Hence, the probabilities of the events $\big \{\bar{\xi}_{n,\alpha}^+\notin \mathcal{I}_n \big\}$ and $\big \{\bar{\xi}_{n,\alpha}^-\notin \mathcal{I}_n \big\}$ can be made arbitrarily small for large $n$. If $\bar{\xi}_{n,\alpha}^+$ and $\bar{\xi}_{n,\alpha}^-$ belong to $\mathcal{I}_n$, then
\begin{align*}
|II|=& \bigg|\sum_{t=1}^n\mathbbm{1}_{\{\tilde{\lambda}_t>1\}} \tilde{\lambda}_t^{-1}  \int_{0}^{(1-\tilde{\lambda}_t)\bar{\xi}_{n,\alpha}}\!\!\!s\big( f(\bar{\xi}_{n,\alpha}+\varepsilon_{t,n})-f(\xi_\alpha)\big)ds\bigg|\\
\leq& \frac{1}{2} \bar{\xi}_{n,\alpha}^2 \sup_{x \in \mathcal{I}_n}\big| f(x)-f(\xi_\alpha)\big| \underbrace{\frac{1}{n}\sum_{t=1}^n \mathbbm{1}_{\{\tilde{\lambda}_t>1\}} \tilde{\lambda}_t^{-1} n\big(\tilde{\lambda}_t-1\big)^2}_{=I_2+I_3}.
\end{align*}
with $\varepsilon_{t,n}$ between $0$ and $(1-\tilde{\lambda}_t)\bar{\xi}_{n,\alpha}$. As 
$\mathcal{I}_n$ shrinks to $\xi_\alpha$ and $f$ is continuous in a neighborhood of $\xi_\alpha$ (see Assumption \ref{as:4.4}(\ref{as:4.4.2})) we have 
$ \sup_{x \in \mathcal{I}_n}\big| f(x)-f(\xi_\alpha)\big|\to 0$. Together with $\bar{\xi}_{n,\alpha}\overset{a.s.}{\to} \xi_{\alpha}$ and $I_2+I_3\overset{a.s.}{\to}\EE[ \mathbbm{1}_{\{D_t'u>0\}} u'D_tD_t'u]$ we establish $II\overset{p}{\to}0$. Focusing on $III$ in \eqref{eq:7821821}, we only consider the case of $\hat{\xi}_{n,\alpha}^+, \hat{\xi}_{n,\alpha}^-\in \mathcal{I}_n$. In this case $\bar{\xi}_{n,\alpha}$ and $\bar{\xi}_{n,\alpha}+ s$ belong to $\mathcal{I}_n$ for all $s$ between $0$ and $(1-\tilde{\lambda}_t)\bar{\xi}_{n,\alpha}$ for all $t$. We obtain
\begin{align*}
|III|\leq & \big|\bar{\xi}_{n,\alpha}\big|\!\sup_{x,y \in \mathcal{I}_n}\!\Big| \sqrt{n}\big(\hat{\mathbbm{F}}_n(x)-\hat{\mathbbm{F}}_n(y-)\big)- \sqrt{n}\big(F(x)-F(y)\big)\Big|\frac{1}{n} \sum_{t=1}^n \sqrt{n}  \big|\tilde{\lambda}_t^{-1}-1\big|\overset{a.s.}{\to}0
\end{align*}
by $\bar{\xi}_{n,\alpha}\overset{a.s.}{\to}\xi_{\alpha}$ and Lemmas \ref{lem:4.2} and \ref{lem:4.3}. We conclude $III\overset{p}{\to}0$ and establish
\begin{align}
\label{eq:4.A.66}
\EE\big[T_{n,1}^*\big] \overset{p}{\to} \frac{1}{2}\xi_\alpha f(\xi_\alpha)\EE\big[ \mathbbm{1}_{\{D_t'u>0\}} u'D_tD_t'u\big].
\end{align}
Employing \eqref{eq:4.A.55}, the bootstrap variance of  $T_{n,1}^*$ is bounded by
\begin{align*}
\Var^*\big[T_{n,1}^*\big]
=& \sum_{t=1}^n \mathbbm{1}_{\{\tilde{\lambda}_t>1\}}\tilde{\lambda}_t^{-2} \Var^*\bigg[\!\int_{0}^{(1-\tilde{\lambda}_t)\bar{\xi}_{n,\alpha}}\!\!\!\!\!(\mathbbm{1}_{\{\eta_t^*-\bar{\xi}_{n,\alpha}\leq s\}}-\mathbbm{1}_{\{\eta_t^*-\bar{\xi}_{n,\alpha}<0\}})ds\bigg]\\
\leq & \sum_{t=1}^n \tilde{\lambda}_t^{-2} \big|\tilde{\lambda}_t-1\big|\:\big|\bar{\xi}_{n,\alpha}\big| \EE^*\bigg[\!\int_{0}^{(1-\tilde{\lambda}_t)\bar{\xi}_{n,\alpha}}\!\!\!\!(\mathbbm{1}_{\{\eta_t^*-\bar{\xi}_{n,\alpha}\leq s\}}-\mathbbm{1}_{\{\eta_t^*-\bar{\xi}_{n,\alpha}<0\}})ds\bigg]\\
= & \big|\bar{\xi}_{n,\alpha}\big|\sum_{t=1}^n \tilde{\lambda}_t^{-2} \big|\tilde{\lambda}_t-1\big|  \int_{0}^{(1-\tilde{\lambda}_t)\bar{\xi}_{n,\alpha}}\!\!\!\!\big(\hat{\mathbbm{F}}_n(\bar{\xi}_{n,\alpha}+s)-\hat{\mathbbm{F}}_n(\bar{\xi}_{n,\alpha}-\!)\big)ds\\
\leq & \bar{\xi}_{n,\alpha}^2\frac{1}{n}\sum_{t=1}^n  n\big|\tilde{\lambda}_t^{-1}-1\big|^2 \big(\hat{\mathbbm{F}}_n(\bar{\xi}_{n,\alpha}^+) -\hat{\mathbbm{F}}_n(\bar{\xi}_{n,\alpha}^-)\big).
\end{align*}
We have $\bar{\xi}_{n,\alpha}^2\overset{a.s.}{\to}\xi_\alpha^2$ and $\frac{1}{n}\sum_{t=1}^n  n\big|\tilde{\lambda}_t^{-1}-1\big|^2\overset{a.s.}{\to}\EE[u'D_tD_t'u]$ by Lemma \ref{lem:4.2}. Moreover,  $\hat{\mathbbm{F}}_n(\bar{\xi}_{n,\alpha}^+) -\hat{\mathbbm{F}}_n(\bar{\xi}_{n,\alpha}^-)\overset{p}{\to}0$ since $\bar{\xi}_{n,\alpha}^+ \overset{p}{\to}\xi_\alpha$, $\bar{\xi}_{n,\alpha}^- \overset{p}{\to}\xi_\alpha$  and  $\sup_{x \in \R}|\hat{\mathbbm{F}}_n(x)-F(x)|\overset{a.s.}{\to}0$ (Lemma \ref{lem:4.1}) and $\Var^*[T_{n,1}^*]\overset{p}{\to}0$ follows. Together with \eqref{eq:4.A.66} we establish $T_{n,1}^*\overset{p^*}{\to}\frac{1}{2}\xi_{\alpha}^2f(\xi_\alpha)\EE[\mathbbm{1}_{\{D_t'u>0\}}u'D_tD_t'u]$ in probability. The proof of $T_{n,2}^*\overset{p^*}{\to}\frac{1}{2}\xi_{\alpha}^2f(\xi_\alpha)\EE[\mathbbm{1}_{\{D_t'u<0\}}u'D_tD_t'u]$ in probability is analogous and hence omitted.

Regarding \textit{Step 2} the triangle inequality yields
\begin{align}
\label{eq:4.A.67}
\begin{split}
\sup_{||u||\leq A}\Big|T_n^*(z,u)-\plim\limits_{n\to \infty}T_{n}^*(z,u)\Big|\leq& \sup_{||u||\leq A}\Big|T_{n,1}^*(z,u)-\plim\limits_{n\to \infty}T_{n,1}^*(z,u)\big]\Big|\\
&\quad +\sup_{||u||\leq A}\Big|T_{n,2}^*(z,u)-\plim\limits_{n\to \infty}T_{n,2}^*(z,u)\big]\Big|.
\end{split}
\end{align}
Let $N\geq 1$ be an integer. We divide the (hyper-)cube $[-A,A]^r$ into $L=(2N)^r$ cubes with side length $A/N$. Let $u_\bullet(\ell)$ and $u^\bullet(\ell)$ denote the lower left and upper right vertex of cube $\ell$. For $u$ satisfying $u_\bullet(\ell)\leq u\leq u^\bullet(\ell)$ (element-by-element comparison) Assumption \ref{as:4.8} implies $\tilde{\lambda}_t(u_\bullet(\ell))\leq \tilde{\lambda}_t(u)\leq \tilde{\lambda}_t(u^\bullet(\ell))$. Further, Theorem \ref{thm:4.1} results in $\bar{\xi}_{n,\alpha}\overset{a.s.}{\to}\xi_\alpha<0$. Thus, we have for $n$ sufficiently large
\begin{align*}
T_{n,1}^*\big(z,u_\bullet(\ell)\big)\leq& T_{n,1}^*(z,u)\leq T_{n,1}^*\big(z,u^\bullet(\ell)\big)\\
T_{n,2}^*\big(z,u^\bullet(\ell)\big)\leq& T_{n,2}^*(z,u)\leq T_{n,2}^*\big(z,u_\bullet(\ell)\big).
\end{align*}
Let $k\in\{1,2\}$; we obtain
\begin{align*}
&\sup_{||u||\leq A}\Big|T_{n,k}^*(z,u)-\plim\limits_{n\to \infty} T_{n,k}^*(z,u)\Big|\\
\leq&   \max_{1 \leq \ell\leq L}  \Big|T_{n,k}^*\big(z,u^\bullet(\ell)\big)-\plim\limits_{n\to \infty}T_{n,k}^*\big(z,u^\bullet(\ell)\big)\Big|+ \underbrace{\max_{1 \leq \ell\leq L} \sup_{u_\bullet(\ell)\leq u\leq u^\bullet(\ell)} \Big|T_{n,k}^*\big(z,u^\bullet(\ell)\big)-T_{n,k}^*(z,u)\Big|}_{A_n}\\
&\:\:+ \underbrace{\max_{1 \leq \ell\leq L} \sup_{u_\bullet(\ell)\leq u\leq u^\bullet(\ell)} \Big|\plim\limits_{n\to \infty} T_{n,k}^*\big(z,u^\bullet(\ell)\big)-\plim\limits_{n\to \infty}T_{n,k}^*(z,u)\Big)\Big|}_{B_n}
\end{align*}
with
\begin{align*}
A_n\leq & \max_{1 \leq \ell\leq L}  \Big|T_{n,k}^*\big(z,u^\bullet(\ell)\big)-T_{n,k}^*\big(z,u_\bullet(\ell)\big)\Big|\\
\leq & \max_{1 \leq \ell\leq L} \Big|T_{n,k}^*\big(z,u^\bullet(\ell)\big)-\plim\limits_{n\to \infty}T_{n,k}^*\big(z,u^\bullet(\ell)\big)\Big|\\
&\qquad + \max_{1 \leq \ell\leq L}  \Big|T_{n,k}^*\big(z,u_\bullet(\ell)\big)-\plim\limits_{n\to \infty}T_{n,k}^*\big(z,u_\bullet(\ell)\big)\Big|\\
&\qquad \qquad + \max_{1 \leq \ell\leq L}  \Big|\plim\limits_{n\to \infty}T_{n,k}^*\big(z,u^\bullet(\ell)\big)-\plim\limits_{n\to \infty}T_{n,k}^*\big(z,u_\bullet(\ell)\big)\Big|\\ \\
B_n\leq& \max_{1 \leq \ell\leq L}  \Big|\plim\limits_{n\to \infty}T_{n,k}^*\big(z,u^\bullet(\ell)\big)-\plim\limits_{n\to \infty}T_{n,k}^*\big(z,u_\bullet(\ell)\big)\Big|.
\end{align*}
Hence, we establish the following bound
\begin{align*}
\sup_{||u||\leq A}\Big|T_{n,k}^*(z,u)-\plim\limits_{n\to \infty}T_{n,k}^*(z,u)\Big|
\leq& 2IV+V+2VI
\end{align*}
with
\begin{align*}
IV=&\max_{1 \leq \ell\leq L}  \Big|\plim\limits_{n\to \infty}T_{n,k}^*\big(z,u^\bullet(\ell)\big)-\plim\limits_{n\to \infty}T_{n,k}^*\big(z,u_\bullet(\ell)\big)\Big|\\
V=& \max_{1 \leq \ell\leq L}  \Big|T_{n,k}^*\big(z,u_\bullet(\ell)\big)-\plim\limits_{n\to \infty}T_{n,k}^*\big(z,u_\bullet(\ell)\big)\Big|\\
VI=& \max_{1 \leq \ell\leq L} \Big|T_{n,k}^*\big(z,u^\bullet(\ell)\big)-\plim\limits_{n\to \infty}T_{n,k}^*\big(z,u^\bullet(\ell)\big)\Big|.
\end{align*}
Regarding $IV$, we have for every $u$ satisfying $||u||\leq A$ that
\begin{align*}
\plim\limits_{n\to \infty}T_{n,k}^*(z,u)=\begin{cases}\frac{1}{2}\xi_{\alpha}^2f(\xi_\alpha)\EE\big[\mathbbm{1}_{\{D_t'u>0\}}u'D_tD_t'u\big]\quad \text{if } k=1\\
\frac{1}{2}\xi_{\alpha}^2f(\xi_\alpha)\EE\big[\mathbbm{1}_{\{D_t'u<0\}}u'D_tD_t'u\big]\quad \text{if } k=2
\end{cases}
\end{align*}
is continuous in $u$. Together with $||u^\bullet(\ell)-u_\bullet(\ell)||\leq \frac{A}{N}$ for every $\ell$, it follows that $IV$ can be made arbitrarily small by choosing $N$ sufficiently large. Given $N$ (and $L$), $V\overset{p^*}{\to}0$ in probability and $VI\overset{p^*}{\to}0$ in probability by \textit{Step 1}, which completes \textit{Step 2}.

Consider \textit{Step 3}; for each $\varepsilon>0$ we obtain
\begin{align*}
&\PP^*\bigg[\bigg|J_{n,1}^*(z)- \frac{1}{2}\xi_\alpha^2f(\xi_\alpha) \sqrt{n}\big(\hat{\theta}_n^*-\hat{\theta}_n\big)'J\sqrt{n}\big(\hat{\theta}_n^*-\hat{\theta}_n\big)\bigg|\geq \varepsilon \bigg]\\
\leq & \PP^*\bigg[\sup_{||u||\leq A}\bigg|T_n^*(u)- \frac{1}{2}\xi_\alpha^2f(\xi_\alpha) u'Ju\bigg|\geq \varepsilon \bigg]+\PP^*\Big[\sqrt{n}||\hat{\theta}_n^*-\hat{\theta}_n||> A\Big].
\end{align*}
With regard to Proposition \ref{prop:4.1}, the second term can be made arbitrarily small for large $n$ by choosing $A$ sufficiently large. Given $A$, the first term vanishes in probability by \textit{Step 2}. Expanding $\frac{1}{2} = \frac{\kappa-1}{8} \frac{4}{\kappa-1}$, we establish
\begin{align*}
J_{n,1}^*(z)= \frac{\kappa-1}{8}\xi_\alpha^2 f(\xi_\alpha) \sqrt{n}\big(\hat{\theta}_n^*-\hat{\theta}_n\big)'\frac{4}{\kappa-1} J \sqrt{n}\big(\hat{\theta}_n^*-\hat{\theta}_n\big) +o_{p^*}(1)
\end{align*}
in probability. Proposition \ref{prop:4.1} implies that $\sqrt{n}(\hat{\theta}_n^*-\hat{\theta}_n)'\frac{4}{\kappa-1} J \sqrt{n}(\hat{\theta}_n^*-\hat{\theta}_n)\overset{d^*}{\to}\chi_r^2$ almost surely, where $\chi_r^2$ denotes the Chi Square distribution with $r$ degrees of freedom. Further, note that $Y=cQ$ with $c>0$ and $Q\sim \chi_r^2$ implies $Y\sim \Gamma(r/2,2c)$. It follows that 
%
$J_{n,1}^*(z) \overset{d^*}{\to}\Gamma \big(\frac{r}{2},\frac{\kappa-1}{4}\xi_\alpha^2f(\xi_\alpha)\big)$
%
in probability, which establishes the lemma's claim.
\end{proof}

\color{black}

\begin{remark}
\label{rem:4.5}
In the preceding proof of Lemma \ref{lem:4.9} a compactness/supremum argument is employed, in which the monotonicity condition of Assumption \ref{as:4.8} plays a central role. In contrast, the proof of \citeauthor{francq2015risk} (\citeyear{francq2015risk}, p.172) rests on a conditional argument involving the density of $\eta_t$ given $\{\hat{\theta}_n-\theta_0,\eta_u : u<t\}$. This argument does not carry over to the residual bootstrap since the probability mass function of $\eta_t^*$ given $\{\hat{\theta}_n^*-\hat{\theta}_n,\eta_u^* : u<t\}$ and $\mathcal{F}_n$ has, almost surely, a single point mass.
\end{remark}

\begin{lemma}
\label{lem:4.10}
Suppose Assumptions \ref{as:4.1}--\ref{as:4.10} with $a=\pm12$, $b=12$ and $c=6$. Then, $J_{n,2}^*(z)$ given in
\eqref{eq:4.4.7} satisfies $J_{n,2}^*(z)=z\xi_\alpha f(\xi_\alpha)\Omega'\sqrt{n}\big(\hat{\theta}_n^*-\hat{\theta}_n\big)+o_{p^*}(1)$ in probability.
\end{lemma}

\begin{proof}
Inserting $\hat{\eta}_t^*=\frac{\tilde{\sigma}_t(\hat{\theta}_n)}{\tilde{\sigma}_t(\hat{\theta}_n^*)}\eta_t^*$ into  \eqref{eq:4.4.7} leads to
\begin{align}
\label{eq:4.A.68}
J_{n,2}^*(z)=\sum_{t=1}^n \bigg(1-\frac{\tilde{\sigma}_t(\hat{\theta}_n)}{\tilde{\sigma}_t(\hat{\theta}_n^*)}\bigg)\underbrace{\eta_t^*\big(\mathbbm{1}_{\{\eta_t^*<\hat{\xi}_{n,\alpha}+\frac{z}{\sqrt{n}}\}}-\mathbbm{1}_{\{\eta_t^*<\hat{\xi}_{n,\alpha}\}}\big)}_{j_{n,t}^{*(2)}(z)}.
\end{align}
A Taylor expansion around $\hat{\theta}_n$ yields 
\begin{align}
\label{eq:4.A.69}
1-\frac{\tilde{\sigma}_t(\hat{\theta}_n)}{\tilde{\sigma}_t(\hat{\theta}_n^*)}=&\frac{1}{\tilde{\sigma}_t(\hat{\theta}_n)}\frac{\partial \tilde{\sigma}_t(\hat{\theta}_n)}{\partial \theta}\big(\hat{\theta}_n^*-\hat{\theta}_n\big)\\
\nonumber
& +\frac{1}{2}\big(\hat{\theta}_n^*-\hat{\theta}_n\big)'\frac{\tilde{\sigma}_t(\hat{\theta}_n)}{\tilde{\sigma}_t(\breve{\theta}_n)}\bigg(\frac{1}{\tilde{\sigma}_t(\breve{\theta}_n)}\frac{\partial^2 \tilde{\sigma}_t(\breve{\theta}_n)}{\partial \theta \partial \theta'}-\frac{2}{\tilde{\sigma}_t^2(\breve{\theta}_n)}\frac{\partial \tilde{\sigma}_t(\breve{\theta}_n)}{\partial \theta}\frac{\partial \tilde{\sigma}_t(\breve{\theta}_n)}{\partial \theta'}\bigg) \big(\hat{\theta}_n^*-\hat{\theta}_n\big)\\
\nonumber
=& \hat{D}_t'\big(\hat{\theta}_n^*-\hat{\theta}_n\big)+\frac{1}{2}\big(\hat{\theta}_n^*-\hat{\theta}_n\big)'\frac{\tilde{\sigma}_t(\hat{\theta}_n)}{\tilde{\sigma}_t(\breve{\theta}_n)}\Big(\tilde{H}_t(\breve{\theta}_n)-2\tilde{D}_t(\breve{\theta}_n)\tilde{D}_t'(\breve{\theta}_n)\Big) \big(\hat{\theta}_n^*-\hat{\theta}_n\big),
\end{align}
where $\breve{\theta}_n$ lies between $\hat{\theta}_n^*$ and $\hat{\theta}_n$. Plugging this result into \eqref{eq:4.A.68} gives 
\begin{align*}
J_{n,2}^*(z)=&\underbrace{\frac{1}{\sqrt{n}}\sum_{t=1}^n j_{n,t}^{*(2)}(z)\hat{D}_t'}_{I}\sqrt{n}\big(\hat{\theta}_n^*-\hat{\theta}_n\big)\\
&\!\! +\frac{1}{2}\sqrt{n}\big(\hat{\theta}_n^*-\hat{\theta}_n\big)'\underbrace{\frac{1}{n}\sum_{t=1}^n\frac{\tilde{\sigma}_t(\hat{\theta}_n)}{\tilde{\sigma}_t(\breve{\theta}_n)}\Big(\tilde{H}_t(\breve{\theta}_n)-2\tilde{D}_t(\breve{\theta}_n)\tilde{D}_t'(\breve{\theta}_n)\Big) j_{n,t}^{*(2)}(z)}_{II}\sqrt{n}\big(\hat{\theta}_n^*-\hat{\theta}_n\big).
\end{align*}
With regard to Proposition \ref{prop:4.1}, it suffices to show that $I\overset{p^*}{\to} \xi_{\alpha}zf(\xi_\alpha)\Omega'$ in probability and $II\overset{p^*}{\to} 0$ in probability. The conditional mean and variance of the first term are
\begin{align}
\label{eq:4.A.70}
\begin{split}
\EE^*[I]=&\sqrt{n}\EE^*\big[j_{n,t}^{*(2)}\big]\frac{1}{n}\sum_{t=1}^n\hat{D}_t'=\sqrt{n}\EE^*\big[j_{n,t}^{*(2)}(z)\big]\hat{\Omega}_n'\\
\Var^*[I]=&\Var^*\big[j_{n,t}^{*(2)}\big]\frac{1}{n}\sum_{t=1}^n\hat{D}_t\hat{D}_t'
= \Var^*\big[j_{n,t}^{*(2)}(z)\big]\hat{J}_n.
\end{split}
\end{align}
Lemma \ref{lem:4.2} states $\hat{\Omega}_n\overset{a.s.}{\to}\Omega$ and $\hat{J}_n\overset{a.s.}{\to}J$. Further, we have $\sqrt{n}\EE^*\big[j_{n,t}^{*(2)}(z)\big]\overset{p}{\to}z\xi_\alpha f(\xi_\alpha)$ and $\sqrt{n}\EE^*\Big[\big(j_{n,t}^{*(2)}(z)\big)^2\Big]\overset{p}{\to}|z| \xi_\alpha^2 f(\xi_\alpha)$, which implies $\Var^*\big[j_{n,t}^{*(2)}(z)\big]\overset{p}{\to}0$. To appreciate why, we obtain for $z\geq 0$ 
\begin{align*}
&\sqrt{n}\EE^*\big[j_{n,t}^{*(2)}(z)\big]
=\sqrt{n} \int_{\big[\hat{\xi}_{n,\alpha},\hat{\xi}_{n,\alpha}+\frac{z}{\sqrt{n}}\big)}x\:d\hat{\mathbbm{F}}_n(x)\\
=& \Big(\hat{\xi}_{n,\alpha}+\frac{z}{\sqrt{n}}\Big) \sqrt{n}\hat{\mathbbm{F}}_n\Big(\hat{\xi}_{n,\alpha}+\frac{z}{\sqrt{n}}-\!\Big)- \hat{\xi}_{n,\alpha} \sqrt{n}\hat{\mathbbm{F}}_n(\hat{\xi}_{n,\alpha}-)-\sqrt{n}\!\int_{\big[\hat{\xi}_{n,\alpha},\hat{\xi}_{n,\alpha}+\frac{z}{\sqrt{n}}\big)}\!\!\hat{\mathbbm{F}}_n(x)\:dx\\
=& \underbrace{\hat{\xi}_{n,\alpha} \sqrt{n}\bigg(\hat{\mathbbm{F}}_n\Big(\hat{\xi}_{n,\alpha}+\frac{z}{\sqrt{n}}-\!\Big)-\hat{\mathbbm{F}}_n\big(\hat{\xi}_{n,\alpha}-\big)\bigg)}_{I_1}+\underbrace{z\hat{\mathbbm{F}}_n\Big(\hat{\xi}_{n,\alpha}+\frac{z}{\sqrt{n}}-\!\Big)}_{I_2}\\
&\qquad -\underbrace{\int_{[0,z)}  \hat{\mathbbm{F}}_n \Big(\hat{\xi}_{n,\alpha}+\frac{y}{\sqrt{n}}\Big)\:dy}_{I_3}.
\end{align*}
Using Lemma \ref{lem:4.3} and the mean value theorem, we find
\begin{align*}
I_1 = \hat{\xi}_{n,\alpha} \sqrt{n}\bigg(F\Big(\hat{\xi}_{n,\alpha}+\frac{z}{\sqrt{n}}-\!\Big)-F\big(\hat{\xi}_{n,\alpha}\big)\bigg)+o_p(1)= z\hat{\xi}_{n,\alpha} f\big(\hat{\xi}_{n,\alpha}+\varepsilon_n\big)+o_p(1),
\end{align*}
where $0\leq \varepsilon_n\leq z/\sqrt{n}$, and together with Theorem \ref{thm:4.1} we establish $I_1\overset{p}{\to} z\xi_{\alpha}f(\xi_\alpha)$. Moreover, Theorem \ref{thm:4.1} and Lemma \ref{lem:4.1} imply $I_2\overset{p}{\to}zF(\xi_\alpha)$ and using additionally the dominated convergence theorem, we obtain $I_3 \overset{p}{\to}zF(\xi_\alpha)$. Hence, $\sqrt{n}\EE^*\big[j_{n,t}^{*(2)}(z)\big]\overset{p}{\to} z\xi_{\alpha}f(\xi_\alpha)$ for $z\geq 0$ and analogously one can show it to hold for $z<0$. Similarly, we find for $z\geq 0$
\begin{align*}
&\sqrt{n}\EE^*\Big[\big(j_{n,t}^{*(2)}(z)\big)^2\Big]=\sqrt{n} \int_{\big[\hat{\xi}_{n,\alpha},\hat{\xi}_{n,\alpha}+\frac{z}{\sqrt{n}}\big)}x^2\:d\hat{\mathbbm{F}}_n(x)\\
=& \Big(\hat{\xi}_{n,\alpha}+\frac{z}{\sqrt{n}}\Big)^2 \sqrt{n}\hat{\mathbbm{F}}_n\Big(\hat{\xi}_{n,\alpha}+\frac{z}{\sqrt{n}}-\!\!\Big)- \hat{\xi}_{n,\alpha}^2 \sqrt{n}\hat{\mathbbm{F}}_n(\hat{\xi}_{n,\alpha}-) -\sqrt{n}\int_{\big[\hat{\xi}_{n,\alpha},\hat{\xi}_{n,\alpha}+\frac{z}{\sqrt{n}}\big)}\!\!\hat{\mathbbm{F}}_n(x)\:d x^2\\ 
=& \bigg(\Big(\hat{\xi}_{n,\alpha}+\frac{z}{\sqrt{n}}\Big)^2-\hat{\xi}_{n,\alpha}^2\bigg)\sqrt{n}\hat{\mathbbm{F}}_n\Big(\hat{\xi}_{n,\alpha}+\frac{z}{\sqrt{n}}-\!\!\Big)+\hat{\xi}_{n,\alpha}^2 \sqrt{n}\bigg(\hat{\mathbbm{F}}_n\Big(\hat{\xi}_{n,\alpha}+\frac{z}{\sqrt{n}}-\!\!\Big)-\hat{\mathbbm{F}}_n(\hat{\xi}_{n,\alpha}-)\bigg)\\ 
&\qquad -2\int_{[0,z)}\Big(\hat{\xi}_{n,\alpha}+\frac{y}{\sqrt{n}}\Big)\hat{\mathbbm{F}}_n\Big(\hat{\xi}_{n,\alpha}+\frac{y}{\sqrt{n}}\Big)\:dy\\ 
=& \bigg(2z\hat{\xi}_{n,\alpha}+\frac{z^2}{\sqrt{n}}\bigg)\hat{\mathbbm{F}}_n\Big(\hat{\xi}_{n,\alpha}+\frac{z}{\sqrt{n}}-\!\Big)+\hat{\xi}_{n,\alpha}^2 \sqrt{n}\bigg(\hat{\mathbbm{F}}_n\Big(\hat{\xi}_{n,\alpha}+\frac{z}{\sqrt{n}}-\!\Big)-\hat{\mathbbm{F}}_n(\hat{\xi}_{n,\alpha}-)\bigg)\\ 
&\qquad -2\bigg(\hat{\xi}_{n,\alpha}\int_{[0,z)}\hat{\mathbbm{F}}_n\Big(\hat{\xi}_{n,\alpha}+\frac{y}{\sqrt{n}}\Big)\:dy+\int_{[0,z)}\frac{y}{\sqrt{n}}\hat{\mathbbm{F}}_n\Big(\hat{\xi}_{n,\alpha}+\frac{y}{\sqrt{n}}\Big)\:dy\bigg)\\
\overset{p}{\to}& 2z\xi_\alpha F(\xi_\alpha)+z \xi_\alpha^2 f(\xi_\alpha)-2z \xi_\alpha  F(\xi_\alpha) = z \xi_\alpha^2 f(\xi_\alpha)
\end{align*}
and analogously for $z<0$.
Combining results we establish $I\overset{p^*}{\to} \xi_{\alpha}zf(\xi_\alpha)\Omega'$ in probability. Consider the second term; since $\hat{\theta}_n\overset{a.s.}{\to}\theta_0$ (Theorem \ref{thm:4.1}) and $\hat{\theta}_n^*\overset{p^*}{\to}\theta_0$  almost surely (Lemma \ref{lem:4.5}), we have $\PP^*\big[\breve{\theta}_n \notin \mathscr{V}(\theta_0)\big]\overset{a.s.}{\to}0$. Thus, for every $\varepsilon>0$ we obtain 
\begin{align*}
&\PP^*\big[||II||\geq \varepsilon \big]\\
\leq & \PP^*\Bigg[\bigg|\bigg|\frac{1}{n}\sum_{t=1}^n\frac{\tilde{\sigma}_t(\hat{\theta}_n)}{\tilde{\sigma}_t(\breve{\theta}_n)}\Big(\tilde{H}_t(\breve{\theta}_n)-2\tilde{D}_t(\breve{\theta}_n)\tilde{D}_t'(\breve{\theta}_n)\Big) j_{n,t}^{*(2)}\bigg|\bigg|\geq \varepsilon \cap \breve{\theta}_n \in \mathscr{V}(\theta_0)\Bigg]+ \PP^*\Big[\breve{\theta}_n \notin \mathscr{V}(\theta_0)\Big]\\
\leq & \PP^*\Bigg[\frac{1}{n}\sum_{t=1}^n\sup_{\theta \in \mathscr{V}(\theta_0)}\frac{\tilde{\sigma}_t(\hat{\theta}_n)}{\tilde{\sigma}_t(\theta)}\bigg(\sup_{\theta \in \mathscr{V}(\theta_0)}\big|\big|\tilde{H}_t(\theta)\big|\big|+2\sup_{\theta \in \mathscr{V}(\theta_0)}\big|\big|\tilde{D}_t(\theta)\big|\big|^2\bigg) \big|j_{n,t}^{*(2)}\big|\geq \varepsilon\Bigg]+o(1)\\
\leq &\frac{1}{\varepsilon} \EE^*\Bigg[\frac{1}{n}\sum_{t=1}^n\sup_{\theta \in \mathscr{V}(\theta_0)}\frac{\tilde{\sigma}_t(\hat{\theta}_n)}{\tilde{\sigma}_t(\theta)}\bigg(\sup_{\theta \in \mathscr{V}(\theta_0)}\big|\big|\tilde{H}_t(\theta)\big|\big|+2\sup_{\theta \in \mathscr{V}(\theta_0)}\big|\big|\tilde{D}_t(\theta)\big|\big|^2\bigg) \big|j_{n,t}^{*(2)}\big|\Bigg]+o(1)\\
= &\frac{1}{\varepsilon} \EE^*\Big[\big|j_{n,t}^{*(2)}\big|\Big]\frac{1}{n}\sum_{t=1}^n\sup_{\theta \in \mathscr{V}(\theta_0)}\frac{\tilde{\sigma}_t(\hat{\theta}_n)}{\tilde{\sigma}_t(\theta)}\bigg(\sup_{\theta \in \mathscr{V}(\theta_0)}\big|\big|\tilde{H}_t(\theta)\big|\big|+2\sup_{\theta \in \mathscr{V}(\theta_0)}\big|\big|\tilde{D}_t(\theta)\big|\big|^2\bigg)+o(1)
\end{align*}
almost surely, where the third inequality follows from Markov's inequality. Because $\EE^*\Big[\big|j_{n,t}^{*(2)}\big|\Big]\leq \EE^*\Big[\big(j_{n,t}^{*(2)}\big)^2\Big]^{\frac{1}{2}}\overset{p}{\to}0$, it remains to show that 
\begin{align}
\label{eq:4.A.71}
\frac{1}{n}\sum_{t=1}^n\sup_{\theta \in \mathscr{V}(\theta_0)}\frac{\tilde{\sigma}_t(\hat{\theta}_n)}{\tilde{\sigma}_t(\theta)}\bigg(\sup_{\theta \in \mathscr{V}(\theta_0)}\big|\big|\tilde{H}_t(\theta)\big|\big|+2\sup_{\theta \in \mathscr{V}(\theta_0)}\big|\big|\tilde{D}_t(\theta)\big|\big|^2\bigg)
\end{align}
is stochastically bounded. Assumptions \ref{as:4.3} and \ref{as:4.4}\eqref{as:4.4.1} together with Theorem \ref{thm:4.1} imply
\begin{align*}
\sup_{\theta \in \mathscr{V}(\theta_0)}\frac{\tilde{\sigma}_t(\hat{\theta}_n)}{\tilde{\sigma}_t(\theta)}\leq
\sup_{\theta \in \mathscr{V}(\theta_0)}\Bigg(\frac{\sigma_t(\hat{\theta}_n)}{\sigma_t(\theta)}+\frac{C_1 \rho^t}{\underline{\omega}}\bigg(1+ \frac{\sigma_t(\hat{\theta}_n)}{\sigma_t(\theta)}\bigg)\Bigg)\overset{a.s.}{\leq} S_t T_t+\frac{C_1 \rho^t}{\underline{\omega}}\big(1+ S_t T_t\big).
\end{align*}
In addition, employing Assumptions \ref{as:4.3} and \ref{as:4.4}  
we obtain
\begin{align*}
\sup_{\theta \in \mathscr{V}(\theta_0)}\big|\big|\tilde{H}_t(\theta)\big|\big|\leq \sup_{\theta \in \mathscr{V}(\theta_0)}\bigg(\big|\big|H_t(\theta)\big|\big|+\frac{C_1 \rho^t}{\underline{\omega}}\Big(1+ \big|\big|H_t(\theta)\big|\big|\Big)\bigg)\leq V_t+\frac{C_1 \rho^t}{\underline{\omega}}\big(1+ V_t\big)
\end{align*}
and similarly we find
\begin{align*}
\sup_{\theta \in \mathscr{V}(\theta_0)}\big|\big|\tilde{D}_t(\theta)\big|\big|^2\leq& \sup_{\theta \in \mathscr{V}(\theta_0)}\bigg(\big|\big|D_t(\theta)\big|\big| +\frac{C_1 \rho^t}{\underline{\omega}}\Big(1+ \big|\big|D_t(\theta)\big|\big|\Big)\bigg)^2\\
\leq&  \sup_{\theta \in \mathscr{V}(\theta_0)}3\bigg(\big|\big|D_t(\theta)\big|\big|^2 +\frac{C_1^2 \rho^{2t}}{\underline{\omega}^2}\Big(1+ \big|\big|D_t(\theta)\big|\big|^2\Big)\bigg)\\
\leq& 3 U_t^2 +\frac{3 C_1^2 \rho^{2t}}{\underline{\omega}^2}\big(1+ U_t^2\big),
\end{align*}
where we also use the inequality $(x+y+z)^2\leq 3(x^2+y^2+z^2)$ for $x,y,z \in \R$. Hence, 
\begin{align*}
&\frac{1}{n}\sum_{t=1}^n\sup_{\theta \in \mathscr{V}(\theta_0)}\frac{\tilde{\sigma}_t(\hat{\theta}_n)}{\tilde{\sigma}_t(\theta)}\bigg(\sup_{\theta \in \mathscr{V}(\theta_0)}\big|\big|\tilde{H}_t(\theta)\big|\big|+2\sup_{\theta \in \mathscr{V}(\theta_0)}\big|\big|\tilde{D}_t(\theta)\big|\big|^2\bigg)\\
\overset{a.s.}{\leq} & \frac{1}{n}\sum_{t=1}^n\bigg(S_t T_t+\frac{C_1 \rho^t}{\underline{\omega}}\big(1+ S_t T_t\big)\bigg)\bigg(V_t+\frac{C_1 \rho^t}{\underline{\omega}}\big(1+ V_t\big)+6 U_t^2 +\frac{6C_1^2 \rho^{2t}}{\underline{\omega}^2}\big(1+ U_t^2\big)\bigg)\\
= & \underbrace{\frac{1}{n}\sum_{t=1}^n S_t T_t V_t}_{II_1}+ \underbrace{\frac{6}{n}\sum_{t=1}^n S_t T_t U_t^2}_{II_2}+ \underbrace{\frac{C_1}{\underline{\omega}}\frac{1}{n}\sum_{t=1}^n \rho^t S_t T_t}_{II_3}+ \underbrace{\frac{C_1}{\underline{\omega}}\frac{1}{n}\sum_{t=1}^n \rho^t S_t T_t V_t}_{II_4} \\
 &+ \underbrace{\frac{C_1}{\underline{\omega}}\frac{1}{n}\sum_{t=1}^n\rho^t V_t}_{II_5}+ \underbrace{\frac{C_1}{\underline{\omega}}\frac{6}{n}\sum_{t=1}^n\rho^t U_t^2}_{II_6}+ \underbrace{\frac{C_1}{\underline{\omega}}\frac{6}{n}\sum_{t=1}^n\rho^t S_t T_t U_t^2}_{II_7}+ \underbrace{\frac{C_1}{\underline{\omega}}\frac{1}{n}\sum_{t=1}^n\rho^t S_t T_t V_t}_{II_8}\\
& + \underbrace{\frac{C_1^2}{\underline{\omega}^2}\frac{1}{n}\sum_{t=1}^n\rho^{2t} V_t}_{II_9}+ \underbrace{\frac{C_1^2}{\underline{\omega}^2}\frac{1}{n}\sum_{t=1}^n\rho^{2t}S_t T_t}_{II_{10}}+ \underbrace{\frac{C_1^2}{\underline{\omega}^2}\frac{1}{n}\sum_{t=1}^n\rho^{2t} S_t T_t  V_t}_{II_{11}}+ \underbrace{\frac{6C_1^2 }{\underline{\omega}^2}\frac{1}{n}\sum_{t=1}^n \rho^{2t} S_t T_t}_{II_{12}}\\
&+ \underbrace{\frac{C_1^3}{\underline{\omega}^2}\frac{6}{n}\sum_{t=1}^n\rho^{3t} U_t^2}_{II_{13}}+ \underbrace{\frac{C_1^3}{\underline{\omega}^2}\frac{6}{n}\sum_{t=1}^n\rho^{3t}S_t T_t}_{II_{14}}+ \underbrace{\frac{6C_1^2 }{\underline{\omega}^2}\frac{1}{n}\sum_{t=1}^n \rho^{2t} S_t T_t U_t^2}_{II_{15}}+ \underbrace{\frac{C_1^3}{\underline{\omega}^2}\frac{6}{n}\sum_{t=1}^n\rho^{3t} S_t T_t U_t^2}_{II_{16}}\\
&+ \underbrace{\frac{C_1^2}{\underline{\omega}^2}\frac{1}{n}\sum_{t=1}^n\rho^{2t}}_{II_{17}}+ \underbrace{\frac{C_1^3}{\underline{\omega}^2}\frac{6}{n}\sum_{t=1}^n\rho^{3t}}_{II_{18}}
\end{align*}
From Assumption \ref{as:4.9}, the uniform ergodic theorem (in the sense of \citealp[page 181]{francq2011garch}) and H\"older's inequality, we obtain
\begin{align*}
II_1\leq \bigg(\frac{1}{n}\sum_{t=1}^n S_t^3 \bigg)^{\frac{1}{3}}\bigg(\frac{1}{n}\sum_{t=1}^n T_t^3 \bigg)^{\frac{1}{3}}\bigg(\frac{1}{n}\sum_{t=1}^n V_t^3\bigg)^{\frac{1}{3}}
\overset{a.s.}{\to} \Big(\EE\big[S_t^3\big] \Big)^{\frac{1}{3}}\Big(\EE\big[T_t^3\big] \Big)^{\frac{1}{3}}\Big(\EE\big[V_t^3\big] \Big)^{\frac{1}{3}}<\infty
\end{align*}
and similarly we can show that $\lim_{n \to \infty} II_2 <\infty$ almost surely. Consider $II_3$; for each $\varepsilon>0$, Markov's inequality and the Cauchy-Schwarz inequality yield
\begin{align*}
&\sum_{t=1}^\infty \PP\Big[\rho^{t}S_t T_t>\varepsilon\Big]  \leq \sum_{t=1}^\infty \rho^{t} \frac{1+\EE[S_t T_t]}{\varepsilon} 
= \frac{1+(\EE[S_t^2])^{\frac{1}{2}}(\EE[T_t^2])^{\frac{1}{2}}}{\varepsilon(1-\rho)}<\infty
\end{align*}
and $\frac{1}{n}\sum_{t=1}^n \rho^{t}S_t T_t\overset{a.s.}{\to}0$ follows from combining the Borel-Cantelli lemma  with Ces\'aro's lemma. Hence, $II_3 \overset{a.s.}{\to}0$. Similarly we can show that the terms $II_4,\dots,II_{16}$ vanish almost surely. Further, $II_{17}\leq \frac{1}{n}\frac{C_1^2}{\underline{\omega}^2(1-\rho^2)}\overset{a.s.}{\to} 0$ and similarly, we can prove that $II_{18}$ vanishes almost surely, which completes the proof.
\end{proof}

\section{Recursive-design Residual Bootstrap}
\label{app:4.B}

This appendix devotes attention to the recursive-design residual bootstrap. 
The bootstrap scheme described in Algorithm \ref{alg:4.3} is the recursive-design counterpart of Algorithm \ref{alg:4.1}. 
Note that the bootstrap observation $\epsilon_t^\star$ is generated recursively on the basis of its past realizations $\epsilon_{t-1}^\star,\dots, \epsilon_1^\star$.

\begin{algorithm}\textit{(Recursive-design residual bootstrap)}
\label{alg:4.3}
\begin{enumerate}

\item For $t=1,\dots, n$ generate  $\eta_t^\star \overset{iid}{\sim} \hat{\mathbbm{F}}_n$ and the bootstrap observation $\epsilon_t^\star = \sigma_t^\star \eta_t^\star$ with $\sigma_t^\star=\sigma_t^\star(\hat{\theta}_n)$ and $\sigma_t^\star(\theta) = \sigma(\epsilon_{t-1}^\star,\dots,\epsilon_1^\star,\tilde{\epsilon}_0,\tilde{\epsilon}_{-1},\dots;\theta)$
\item Calculate the bootstrap estimator 
\begin{align*}
\hat{\theta}_n^\star = \arg \max_{\theta \in \Theta}\frac{1}{n}\sum_{t=1}^n \ell_t^\star(\theta) \qquad \text{with} \qquad \ell_t^\star(\theta)=-\frac{1}{2}\bigg(\frac{\epsilon_t^{\star}}{\sigma_t^\star(\theta)}\bigg)^2-\log \tilde{\sigma}_t(\theta).
\end{align*}

\item For $t=1,\dots,n$ compute the bootstrap residual $\hat{\eta}_t^\star = \epsilon_t^\star/\sigma_t^\star(\hat{\theta}_n^\star)$
and obtain 
\begin{align*}
\hat{\xi}_{n,\alpha}^\star = \arg\min_{z \in \R} \frac{1}{n}\sum_{t=1}^n\rho_\alpha(\hat{\eta}_t^\star-z).
\end{align*}

\item Obtain the bootstrap estimator of the conditional VaR
\begin{align*}
\reallywidehat{VaR}_{n,\alpha}^{\star}=-\hat{\xi}_{n,\alpha}^{\star}\: \tilde{\sigma}_{n+1}\big(\hat{\theta}_n^{\star}\big).
\end{align*}

\end{enumerate}

\end{algorithm}

In contrast to the fixed-design bootstrap, the bootstrap sample $\epsilon_1^\star,\dots,\epsilon_n^\star$, conditional on the original sample, is a dependent sequence. Therefore one likely needs a stronger set of conditions to show the validity of the recursive-design bootstrap. Moreover, whether the recursive bootstrap scheme is valid is contingent on the specific conditional volatility model, e.g.\ GARCH($1,1$), and as such needs to be investigated on a case-by-case basis. This is therefore outside the scope of the current paper. 

\newpage 

\section{Additional tables with simulation results}\label{sec add simulation}

\bgroup
\def\arraystretch{0.95}
\begin{table}[h]
\caption{\textbf{Fixed-design} bootstrap confidence intervals and asymptotic confidence interval for \textbf{GARCH($1$,$1$)} with Gaussian innovations}\label{tab:4.2}
\centering
\resizebox{\textwidth}{!}{\begin{tabular}{rc:ccc:ccc}
\hline \hline
\multicolumn{1}{c}{\begin{tabular}[c]{@{}c@{}}Sample\\ Size\end{tabular}} &                         & \begin{tabular}[c]{@{}c@{}}Average \\ coverage\end{tabular} & \begin{tabular}[c]{@{}c@{}}Av. coverage\\ below/above\end{tabular} & \begin{tabular}[c]{@{}c@{}}Average\\ length\end{tabular} & \begin{tabular}[c]{@{}c@{}}Average\\ coverage\end{tabular} & \begin{tabular}[c]{@{}c@{}}Av. coverage\\ below/above\end{tabular} & \begin{tabular}[c]{@{}c@{}}Average\\ length\end{tabular} \\ \hline
\multicolumn{1}{c}{} & & & & & & & \\[-9pt]
\multicolumn{1}{c}{} & & \multicolumn{3}{c}{low persistence} & \multicolumn{3}{c}{high persistence} \\
250 & EP & 
81.17 & 7.80/11.03 & 0.509 & 80.40 & 7.89/11.71 & 0.654 \\
& RT & 
90.24 & 2.23/7.53 & 0.509 & 90.03 & 2.48/7.49 & 0.654 \\
& SY & 
89.40 & 2.94/7.66 & 0.531 & 88.53 & 3.29/8.18 & 0.681 \\
& AS & 
86.81 & 3.30/9.89 & 0.514 & 85.94 & 3.90/10.16 & 0.644 \\
\hdashline
500 & EP & 
84.30 & 6.79/8.91 & 0.367 & 83.84 & 7.14/9.02 & 0.468 \\
& RT & 
91.16 & 2.90/5.94 & 0.367 & 90.72 & 3.12/6.16 & 0.468 \\
& SY & 
90.03 & 3.70/6.27 & 0.378 & 89.68 & 3.76/6.56 & 0.479 \\
& AS & 
88.69 & 3.60/7.71 & 0.372 & 87.75 & 3.98/8.27 & 0.455 \\
\hdashline
1,000 & EP &
85.64 & 6.37/7.99 & 0.259 & 86.09 & 5.79/8.12 & 0.333 \\
& RT & 
90.66 & 3.54/5.80 & 0.259 & 90.94 & 3.34/5.72 & 0.333 \\
& SY &
89.66 & 4.18/6.16 & 0.264 & 89.73 & 3.96/6.31 & 0.338 \\
& AS & 
89.24 & 3.78/6.98 & 0.262 & 88.55 & 3.99/7.46 & 0.324 \\
\hdashline
5,000 & EP & 
87.75 & 5.39/6.86 & 0.120 & 88.35 & 5.23/6.42 & 0.149 \\
& RT & 
90.42 & 4.18/5.40 & 0.120 & 89.83 & 4.41/5.76 & 0.149 \\
& SY & 
89.77 & 4.45/5.78 & 0.121 & 89.84 & 4.46 5.70 & 0.150 \\
& AS & 
89.45 & 4.42/6.13 & 0.120 & 89.22 & 4.61/6.17 & 0.146 \\
\hline \hline
\end{tabular}
}
\vspace{0.15cm}
\caption*{Table \ref{tab:4.2} is exactly as Table \ref{tab:4.1} but with \textbf{Gaussian innovations} instead of Student-t innovations.
}
\end{table}
 \egroup

\bgroup
\def\arraystretch{0.95}
\begin{table}[!htbp]
\caption{\textbf{Fixed-design} bootstrap confidence intervals and asymptotic confidence interval for \textbf{T-GARCH($1$,$1$)} with Gaussian innovations}\label{tab:4.2a}
\centering
\resizebox{\textwidth}{!}{\begin{tabular}{rc:ccc:ccc}
\hline \hline
\multicolumn{1}{c}{\begin{tabular}[c]{@{}c@{}}Sample\\ Size\end{tabular}} &                         & \begin{tabular}[c]{@{}c@{}}Average \\ coverage\end{tabular} & \begin{tabular}[c]{@{}c@{}}Av. coverage\\ below/above\end{tabular} & \begin{tabular}[c]{@{}c@{}}Average\\ length\end{tabular} & \begin{tabular}[c]{@{}c@{}}Average\\ coverage\end{tabular} & \begin{tabular}[c]{@{}c@{}}Av. coverage\\ below/above\end{tabular} & \begin{tabular}[c]{@{}c@{}}Average\\ length\end{tabular} \\ \hline
\multicolumn{1}{c}{} & & & & & & & \\[-9pt]
& \multicolumn{1}{l}{} & \multicolumn{3}{c}{low persistence} & \multicolumn{3}{c}{high persistence} \\
250 & EP & 
80.91 & 7.07/12.02 & 0.116 & 80.02 & 7.53/12.45 & 0.238 \\
& RT & 
90.30 & 2.22/7.48 & 0.116 & 90.43 & 1.86/7.71 & 0.238 \\
& SY & 
88.85 & 2.84/8.31 & 0.121 & 88.97 & 2.54/8.49 & 0.249 \\
& AS & 
88.32 & 2.62/9.06 & 0.119 & 88.40 & 2.62/8.98 & 0.247 \\
\hdashline
500 & EP & 
84.34 & 6.17/9.49 & 0.085 & 82.90 & 6.93/10.17 & 0.173 \\
& RT & 
90.45 & 3.03/6.52 & 0.085 & 90.84 & 2.55/6.61 & 0.173 \\
& SY & 
89.21 & 3.68/7.11 & 0.087 & 89.13 & 3.41/7.46 & 0.178 \\
& AS & 
89.06 & 3.48/7.46 & 0.087 & 89.00 & 3.21/7.79 & 0.177 \\
\hdashline
1,000 & EP & 
85.79 & 5.93/8.28 & 0.061 & 84.88 & 6.53/8.59 & 0.124 \\
& RT & 
90.02 & 3.90/6.08 & 0.061 & 90.42 & 3.53/6.05 & 0.124 \\
& SY & 
89.74 & 4.03/6.23 & 0.062 & 89.55 & 4.02/6.43 & 0.127 \\
& AS & 
89.72 & 3.82/6.46 & 0.062 & 89.64 & 3.86/6.50 & 0.127 \\
\hdashline
5,000 & EP & 
88.60 & 4.80/6.60 & 0.028 & 88.29 & 5.12/6.59 & 0.058 \\
& RT &
90.23 & 4.13/5.64 & 0.028 & 90.31 & 4.18/5.51 & 0.058 \\
& SY &
90.07 & 4.04/5.89 & 0.028 & 90.27 & 4.05/5.68 & 0.058 \\
& AS &
90.67 & 3.76/5.57 & 0.029 & 90.65 & 3.88/5.47 & 0.059 \\
\hline \hline
\end{tabular}
}
\vspace{0.15cm}
\caption*{Table \ref{tab:4.2a} is exactly as Table \ref{alg:4.2} but with with \textbf{Gaussian innovations} instead of Student-t innovations.
}

\end{table}
 \egroup


\begin{thebibliography}{}

\bibitem[\protect\citeauthoryear{Bahadur}{Bahadur}{1966}]{bahadur1966note}
Bahadur, R.R. (1966).
\newblock A note on quantiles in large samples.
\newblock {\em The Annals of Mathematical Statistics\/}~{\em 37\/}(3),
  577--580.

\bibitem[\protect\citeauthoryear{Bardet, Kamila, and Kengne}{Bardet
  et~al.}{2020}]{Bardet2020}
Bardet, J.M., K.~Kamila, and W.~Kengne (2020).
\newblock Consistent model selection criteria and goodness-of-fit test for
  common time series models.
\newblock {\em Electronic Journal of Statistics\/}~{\em 14\/}(1), 2009--2052.

\bibitem[\protect\citeauthoryear{Berkes and Horv{\'a}th}{Berkes and
  Horv{\'a}th}{2003}]{berkes2003limit}
Berkes, I. and L.~Horv{\'a}th (2003).
\newblock Limit results for the empirical process of squared residuals in
  \uppercase{GARCH} models.
\newblock {\em Stochastic Processes and their Applications\/}~{\em 105\/}(2),
  271--298.

\bibitem[\protect\citeauthoryear{Beutner, Heinemann, and Smeekes}{Beutner
  et~al.}{2019}]{beutner2019technical}
Beutner, E., A.~Heinemann, and S.~Smeekes (2019).
\newblock A general framework for prediction in time series models.
\newblock Working paper, Maastricht University,
  \url{https://arxiv.org/pdf/1902.01622.pdf}.

\bibitem[\protect\citeauthoryear{Beutner, Heinemann, and Smeekes}{Beutner
  et~al.}{2021}]{beutner2017justification}
Beutner, E., A.~Heinemann, and S.~Smeekes (2021).
\newblock A justification of conditional confidence intervals.
\newblock {\em Electronic Journal of Statistics\/}~{\em 15\/}(1), 2517--2565.

\bibitem[\protect\citeauthoryear{Billingsley}{Billingsley}{1986}]{billingsley1986probability}
Billingsley, P. (1986).
\newblock {\em Probability and Measure\/} (2nd ed.).
\newblock New York: John Wiley \& Sons.

\bibitem[\protect\citeauthoryear{Bollerslev}{Bollerslev}{1986}]{bollerslev1986generalized}
Bollerslev, T. (1986).
\newblock Generalized autoregressive conditional heteroskedasticity.
\newblock {\em Journal of Econometrics\/}~{\em 31\/}(3), 307--327.

\bibitem[\protect\citeauthoryear{Cavaliere, Nielsen, Pedersen, and
  Rahbek}{Cavaliere et~al.}{2022}]{cavaliere2020bootstrap}
Cavaliere, G., H.B. Nielsen, R.S. Pedersen, and A.~Rahbek (2022).
\newblock Bootstrap inference on the boundary of the parameter space, with
  application to conditional volatility models.
\newblock {\em Journal of Econometrics\/}~{\em 227\/}(1), 241--263.

\bibitem[\protect\citeauthoryear{Cavaliere, Pedersen, and Rahbek}{Cavaliere
  et~al.}{2018}]{cavaliere2018fixed}
Cavaliere, G., R.S. Pedersen, and A.~Rahbek (2018).
\newblock The fixed volatility bootstrap for a class of \uppercase{ARCH}($q$)
  models.
\newblock {\em Journal of Time Series Analysis\/}~{\em 39}, 920--941.

\bibitem[\protect\citeauthoryear{Cho and White}{Cho and
  White}{2011}]{cho2011generalized}
Cho, J.S. and H.~White (2011).
\newblock Generalized runs tests for the iid hypothesis.
\newblock {\em Journal of Econometrics\/}~{\em 162\/}(2), 326--344.

\bibitem[\protect\citeauthoryear{Christoffersen and
  Gon{\c{c}}alves}{Christoffersen and
  Gon{\c{c}}alves}{2005}]{christoffersen2005estimation}
Christoffersen, P. and S.~Gon{\c{c}}alves (2005).
\newblock Estimation risk in financial risk management.
\newblock {\em The Journal of Risk\/}~{\em 7\/}(3), 1--28.

\bibitem[\protect\citeauthoryear{Corradi and Iglesias}{Corradi and
  Iglesias}{2008}]{corradi2008bootstrap}
Corradi, V. and E.M. Iglesias (2008).
\newblock Bootstrap refinements for \uppercase{QML} estimators of the
  \uppercase{GARCH}(1,1) parameters.
\newblock {\em Journal of Econometrics\/}~{\em 144\/}(2), 500--510.

\bibitem[\protect\citeauthoryear{Cs{\"o}rg{\H o} and
  R{\'e}v{\'e}sz}{Cs{\"o}rg{\H o} and R{\'e}v{\'e}sz}{1981}]{csorgo1981strong}
Cs{\"o}rg{\H o}, M. and P.~R{\'e}v{\'e}sz (1981).
\newblock {\em Strong Approximations in Probability and Statistics}.
\newblock Budapest: Akad{\'e}miai Kiad{\'o}.

\bibitem[\protect\citeauthoryear{Davidson and Flachaire}{Davidson and
  Flachaire}{2008}]{DavidsonFlachaire08}
Davidson, R. and E.~Flachaire (2008).
\newblock The wild bootstrap, tamed at last.
\newblock {\em Journal of Econometrics\/}~{\em 146}, 162--169.

\bibitem[\protect\citeauthoryear{Ding, Granger, and Engle}{Ding
  et~al.}{1993}]{ding1993long}
Ding, Z., C.W. Granger, and R.F. Engle (1993).
\newblock A long memory property of stock market returns and a new model.
\newblock {\em Journal of Empirical Finance\/}~{\em 1\/}(1), 83--106.

\bibitem[\protect\citeauthoryear{Engle}{Engle}{1982}]{engle1982autoregressive}
Engle, R.F. (1982).
\newblock Autoregressive conditional heteroscedasticity with estimates of the
  variance of {U}nited {K}ingdom inflation.
\newblock {\em Econometrica\/}~{\em 50\/}(4), 987--1007.

\bibitem[\protect\citeauthoryear{Escanciano}{Escanciano}{2009}]{Escanciano2009}
Escanciano, J.C. (2009).
\newblock Quasi-maximum likelihood estimation of semi-strong \uppercase{GARCH}
  model.
\newblock {\em Econometric Theory\/}~{\em 25\/}(2), 561--570.

\bibitem[\protect\citeauthoryear{Falk and Kaufmann}{Falk and
  Kaufmann}{1991}]{falk1991coverage}
Falk, M. and E.~Kaufmann (1991).
\newblock Coverage probabilities of bootstrap-confidence intervals for
  quantiles.
\newblock {\em The Annals of Statistics\/}~{\em 19\/}(1), 485--495.

\bibitem[\protect\citeauthoryear{Francq, Horv{\'a}th, and Zako{\"\i}an}{Francq
  et~al.}{2016}]{francq2016variance}
Francq, C., L.~Horv{\'a}th, and J.M. Zako{\"\i}an (2016).
\newblock Variance targeting estimation of multivariate \uppercase{GARCH}
  models.
\newblock {\em Journal of Financial Econometrics\/}~{\em 14\/}(2), 353--382.

\bibitem[\protect\citeauthoryear{Francq and Zako{\"\i}an}{Francq and
  Zako{\"\i}an}{2004}]{francq2004maximum}
Francq, C. and J.M. Zako{\"\i}an (2004).
\newblock Maximum likelihood estimation of pure \uppercase{Garch} and
  \uppercase{Arma}-\uppercase{Garch} processes.
\newblock {\em Bernoulli\/}~{\em 10\/}(4), 605--637.

\bibitem[\protect\citeauthoryear{Francq and Zako\"ian}{Francq and
  Zako\"ian}{2011}]{francq2011garch}
Francq, C. and J.M. Zako\"ian (2011).
\newblock {\em \uppercase{GARCH} Models: Structure, Statistical Inference and
  Financial Applications}.
\newblock Chichester: John Wiley \& Sons.

\bibitem[\protect\citeauthoryear{Francq and Zako{\"\i}an}{Francq and
  Zako{\"\i}an}{2015}]{francq2015risk}
Francq, C. and J.M. Zako{\"\i}an (2015).
\newblock Risk-parameter estimation in volatility models.
\newblock {\em Journal of Econometrics\/}~{\em 184\/}(1), 158--173.

\bibitem[\protect\citeauthoryear{Francq and Zako{\"\i}an}{Francq and
  Zako{\"\i}an}{2016}]{FrancqZakoian2016}
Francq, C. and J.M. Zako{\"\i}an (2016).
\newblock Estimating multivariate volatility models equation by equation.
\newblock {\em Journal of the Royal Statistical Society, Series B\/}~{\em
  76\/}(3), 613--635.

\bibitem[\protect\citeauthoryear{Francq and Zakoïan}{Francq and
  Zakoïan}{2022}]{FRANCQ202247}
Francq, C. and J.M. Zakoïan (2022).
\newblock Testing the existence of moments for garch processes.
\newblock {\em Journal of Econometrics\/}~{\em 227\/}(1), 47--64.

\bibitem[\protect\citeauthoryear{Friedrich, Smeekes, and Urbain}{Friedrich
  et~al.}{2020}]{FSU20}
Friedrich, M., S.~Smeekes, and J.P. Urbain (2020).
\newblock Autoregressive wild bootstrap inference for nonparametric trends.
\newblock {\em Journal of Econometrics\/}~{\em 214}, 81--109.

\bibitem[\protect\citeauthoryear{Gao and Song}{Gao and
  Song}{2008}]{gao2008estimation}
Gao, F. and F.~Song (2008).
\newblock Estimation risk in \uppercase{GARCH} \uppercase{V}a\uppercase{R} and
  \uppercase{ES} estimates.
\newblock {\em Econometric Theory\/}~{\em 24\/}(5), 1404--1424.

\bibitem[\protect\citeauthoryear{Geweke}{Geweke}{1986}]{geweke1986comment}
Geweke, J. (1986).
\newblock Comment on: modelling the persistence of conditional variances.
\newblock {\em Econometric Reviews\/}~{\em 5}, 57--61.

\bibitem[\protect\citeauthoryear{Glosten, Jagannathan, and Runkle}{Glosten
  et~al.}{1993}]{glosten1993relation}
Glosten, L.R., R.~Jagannathan, and D.E. Runkle (1993).
\newblock On the relation between the expected value and the volatility of the
  nominal excess return on stocks.
\newblock {\em The Journal of Finance\/}~{\em 48\/}(5), 1779--1801.

\bibitem[\protect\citeauthoryear{Gonçalves and Kilian}{Gonçalves and
  Kilian}{2004}]{GONCALVES2004}
Gonçalves, S. and L.~Kilian (2004).
\newblock Bootstrapping autoregressions with conditional heteroskedasticity of
  unknown form.
\newblock {\em Journal of Econometrics\/}~{\em 123\/}(1), 89--120.

\bibitem[\protect\citeauthoryear{Hall, DiCiccio, and Romano}{Hall
  et~al.}{1989}]{hall1989smoothing}
Hall, P., T.J. DiCiccio, and J.P. Romano (1989).
\newblock On smoothing and the bootstrap.
\newblock {\em The Annals of Statistics\/}~{\em 17\/}(2), 692--704.

\bibitem[\protect\citeauthoryear{Hall and Heyde}{Hall and
  Heyde}{1980}]{hall1980martingale}
Hall, P. and C.C. Heyde (1980).
\newblock {\em Martingale Limit Theory and its Application}.
\newblock New York: Academic Press.

\bibitem[\protect\citeauthoryear{Hall and Martin}{Hall and
  Martin}{1988}]{hall1988bootstrap}
Hall, P. and M.A. Martin (1988).
\newblock On bootstrap resampling and iteration.
\newblock {\em Biometrika\/}~{\em 75\/}(4), 661--671.

\bibitem[\protect\citeauthoryear{Hall and Yao}{Hall and
  Yao}{2003}]{hall2003inference}
Hall, P. and Q.~Yao (2003).
\newblock Inference in \uppercase{ARCH} and \uppercase{GARCH} models with
  heavy--tailed errors.
\newblock {\em Econometrica\/}~{\em 71\/}(1), 285--317.

\bibitem[\protect\citeauthoryear{Hamadeh and Zako{\"\i}an}{Hamadeh and
  Zako{\"\i}an}{2011}]{hamadeh2011asymptotic}
Hamadeh, T. and J.M. Zako{\"\i}an (2011).
\newblock Asymptotic properties of \uppercase{LS} and \uppercase{QML}
  estimators for a class of nonlinear \uppercase{GARCH} processes.
\newblock {\em Journal of Statistical Planning and Inference\/}~{\em 141\/}(1),
  488--507.

\bibitem[\protect\citeauthoryear{Hartz, Mittnik, and Paolella}{Hartz
  et~al.}{2006}]{hartz2006accurate}
Hartz, C., S.~Mittnik, and M.~Paolella (2006).
\newblock Accurate value-at-risk forecasting based on the
  normal-\uppercase{GARCH} model.
\newblock {\em Computational Statistics \& Data Analysis\/}~{\em 51\/}(4),
  2295--2312.

\bibitem[\protect\citeauthoryear{Heinemann and Telg}{Heinemann and
  Telg}{2018}]{heinemann2018residual}
Heinemann, A. and S.~Telg (2018).
\newblock A residual bootstrap for conditional expected shortfall.
\newblock arXiv Preprint 1811.11557.

\bibitem[\protect\citeauthoryear{Hetland, Pedersen, and Rahbek}{Hetland
  et~al.}{ress}]{HETLAND2021}
Hetland, S., R.S. Pedersen, and A.~Rahbek (in press).
\newblock Dynamic conditional eigenvalue garch.
\newblock {\em Journal of Econometrics\/}.

\bibitem[\protect\citeauthoryear{Hidalgo and Zaffaroni}{Hidalgo and
  Zaffaroni}{2007}]{hidalgo2007goodness}
Hidalgo, J. and P.~Zaffaroni (2007).
\newblock A goodness-of-fit test for \uppercase{ARCH}($\infty$) models.
\newblock {\em Journal of Econometrics\/}~{\em 141\/}(2), 835--875.

\bibitem[\protect\citeauthoryear{Hjort and Pollard}{Hjort and
  Pollard}{2011}]{hjort2011asymptotics}
Hjort, N.L. and D.~Pollard (2011).
\newblock Asymptotics for minimisers of convex processes.
\newblock {\em Preprint arXiv:1107.3806v1\/}.

\bibitem[\protect\citeauthoryear{Hoga and Demetrescu}{Hoga and
  Demetrescu}{2023}]{HogaDemetrescu}
Hoga, Y. and M.~Demetrescu (2023).
\newblock Monitoring value-at-risk and expected shortfall forecasts.
\newblock {\em Management Science\/}~{\em 69\/}(5), 2954--2971.

\bibitem[\protect\citeauthoryear{Jeong}{Jeong}{2017}]{jeong2017residual}
Jeong, M. (2017).
\newblock Residual-based \uppercase{GARCH} bootstrap and second order
  asymptotic refinement.
\newblock {\em Econometric Theory\/}~{\em 33\/}(3), 779--790.

\bibitem[\protect\citeauthoryear{Jim\'{e}nez-Gamero, Lee, and
  Meintanis}{Jim\'{e}nez-Gamero et~al.}{2020}]{Maria2019}
Jim\'{e}nez-Gamero, M.D., S.~Lee, and S.G. Meintanis (2020).
\newblock Goodness-of-fit tests for parametric specifications of conditionally
  heteroscedastic models.
\newblock {\em TEST\/}~{\em 29}, 682--703.

\bibitem[\protect\citeauthoryear{Koenker and Xiao}{Koenker and
  Xiao}{2006}]{koenker2006quantile}
Koenker, R. and Z.~Xiao (2006).
\newblock Quantile autoregression.
\newblock {\em Journal of the American Statistical Association\/}~{\em
  101\/}(475), 980--990.

\bibitem[\protect\citeauthoryear{Kreiss}{Kreiss}{2016}]{kreiss2015discussion}
Kreiss, J.P. (2016).
\newblock Discussion: bootstrap prediction intervals for linear, nonlinear and
  nonparametric autoregressions.
\newblock {\em Journal of Statistical Planning and Inference\/}~{\em 177},
  28--30.

\bibitem[\protect\citeauthoryear{Lahiri}{Lahiri}{2003}]{Lahiri03}
Lahiri, S.N. (2003).
\newblock {\em Resampling Methods for Dependent Data}.
\newblock New York: Springer-Verlag.

\bibitem[\protect\citeauthoryear{Li, Peng, and Song}{Li
  et~al.}{ress}]{li_peng_song_2022}
Li, S., L.~Peng, and X.~Song (in press).
\newblock Simultaneous confidence bands for conditional value-at-risk and
  expected shortfall.
\newblock {\em Econometric Theory\/}.

\bibitem[\protect\citeauthoryear{Linton, Pan, and Wang}{Linton
  et~al.}{2010}]{Linton2010}
Linton, O., J.~Pan, and H.~Wang (2010).
\newblock Estimation for a nonstationary semi-strong {GARCH}(1,1) model with
  heavy-tailed errors.
\newblock {\em Econometric Theory\/}~{\em 26\/}(1), 1--28.

\bibitem[\protect\citeauthoryear{Mammen}{Mammen}{1993}]{Mammen93}
Mammen, E. (1993).
\newblock Bootstrap and wild bootstrap for high dimensional linear models.
\newblock {\em Annals of Statistics\/}~{\em 21}, 255--285.

\bibitem[\protect\citeauthoryear{McNeil and Frey}{McNeil and
  Frey}{2000}]{McNeil2002}
McNeil, A.J. and R.~Frey (2000).
\newblock Estimation of tail-related risk measures for heteroscedastic
  financial time series: an extreme value approach.
\newblock {\em Journal of Empirical Finance\/}~{\em 7\/}(3), 271 -- 300.

\bibitem[\protect\citeauthoryear{Nelson}{Nelson}{1991}]{nelson1991conditional}
Nelson, D.B. (1991).
\newblock Conditional heteroskedasticity in asset returns: A new approach.
\newblock {\em Econometrica: Journal of the Econometric Society\/}~{\em
  59\/}(2), 347--370.

\bibitem[\protect\citeauthoryear{Pantula}{Pantula}{1986}]{pantula1986modeling}
Pantula, S.G. (1986).
\newblock Modeling the persistence of conditional variances: a comment.
\newblock {\em Econometric Reviews\/}~{\em 5}, 79--97.

\bibitem[\protect\citeauthoryear{Pascual, Romo, and Ruiz}{Pascual
  et~al.}{2006}]{pascual2006bootstrap}
Pascual, L., J.~Romo, and E.~Ruiz (2006).
\newblock Bootstrap prediction for returns and volatilities in
  \uppercase{garch} models.
\newblock {\em Computational Statistics \& Data Analysis\/}~{\em 50\/}(9),
  2293--2312.

\bibitem[\protect\citeauthoryear{Pesaran}{Pesaran}{2015}]{pesaran2015time}
Pesaran, M.H. (2015).
\newblock {\em Time Series and Panel Data Econometrics}.
\newblock Oxford: Oxford University Press.

\bibitem[\protect\citeauthoryear{Phillips}{Phillips}{1979}]{phillips1979sampling}
Phillips, P.C.B. (1979).
\newblock The sampling distribution of forecasts from a first-order
  autoregression.
\newblock {\em Journal of Econometrics\/}~{\em 9\/}(3), 241--261.

\bibitem[\protect\citeauthoryear{Roussas}{Roussas}{1997}]{roussas1997course}
Roussas, G.G. (1997).
\newblock {\em A Course in Mathematical Statistics\/} (2nd ed.).
\newblock San Diego: Academic Press.

\bibitem[\protect\citeauthoryear{Shao}{Shao}{2010}]{Shao10}
Shao, X. (2010).
\newblock The dependent wild bootstrap.
\newblock {\em Journal of the American Statistical Association\/}~{\em 105},
  218--235.

\bibitem[\protect\citeauthoryear{Shimizu}{Shimizu}{2009}]{shimizu2009bootstrapping}
Shimizu, K. (2009).
\newblock {\em Bootstrapping Stationary \uppercase{ARMA}--\uppercase{GARCH}
  Models}.
\newblock Springer.

\bibitem[\protect\citeauthoryear{Silverman}{Silverman}{1986}]{silverman1986density}
Silverman, B. (1986).
\newblock {\em Density Estimation for Statistics and Data Analysis Estimation
  Density}.
\newblock London: Chapman and Hall.

\bibitem[\protect\citeauthoryear{Spierdijk}{Spierdijk}{2016}]{spierdijk2016confidence}
Spierdijk, L. (2016).
\newblock Confidence intervals for \uppercase{ARMA}--\uppercase{GARCH}
  value-at-risk: the case of heavy tails and skewness.
\newblock {\em Computational Statistics \& Data Analysis\/}~{\em 100},
  545--559.

\bibitem[\protect\citeauthoryear{Xiong and Li}{Xiong and
  Li}{2008}]{xiong2008some}
Xiong, S. and G.~Li (2008).
\newblock Some results on the convergence of conditional distributions.
\newblock {\em Statistics \& Probability Letters\/}~{\em 78\/}(18), 3249--3253.

\bibitem[\protect\citeauthoryear{Zako{\"\i}an}{Zako{\"\i}an}{1994}]{zakoian1994threshold}
Zako{\"\i}an, J.M. (1994).
\newblock Threshold heteroskedastic models.
\newblock {\em Journal of Economic Dynamics and Control\/}~{\em 18\/}(5),
  931--955.

\end{thebibliography}
\end{document}